\documentclass[11pt]{amsart}
\usepackage[margin=1.2in]{geometry}
\usepackage{amscd,amssymb,bm}
\usepackage{hyperref}
\usepackage{verbatim}

\usepackage{enumerate}

\usepackage{url}

\usepackage[T2A,T1]{fontenc}
\usepackage[utf8]{inputenc}
\usepackage[russian,english]{babel}
\usepackage{graphicx}

\usepackage{epstopdf}
\numberwithin{equation}{section}

\newcommand{\ndash}{\nobreakdash-\hspace{0pt}}
\newcommand{\Ndash}{\nobreakdash--}

\newcommand{\ii}{{\mathrm{i}}}
\newcommand{\dd}{{\mathrm{d}}}
\newcommand{\DD}{{\mathrm{D}}}
\newcommand{\NN}{{\mathrm{N}}}
\newcommand{\hDD}{{\Hat\DD}}
\newcommand{\hNN}{{\Hat\NN}}

\newcommand{\mr}{\mathrm}

\DeclareMathOperator{\Fun}{Fun}
\DeclareMathOperator{\Harm}{Harm}

\newcommand{\bt}{\bullet}
\newcommand{\ra}{\rightarrow}
\newcommand{\be}{\begin{equation}}
\newcommand{\ee}{\end{equation}}

\newcommand{\backgrounds}{residual fields}
\newcommand{\pseudovacua}{residual fields}
\newcommand{\DDD}{\mathfrak{D}}

\newcommand{\bA}{{\mathbb{A}}}
\newcommand{\bB}{{\mathbb{B}}}
\newcommand{\bD}{{\mathbb{D}}}
\newcommand{\bE}{{\mathbb{E}}}

\newcommand{\tbA}{\Tilde\bA}
\newcommand{\tbB}{\Tilde\bB}

\newcommand{\calL}{{\mathcal{L}}}

\newcommand{\calZ}{{\mathcal{Z}}}
\newcommand{\calF}{{\mathcal{F}}}
\newcommand{\calT}{{\mathcal{T}}}
\newcommand{\calEL}{\mathcal{EL}}
\newcommand{\calN}{{\mathcal{N}}}

\newcommand{\EE}{\mathrm{e}}

\DeclareMathOperator{\Map}{Map}

\DeclareMathOperator{\Ber}{Ber}

\newcommand{\hHarm}{{\Hat{\mathrm{Harm}}}}

\newcommand{\id}{\mathrm{id}}

\DeclareMathOperator{\tr}{Tr}

\DeclareMathOperator{\ad}{ad}

\DeclareMathOperator{\Ad}{Ad}

\DeclareMathOperator{\Hom}{Hom}

\newtheorem{Thm}{Theorem}[section]

\newtheorem{Prop}[Thm]{Proposition}
\newtheorem{Lem}[Thm]{Lemma}
\newtheorem{Cor}[Thm]{Corollary}

\newtheorem*{Thm*}{Theorem}
\newtheorem*{Lem*}{Lemma}

\theoremstyle{remark}
\newtheorem{Rem}[Thm]{Remark}
\newtheorem*{Ack}{Acknowledgment}

\newtheorem*{Rem*}{Remark}

\theoremstyle{definition}

\newtheorem{Def}[Thm]{Definition}
\newtheorem{Exa}[Thm]{Example}

\newtheorem{Ass}[Thm]{Assumption}

\newtheorem{cor}[Thm]{Corollary}

\newcommand{\braket}[2]{\left\langle{\,{#1}\,,\,{#2}\,}\right\rangle}

\newcommand{\bbR}{{\mathbb{R}}}

\newcommand{\bbZ}{{\mathbb{Z}}}

\newcommand{\de}{\partial}

\newcommand{\calB}{\mathcal{B}}
\newcommand{\calH}{\mathcal{H}}
\newcommand{\calS}{\mathcal{S}}
\newcommand{\calC}{\mathcal{C}}

\newcommand{\calM}{\mathcal{M}}
\newcommand{\calV}{\mathcal{V}}

\newcommand{\calU}{\mathcal{U}}

\newcommand{\calP}{\mathcal{P}}
\newcommand{\calY}{\mathcal{Y}}
\newcommand{\Dens}{\mathrm{Dens}}
\newcommand{\HDens}{\mathrm{Dens}^{\frac12}}

\newcommand{\sfeta}{\boldsymbol{\eta}}

\newcommand{\sfA}{{\mathsf{A}}}
\newcommand{\sfB}{{\mathsf{B}}}

\newcommand{\sfa}{{\mathsf{a}}}
\newcommand{\sfb}{{\mathsf{b}}}

\newcommand{\sfe}{{\mathsf{e}}}
\newcommand{\sfX}{{\mathsf{X}}}

\newcommand{\frg}{{\mathfrak{g}}}

\def\gpd{\,\lower1pt\hbox{$\longrightarrow$}\hskip-.24in\raise2pt
               \hbox{$\longrightarrow$}\,}

\let\Tilde=\widetilde
\let\Bar=\overline

\let\Hat=\widehat

\newcommand\qq{}
\newcommand\cmp[1]{{\qq Commun.\ Math.\ Phys.\ \bf #1}}

\newcommand\lmp[1]{{\qq Lett.\ Math.\ Phys.\ \bf #1}}

\newcommand\ijmp[1]{{\qq Int.\ J. Mod.\ Phys.\ \bf #1}}

\newcommand\anm[1]{{\qq Ann.\ Math.\ \bf #1}}

\newcommand\jdg[1]{{\qq J.\ Diff.\ Geom.\ \bf #1}}

\newcommand\proma[1]{{\qq Progress in Mathematics \bf #1}}
\newcommand\conm[1]{{\qq Cont.\  Math.\  \bf #1}}

\newcommand{\bulkzm}{{\mathcal{V}_M}}
\newcommand{\vac}{\mathrm{res}}
\newcommand{\wavy[1]}{\stackrel{#1}{\rightsquigarrow}}
\newcommand{\g}{\mathfrak{g}}
\newcommand{\bR}{\mathbb{R}}
\newcommand{\bC}{\mathbb{C}}

\allowdisplaybreaks

\begin{document}
\title{Perturbative quantum gauge theories on manifolds with boundary}


\author[A.~S.~Cattaneo]{Alberto~S.~Cattaneo}
\address{Institut f\"ur Mathematik, Universit\"at Z\"urich,
Winterthurerstrasse 190, CH-8057 Z\"urich, Switzerland}
\email{alberto.cattaneo@math.uzh.ch}

\author[P. Mnev]{Pavel Mnev}
\address{
Max Planck Institute for Mathematics, Vivatsgasse 7,
53111 Bonn, Germany}
\address{
St. Petersburg Department of V. A. Steklov Institute of Mathematics of the Russian Academy of Sciences, Fontanka 27, St. Petersburg, 191023 Russia
}
\email{pmnev@pdmi.ras.ru}

\author[N. Reshetikhin]{Nicolai Reshetikhin}
\address{Department of Mathematics,
University of California, Berkeley,
CA 94720, USA \& ITMO University, Kronverkskii ave. 49, Saint Petersburg 197101, Russian Federation
\& KdV Institute for Mathematics, University of Amsterdam,
Science Park 904, 1098 XH Amsterdam, The Netherlands}
\email{reshetik@math.berkeley.edu}

\thanks{
A.~S.~C. acknowledges partial support of SNF Grant No.~200020-149150/1. This research was (partly) supported by the NCCR SwissMAP, funded by the Swiss National Science Foundation, and by the COST Action MP1405 QSPACE, supported by COST (European Cooperation in Science and Technology).
P.~M. acknowledges partial support of RFBR
Grant No.~13-01-12405-ofi-m
and of SNF Grant No.~200021-137595.
The work of N.~R. was supported by the NSF grant DMS-0901431, by the Chern-Simons research grant.
The work on section 2 was supported by
RSF project no. 14-11-00598.
}


\begin{abstract}

This paper introduces a general perturbative quantization scheme for gauge theories on manifolds with boundary, compatible with cutting and gluing, in the cohomological symplectic (BV-BFV) formalism.
Explicit examples, like abelian BF theory and its perturbations, including nontopological ones, are presented.
\end{abstract}

\maketitle


\setcounter{tocdepth}{3}
\tableofcontents

\section{Introduction}\label{intro}

The goal of this paper is to lift Atiyah--Segal's functors to the cochain level. We show how to construct the data of such functors in terms of perturbative path integrals.

The natural framework for this construction is the Batalin--Vilkovisky formalism, or, more precisely, its natural extension to the setting of spacetime manifolds with boundary \cite{CMR,CMR2}.

The formalism we propose also incorporates the idea of Wilsonian effective action. In particular, partition functions for closed manifolds in our approach, rather than being numbers, are half-densities on the space of residual fields (if the latter can be chosen to be a point, we do get a number).  Models for the space of residual fields are partially ordered and one can pass from a larger to a smaller model by a certain fiber integration procedure -- in this way a version of Wilson's renormalization flow is built into the picture. Also, in this context, the reduced spaces of states in the case of topological field theories are not forced to be finite-dimensional, which allows one to accommodate for interesting examples (e.g. $BF$ theory) which do not fit into Atiyah's axiomatics in its usual form.


\begin{Rem}
In the text, manifolds, possibly with boundary, are always assumed to be smooth, compact and oriented.
\end{Rem}

\subsection{Functorial quantum field theory}
The functorial point of view on quantum field theory
was first outlined in \cite{At,Se} in the context of topological and conformal field theories,
however it is quite general and can be taken as a universal structure which is present in any quantum field theory.

In this framework a quantum field theory is a monoidal functor from a category of cobordisms
to a given monoidal category.
The target category is, usually, the category of complex vector spaces, or appropriate infinite-dimensional versions.
The category of cobordisms depends on the type of field theory. For example, for topological field theories these
are usually smooth oriented cobordisms.
For Yang-Mills theory and sigma models this is a category of smooth Riemannian manifolds with a collar at the boundary.
Other examples of geometric structures on cobordisms are: framing, volume form, conformal structure, spin and $\mbox{spin}^\mathbb{C}$-structures (on a Riemannian manifold).

When the target category is the category of vector spaces, such a functor does the following. To an $(n-1)$\ndash
dimensional manifold $\Sigma$ (equipped with collars \cite{ST} if we want to have smooth compositions)
it assigns a vector space:
\[
\Sigma\mapsto H(\Sigma)
\]
It should agree with the orientation reversing mapping
\[
H(\Sigma)\cong H(\overline{\Sigma})^*
\]
and should have the monoidal property
\[
H(\Sigma_1\sqcup \Sigma_2)=H(\Sigma_1)\otimes H(\Sigma_2)
\]
where the tensor product should be appropriately completed
in the infinite-dimensional case.
Here $H^*$ is the dual vector space. Typically these vector spaces are infinite-dimensional and the
notion of the dual vector space may depend on the construction of QFT.

To an $n$-dimensional cobordism $M: \de_-M\mapsto \de_+M$ the functor assigns a linear  map
\[
M\mapsto \psi_M: H(\de_-M)\to H(\de_+M)
\]
Taking into account the orientation reversing mapping and the monoidal property,
the mapping $\psi_M$ can be regarded as a vector:
\[
\psi_M\in H(\de M)
\]
Here $\de M=\overline{\de_-M}\sqcup\de_+M$ is the boundary of $M$.
For a given $M$ the space $H(\de M)$ is called the space of boundary states\footnote{To be
precise, as usual, states are density matrices on this space.}. The vector $\psi_M$ is called the state (a.k.a.\ the amplitude or the partition function or the wave function).

\subsection{The functional integral}
In the case of a theory without gauge symmetries, the space of states associated to the boundary and the state
associated to the bulk can be obtained as follows in the functional integral formalism.
We start from
a field theory on a manifold $M$ defined in terms of a space of fields $F_M$ on $M$ and an action functional $S_M$, which is a functional on $F_M$. We refer to $M$ as the space-time manifold as this is its physical meaning in field theory (but not in string theory where space-time is the target of maps defined on the worldsheet $M$).

Under mild assumptions, a local classical field theory naturally defines a symplectic manifold $F^\de_\Sigma$ of boundary fields on a boundary manifold $\Sigma$. The space of states is then defined as a quantization of $F^\de_\Sigma$. In the simple, but common, situation when $F^\de_\Sigma$ is an affine space, the quantization can be defined by choosing a Lagrangian polarization with a smooth leaf space $B_\Sigma$. The space of states is then defined as the space of functions on $B_\Sigma$. If $\Sigma=\de M$, there is a surjective submersion from the space of fields $F_M$ to the space of boundary fields $F^\de_{\de M}$. We denote by $p_M$ the composition of this map with the projection $F^\de_{\de M}\to B_{\de M}$. Then the state associated to $M$ may be heuristically defined as
\[
\psi_M(\beta) = \int_{\Phi\in p_M^{-1}(\beta)} \EE^{\frac\ii\hbar S_M(\Phi)} D\Phi,
\]
where $\beta$ is a point in $B_{\de M}$.

The gluing procedure is formally obtained by pairing the two states coming from two manifolds with the same boundary (component) $\Sigma$ via integration over $B_{\Sigma}$.\footnote{This procedure relies implicitly on a version of Fubini theorem which is heuristically expected to hold, cf. Remark \ref{rem: Fubini} and the preamble of Appendix \ref{a:comp}.}. This integral is not defined measure theoretically, but
as a formal power series modelled on the asymptotic expansion of an oscillatory integral around a critical point, with coefficients given by Feynman diagrams.\footnote{This formal power series is expected to be the asymptotic series for the non-perturbative state defined for finite values of $\hbar$.}
Sometimes it is also convenient to ``linearize'' the space of fields. Then the
procedure consists in splitting the action into a sum
$S_M=S_M^0+S_M^\text{pert}$, where $S_M^0$ is quadratic in the fields and
$S_M^\text{pert}$ is a small perturbation. One defines the Gaussian integral for $S_M^0$ as usual
 and then computes the effects of the perturbation in terms
of expectation values of powers of $S_M^\text{pert}$ in the Gaussian theory.

\subsection{Gauge theories and the BV formalism}
One of the results of this paper is the lift of the above construction to the cochain level, which is needed to treat gauge theories (or, more generally, theories with degenerate action functionals). The idea is to replace the vector space $H(\Sigma)$ by a cochain complex $\calH^\bullet(\Sigma)$
(whose cohomology in degree zero is $H(\Sigma)$). The state associated to a bulk $M$ in such a theory is a cocycle
in $\calH^0(\de M)$. The reason for this is that the construction of a state   usually depends on gauge choices and as a consequence the state
is defined up to a coboundary.

The functional integral approach outlined above has to be modified to accommodate for these changes. At first we assume that $M$
has no boundary. In this case
the most general framework is the Batalin--Vilkovisky (BV) formalism \cite{BV81}. It requires two steps: extending the space of fields on a manifold $M$ 
to an odd-symplectic supermanifold of fields $\calF_M$, and then extending the action functional to a function
$\calS_M$ on $\calF_M$ which satisfies a certain condition called the master equation. The space of fields $\calF_M$ usually comes with a special Lagrangian submanifold $\calL_0$ 
which corresponds to
the classical fields of the theory and the infinitesimal generators of symmetry.
The main result
of Batalin and Vilkovisky is that the integral of $\exp(\ii\calS_M/\hbar)$ over a Lagrangian submanifold $\calL$ of
$\calF_M$ is invariant under deformations of $\calL$. The application to field theory consists in replacing the, usually ill-defined, integral over $\calL_0$ with a well-defined integral over a deformation $\calL$ (this procedure is called the gauge-fixing).

Under mild assumptions, one can show \cite{CMR,CMR2} that a local BV theory naturally defines an even symplectic supermanifold $\calF^\de_\Sigma$ of boundary fields on a boundary manifold $\Sigma$
endowed with an odd function $\calS^\de_\Sigma$ that Poisson commutes with itself
(this structure is familiar from the BFV formalism; see \cite{BFV} and, for a more recent mathematical treatment, \cite{JS,FS}).
Again, we assume that we have a Lagrangian polarization on $\calF^\de_\Sigma$ with a smooth
leaf space $\calB_\Sigma$.
The space of states, now a cochain complex, is defined as the space of functions\footnote{The construction is in fact canonical if one works with half-densities instead of functions, which we will actually do in the paper. For simplicity of exposition we consider functions in this Introduction.} on
$\calB_\Sigma$ (in order to have a $\bbZ$\ndash graded complex, one needs a $\bbZ$\ndash grading, a.k.a.\ ghost number, on the supermanifolds of fields, which is usually the case). The coboundary operator $\Omega_\Sigma$ on the space of states is constructed
as  a quantization
of $\calS^\de_\Sigma$ which we assume to square to zero (otherwise the theory is called anomalous).

If $\Sigma=\de M$, there is a surjective submersion from the space of fields $\calF_M$ to the space of boundary fields
$\calF^\de_{\de M}$. The master equation for $\calS_M$ turns out to be modified by terms coming from
$\calF^\de_{\de M}$ (the classical master equation in this situation was analyzed in \cite{CMR} in the framework of BV-BFV theory). As we explain below (see Section \ref{ss:pertq}),  if we denote by $p_M$ the composition of the map
from $\calF_M$ to
$\calF^\de_{\de M}$ with the projection to the leaf space $\calB_{\de M}$, the fibers of $p_M$ inherit an odd-symplectic structure and the restriction of $\calS_M$ to the fibers satisfies the master equation modified by a boundary term.
The state
associated to $M$ is then defined by integrating the exponentiated action over a Lagrangian submanifold $\calL$ in the fibers
\begin{equation}\label{e:psi intro}
\psi_M(b) = \int_{\Phi\in\calL\subset p_M^{-1}(b)} \EE^{\frac\ii\hbar S_M(\Phi)} D\Phi,\quad b\in\calB_{\de M}.
\end{equation}
Notice that in principle we need a choice of $\calL$ in each fiber $p_M^{-1}(b)$. We refrain from using the notation $\calL_b$, for
we will see that for the formalism to make sense one actually has to assume that, at least locally, the fibration is
a product manifold and that $\calL$ is a Lagrangian submanifold of the fibers independent of the base point.\footnote{This assumption is natural in the setting of perturbative quantization in the formal neighborhood of a fixed critical point of the action when the relevant spaces of fields/boundary fields are automatically equipped with a linear structure.
} Notice that the functional integral corresponds to a choice of ordering. This yields a preferred quantization $\Omega_{\de M}$ of the
boundary action.

One of the goals of this paper is to show that, under natural assumptions, this is a well-defined procedure and that a change of gauge fixing (i.e., a deformation of the Lagrangian submanifolds $\calL$) changes the state $\psi_M$ by an
$\Omega_{\de M}$\ndash exact term.

\subsection{Perturbation theory and residual fields}
The functional integral (\ref{e:psi intro}) is understood as an expansion in Feynman diagrams corresponding to the
asymptotic
expansion around a critical point. We also consider perturbation theory where $\calS_M=\calS_M^0+\calS_M^\text{pert}$, where $\calS^0_M$ is quadratic and $\calS_M^\text{pert}$ is a small perturbation. In this case, it is also interesting to allow for non-isolated critical
points of $\calS_M^0$. The idea is to consider critical points of $\calS_M^0$ modulo its own gauge symmetry
as {\backgrounds} and to integrate 
in transversal directions to the space of {\backgrounds}.
The resulting state is a function on the space of {\backgrounds},
which is a finite-dimensional supermanifold and comes equipped with a BV Laplacian, i.e., an odd second order operator $\Delta$
that squares to zero (and anticommutes with $\Omega_{\de M}$). The main result is that, under certain assumptions, the state is now closed
under the coboundary operator $\hbar^2\Delta+\Omega_{\de M}$ and changes by
$\hbar^2\Delta+\Omega_{\de M}$\ndash exact terms under changes of gauge-fixing. This has profound consequences, e.g., when one wants to globalize the results (i.e., define the state as a function on the whole space of
solutions of Euler-Lagrange equations for $\calS_M$ modulo gauge symmetry,
and
not just on a formal neighborhood of each point as in perturbation theory), cf. \cite{BCM} for the detailed treatment of globalization for the Poisson sigma model.

More general spaces of {\backgrounds} may be defined as submanifolds of $\calF_M$ compatible with the BV structure.
This leads, e.g., to a Wilsonian picture, where one has a hierarchy of spaces of residual (``low energy'') fields and can pass from larger to smaller models by fiber BV integrals, see Appendix \ref{appendix: semi-classics and eff actions} for more details. Choosing appropriate spaces of {\backgrounds} is also important for the gluing procedure, see Section \ref{sss:gluing}.


\subsection{Main results}

This paper contains two main results. The first one is the construction of a
general framework of perturbative
quantization of any local QFT with gauge symmetry on manifolds with boundary.
Some of the assumptions may be too strong for specific examples. This is why
the application of this framework requires extra
work. In particular, we do not address possible issues with renormalization which would
be very important for non-topological theories.

Our second main result concerns the application of the general framework to a class of
topological field theories, \textit{BF-like theories}, see Sections \ref{s:abeBF} and \ref{sec: BF-like}. The result can be formulated as the following
theorem.

\begin{Thm*} The following holds for BF-like theories such as
BF theories, 2D Yang-Mills theory, the nonlinear Poisson sigma-model and the first order formalism
for quantum mechanics
where the construction of $\Omega_{\de M}$ and of the state $\psi_M$
is given in terms of configuration space integrals:

\begin{enumerate}
\item  $\Omega_{\de M}^2=0$.
\item The quantum master equation modified by the boundary term holds:\\
$(\hbar^2\Delta+\Omega_{\de M})\psi_M=0$.
\item A change of gauge changes the state (the partition function) by a coboundary:\\
$\psi_M\mapsto \psi_M+(\hbar^2\Delta+\Omega_{\de M})\phi$.
\item The gluing axiom holds, i.e. if $M=M_1\cup_{\Sigma} M_2$, then
\[
\psi_M=P_*(\psi_{M_1}\underset\Sigma\ast\psi_{M_2}),
\]
where $P_*$ is the BV-pushforward with respect to the odd-symplectic fibration
of residual fields $P\colon\calV_{M_1}\times \calV_{M_2}\rightarrow \calV_M$ (see section 2.4.4),
and $\ast_\Sigma$ is the pairing of states in $\calH(\Sigma)$.
\end{enumerate}
\end{Thm*}

This theorem is proven in Sections \ref{s:abeBF} and \ref{sec: BF-like} with the beginning of Section \ref{sec: BF-like} (specifically, Subsections \ref{s:BF-like pert},\ref{s:mQME}) being the core
part.

\subsection{Summary}

In Section ~\ref{s:BVBFV} we give the framework of quantum BV-BFV theory.
Later, in Sections 3 and 4 we use Feynman diagrams and integrals over configuration
spaces to make a precise mathematical construction of abelian BF theories and their perturbations,
including, in particular, non-abelian BF theories, 2D Yang-Mills theiory and the Poisson sigma model.
Section ~\ref{s:BVBFV}  begins with a short review of the classical BV-BFV formalism for Lagrangian
field theories on manifolds
with boundaries \cite{CMR,CMR2}. Then, after introducing in Section \ref{sec: BV pushforward in families}
the main construction underlying our quantization scheme --  BV pushforward in a family --
we continue with an abstract formulation of its quantum version (Section \ref{ss:qBVBFV}) which will be substantiated by examples in the rest of the paper.
Then we present the construction of perturbative quantization which starts with a classical BV-BFV theory and returns a quantum BV-BFV theory (Section \ref{ss:pertq}).
 Here we focus on finite-dimensional integrals and comment on the infinite-dimensional version defined via the stationary phase asymptotical formula, with integrals defined by their Feynman diagram expansions.
In particular, we show how the functional integral formalism yields a preferred quantization of the BFV action
-- i.e., roughly speaking, of the constraints on boundary fields -- which is compatible with the quantization in the bulk.

In Section~\ref{s:abeBF}, we consider the case of abelian $BF$ theories. We discuss the space of {\backgrounds},
the choice of gauge fixings (by Hodge theory on manifolds with boundary), and the construction of
propagators (some of the results of sections~\ref{sec: abBF backgrounds}
and~\ref{ss:prop}  were already described in \cite{Wu}).
We compute
the state explicitly, see \eqref{e:state},  as
\[
\Hat\psi_M= T_M\, \EE^{\frac\ii\hbar\calS^\text{eff}_M},
\]
where $T_M$ is, up to a coefficient depending on Betti numbers of $M$, the torsion of $M$ (to the power $\pm1$) and
\[
\calS^\text{eff}_M=
\pm\left(\int_{\de_2 M}\bB\sfa -\int_{\de_1 M}\sfb\bA\right)
\pm \int_{\de_2M\times\de_1M}\pi_1^*\bB\,\eta\,\pi_2^*\bA
\]
is the effective action, where $\bA$ and $\bB$ denote the boundary fields, $\sfa$ and $\sfb$ the residual fields and
$\eta$ the propagator ($\pi_1$ and $\pi_2$ are just projections to the factors in the Cartesian product).
We show that the quantum BV-BFV axioms are satisfied. Finally, we
discuss the gluing procedure and show that it is a combination of the gluing formula for torsions and of Mayer--Vietoris. In particular, we derive a formula for the gluing of propagators (see also Appendix~\ref{We now do the final step in computing the propagator on $M$ for the reduced  space of pseudovacua}).

In Section~\ref{sec: BF-like} 
we discuss examples of quantum BV-BFV theories that arise as a perturbation of abelian $BF$ theories.
These include non-abelian $BF$ theories, quantum mechanics, the Poisson sigma model, two-dimensional Yang--Mills theory and
particular cases of Chern--Simons theory. For this class of examples we show that gluing and
the quantum BV-BFV axioms are satisfied.
Notice that, with the exception of quantum mechanics and two-dimensional Yang--Mills theory, we only present topological field theories, yet recall that the formalism of Section \ref{s:BVBFV} is general. In the context of two-dimensional Yang--Mills theory we also present a nontrivial example of the generalized Segal--Bargmann transform. The Poisson sigma model provides an example where the boundary structure gets quantum corrections.

Appendix \ref{a:Hodge} introduces the necessary background on Hodge theory on manifolds with boundary. In Appendix \ref{a:prop} we present a construction of propagators on manifolds with boundary by a version of the method of image charges. In Appendix \ref{a:comp} we present the details of the gluing procedure for propagators. In Appendices \ref{a: exa of prop} and \ref{a: exa of gluing of prop} we provide examples of propagators and of the gluing construction for propagators. In Appendix \ref{appendix: semi-classics and eff actions} we comment on the globalization aspect of our formalism where perturbative quantization is performed in a family over the moduli space of solutions of the Euler--Lagrange equations of the classical system modulo gauge symmetry.

\subsection{Final comments}
The general setting described in Section~\ref{s:BVBFV} has a much wider scope than the few examples
presented in this paper, which are however particularly suitable to point out the various features of the formalism.
Depending on the reader's taste, it might actually be useful to start with the examples, at least Section~\ref{s:abeBF}, first and to return to Sections~ \ref{sec: BV pushforward in families}, \ref{ss:qBVBFV} and \ref{ss:pertq} later. 

Whereas we discuss abelian $BF$ theory in full details, we only present the general structure for its perturbations.
Explicit computations of states are of course important in relevant examples (see \cite{CMW} for a computation in
split Chern--Simons theory).

We also plan to present another instantiation of the general theory in the case of the discrete version of $BF$ theories
in a separate paper \cite{cell_ab_BF}.

The application of the formalism to other classes of theories, in particular to physical theories like Yang--Mills with and without matter, is part of a long standing program.

\begin{Ack} We thank Francesco Bonechi, Ivan Contreras, Santosh Kandel, Thomas Kappeler, Samuel Monnier, Albert S. Schwarz, Jim Stasheff and especially Konstantin Wernli for useful discussions and comments.
A.~S.~C. and P.~M. gratefully acknowledge support from the University of California at Berkeley, the QGM centre at the University of Aarhus
and the Simons Center for Geometry and Physics where parts of research for this paper were performed. N.~R. is also grateful for the hospitality at QGM, Aarhus University and at Universite Paris 7, where an important part of the work has been done.
Also, P.~M. thanks the University of Zurich, where he was affiliated until mid 2014 and where a substantial part of the work was done, for providing an excellent work environment.
\end{Ack}

\section{The BV-BFV formalism}\label{s:BVBFV}
The aim of this Section is to describe a perturbative quantization scheme for gauge theories on manifolds with boundary in the framework of the BV-BFV formalism introduced in \cite{CMR,CMR2}.
For the reader's convenience, we start by recalling the classical BV-BFV construction (Section \ref{sec: class BV-BFV}). In Section \ref{ss:qBVBFV} we describe the mathematical structure of a quantum BV-BFV theory, and in Section \ref{ss:pertq} we develop the perturbative quantization scheme which starts with a classical BV-BFV theory and lands in the quantum one. The main technical tool underlying the construction of quantization is the family (parametric) version of the construction of pushforward for solutions of quantum master equation along odd-symplectic fibrations; we present this construction in Section \ref{sec: BV pushforward in families}.

\subsection{The classical BV-BFV formalism}\label{sec: class BV-BFV}
Here we will recall basic definitions of BV-BFV manifolds which are the fundamental structure for classical gauge theories on space-time manifolds with boundary. The reader is referred to \cite{CMR} for details and examples.
\subsubsection{BV-BFV manifolds}
Let $\calF$ be a supermanifold with an additional
$\bbZ$\ndash grading; we will speak of a graded manifold. An  odd vector field $Q$ of degree $+1$ on $\calF$ is called \textsf{cohomological} if it
commutes with itself, i.e., $[Q,Q]=0$. A symplectic form (i.e., a closed, nondegenerate $2$\ndash form) $\omega$ is called a \textsf{BV form} if it is odd and has degree $-1$ and a
\textsf{BFV form} if it is even and has degree $0$. If $\omega$ is exact, a specific $\alpha$ of the same parity and degree with $\omega=\delta\alpha$ will be called a
\textsf{BV/BFV $1$\ndash form}.
\begin{Rem}
In the application to field theory, the coordinates on the BV manifold are the classical fields, the ghosts and the antifields for all of them. In particular, the de Rham differential on such a supermanifold will correspond to the variation and for this reason we use the symbol $\delta$. This will also avoid confusion with the de Rham differential $\dd$ on the underlying spacetime manifold. Finally, observe that the degree in this context is what is usually called ghost number. In the case when no classical fermionic fields are present, the parity is equal to the ghost number modulo $2$. This is the case in all the examples discussed in this paper, but in this introductory Section we prefer to be general. As a result $\omega$ is tri-graded: form degree $2$, parity odd, ghost number $-1$.
\end{Rem}

A vector field $Q$ is called symplectic if $L_Q\omega=0$ and Hamiltonian if $\iota_Q\omega=\delta S$ for a function $S$. In the BFV case, by degree reasons,
if the cohomological vector field is symplectic, it is also automatically Hamiltonian with a uniquely defined function $S$ of degree $+1$ called the \textsf{BFV action}.
In the BV case, a Hamiltonian function of degree $0$ for the cohomological vector field is called a \textsf{BV action}.

\begin{Def}
A \textsf{BFV manifold} is a triple $(\calF,\omega,Q)$ where $\calF$ is a graded manifold,
$\omega$ is a BFV form and $Q$ is a cohomological, symplectic vector field on $\calF$.
A BFV manifold is called \textsf{exact} if a BFV $1$\ndash form $\alpha$ is specified.
\end{Def}

\begin{Def}\label{def: BV-BFVmfd}
A \textsf{BV-BFV manifold} over a given exact BFV manifold $(\calF^\de,\omega^\de=\delta\alpha^\de,Q^\de)$ is a quintuple
$(\calF,\omega,\calS,Q,\pi)$ where $\calF$ is a graded manifold, $\omega$ is a BV form,
$\calS$ is an even function of degree $0$, $Q$ is a cohomological vector field and $\pi\colon\calF\to\calF^\de$ is a surjective submersion such that
\begin{enumerate}[(i)]
\item $\iota_Q\omega = \delta\calS + \pi^*\alpha^\de$,
\item $Q^\de=
\delta\pi\, Q$.
\end{enumerate}
Here 
$\delta\pi$
denotes the differential of the map $\pi$.
If $\calF^\de$ is a point, $(\calF,\omega,\calS)$ is called a \textsf{BV manifold}.
\end{Def}

A 
consequence of the conditions of Definition \ref{def: BV-BFVmfd}
is the 
\textsf{modified Classical Master Equation (mCME)}:
\begin{equation}\label{e:mCME}
Q(\calS) = \pi^*(2\calS^\de -\iota_{Q^\de}\alpha^\de).
\end{equation}
In the case 
when $\calF^\de$ is a point,
it reduces to the usual CME, $Q(S)=0$. The latter is normally written as $(S,S)=0$,
where $(\ ,\ )$ is the BV bracket defined by $\omega$.\footnote{
Note that $(\ , \ )$ is a Gerstenhaber bracket due to the odd degree of $\omega$. In the literature it is also called the
anti-bracket.
}
The modified CME (\ref{e:mCME}) can equivalently be rewritten as
\begin{equation}\label{e:mCME_1}
\frac{1}{2}\,\iota_Q \iota_Q \omega = \pi^*\calS^\de.
\end{equation}

\subsubsection{Classical BV-BFV theories}
An exact \textsf{BV-BFV $d$\ndash dimensional field theory} is the local association of an exact BFV manifold
$(\calF^\de_\Sigma,\omega^\de_\Sigma=\delta\alpha^\de_\Sigma,Q^\de_\Sigma)$ to every $(d-1)$\ndash dimensional compact manifold $\Sigma$ and
of a BV manifold $(\calF_M,\omega_M,\calS_M,Q_M,\pi_M)$  over the BFV manifold
$(\calF^\de_{\de M},\omega^\de_{\de M}=\delta\alpha^\de_{\de M},Q^\de_{\de M})$ to every $d$\ndash dimensional compact manifold $M$ with boundary
$\de M$.
Here $\calF_M$ is the space of fields on $M$ (in the bulk) and $\calF^\de_{\de M}$ is the space fields on the boundary $\de M$  (or the phase space).
{\sf Local association} means that the graded manifolds $\calF_M$ and $\calF^\de_\Sigma$ 
are modeled on
spaces of sections of bundles (or, more generally, sheaves) over $M$ and $\Sigma$, whereas the function, symplectic forms and cohomological vector fields are local (i.e., they are defined as integrals of functions of finite jets of the fields). In particular, $\calF_M$, $\calF_\Sigma^\de$ are, typically, infinite-dimensional Banach or Fr\'echet manifolds (depending on the allowed class of sections).

\begin{Rem}
The BV-BFV formalism may be generalized to the nonexact case (see \cite{CMR,CMR2}), but we will not need it in this paper.
\end{Rem}

A classical BV-BFV theory can be seen, in the spirit of Atiyah-Segal axioms, as a functor from the category of $d$-dimensional cobordisms endowed with some geometric structure (depending on a particular model, it can be a Riemannian metric, a conformal structure, a volume form, a principal bundle, a cell decomposition, etc.) with composition given by gluing along common boundary, to the category with objects the BFV manifolds and morphisms the BV-BFV manifolds over direct products of BFV manifolds, with composition given by homotopy fiber products. This functor is compatible with the monoidal structure on source (space-time) and target (BFV) categories, given by disjoint unions and direct products, respectively (in particular, $\calF^\de_\varnothing$ is a point).   See \cite{CMR} for details. See also \cite{Schlegel} for the approach to gluing via synthetic geometry.

\subsection{Finite-dimensional BV pushforward in families}
\label{sec: BV pushforward in families}

Here we will recall the notion of 
the BV integral (Section \ref{sec: BV int}) and its refined version, the BV pushforward construction, or fiber BV integral (Section \ref{sec: BV pushforward}). The latter is a model for a path integral 
over ``fast'' (or ``ultraviolet'') fields, depending on the ``slow'' (or ``infrared'') residual fields (Wilson's effective action), within the Batalin-Vilkovisky approach to gauge theories. We then introduce the family (or parametric) version of BV pushforward (Section \ref{sec: family version}), which models the computation of matrix elements of the evolution operator in the effective action framework. In this sense, the BV pushforward in families can be regarded as a ``hybrid effective action'' formalism (i.e. a hybrid between effective action in BV formalism and an evolution operator/partition function, as in Atiyah-Segal axiomatics). In Section \ref{sec: BV pushforward, case of exp} we specialize the construction of BV pushforward in family to ``exponential'' half-densities, i.e. those of the form $\mathfrak{m}^{\frac12}\, \EE^{\frac\ii\hbar \calS}$ and consider the asymptotics $\hbar\rightarrow 0$, which sets the stage for the perturbative quantization scheme that is the focus of this paper.

Within this Section, unless explicitly stated otherwise, we are assuming that all manifolds are finite-dimensional and all integrals are convergent (see also \cite{ABF} for the discussion of finite-dimensional BV integrals; for classical BV formalism in finite-dimensional setting, see \cite{FK}). We also assume that manifolds are equipped with orientations, so that we can ignore the distinction between densities and Berezinians.

The logic is that we develop all the constructions in the setting of finite-dimensional integrals, which are defined within measure theory. Then we can consider the fast oscillating ($\hbar\rightarrow 0$) asymptotics of our integral and write it, using stationary phase formula, as a sum of Feynman diagrams. In the case of path integrals over infinite-dimensional spaces of fields, we instead {\sf define} the integral perturbatively, i.e. as a formal power series in $\hbar$ with coefficients given by sums of Feynman diagrams. In this perturbative setting, theorems that are proven for measure-theoretic integrals have to be checked, model by model, on the level of Feynman diagrams.

\subsubsection{BV integral}
\label{sec: BV int}
Let $\calY$ be a $\mathbb{Z}$-graded manifold with a degree $-1$
odd symplectic form $\omega$. 

\begin{Thm}[\cite{Khudaverdian,Severa}]
The space $\Dens^{\frac12}(\calY)$ of half-densities on $\calY$ carries a degree $+1$ odd 
coboundary operator, the canonical
BV Laplacian $\Delta$, such that in any local Darboux coordinate chart $(x^i,\xi_i)$ on $\calY$, the operator $\Delta$ has the form $\sum_i\frac{\de}{\de x^i}\frac{\de}{\de \xi_i}$. 
\end{Thm}

\begin{Def} We say that a Berezinian $\mu$ on $\calY$ is {\sf compatible} with the odd symplectic structure $\omega$, if $\Delta \mu^{\frac12}=0$ with $\Delta$ the canonical BV Laplacian.
\end{Def}

\begin{Rem}
Given a compatible Berezinian $\mu$ on $(\calY,\omega)$,
one can construct a $\mu$-dependent BV Laplacian on functions on $\calY$ (as opposed to half-densities), $\Delta_\mu:\; C^\infty(\calY)\rightarrow C^\infty(\calY)$ defined by $\mu^{\frac12}\Delta_\mu f=\Delta(\mu^{\frac12}f)$ for any $f\in C^\infty(\calY)$. See \cite{SchwBV} for details.
\end{Rem}

Given a Lagrangian submanifold $\calL\subset \calY$, a half-density $\xi$ on $\calY$ can be restricted to a 1-density $\xi|_\calL$ on $\calL$, which can in turn be integrated over $\calL$. The BV integral is the composition
$$\int_\calL: \quad \Dens^{\frac12}(\calY)\xrightarrow{\bullet|_\calL} \Dens(\calL)\xrightarrow{\int}\mathbb{C},\qquad \xi\mapsto \int_\calL \xi|_\calL.$$

\begin{Thm}[Batalin-Vilkovisky-Schwarz, \cite{BV81,SchwBV}]
\label{thm: BV Stokes}
\begin{enumerate}[(i)]
\item For every half-density $\xi$ on $\calY$ and every Lagrangian submanifold
$\calL\subset \calY$, one has
$$\int_{\calL} \Delta\xi=0$$
assuming convergence of the integral.
\item For a half-density $\xi$ on $\calY$ satisfying $\Delta\xi=0$ and a smooth family of Lagrangian submanifolds $\calL_t\subset \calY$ parametrized by $t\in [0,1]$, one has
$$ \int_{\calL_0} \xi = \int_{\calL_1}\xi $$
assuming convergence of $\int_{\calL_t}\xi$ for all $t\in[0,1]$.\footnote{In fact, in \cite{SchwBV} a stronger version of this statement is proven.
}
\end{enumerate}

\end{Thm}

\subsubsection{BV pushforward}
\label{sec: BV pushforward}
Assume that $(\calY,\omega)$ is a direct product of two odd-symplectic manifolds $(\calY',\omega')$ and $(\calY'',\omega'')$, i.e. $\calY=\calY'\times\calY''$, $\omega=\omega'+\omega''$. Then the space of half-densities on $\calY$ factorizes as
$$\HDens(\calY)=\HDens(\calY')\widehat\otimes \HDens(\calY'').$$
BV integration in the second factor, over a Lagrangian submanifold $\calL\subset \calY''$, defines a pushforward map on half-densities
\begin{equation}\label{e:BV pushforward}
\int_{\calL}:\quad \HDens(\calY)\xrightarrow{\mathrm{id}\otimes\int_{\calL}} \HDens(\calY').
\end{equation}
This map is also known as the fiber BV integral.\footnote{Here $\calY'$ is a model for ``slow fields'', or ``zero-modes'', or ``classical backgrounds'', or ``residual fields'' in the effective action formalism.}
The version of Theorem \ref{thm: BV Stokes} in the context of BV pushforwards is as follows.
\begin{Thm}
\label{thm: BV Stokes for pushforward}
\begin{enumerate}[(i)]
\item \label{thm: BV Stokes for pushforward, (i)}
For $\xi$ a half-density on $\calY$,
$$\int_{\calL}\Delta\xi=\Delta'\int_{\calL}\xi$$.
\item \label{thm: BV Stokes for pushforward, (ii)}
For $\calL_t\subset \calY''$ a smooth family of Lagrangian submanifolds parametrized by $t\in [0,1]$, and a half-density $\xi$ on $\calY$ satisfying $\Delta \xi=0$, one has
\begin{equation}\label{e:BV pushforward L dependence}
\int_{\calL_1}\xi-\int_{\calL_0}\xi=\Delta' \Psi
\end{equation}
for some $\Psi\in \HDens(\calY')$. Moreover, if $\calL_{t+\epsilon}$ is given, in the first order in $\epsilon$, as the flow in time $\epsilon$ of a Hamiltonian vector field $(\bullet,H_t)$ with $H_t\in C^\infty(\calL_t)_{-1}$, then $\Psi$ in (\ref{e:BV pushforward L dependence}) is given by
$$\Psi=\int_0^1 dt \int_{\calL_t}\xi H_t.$$
\end{enumerate}
\end{Thm}
\begin{proof}
While (\ref{thm: BV Stokes for pushforward, (i)}) follows immediately from (\ref{thm: BV Stokes}) and from the splitting of Laplacians $\Delta=\Delta'+\Delta''$, part (\ref{thm: BV Stokes for pushforward, (ii)}) is implied by the following calculation. Let $\mu$ be a Berezinian on $\calY$ compatible with $\omega$. Then it defines a BV Laplacian $\Delta_\mu=\mu^{-\frac12}\Delta(\mu^{\frac12}\bullet)$ on functions on $\calY$. Expressing the half-density $\xi$ as $\xi=\mu^{\frac12}f$ with $f$ a function, we have
\begin{multline}\label{e:BV pushforward infin symplectomorphism}
\frac{\partial}{\partial t}\int_{\calL_t} \mu^{\frac12}f=\int_{\calL_t}\mu^{\frac12} ((f,H_t)+f\underbrace{\frac12 \mathrm{div}_\mu (\bullet,H_t)}_{\Delta_\mu H_t})=\\
=\int_{\calL_t} \mu^{\frac12}(\Delta_\mu(f H_t)-\underbrace{\Delta_\mu(f)}_0 H_t)=\Delta'\int_{\calL_t}\mu^{\frac12}f H_t,
\end{multline}

using (\ref{thm: BV Stokes for pushforward, (i)}) and the assumption that $\Delta\xi=0$ or equivalently $\Delta_\mu f=0$.
\end{proof}
We refer the reader to \cite{discrBF,CMcs} for more details.

Theorem \ref{thm: BV Stokes for pushforward} implies in particular that the BV pushforward defines a pushforward map from the cohomology of $\Delta$ to the cohomology of $\Delta'$ dependent on a choice of a Lagrangian $\calL$ modulo \textsf{Lagrangian homotopy}.\footnote{We say that two Lagrangians are Lagrangian homotopic if they can be connected by a smooth family of Lagrangians.}

Of particular interest is the case when the odd-symplectic manifolds $\calY,\calY',\calY''$ are equipped with compatible Berezinians $\mu,\mu',\mu''$ which then give rise to BV Laplacians $\Delta_\mu,\Delta'_{\mu'},\Delta''_{\mu''}$ on functions on the respective manifolds. Assuming that $\mu=\mu'\otimes\mu''$, the BV Laplacians satisfy $\Delta_{\mu}=\Delta'_{\mu'}+\Delta''_{\mu''}$. We can apply the BV pushforward to a half-density of the form $\xi=\EE^{\frac{\ii}{\hbar}\calS}\mu^{1/2}$. It is $\Delta$-closed if and only if $\calS\in C^\infty(\calY)_{0}$ satisfies the \textsf{quantum master equation} (QME):
\begin{equation}\label{e:QME}
\Delta_\mu \EE^{\frac{\ii}{\hbar}\calS}=0 \quad \Leftrightarrow \quad \frac12 (\calS,\calS)-\ii\hbar \Delta_\mu \calS=0.
\end{equation}
\begin{Rem}
Assume that $\calS$ has the form $\calS=\calS^0+\hbar \calS^1+\cdots \in C^\infty(\calY)_0[[\hbar]]$. Then (\ref{e:QME}) implies, by expanding in powers in $\hbar$ and looking at the lowest order term, the {\sf classical master equation} (CME)
\begin{equation}\label{e:CME}
(\calS^0,\calS^0)=0.
\end{equation}
\end{Rem}

\begin{Def}
We define the \textsf{effective BV action} $\calS'\in C^\infty(\calY')_{0}$ via BV pushforward (\ref{e:BV pushforward}):
\begin{equation}\label{e:fiber BV int}
\EE^{\frac{\ii}{\hbar}\calS'}\mu'^{\frac12}\colon =\int_\calL \EE^{\frac{\ii}{\hbar}\calS}\mu^{\frac12}.
\end{equation}
\end{Def}
Theorem \ref{thm: BV Stokes for pushforward} implies the following.
\begin{cor}[\cite{discrBF,CMcs}]
\label{cor: BV pushforward for exp(S) properties}
\begin{enumerate}[(i)]
\item If $\calS\in C^\infty(\calY)_0$ satisfies the quantum master equation
on $\calY$, then $\calS'\in C^\infty(\calY')_0$ defined by (\ref{e:fiber BV int}) satisfies the quantum master equation on $\calY'$.
\item Assume that $\calL_t\subset \calY''$ is a smooth family of Lagrangian submanifolds parametrized by $t\in [0,1]$ and $\calS$ satisfies the quantum master equation on $\calY$. Let $\calS_t$ be the effective BV action  defined using $\calL_t$. Then $\calS_1$ is a \textsf{canonical BV transformation} of $\calS_0$, i.e.
\begin{equation}\label{e:can transf}
\EE^{\frac{\ii}{\hbar}\calS'_1}-\EE^{\frac{\ii}{\hbar}\calS'_0}=\Delta'_{\mu'}\Psi
\end{equation}
for some $\Psi\in C^\infty(\calY')_{-1}$. Infinitesimally, one has
$$\frac{\partial}{\partial t}\calS'_t=(\calS'_t,\phi_t)-\ii\hbar\Delta'_{\mu'}\phi_t$$
where the generator of the infinitesimal canonical transformation is
\begin{equation}\label{e:can tranf generator induced from Lagr homotopy}
\phi_t=\mu'^{-\frac12}\EE^{-\frac{\ii}{\hbar}\calS'_t}\cdot\int_{\calL_t} \mu^{\frac12}\EE^{\frac{\ii}{\hbar}\calS} H_t, 
\end{equation}
with $H_t$ as in (\ref{thm: BV Stokes for pushforward, (ii)}) of Theorem \ref{thm: BV Stokes for pushforward}.
The generator $\Psi$ of the finite canonical transformation (\ref{e:can transf}) is:
$$\Psi=\int_0^1 dt \int_{\calL_t} \mu^{\frac12}\EE^{\frac{\ii}{\hbar}\calS} H_t.$$
\item If $\calS$, $\Tilde\calS$ are solutions of the quantum master equation on $\calY$ differing by a canonical transformation, then the corresponding effective actions $\calS',\Tilde\calS'$ also differ by a canonical transformation on $\calY'$.
\end{enumerate}

\end{cor}

As a consequence, the BV pushforward gives a map from solutions of the QME on $\calY$ modulo canonical transformations to solutions of the QME on $\calY'$ modulo canonical transformations. This map depends on the choice of a class of Lagrangians $\calL\subset \calY''$ modulo Lagrangian homotopy:
$$\frac{\mbox{Solutions of QME on}\;\calY}{\mbox{can. transf.}}\qquad\stackrel{[\calL]}{\longrightarrow}\qquad \frac{\mbox{Solutions of QME on}\;\calY'}{\mbox{can. transf.}}.$$

\begin{Rem}\label{rem: hedgehog}
The direct product $\calY=\calY'\times \calY''$ setting for the BV pushforward introduced above admits the following generalization.
Let $\calY''$ be an odd-symplectic manifold.
An odd-symplectic fiber bundle with typical fiber $\calY''$ over an odd-symplectic manifold $\calY'$
 consists of a pair $(\calY,\calY')$ of odd-symplectic manifolds together with a surjective submersion $\pi\colon\calY\to\calY'$ such that each point of $\calY'$ has a neighborhood $\calU$ with a symplectomorphism
 $\phi_{\calU}\colon \pi^{-1}(\calU)\to \calU\times\calY''$. Notice that, by the nondegeneracy of the symplectic forms, on the overlaps of two such neighborhoods
 $\calU_{\alpha}$ and $\calU_\beta$ the transition functions
 $\phi_{\alpha\beta}\colon \pi^{-1}(\calU_\alpha\cap\calU_\beta)\to \pi^{-1}(\calU_\alpha\cap\calU_\beta)$
 are given by symplectomorphisms of $\calY$ constant over $\calY'$. If all these symplectomorphisms are connected to the identity, the BV pushforward may be defined and we call such a fiber bundle a
\textsf{hedgehog}, or a hedgehog fibration.\footnote{
An explanation for this terminology may be found on
\href{https://youtu.be/BhPtkIMEnjk}{YouTube: Hedgehog BV}.
}
\end{Rem}

\begin{Rem}
A more general version of BV pushforward is the following. Suppose we have a coisotropic submanifold
$\calC$ of $\calY$ with a smooth reduction $\underline{\calC}$. A half-density on $\calY$ can then be integrated
to a half density on $\underline{\calC}$. This pushforward is also a chain map for the BV Laplacians.
An example of this is when we have a hedgehog fibration $\calY\to\calY'$; the total space of the fiber bundle
over $\calY'$ consisting of the (locally constant) choice of Lagrangian submanifolds in the hedgehog fibers is
a coisotropic submanifold of $\calY$ with reduction $\calY'$. The hedgehog version, though less general, is more
suitable for applications to field theory as on the one hand we want to fix the reduction and on the other hand the choice of Lagrangian submanifolds is an auxiliary piece of data.
\end{Rem}

\subsubsection{Family version}
\label{sec: family version}
Let $(\calY,\omega)$ be an odd-symplectic manifold as above and let $\calB$ be a $\mathbb{Z}$-graded supermanifold endowed with a degree $+1$ odd differential operator $\Omega$ acting on half-densities on $\calB$ satisfying $\Omega^2=0$.\footnote{In the setting of field theory, $\calB$ will become the space of leaves of a Lagrangian foliation of the space of boundary fields, i.e. the space parameterizing admissible boundary conditions for the path integral over field configurations.}

Let
$\calF=\calB\times\calY$ be the product manifold. Then we have a coboundary operator $\hbar^2 \Delta+\Omega$ acting on
\begin{equation}
\HDens(\calF)=\HDens(\calB)\widehat\otimes\HDens(\calY).
\label{e:half-densities splitting in B and Y}
\end{equation}
Assuming, as in Section \ref{sec: BV pushforward}, that $\calY$ is split as a product of two odd-symplectic manifolds $(\calY',\omega')$ and $(\calY'',\omega'')$, we have a version of the BV pushforward (\ref{e:BV pushforward}) in family over $\calB$:
\begin{equation}\label{e: BV pushforward in family}
\int_{\calL}:\quad  \HDens(\calF)\rightarrow \HDens(\calF')
\end{equation}
where $\calF'=\calB\times \calY'$ and $\calL\subset \calY''$ is a Lagrangian submanifold. Half-densities on $\calF'$ are equipped with a coboundary operator $\hbar^2\Delta'+\Omega$ where $\Delta'$ is the canonical BV Laplacian on $\calY'$.
We have the following family version of Theorem \ref{thm: BV Stokes for pushforward}.

\begin{Thm}\label{thm: BV pushforward in family}
\begin{enumerate}[(i)]
\item \label{thm: BV pushforward in family (i)}
For every Lagrangian $\calL\subset \calY''$ and every $\xi\in\HDens(\calF)$, we have
$$\int_\calL (\hbar^2\Delta+\Omega)\xi=(\hbar^2\Delta'+\Omega)\int_\calL \xi.$$
\item \label{thm: BV pushforward in family (ii)}
For a half-density $\xi$ on $\calF$ satisfying $(\hbar^2\Delta+\Omega)\xi=0$ and a smooth family of Lagrangians $\calL_t\subset \calY''$ parametrized by $t\in [0,1]$, we have
$$\int_{\calL_1}\xi-\int_{\calL_0}\xi=(\hbar^2\Delta'+\Omega)\Psi$$
for some $\Psi\in \HDens(\calF')$. Explicitly, the generator is
$$\Psi=\hbar^{-2}\int_0^1 dt \int_{\calL_t} \xi H_t$$
with $H_t\in C^\infty(\calL_t)_{-1}$ as in (\ref{thm: BV Stokes for pushforward, (ii)}) of Theorem \ref{thm: BV Stokes for pushforward}.
\end{enumerate}

\end{Thm}
\begin{proof}
Part (\ref{thm: BV pushforward in family (i)}) follows immediately from (\ref{thm: BV Stokes for pushforward, (i)}) of Theorem \ref{thm: BV Stokes for pushforward}
and the fact that the map (\ref{e: BV pushforward in family}) is trivial in the first factor of (\ref{e:half-densities splitting in B and Y}) and hence commutes with $\Omega$. The proof of (\ref{thm: BV pushforward in family (ii)}) is a minor modification of the proof of (\ref{thm: BV Stokes for pushforward, (ii)}): choose a Berezinian $\mu$ on $\calY$ compatible with $\omega$. We can write $\xi=\mu^{\frac12}f$ for some $f\in \HDens(\calB)\widehat\otimes C^\infty(\calY)$. Repeating the calculation (\ref{e:BV pushforward infin symplectomorphism}) in the family setting we have
$$ \frac{\partial}{\partial t}\int_{\calL_t} \mu^{\frac12}f=\int_{\calL_t}\mu^{\frac12}(\Delta_{\mu} (fH_t)-
\Delta_{\mu}(f)
H_t)=
\hbar^{-2}(\hbar^2\Delta'+\Omega)\int_{\calL_t}\xi H_t. $$
Here we used that $\Delta_\mu (f)=\mu^{-\frac12}\Delta\xi=-\hbar^{-2}\mu^{-\frac12}\Omega\xi$.
\end{proof}

\subsubsection{
Case of exponential half-densities and  asymptotics $\hbar\rightarrow 0$}
\label{sec: BV pushforward, case of exp}
Now consider the case when $\calF$ is equipped with a Berezinian $\mathfrak{m}=\mu\cdot \nu$ where $\mu$ is a Berezinian on $\calY$ compatible with $\omega$ and $\nu$ is a Berezinian on $\calB$ (we do not require any compatibility between $\nu$ and $\Omega$), and consider half-densities on $\calF$ of the form
\begin{equation}\label{e:xi=exp}
\xi=\mathfrak{m}^{\frac12}\EE^{\frac{\ii}{\hbar}\calS}
\end{equation}
with $\calS=\calS^0+\hbar \calS^1+\cdots \in C^\infty(\calF)[[\hbar]]$. Using Berezinians $\mu$, $\nu$, we define the BV Laplacian $\Delta_\mu=\mu^{-\frac12}\Delta(\mu^{\frac12}\bullet)$ on $C^\infty(\calY)$ and the coboundary operator $\Omega_\nu=\nu^{-\frac12}\Omega(\nu^{\frac12}\bullet)$ on $C^\infty(\calB)$. Assume that $\Omega_\nu=\sum_{p\geq 0}(-\ii\hbar)^p \Omega_{(p)}$ where $\Omega_{(p)}=\Omega_{(p)}^0+\hbar \Omega_{(p)}^1+\cdots \in \mathrm{Diff}(\calB)[[\hbar]]$ is a differential operator on $\calB$ of order at most $p$. Denote by $\mathrm{Symb}\,\Omega^0_{(p)}\in \Gamma(\calB,S^p T\calB)$ the leading symbol of $\Omega^0_{(p)}$, and set $\mathrm{Symb}\,\Omega^0=\sum_{p\geq 0}\mathrm{Symb}\,\Omega^0_{(p)}$. Viewing $\mathrm{Symb}\, \Omega^0$ as a function on $T^*\calB$, we can  define a function $\mathrm{Symb}\, \Omega^0\circ \delta_\calB \calS^0\in C^\infty(\calF)$  where $\delta_\calB$ is the de Rham differential on $\calB$. Then the 
\textsf{modified quantum master equation} (\textsf{mQME})
\begin{equation}\label{e:mQME for half-density}
(\hbar^2 \Delta+\Omega)\;\mathfrak{m}^{\frac12}\EE^{\frac{\ii}{\hbar}\calS}=0
\end{equation}
can be expanded, as $\hbar\rightarrow 0$, as
\begin{equation}\label{e:mQME 0 order}
\mathfrak{m}^{\frac12} \left(-\frac12 (\calS^0,\calS^0)	+\mathrm{Symb}\, \Omega^0\circ \delta_\calB \calS^0+\mathcal{O}(\hbar)\right)\EE^{\frac{\ii}{\hbar}\calS}=0.
\end{equation}
If $b^\alpha$ are local coordinates on $\calB$, one has
$$\Omega_\nu=\sum_{p\geq 0} (-\ii\hbar)^p \underbrace{\frac{1}{p!}\sum_{\alpha_1,\ldots,\alpha_p}  \Omega^{\alpha_1\cdots\alpha_p}(b;\hbar)\;\frac{\partial}{\partial b^{\alpha_1}}\cdots \frac{\partial}{\partial b^{\alpha_p}}}_{\Omega_{(p)}\in \mathrm{Diff}(\calB)[[\hbar]]}.$$
Then (\ref{e:mQME 0 order}) gives, in the lowest order in $\hbar$, the equation
\begin{equation}\label{e:CME family}
\frac{1}{2}(\calS^0,\calS^0)-
\sum_{p\geq 0} \frac{1}{p!}\sum_{\alpha_1,\ldots,\alpha_p} \Omega^{\alpha_1\cdots\alpha_p}(b;0)\;\frac{\partial \calS^0}{\partial b^{\alpha_1}}\cdots \frac{\partial \calS^0}{\partial b^{\alpha_p}}
=0.
\end{equation}
This equation 
is the replacement of the classical master equation (\ref{e:CME}) in the family setting.

\begin{Rem}
Note that the Poisson bracket $(,)$ on $\calY$ and the symbol $\mathrm{Symb}\,\Omega^0$ do not depend on the choice of Berezinians $\mu$, $\nu$. Thus, equation (\ref{e:CME family}) is also independent of Berezinians.	
\end{Rem}

In analogy with (\ref{e:can transf}), we say that two solutions $\calS_0$, $\calS_1$ of the mQME (\ref{e:mQME for half-density}) differ by a canonical BV transformation, if
\begin{equation}\label{e:can transf family}
\EE^{\frac\ii\hbar \calS_1}-\EE^{\frac\ii\hbar \calS_0}=(\hbar^2\Delta_\mu+\Omega_\nu) \Psi
\end{equation}
for some $\Psi\in C^\infty(\calF)_{-1}$. This is equivalent to having a family $\calS_t$ of solutions of the mQME for $t\in [0,1]$, satisfying
\begin{equation}\label{e:can transf family inf}
\frac{\de}{\de t} \EE^{\frac{\ii}{\hbar}\calS_t}=(\hbar^2\Delta_\mu+\Omega_\nu)\left(\hbar^{-2}\EE^{\frac\ii\hbar \calS_t}\phi_t\right)
\end{equation}
with $\phi_t\in C^\infty(\calF)_{-1}[[\hbar]]$. Note that equation (\ref{e:can transf family inf}) together with the mQME can be packaged into an extended version of the mQME satisfied by $\calS_t+dt\cdot \phi_t
$ viewed as a non-homogeneous differential form on the interval $[0,1]$ with values in functions on $\calF$:
$$
\left(\hbar^2 dt\;\frac{\de}{\de t}+\hbar^2\Delta_\mu+\Omega_\nu\right)\EE^{\frac\ii\hbar (\calS_t+dt\cdot\phi_t)}=0.
$$
In the lowest order in $\hbar$, equation (\ref{e:can transf family inf}) reads
\begin{equation}\label{e:can transf family class}
\frac{\de}{\de t} \calS_t^0= (\calS_t^0,\phi_t^0)-\sum_{p\geq 0}\frac{1}{p!}\sum_{\alpha_1,\ldots,\alpha_p,\beta} \Omega^{\alpha_1\cdots\alpha_p\beta}(b;0)\;\frac{\de \calS_t^0}{\de b^{\alpha_1}}\cdots \frac{\de \calS_t^0}{\de b^{\alpha_p}}\; \frac{\de \phi_t^0}{\de b^{\beta}}.
\end{equation}
Here $\phi_t^0=\phi_t\bmod\hbar$.

\begin{Rem}
One can introduce a sequence of multi-derivations with $p\geq 0$ inputs,
\begin{align}\label{e:Omega-brackets}
&[\bullet,\cdots,\bullet]_{\Omega}:\quad  \underbrace{C^\infty(\calB)\times \cdots\times C^\infty(\calB)}_p\rightarrow  C^\infty(\calB), \\
\nonumber &[f_1,\ldots,f_p]_{\Omega}= \sum_{\alpha_1,\ldots,\alpha_p} \Omega^{\alpha_1\cdots\alpha_p}(b;0)\; \frac{\de f_1}{\de b^{\alpha_1}}\cdots \frac{\de f_p}{\de b^{\alpha_p}},
\end{align}
generated by the symbols $\mathrm{Symb}\,\Omega^0_{(p)}$.\footnote{
As a consequence of $\Omega^2=0$, the operations (\ref{e:Omega-brackets}) define on $C^\infty(\calB)[-1]$ the structure of a curved $L_\infty$ algebra (which is flat if $\Omega^0_{(0)}=0$).
} Then equations (\ref{e:CME family},\ref{e:can transf family class}) can be written, respectively, as
\begin{align*}
\frac12 (\calS^0,\calS^0)-\sum_{p\geq 0}\frac{1}{p!}\; [\underbrace{\calS^0,\ldots,\calS^0}_p]_\Omega = 0,\\
\frac{\de}{\de t}\;\calS^0_t= (\calS^0_t,\phi^0_t)-\sum_{p\geq 0}\frac{1}{p!}\;[\underbrace{\calS^0_t,\ldots,\calS^0_t}_p,\phi_t^0]_\Omega.
\end{align*}
\end{Rem}

Assume again that $(\calY,\omega)$ is a product of odd-symplectic manifolds $(\calY',\omega')$ and $(\calY'',\omega'')$ and that the Berezinian  $\mu$ on $\calY$ is of form $\mu=\mu'\cdot\mu''$ with $\mu',\mu''$ compatible Berezinians on $\calY',\calY''$.

Given a half-density on $\calF$ of the form 
$\xi=\mathfrak{m}^{\frac12}\EE^{\frac\ii\hbar \calS}$, we can apply the pushforward construction (\ref{e: BV pushforward in family}), producing a half-density on $\calF'=\calB\times\calY'$ of form $\xi'=\mathfrak{m}'^{\frac12}\EE^{\frac\ii\hbar\calS'}$ where $\mathfrak{m}'=\mu'\cdot \nu$. The effective BV action $\calS'\in C^\infty(\calF')[[\hbar]]$ can be calculated by stationary phase formula for the integral
$$\EE^{\frac\ii\hbar\calS'}=\int_\calL \mu''^{\frac12}\; \EE^{\frac\ii\hbar \calS}.$$
Here we assume that there is a single simple isolated critical point of $\calS$ on $\calL$. The asymptotics $\hbar\rightarrow 0$ of the integral yields $\calS'$ as a formal power series in $\hbar$ with coefficients given by Feynman diagrams.

Corollary \ref{cor: BV pushforward for exp(S) properties} translates to the family setting in the following way.
\begin{Cor}
\begin{enumerate}[(i)]
\item 
The modified quantum master equation (\ref{e:mQME for half-density}) on $\calS$ implies the mQME on the effective BV action $\calS'$:
$$(\hbar^2\Delta'+\Omega)\; \mathfrak{m}'^{\frac12}\EE^{\frac\ii\hbar\calS'}=0.$$
In particular,  $\calS'^0$ satisfies equation (\ref{e:CME family}) with Poisson bracket $(,)$ on $\calY$ replaced by the one on $\calY'$.
\item If $\calL_t\subset \calY''$ is a family of Lagrangian submanifolds, the respective effective actions $\calS'_t$ satisfy equation (\ref{e:can transf family inf}) on $\calF'$, with the generator $\phi_t$ of the infinitesimal canonical BV transformation given  by (\ref{e:can tranf generator induced from Lagr homotopy}).
\item If $\calS$, $\Tilde\calS$ are solutions of the mQME on $\calF$ related by a canonical transformation (\ref{e:can transf family}), then the respective effective actions $\calS'$, $\Tilde\calS'$ are also related by a canonical transformation on $\calF'$, with generator given by the pushforward $\Psi'=\mathfrak{m'}^{-\frac12}\int_\calL \mathfrak{m}^{\frac12}\Psi$.
\end{enumerate}
\end{Cor}

\begin{Def}
\label{def: BV bundle}
We call a fiber bundle $\calF$ over a base $\calB$ with odd-symplectic fiber $(\calY,\omega)$ a {\sf BV bundle} if the transition functions  of $\calF$ are given by locally constant fiberwise symplectomorphisms.
\end{Def}

Throughout this Section, the direct product $\calF=\calB\times\calY$
can be replaced by a more general BV bundle.
For the family BV pushforward, we can allow $\calF$ to be a BV bundle over $\calB$ with fiber $\calY$ a hedgehog (cf. Remark \ref{rem: hedgehog}; recall that a hedgehog is the same as a BV bundle with an odd-symplectic base, satisfying the extra assumption that the transition functions are homotopic to the identity). 
In this case we have a tower of BV bundles $\calF\rightarrow \calF'\rightarrow \calB$.

\begin{Rem}\label{rem: family pushforward for Omega=0}
In the special case $\Omega=0$, Theorem \ref{thm: BV pushforward in family} holds in a more general setting where $\calF\rightarrow \calB$ is 
a general fiber bundle with fiber $\calY$ a hedgehog (i.e. no requirement on transition functions to be constant on $\calB$). The Lagrangian submanifold $\calL$ in this setting also does not have to be locally constant as a function on $\calB$.
\end{Rem}

\subsubsection{Half-densities on an elliptic complex}
\label{rem: half-densities on an elliptic complex}
For $X=(X^\bullet,d)$ a cochain complex, one can use the canonical isomorphism of determinant lines $\mathrm{Det}\,X^\bullet\cong \mathrm{Det}\,H^\bullet(X)$ to define the space of densities of weight $\alpha\in\mathbb{R}$ on $X$ as
\begin{equation}\label{e: Dens via DetH}
\mathrm{Dens}^\alpha(X)=C^\infty(X)\otimes (\mathrm{Det}\,H^\bullet(X)/\{\pm 1\})^{\otimes -\alpha}.
\end{equation}
Here the second factor represents positive, {\sf constant} (coordinate-independent) $\alpha$-densities on $X$.\footnote{In other words, an $\alpha$-density $\xi$ prescribes a number $\xi(x,\{\chi_i\})$ to an element $x\in X$ and a basis $\{\chi_i\}$ in cohomology $H^\bullet(X)$, in such a way that, for $\{\chi'_i\}$ another basis, related to $\{\chi_i\}$ by a linear transformation $\theta\in GL(H^\bullet(X))$, one has $\xi(x,\{\chi'_i\})=|\mathrm{Ber}\;\theta|^\alpha \xi(x,\{\chi_i\})$. Here $\mathrm{Ber}\;\theta\in\mathbb{R}$ is the Berezinian (superdeterminant) of the linear transformation.
}

In the case of infinite-dimensional elliptic complexes,
(\ref{e: Dens via DetH})
gives a definition of the space of densities,
which is suitable for the setting of field theory on compact manifolds. Here the typical $X$ is the de Rham complex of the space-time manifold tensored with some graded vector space of coefficients, $X=\Omega^\bullet(M)\otimes V$ (which corresponds to abelian $BF$ theory and its perturbations). 
In this case $C^\infty(X)$ in (\ref{e: Dens via DetH}) should be
understood as the space of smooth functions on $X$ in Fr\'echet sense.
In perturbative computations one typically encounters $\hbar$-dependent asymptotic families of functions on $X$ of the form
$$
f_\hbar=\EE^{\frac{\ii}{\hbar}\varphi}\cdot \rho,\quad\mbox{where}\quad \varphi
\in (\Hat S^\bullet X^*)_0 ,\quad
\rho=\rho^0+\hbar \rho^1+\cdots\in \Hat S^\bullet X^* [[\hbar]].
$$
Here $\varphi(\theta)=\sum_{n\geq 0}\int_{M^{n}}\Phi_n\wedge \pi_1^*\theta\wedge\cdots\wedge \pi_n^*\theta$ for $\theta\in X=\Omega^\bullet(M)\otimes V$ a test differential form, and likewise $\rho^j(\theta)=\sum_{n\geq 0}\int_{M^{n}} R^j_n\wedge \pi_1^*\theta\wedge\cdots\wedge \pi_n^*\theta$. Here $\Phi_n$, $R^j_n$ are distributional differential forms (de Rham currents) on $M^n=\underbrace{M\times\cdots\times M}_n$
with values in $S^n V^*$ and $\pi_i: M^n\rightarrow M$ is the projection to the $i$-th copy of $M$.

Note that the Reidemeister-Ray-Singer torsion $\tau(M)$ of $M$ provides a natural reference constant density on $X$ (in the sense above) and thus fixes an isomorphism
\begin{eqnarray*}
C^\infty(X) &\simeq & \mathrm{Dens}^\alpha(X)	\\
f & \mapsto & f\cdot {\tau(M)}^{-\alpha\cdot \mathrm{Sdim}(V)}
\end{eqnarray*}
where $\mathrm{Sdim}(V)=\sum_i (-1)^i \dim V^i$ is the superdimension of the space of coefficients. More generally, instead of $\Omega^\bullet(M)\otimes V$ one can allow $X$ to be the space of differential forms with coefficients in a flat graded vector bundle over $M$.\footnote{This is the case e.g. for perturbative Chern-Simons theory evaluated around a non-trivial flat connection, see \cite{AS} and Remark \ref{rem: globalization} below. The bundle in this case is $\mathrm{ad}(P)[1]$ -- the adjoint of the principal $G$-bundle $P$ carrying the flat connection, with a homological degree shift by $1$.}

\subsection{The quantum BV-BFV formalism}\label{ss:qBVBFV}
The goal of this Section is to propose the definition of perturbative quantum BV-BFV theory. 

Given a classical BV-BFV theory, its
perturbative quantization consists of the following data:
\begin{enumerate}
\item A graded vector space $\calH_\Sigma^\calP$, the {\sf space of states}, associated to each $(d-1)$\ndash manifold $\Sigma$ with a choice of polarization $\calP$ on $\calF^\de_\Sigma$ (to be constructed as a geometric quantization\footnote{Under the assumption that the 1-form $\alpha^\de_\Sigma$ vanishes along $\calP$, the space of states is (a suitable model for) the space of functions on $\calF^\de_\Sigma$ constant in $\calP$-directions. Furthermore, in the case of $\calP$ a real fibrating polarization, the space of states can be identified with the space of functions on the quotient (space of leaves) $\calB_\Sigma^\calP=\calF^\de_\Sigma/\calP$.  
A correction to this picture is that, instead of functions on $\calB^\calP_\Sigma$, we should consider half-densities on $\calB^\calP_\Sigma$, i.e. $\calH^\calP_\Sigma=\HDens(\calB^\calP_\Sigma)$.
More generally, the space of states is the space of $\calP$-horizontal sections of the trivial (since we consider an {\sf exact} boundary BFV theory) prequantum line bundle $L$ over $\calF^\de_\Sigma$, with global connection 1-form $\frac{\ii}{\hbar}\alpha^\de_\Sigma$, tensored with the appropriate bundle of half-densities (see e.g. \cite{BatesWeinstein}).
}
of the symplectic manifold $\calF^\de_\Sigma$).
\item A coboundary operator $\Omega^\calP_\Sigma$ on  $\calH_\Sigma^\calP$, the {\sf quantum BFV operator}, 
    which is a quantization of
    $\calS^\de_\Sigma$.
\item \label{qBVBFV data: residual fields} A finite-dimensional graded manifold $\calV_M$ endowed with a degree $-1$ symplectic form -- the {\sf space of residual fields} -- associated to a $d$\ndash manifold $M$ and a polarization $\calP$ on
$\calF^\de_{\de M}$. We define the graded vector space $\Hat\calH_M^\calP=\calH^\calP_{\de M}\Hat\otimes \HDens(\calV_M)$ endowed with two commuting coboundary operators  $\Hat\Omega^\calP_M=\Omega^\calP_{\de M}\otimes \mathrm{id}$ and $\Hat\Delta^\calP_M=\mathrm{id}\otimes \Delta_{\calV_M}$. Here $\Delta_{\calV_M}$ is the canonical BV Laplacian on half-densities on residual fields.\footnote{In our notational system, objects depending on
{\backgrounds} are decorated with hats.}
\item A state\footnote{As we will presently see, the state $\Hat\psi_M$ is not uniquely defined as it depends on the additional choice of a ``gauge fixing''.}
 $\Hat\psi_M\in\Hat\calH^\calP_M$ which satisfies the \textsf{modified quantum master equation (mQME)}
\begin{equation}\label{e:mQME}
(\hbar^2\Hat\Delta^\calP_M+\Hat\Omega^\calP_M)\Hat\psi_M=0,
\end{equation}
which is the quantum version of \eqref{e:mCME}. 
\end{enumerate}

\begin{Rem}\label{rem: hat H}
The space $\Hat\calH_M^\calP$ results from a partial integration of bulk fields. 
Hence one can think of $\Hat\calH_M^\calP$ as of the space of boundary states with values in
half-densities on the space of {\backgrounds} $\calV_M$. In the case of a real fibrating polarization on the boundary, we have a trivial bundle of
{\backgrounds}
$\calZ_M= \calB^\calP_{\de M}\times\bulkzm \ra \calB^\calP_{\de M}$ with fiber $\bulkzm$. One has then $\Hat\calH_M^\calP=\HDens(\calZ_M)= \calH_{\de M}^\calP \Hat\otimes \HDens(\bulkzm)$. Note that the triviality of the bundle $\calZ_M$ is implicitly built into the data (\ref{qBVBFV data: residual fields}) above.
\end{Rem}

\begin{Rem}[Change of data]\label{r:psidot}
The coboundary operator 
$\Omega^\calP_{\de M}$ and the state $\Hat\psi_M$ are not uniquely defined, but are allowed to change, infinitesimally, as follows\footnote{
The ambiguity stems from the freedom to choose different gauge-fixing Lagrangians in fiber BV integrals which produce the coboundary operators and the state. 
}
\begin{align*}
\frac\dd{\dd t}{\Omega^\calP_{\de M}}&=[\Omega^\calP_{\de M},\tau],\\
\frac\dd{\dd t}\Hat\psi_M&=(\hbar^2\Hat\Delta^\calP_M+\Hat\Omega^\calP_M)\Hat\chi-\Hat\tau\Hat\psi_M,
\end{align*}
where $\Hat\chi$ is an element of $\Hat\calH_M^\calP$,
$\tau$ is an operator on $\calH_{\de M}^\calP$ and $\Hat\tau=\tau\otimes\mathrm{id}$ is its extension to $\Hat\calH_M^\calP$.
\end{Rem}

\begin{Def}\label{r:equivhat}
We say that the space $\Hat\calH_M^\calP$ is {\sf equivalent} to $\Tilde\calH_M^\calP$ if there is a quasi-isomorphism 
of bi-complexes
$I\colon (\Hat\calH_M^\calP,\Hat\Delta^\calP_M,\Hat\Omega^\calP_M)\to (\Tilde\calH_M^\calP,\Tilde\Delta^\calP_M,\Tilde\Omega^\calP_M)$.
\end{Def}

\begin{Rem}
If $\calV_M$ is a point (and thus $\Hat\calH_M^\calP=\calH_M^\calP$ and $\Hat\Delta^\calP_M=0$), we call $\psi_M=\Hat\psi_M$ the boundary state. It satisfies
$\Omega^\calP_{\de M}\psi_M=0$. Its $\Omega^\calP_{\de M}$\ndash cohomology class is called the physical state.
\end{Rem}

An example where this program  has been successfully completed is the one\ndash dimensional Chern--Simons theory \cite{AM}.
Several other examples are presented in the rest of this paper.

\begin{Rem} For $M$ a closed manifold, the boundary space of states is $\calH_{\de M=\varnothing}^\calP=\mathbb{C}$. In this case the state $\Hat \psi_M=\EE^{\frac{\ii}{\hbar}\calS_\mathrm{eff}}$ is the exponential of the BV effective action induced on the space of {\backgrounds} (see \cite{CFvacua,CMcs,AM,BCM} for examples).\footnote{More pedantically, one should write $\Hat\psi_M=\mu^{\frac12}_{\calV_M}\cdot\EE^{\frac\ii\hbar \calS_\mathrm{eff}}$ with $\mu^{\frac12}_{\calV_M}$ a reference half-density on $\calV_M$.}
If additionally there are no {\backgrounds}, i.e. $\Hat\calH_M^\calP=\calH_{\de M}^\calP=\mathbb{C}$, then the state $\Hat\psi_M=\psi_M\in \mathbb{C}$ is the usual partition function.
\end{Rem}

\subsection{Perturbative quantization of classical BV-BFV theories}\label{ss:pertq}
In this Section we outline a quantization scheme which produces a realization of quantum BV-BFV formalism of Section \ref{ss:qBVBFV} out of the data of a classical BV-BFV theory.

In this Section we appeal to the intuition of the finite-dimensional setting. The following discussion is absolutely correct in the finite-dimensional case and provides a motivating construction for the infinite-dimensional case where the formal reasoning has to be checked, e.g., at the level of Feynman diagrams.
Concrete examples will be presented in Sections \ref{s:abeBF} and \ref{sec: BF-like}.

\subsubsection{From classical to quantum modified master equation}
For the purposes of this paper, it is enough to consider the special situation where the polarization $\calP$
is given by a Lagrangian foliation  
with smooth leaf space, denoted by 
$\calB^\calP_\Sigma$, 
and with the property that, for an appropriately chosen local functional $f^\calP_\Sigma$, the restriction of the $1$\ndash form
$\alpha^{\de,\calP}_\Sigma :=\alpha^\de_\Sigma-\delta f^\calP_\Sigma$ to the fibers of $\calP$ vanishes (see Section~\ref{s:pol}).
In this case,
$\calH_\Sigma^\calP$ may be identified, via multiplication by $\EE^{\frac\ii\hbar f^\calP_\Sigma}$,
with the space of 
half-densities
on $\calB^\calP_\Sigma$.

Next we assume that $\Sigma$ is the boundary $\de M$ of $M$.
Notice that we may change the BV action $\calS_M$ to
\[
\calS_M^\calP:=\calS_M + \pi_M^*f^\calP_{\de M},
\]
This way we get a new BV-BFV manifold (simply  replacing
$\calS_M$ and $\alpha^\de_{\de M}$ by $\calS_M^\calP$ and $\alpha^{\de,\calP}_\Sigma$). In particular, we still have
the fundamental BV-BFV equation 
\begin{equation}\label{e:fundP}
\iota_{Q_M}\omega_M = \delta\calS_M^\calP + \pi_M^*\alpha^{\de,\calP}_{\de M}
\end{equation}
and the mCME
\begin{equation}\label{e:mCME'}
\frac{1}{2}\,\iota_Q \iota_Q \omega_M = \pi_M^*\calS^\de_{\de M}.
\end{equation}
Denoting by $p^\calP_{\de M}$ the projection $\calF^\de_{\de M}\to\calB^\calP_{\de M}$, we have a surjective submersion
\be p^\calP_{\de M}\circ\pi_M\colon\quad \calF_M\to\calB^\calP_{\de M}. \label{e: p pi: F to B}\ee
We now assume that we have a section so that we can write
\begin{equation}\label{e:F=YxB}
\calF_M=\calB^\calP_{\de M}\times \calY
\end{equation}
(we actually need this only locally; more generally, we could allow $\calF_M$ to be a BV bundle over $\calB^\calP_{\de M}$, cf. Definition \ref{def: BV bundle}).
\begin{Ass}\label{assump: omega compatible with splitting}
We assume that, in the splitting (\ref{e:F=YxB}), $\omega_M$ is a weakly nondegenerate $2$\ndash form on $\calY$
extended to the product $\calB^\calP_{\de M}\times\calY$.\footnote{
In the setting of local field theory this assumption forces one to choose a
section
$\calB_{\de M}^\calP\rightarrow \calF_M$ of (\ref{e: p pi: F to B}) which extends boundary fields by zero in the bulk, see Remark \ref{rem: extension} below.
}
\end{Ass}
There is no contradiction between this assumption and
$\omega_M$ being weakly nondegenerate on the whole space $\calF_M$ (in the finite-dimensional setting, instead,
the BV-BFV formalism is not consistent with nondegeneracy of $\omega$ on the whole space and one precisely has to assume nondegeneracy along the fibers). We may then write $Q_M=Q_\calY+Q_\calB$ (the decomposition induced by the splitting of the tangent bundle $T\calF_M=T_\calY\calF_M\oplus T_\calB \calF_M$) and $\delta=\delta_\calY+\delta_\calB$.
The fundamental equation \eqref{e:fundP} now splits into two equations:
\begin{subequations}
\begin{align}
\delta_\calY\calS_M^\calP &= \iota_{Q_\calY}\omega_M,\label{e:mCMEa}\\
\delta_\calB\calS_M^\calP&=-\pi_M^*\alpha^{\de,\calP}_{\de M}.\label{e:mCMEb}
\end{align}
\end{subequations}
The first equation implies
$\iota_{Q_\calY} \iota_{Q_\calY}\omega_M=Q_\calY\calS_M^\calP=:(\calS_M^\calP,\calS_M^\calP)$ (on the r.h.s. is the fiberwise BV bracket, defined using the odd-symplectic structure on $\calY$-fiber).
By \eqref{e:mCME'}, which now reads $\frac12\iota_{Q_\calY} \iota_{Q_\calY}\omega_M=\pi_M^*\calS^\de_{\de M}$,
we then have
\begin{equation}\label{e:mCMEf}
\frac12\, (\calS_M^\calP,\calS_M^\calP)=\pi_M^*\calS^\de_{\de M},
\end{equation}
which is the fiberwise version of the modified classical master equation.

To interpret \eqref{e:mCMEb}, we assume we have Darboux coordinates $(b^i,p_i)$ for $\omega^\de_{\de M}$,
where
the $b^i$'s are coordinates on $\calB^\calP_{\de M}$ and the $p_i$'s are coordinates on 
the fiber of $p^\calP_{\de M}\colon\; \calF^\de_{\de M}\rightarrow \calB^\calP_{\de M}$
(which is part of $\calY$),
such that $\alpha^{\de,\calP}_{\de M}=-\sum_i p_i\delta b^i$ (indices may also denote ``continuous'' coordinates here). Then we
have
\begin{equation}\label{e:bp}
\frac\de{\de b_i}\calS_M^\calP = p_i.
\end{equation}
In the infinite-dimensional case, partial derivatives here should be replaced by variational derivatives.
This in particular shows that in a splitting with these properties $\calS_M^\calP$ is linear in the $b_i$\ndash coordinates.
It follows that, if we define $\Omega_{\de M}^\calP$ as the standard ordering quantization of $\calS^\de_{\de M}$, obtained by
replacing each $p_i$ by $-\ii\hbar\frac\de{\de b^i}$,
\begin{equation}\label{e:Omegastandardordering}
\Omega_{\de M}^\calP := \calS^\de_{\de M}\left(b,-\ii\hbar\frac\de{\de b}\right),
\end{equation}
and putting all derivatives to the right, we get
\begin{equation}\label{e:Omegaunp}
\Omega_{\de M}^\calP\EE^{\frac\ii\hbar\calS_M^\calP}=\pi_M^*\calS^\de_{\de M}\cdot\EE^{\frac\ii\hbar\calS_M^\calP}.
\end{equation}

We now assume that $\calY$ has a compatible Berezinian (in the infinite-dimensional case this is formal), so we can define
the BV Laplacian $\Delta$. As usual we have
\[
\Delta \EE^{\frac\ii\hbar\calS_M^\calP} =
\left(
\frac\ii\hbar\Delta \calS_M^\calP + \frac12 \left(\frac\ii\hbar\right)^2(\calS_M^\calP,\calS_M^\calP)
\right)\EE^{\frac\ii\hbar\calS_M^\calP}.
\]
If $\Delta \calS_M^\calP=0$, as is usually assumed, then \eqref{e:mCMEf} and \eqref{e:Omegaunp} imply
the modified quantum master equation
\begin{equation}\label{e:mQMEP}
(\hbar^2\Delta + \Omega_{\de M}^\calP)\;\EE^{\frac\ii\hbar\calS_M^\calP} = 0.
\end{equation}
\begin{Rem}\label{r:mQMEnoDeltazero}
If $\calS_M^\calP$ depends on $\hbar$ and/or $\Delta \calS_M^\calP\not=0$, from the assumption that the QME holds in the
bulk, we get the modified quantum master equation anyway by defining a new boundary action $\Tilde\calS^\de_{\de M}=\calS^\de_{\de M}+O(\hbar)$ via
\[
\pi_M^*\Tilde\calS^\de_{\de M} = \frac12 (\calS_M^\calP,\calS_M^\calP)
-\ii\hbar\Delta \calS_M^\calP
\]
and setting $\Omega_{\de M}^\calP$ to be the standard ordering quantization of $\Tilde\calS^\de_{\de M}$.
\end{Rem}
\begin{Rem}
To make sense of the interpretation of physical states as the cohomology in degree zero of
the operator $\hbar^2\Delta + \Omega_{\de M}^\calP$, we have to assume that it is a coboundary operator.
This is equivalent to the requirement $(\Omega_{\de M}^\calP)^2=0$. If this is not the case, one might still
try to correct $\Omega_{\de M}^\calP$ (and $\calS_M^\calP$)
 with higher order terms in $\hbar$ so as to make it square to zero. There may be cohomological obstructions
 (anomalies) to do that.
 \end{Rem}
As a consequence of the two previous remarks, $\Omega_{\de M}^\calP$ is a quantization
of $\calS^\de_{\de M}$ but not necessarily the one obtained by standard ordering. In Section~\ref{sec: BF-like} we will actually
see examples (notably the Poisson sigma model) where this phenomenon occurs (as $\Delta \calS_M^\calP=0$
is not compatible with the regularization).


\begin{Rem} Using the coordinate reference half-density $\nu^{\frac12}=\prod_i|db^i|^{\frac12}$ on $\calB^\calP_{\de M}$, we can identify
$C^\infty(\calB^\calP_{\de M})\stackrel{\cdot \nu^{\frac12}}{\simeq} \HDens(\calB^\calP_{\de M})$
and thus allow the operator $\Omega^\calP_{\de M}$ to act on half-densities on $\calB^\calP_{\de M}$. Then we can write the equivalent half-density version of (\ref{e:mQMEP}):
\begin{equation}\label{e:mQMEP 1/2dens}
(\hbar^2\Delta + \Omega_{\de M}^\calP)\;\mathfrak{m}^{\frac12}\EE^{\frac\ii\hbar\calS_M^\calP}= 0,
\end{equation}
where $\mathfrak{m}^{\frac12}=\mu^{\frac12}\cdot \nu^{\frac12}$ is the reference half-density on $\calF$ comprised of $\nu^{\frac12}$ and half-density $\mu^{\frac12}$ on $\calY$ corresponding to the chosen Berezinian on $\calY$; $\Delta$ in (\ref{e:mQMEP 1/2dens}) is the canonical BV operator on half-densities on $\calY$.
\end{Rem}

\begin{Rem}
In the setting of local quantum field theory,
the 
modified quantum master equation (\ref{e:mQMEP}) is formal and requires a regularization (in particular
higher order functional derivatives have to be regularized).
However, in some examples (see \cite{AM},\cite{cell_ab_BF}) one can replace the continuum theory by a cellular model, with finite-dimensional space of fields, where equation \eqref{e:mQMEP} holds directly.
\footnote{One can indeed say that the discretization is the regularization here. An important point in the cellular examples of \cite{AM},\cite{cell_ab_BF} is that a cellular aggregation (the inverse of subdivision) corresponds to a fiber BV integral, and therefore these discretizations are {\sf exact}: one does not have to take an asymptotical subdivision with mesh tending to zero to recover the state/partition function of the theory as a limit -- any cellular structure on the space-time manifold gives the correct result outright.}
\end{Rem}

\subsubsection{The state}
The state is now produced by a perturbative BV pushforward in a family over $\calB^\calP_{\de M}$. For this we have to assume that $\calY\to\calV_M$ is a hedgehog, where
$\calV_M$ denotes the {\sf space of {\backgrounds}}, which we assume to be finite-dimensional. For simplicity of notations, and also because this is the case in all the examples we discuss in this paper, we assume that actually $\calY= \calV_M\times \calY''$ and
$\calF_M=\calB^\calP_{\de M}\times \calV_M\times \calY''$
The gauge fixing then consists in choosing a Lagrangian submanifold $\calL$ in $\calY''$.
We set $\calZ_M=\calB^\calP_{\de M}\times \calV_M$ (the {\sf bundle of {\backgrounds}} over $\calB^\calP_{\de M}$)
and denote $\Tilde\calZ_M=\calZ_M\times \calL$.
We define the space $\widehat \calH^\calP_M=\HDens(\calZ_M)=\HDens(\calB^\calP_{\de M})\Hat\otimes \HDens(\calV_M)$ and the BV Laplacian $\widehat\Delta^\calP_M=\mathrm{id}\otimes\Delta_{\calV_M}$, as in Remark \ref{rem: hat H}.

\begin{Ass}\label{assump: isolated crit points} For any $\phi\in \calZ_M$, the restriction of the action $\calS_M^\calP$ to $\calL_\phi=\{\phi\}\times\calL$ has isolated critical points on $\calL_\phi$.
\end{Ass}

We finally define the state $\Hat\psi_M$ as the
perturbative (Feynman diagram) computation of the family BV pushforward from $\calF_M$ to $\calZ_M$:
\begin{equation}\label{e:psi}
\Hat\psi_M(\phi) = \EE^{\frac\ii\hbar f^\calP_{\de M}}\int_{\calL} \EE^{\frac\ii\hbar\calS_M}=
\int_{\calL} \EE^{\frac\ii\hbar\calS_M^\calP},\qquad \phi\in\calZ_M.
\end{equation}
In the finite-dimensional setting, it now follows from the preceding discussion that $\Hat\psi_M$ solves the modified QME (\ref{e:mQME}):
\[
(\hbar^2\Delta_{\calV_M} + \Omega_{\de M}^\calP)\;\Hat\psi_M =0.
\]
In the infinite-dimensional setting, where integration is replaced by Feynman diagram computations, this equation is only expected to hold  
and requires an independent proof.

\begin{Rem}
\label{rem:poset of realizations}
If $\calV'_M$ is a different choice of the space of {\backgrounds} and $\calV_M$ fibers over $\calV'_M$ as a hedgehog, then 
$\calZ'_M=\calB^\calP_{\de M}\times \calV'_M$ is a BV subbundle of $\calZ_M$ and the corresponding quantum BV-BFV theories are equivalent in the sense of Definition \ref{r:equivhat}, with the map $I$ given by the BV pushforward from $\calZ_M$ to $\calZ'_M$ (in a family over $\calB^\calP_{\de M}$). Generally, one can have a partially ordered set of realizations of the space of {\backgrounds}, with partial order given by hedgehog fibrations acting on states by BV pushforwards (cf.~the setting of cellular $BF$ theory of \cite{cell_ab_BF} where one can vary cell decompositions $T$ in the bulk while keeping the decomposition on the boundary $T_\de$ unchanged; different $T$'s correspond to different choices of the space of {\backgrounds} $\calV_{M,T}$; cellular aggregations $T\rightarrow T'$ correspond to hedgehog fibrations/BV pushforwards). The poset of realizations has a minimal (final) object, corresponding to the minimal choice of the space of {\backgrounds} $\calV^\mathrm{min}_M$ for which Assumption \ref{assump: isolated crit points} can be satisfied by a judicious choice of $\calL$. In the case of abelian $BF$ theory, $\calV^\mathrm{min}_M$ is expressed in terms of de Rham cohomology of $M$, see Section \ref{sec: abBF backgrounds}.
\end{Rem}

\begin{Rem}\label{rem: extension}
In the typical situation of local field theory, we have $\calF_M=\Gamma(M,E)$, $\calB_{\de M}^{\calP}=\Gamma(\de M,E')$ -- spaces of smooth sections of graded vector bundles $E,E'$ over $M$, $\de M$, respectively, with the odd-symplectic structure given by $\omega_M=\int_M \langle \delta x, \delta x\rangle$. Here $\langle,\rangle$ is a fiberwise inner product on $E$ with values in densities on $M$.
Assumption \ref{assump: omega compatible with splitting} and equation \eqref{e:bp} imply that the extension of the boundary fields $\sigma: \calB_{\de M}^\calP\rightarrow \calF_M$ has been done by discontinuously extending them by zero outside the boundary.\footnote{
One can write the action for a general extension and make sure, by integrating by parts, that no derivative of the extension appears in the action (this is certainly possible if the theory is written in the first order formalism). Then we see that the discontinuous extension by zero is enforced by (\ref{e:bp}).
}
A more formal way consists in 
choosing a sequence of regular extensions $\sigma_n:\; \calB_{\de M}^{\calP}\rightarrow \calF_M$ that converges to the discontinuous one as $n\rightarrow\infty$. Each element of this sequence defines a state $\Hat\psi_n$
that in general will not satisfy the mQME.\footnote{Notice that the choice of a good splitting, compatible with Assumption \ref{assump: omega compatible with splitting} and leading to
\eqref{e:bp}, is a sufficient but not necessary condition for the formalism to work. For example, if we treat by BV a theory with no symmetries, then $\Omega_{\de M}^\calP$ will be zero, which puts us in the setting of Remark \ref{rem: family pushforward for Omega=0}. A change of extension is equivalent to a 
$\calB^\calP_{\de M}$-dependent translation
on the space of bulk fields $\calY$, and, in particular, mQME for a good splitting implies mQME for arbitrary splitting.}
\end{Rem}

\begin{Rem}\label{r:S_0+S_pert}
In many examples the action has the form
\begin{equation}\label{e:S_0+S_pert}
\calS=\calS_0+\calS_\mr{pert},
\end{equation}
a sum of a ``free'' (quadratic) part and a ``perturbation''. The splitting carries over to the cohomological vector field and the boundary BFV action. Then a choice of gauge-fixing data for the free theory can also be used for the perturbed theory with action (\ref{e:S_0+S_pert}), under certain ``smallness'' assumption on the perturbation. E.g. 
one can scale the perturbation $\calS_\mr{pert}$ with a parameter $\epsilon$ and calculate the path integral (\ref{e:psi}) by perturbation theory in $\epsilon$, instead of looking for $\epsilon$-dependent critical points of the perturbed action and calculating their stationary phase contributions as series in $\hbar$. For example, the Poisson sigma model is a perturbation of the 2-dimensional abelian $BF$ theory, and one can use the gauge-fixing for the latter to define the perturbation theory (cf. e.g. \cite{CFdq}). Likewise, one can use gauge-fixing for abelian Chern-Simons theory to define the perturbation theory for the non-abelian Chern-Simons (cf. e.g. \cite{CMcs}).
In this context, one first considers \eqref{e:mQMEP} for the free theory. The functional integral \eqref{e:psi}
for $\calS_0$ defines the unperturbed state $\Hat\psi_{M,0}(\phi)$ which satisfies the mQME for the operator
 $\Omega^\calP_{\de M,0}$. One then computes the state $\Hat\psi_M(\phi)$ for the whole theory perturbatively and looks
 for a deformation $\Omega_{\de M}^\calP$ of $\Omega^\calP_{\de M,0}$ so that the mQME is satisfied. The further
 condition that this deformation squares to zero must be checked separately, and there might be obstructions for it to
 be satisfied.
\end{Rem}

\begin{Rem} In the case of Chern-Simons theory with gauge group $G$ on a closed 3-manifold $M$, the gauge-fixing of Remark \ref{r:S_0+S_pert} corresponds to choosing a Riemannian metric on $M$. The metric induces the Hodge--de Rham decomposition of differential forms into exact, harmonic and $\dd^*$-exact (coexact) forms. We set $\Tilde \calZ_M=\Omega^\bullet_\mathrm{coclosed}(M,\mathfrak{g})[1]$ with $\mathfrak{g}$ the Lie algebra of $G$. Then $\calZ_M=H^\bullet(M,\mathfrak{g})[1]$, the $\mathfrak{g}$-valued de Rham cohomology of $M$ represented by harmonic forms. For every sufficiently small harmonic 1-form $a_\mathrm{harm}$, there is an isolated critical point of the  Chern-Simons action on the subspace $a_\mathrm{harm}+\Omega^1_\mathrm{coexact}(M,\mathfrak{g})$. But only if $a_\mathrm{harm}$ satisfies the (homotopy) Maurer-Cartan equation on cohomology, the corresponding critical point will be a flat connection. We refer the reader to \cite{CMcs} for details.
\end{Rem}

\begin{Rem}\label{rem: globalization}
The framework described above assumes that one can introduce a {\sf global} gauge-fixing. A more general 
technique is to allow a family, parametrized by a  choice $x_0$ of ``background'' (or ``reference'') solution of the
Euler--Lagrange equations,  of local gauge-fixings, in a formal neighborhood of $x_0$ (e.g. one can have an $x_0$-dependent splitting (\ref{e:S_0+S_pert}) and infer the local gauge-fixing as in Remark \ref{r:S_0+S_pert}). This produces a family of ``local states'' --- a horizontal section of the vector bundle of local states over the base (the space of allowed
$x_0$'s), with respect to a version of the flat Grothendieck connection on the base. In this framework, the global state is this family. See \cite{BCM} for details on how this technology applies to the Poisson sigma model on a closed surface, where one has a family of gauge-fixings for fields in the neighborhood of a constant map to the Poisson manifold (thus the parameter of the family here is the value of the constant map). The treatment of non-abelian Chern-Simons theory by Axelrod-Singer \cite{AS} is also very much in this vein, where $x_0$ is the background flat connection. Since in \cite{AS} the background flat connection is assumed to be acyclic, there is no need for formal-geometric gluing with the Grothendieck connection, as the base of the family is a discrete set. See Appendix \ref{appendix: semi-classics and eff actions} for further discussion of the matter.
\end{Rem}

\subsubsection{Transversal polarizations}
A special case of gauge-fixing occurs when the polarization $\calP$ on $\calF^\de_{\de M}$ is transversal to the Lagrangian submanifod $\calL_M:=\pi_M(\calEL_M)$,
where $\calEL_M$ is the zero locus of $Q_M$ (the ``Euler--Lagrange space''). In this case, one may take $\calB^\calP_{\de M}=\calL_M$ and $\calZ_M=\underline{\calEL}_M$.\footnote{Here we consider the {\sf fiberwise} coisotropic reduction $\underline{\calEL}_M$ which is a symplectic fiber bundle over $\calL_M$. It is different from the full ``$Q$-reduction'' $\calEL_M/Q_M$ (which is a bundle over the reduction $\calL_M/Q_{\de M}$) and from the coisotropic reduction of the total space of $\calEL_M$ in $\calF_M$ (called the ``symplectic $\calEL$-moduli space'' in \cite{CMR}). The reduction $\underline{\calEL}_M$ can be seen as an appropriate BV extension of the space of gauge equivalence classes of solutions of equations of motion, with gauge transformations
acting trivially on boundary fields.
See \cite{CMR} for details.}
The fibers of $\underline{\calEL}_M$ are the moduli spaces of the vacua of the theory. Note that, by this construction, we have a preferred (``minimal'') choice of $\calZ_M$.

Despite having this preferred choice,
it is convenient to allow for more general $\calZ_M$'s as they are useful for gluing. Also, it is convenient to consider
polarizations that are not transversal to $\calL_M$, as we will see in the following.

\subsubsection{Gluing}
\label{sss:gluing}
If a $d$\ndash manifold $M$ with boundary is cut along a $(d-1)$\ndash submanifold $\Sigma$ into components $M_1$ and $M_2$ (i.e. $M=M_1\cup_\Sigma M_2$), then we can obtain
the state $\Hat\psi_M$ from the states $\Hat\psi_{M_1}$ and $\Hat\psi_{M_2}$.
The product of the spaces of {\backgrounds} $\calV_{M_1}\times \calV_{M_2}$ is a hedgehog fibration over $\calV_M$, and
the gluing formula has the structure
\begin{equation}\label{e: gluing formula}
\Hat\psi_M=P_* \left(\Hat\psi_{M_1}\underset\Sigma\ast \Hat\psi_{M_2}\right)
\end{equation}
where $\underset\Sigma\ast$ denotes the pairing in $\calH_\Sigma^\calP$ and $P_*$ stands for the BV pushforward corresponding to 
$P\colon\calV_{M_1}\times \calV_{M_2}\rightarrow \calV_M$.
Observe that \eqref{e:mQME} is automatically satisfied.

Also note that it is convenient to choose two different, transversal polarizations $\calP_1$ and $\calP_2$ to define the states $\Hat\psi_{M_1}$ and $\Hat\psi_{M_2}$. If we can realize
$\calF^\de_\Sigma$ as $\calB^{\calP_1}_\Sigma\times\calB^{\calP_2}_\Sigma$, then (for simplicity we ignore the distinction between functions and half-densities) the pairing is the integral
over $\calF^\de_\Sigma$ of the product of a function on $\calB^{\calP_1}_\Sigma$ times a function on $\calB^{\calP_2}_\Sigma$ times the Segal--Bargmann kernel $\EE^{\frac\ii\hbar(f^{\calP_2}_\Sigma-f^{\calP_1}_\Sigma)}$. The latter term may be used
to define the perturbative computation of the pairing.

To explain (\ref{e: gluing formula}), one can consider gluing at the level of exponentials of actions. For simplicity we assume that $\de M_1=\Sigma=(\de M_2)^\mathrm{opp}$ (i.e. the glued manifold $M$ is closed); the discussion generalizes straightforwardly to $M$ with boundary. Let $b^i,b'_i$ be Darboux coordinates on $\calF^\de_\Sigma$ such that the polarizations $\calP_1,\calP_2$ are spanned by vector fields $\frac{\de}{\de b'_i}$ and $\frac{\de}{\de b^i}$, respectively.
Thus the $b^i$ are coordinates on $\calB:=\calB_\Sigma^{\calP_1}$ and the $b'_i$ are coordinates on $\calB':=\calB_\Sigma^{\calP_2}$.
We assume additionally that $\alpha_{\Sigma}^{\de,\calP_1}=-\sum_i b'_i\, \delta b^i$ and $\alpha_{\Sigma}^{\de,\calP_2}=\sum_i b^i\, \delta b'_i$; then in these coordinates we have $\EE^{\frac\ii\hbar (f_\Sigma^{\calP_2}-f_\Sigma^{\calP_1})}=\EE^{-\frac\ii\hbar \langle b, b'\rangle}$. The spaces of fields decompose as
$\calF_{M_1}=\calY_1\times \calB_{(1)}=(\Tilde\calY_1\times \calB'_{(1)})\times \calB_{(1)}$,\quad $\calF_{M_2}=\calY_2\times \calB'_{(2)}=(\Tilde\calY_2\times \calB_{(2)})\times \calB'_{(2)}$ and
$\calF_M=\calY_1\times \calY_2$.
The subscripts $(1)$, $(2)$ are here to distinguish between the copies of $\calB$, $\calB'$ appearing in $\calF_{M_1}$ and $\calF_{M_2}$. Then we have the identity
\begin{equation} \label{e:gluing for exp(S)}
\EE^{\frac\ii\hbar \calS^{\calP_1}_{M_1}}\mathfrak{m}_1^{\frac12}\;\underset\Sigma\ast\;
\EE^{\frac\ii\hbar \calS^{\calP_2}_{M_2}}\mathfrak{m}_2^{\frac12}=
\EE^{\frac\ii\hbar \calS_{M}}\mathfrak{m}^{\frac12}.
\end{equation}
Here the notations are: $\mathfrak{m}_1^{\frac12}=\mu_1^{\frac12}\cdot |db_{(1)}|^{\frac12}$, $\mathfrak{m}_2^{\frac12}=\mu_2^{\frac12}\cdot |db'_{(2)}|^{\frac12}$, $\mathfrak{m}^{\frac12}=\mu_1^{\frac12}\cdot \mu_2^{\frac12}$ with $\mu_1^{\frac12}$, $\mu_2^{\frac12}$ reference half-densities on $\calY_1$, $\calY_2$; the operation $\underset\Sigma\ast$ is defined as the pairing
$$ \Psi_1\underset\Sigma\ast\Psi_2:=\int_{\calB_{(1)}\times \calB'_{(2)}} \EE^{-\frac\ii\hbar \langle b_{(1)}, b'_{(2)}\rangle} |db_{(1)}|^{\frac12} |db'_{(2)}|^{\frac12}\; \Psi_1 \Psi_2. $$
The integral over $b_{(1)},b'_{(2)}$ in (\ref{e:gluing for exp(S)}) is Gaussian (since the actions are linear in the integration variables, by (\ref{e:bp})) and boils down to evaluating the integrand at the critical point which, due to (\ref{e:bp}), is given by $b_{(1)}=b_{(2)}$, $b'_{(2)}=b'_{(1)}$. Thus (\ref{e:gluing for exp(S)}) comes from
$$\left.\calS^{\calP_1}_{M_1}+\calS^{\calP_2}_{M_2}-\langle b_{(1)},b'_{(2)}\rangle \right|_{b_{(1)}=b_{(2)},b'_{(2)}=b'_{(1)}}
=\left.\calS_{M_1}+\calS_{M_2} \right|_{b_{(1)}=b_{(2)},b'_{(2)}=b'_{(1)}}
= \calS_M$$
which is simply the statement of additivity of the action with respect to gluing.
Performing the BV pushforwards $\calY_1\times \calY_2\rightarrow \calV_{M_1}\times \calV_{M_2}\rightarrow \calV_M$ in (\ref{e:gluing for exp(S)}), we obtain the gluing formula (\ref{e: gluing formula}).

\begin{Rem}\label{r:SBOmega}
We assume that the states are $(\hbar^2\Delta+\Omega)$\ndash closed and that, on the boundary component where we glue,
the $\Omega$ for one polarization is the Segal--Bargmann transform with kernel $\EE^{\frac\ii\hbar (f_\Sigma^{\calP_2}-f_\Sigma^{\calP_1})}=\EE^{-\frac\ii\hbar \langle b, b'\rangle}$
of the $\Omega$ for the other polarization.\footnote{This is automatically satisfied if $\Omega$ is constructed
as in equation \eqref{e:Omegastandardordering}. It is also satisfied in all the examples considered in Sections~\ref{s:abeBF} and \ref{sec: BF-like}, also in the presence of quantum corrections. This is essentially due to locality: the quantum corrections may be seen as arising from the standard quantization of a modified BFV boundary action.}
As a consequence of Theorem~\ref{thm: BV pushforward in family}
the glued state will also be $(\hbar^2\Delta+\Omega)$\ndash closed. Moreover, if we change one state by
an $(\hbar^2\Delta+\Omega)$\ndash exact term, the glued state will also change by an $(\hbar^2\Delta+\Omega)$\ndash exact term, e.g. if $\Hat\psi_{M_1}$ is shifted by $(\hbar^2\Delta_{\calV_{M_1}}+\Omega^{\calP_1
}_{\de M_1})\;\Hat\alpha_{M_1}$ with $\Hat\alpha_{M_1}$ some degree $-1$ element of $\Hat\calH_{M_1}^{\calP_1}$, then the glued state (\ref{e: gluing formula}) gets shifted by $(\hbar^2\Delta_{\calV_M}+\Omega_{\de M})\;P_*(\Hat\alpha_{M_1}\ast_\Sigma \Hat\psi_{M_2})$. Here we suppress in the notation the polarizations on the boundary components of $M$, only denoting explicitly the polarization on the gluing interface $\Sigma$; BV pushforward $P_*$ and the pairing $*_\Sigma$ are as in (\ref{e: gluing formula}).
\end{Rem}

\begin{Rem}\label{rem: changing P by attaching a cylinder}
The gluing procedure may also be used to change the polarization by the use of cylinders.
Namely, suppose that that we have a boundary component $\Sigma$ on which we choose a polarization $\calP_1$
to compute the state. If we want to get the state in a polarization $\calP'$, we glue in a cylinder
$\Sigma\times I$, $I$ an interval, with polarization $\calP'$ on one side and a polarization $\calP_2$ transversal to
$\calP_1$ on the other side, the one we glue in. In a topological field theory it does not matter which interval we take.
In a non-topological theory, one has to take the limit for the length of the interval going to zero; an alternative procedure consists in putting on the cylinder a theory that is topological in the interval direction and has the same BFV boundary structure. A canonical way to do this is by the AKSZ formalism \cite{AKSZ} with source $T[1]I$ and target the BFV manifold associated to $\Sigma$ (notice that in this version of the AKSZ model the target is usually infinite-dimensional).
We call this construction the generalized Segal-Bargmann transform.
\end{Rem}

\begin{Rem} The possibility to pass between different polarizations of $\calF^\de_\Sigma$ via (generalized) Segal-Bargmann transform leads, infinitesimally, to a projectively flat connection $\nabla_\mr{H}$ on the vector bundle of spaces of states $\calH_\Sigma^\calP$ over the space of polarizations $\mathfrak{P}_\Sigma$ -- the generalized Hitchin connection -- so that the parallel transport of $\nabla_\mr{H}$ is the Segal-Bargmann transform $\calH_\Sigma^{\calP_1}\to \calH_\Sigma^{\calP_2}$. E.g.~in the case of Chern-Simons theory, the moduli space of conformal structures $\calM^\mr{conf}_\Sigma$ on a surface $\Sigma$ embeds into $\mathfrak{P}_\Sigma$, and the pullback of $\nabla_\mr{H}$ to  $\calM^\mr{conf}_\Sigma$ is the Hitchin connection on the bundle of WZW conformal blocks over the moduli of conformal structures (see e.g.~\cite{ADPW}). In the case of perturbed $BF$ theories that are the focus of this paper, we prefer to work with a discrete subset $\mathfrak{P}_\Sigma^{A,B}$ of $\mathfrak{P}_\Sigma$ consisting of $2^{\#\pi_0(\Sigma)}$ points which correspond to choosing either $\frac{\delta}{\delta\sfA}$ or $\frac{\delta}{\delta\sfB}$ polarization (see Section \ref{s:pol}) on each connected component of $\Sigma$. In this situation we do not have infinitesimal transitions between points of $\mathfrak{P}_\Sigma^{A,B}$ and so it does not make sense to speak of the connection $\nabla_\mr{H}$, only of the (finite) Segal-Bargmann transform between the polarizations.
\end{Rem}

\begin{Rem}\label{rem: Fubini}
Note that our proof of the gluing formula (\ref{e: gluing formula}) implicitly uses Fubini theorem which is automatic for finite-dimensional integrals and which we expect to hold for path integrals representing states in field theory.
We follow this heuristics to derive the gluing formulae for the states and the propagators in abelian $BF$ theory (see Section \ref{ss:gluing} and Appendix \ref{a:comp}). However, these formulae can be proved to hold a posteriori (see Theorem \ref{thm: gluing of propagators} for propagators and Section \ref{ss:reducing backgrounds} for states). This immediately implies the gluing formulae for expectation values as they are determined by states and propagators.
Finally note that, as a consequence, gluing in perturbation theory (for $BF$-like theories of Section \ref{sec: BF-like}) also automatically holds once we have proved it to hold for states and propagators of the unperturbed theory.
\end{Rem}

\section{Abelian $BF$ theory}\label{s:abeBF}
Here we recollect basic notions on the BV-BFV formalism for the abelian $BF$ theory \cite{SchwBF}, which occurs as the unperturbed part in  many AKSZ \cite{AKSZ} theories, but also in  quantum mechanics and in Yang--Mills theory in the first-order formalism.

Fix a dimension $d$ and an integer $k$. The $d$\ndash dimensional abelian
$BF$ theory (with shift $k$) associates to a compact $d$\ndash manifold $M$ (possibily with boundary) the space of fields
$\calF_M=\Omega^\bullet(M)[k]\oplus\Omega^\bullet(M)[d-k-1]$. Using the customary notation $\sfA\oplus \sfB\in \Omega^\bullet(M)[k]\oplus\Omega^\bullet(M)[d-k-1]$ for the fields,
we have the following odd-symplectic form, action and cohomological vector field on $\calF_M$:
\begin{align}
\omega_M &= \int_M \delta\sfB\,\delta\sfA,\label{e:omega BF}\\
\calS_M &= \int_M \sfB\,\dd\sfA,\\
Q_M &= (-1)^d \int_M  \dd\sfB\,\frac{\delta}{\delta\sfB}  +\dd\sfA\,\frac{\delta}{\delta\sfA},
\end{align}
where $\delta$ denotes the de~Rham differential on $\calF_M$, $\dd$ the de~Rham differential on $M$, and we omit the wedge symbols.

\begin{Rem}
One way to read the formulae above is to understand $\sfA,\sfB$ as arguments. A more formal way, which helps understanding grading conventions, consists in
viewing $\sfA$ and $\sfB$ as maps $\sfA\colon \calF_M\rightarrow \Omega^\bullet(M)$, $\sfB\colon\calF_M\rightarrow \Omega^\bullet(M)$ obtained by composing the projections from $\calF_M$  to the first and second summand with the shifted identity maps $\Omega^\bullet(M)[k]\rightarrow \Omega^\bullet(M)$ and  $\Omega^\bullet(M)[d-k-1]\rightarrow \Omega^\bullet(M)$, respectively.
The intrinsic degree (``ghost number'') of the $p$-form components $\sfA^{(p)},\sfB^{(p)}$ (i.e. $\sfA$, $\sfB$ composed with the projection $\Omega^\bullet(M)\rightarrow \Omega^p(M)$) corresponding to the $\mathbb{Z}$-grading on $\calF_M$ is $k-p$ for $\sfA^{(p)}$ and $d-k-1-p$ for $\sfB^{(p)}$.
\end{Rem}

\begin{Rem}
If $k=1$, one simply speaks of abelian $BF$ theory. In this case the degree zero component $A$ of $\sfA$ is a $1$\ndash form, which can also be thought of as a connection for a line bundle. The action restricted to the degree zero fields $A$ and $B$---the latter being now a $(d-2)$\ndash form---is just
$\int_M BF$, where $F=\dd A$ is the curvature of $A$. This explains the name $BF$ theory.
\end{Rem}

The exact  BFV manifold
$(\calF^\de_\Sigma,\omega^\de_\Sigma=\delta\alpha^\de_\Sigma,Q^\de_\Sigma)$ assigned to
a $(d-1)$\ndash dimensional compact manifold $\Sigma$ is given by
$\calF^\de_\Sigma=\Omega^\bullet(\Sigma)[k]\oplus\Omega^\bullet(\Sigma)[d-k-1]$ and
\begin{align*}
\alpha^\de_\Sigma &= (-1)^d\int_\Sigma \sfB\,\delta\sfA,\\
Q^\de_\Sigma &= (-1)^d\int_\Sigma
\dd\sfB\,\frac{\delta}{\delta\sfB}  + \dd\sfA\,\frac{\delta}{\delta\sfA},
\end{align*}
where we denote again a field by
$\sfA\oplus\sfB\in \Omega^\bullet(\Sigma)[k]\oplus \Omega^\bullet(\Sigma)[d-k-1]$
(or regard
$\sfA, \sfB$ as maps $\calF^\de_\Sigma\rightarrow\Omega^\bullet(\Sigma)$).
The BFV action is
$$\calS^\de_\Sigma = \int_\Sigma \sfB\,\dd\sfA.$$

Finally, the surjective submersion $\pi_M\colon\calF_M\to\calF^\de_{\de M}$ is just given by the restriction of forms to the boundary.


\subsection{Polarizations}\label{s:pol}
Let $\de M$ be the disjoint union of the two compact (possibly empty) manifolds $\de_1 M$ and $\de_2 M$,
so $\calF^\de_{\de M}=\calF^\de_{\de_1 M}\times \calF^\de_{\de_2 M}$.
We consider polarizations $\calP$ on $\calF^\de_{\de M}$ given as direct products of polarizations on each factor.

On $\de_1 M$ we choose the $\frac{\delta}{\delta\sfB}$\ndash polarization and identify the quotient (space of leaves of the associated foliation) with $\calB_1:=\Omega^\bullet(\de_1 M)[k]$,
whose coordinates are the $\sfA$\ndash fields.
On $\de_2 M$ we choose the $\frac{\delta}{\delta \sfA}$\ndash polarization and identify the quotient with $\calB_2:=\Omega^\bullet(\de_2 M)[d-k-1]$,
whose coordinates are the $\sfB$\ndash fields.\footnote{One can alternatively call these two polarizations the $\sfA$\ndash\  and $\sfB$\ndash {\sf representations}, respectively, by analogy with the coordinate and momentum representations in quantum mechanics.} Then $\calB^\calP_{\de M}=\calB_1\times\calB_2$.
We have to subtract the differential of
\[
f^\calP_{\de M}=(-1)^{d-k}\int_{\de_2 M}\sfB\sfA,
\]
from $\alpha^\de_{\de M}$ to get the adapted BFV $1$\ndash form
\[
\alpha^{\de,\calP}_{\de M}=(-1)^d\int_{\de_1 M}\sfB\,\delta\sfA + (-1)^{k}\int_{\de_2 M} \delta\sfB\,\sfA.
\]
We then get the modified action
\[
\calS^\calP_M=\int_M \sfB\,\dd\sfA +(-1)^{d-k}\int_{\de_2 M}\sfB\sfA.
\]

We will denote by $\bA$ the coordinate on $\calB_1$ and by $\bB$ the coordinate on $\calB_2$ and
by $\tbA$ and $\tbB$ some prescribed extensions of these fields to $\calF_M$. We write the fields in $\calF_M$ as
\begin{equation}\label{e:splitting of fields into bdry and fluct}
\begin{split}
\sfA&=\tbA + \Hat\sfA,\\
\sfB&=\tbB+\Hat\sfB,
\end{split}
\end{equation}
where $\Hat\sfA$ is required  to restrict to zero on $\de_1 M$, whereas  $\Hat\sfB$ is required  to restrict to zero on $\de_2 M$.
This is our choice of a section of $\calF_M\to\calB^{\calP}_{\de M}$.
See Section~\ref{ss:ext} for a further discussion.
We then get
\begin{equation}\label{e:SPAB}
\calS^\calP_M=\int_M \left(\tbB\,\dd\tbA + \tbB\,\dd\Hat\sfA + \Hat\sfB\,\dd\tbA + \Hat\sfB\,\dd \Hat\sfA \right)
+(-1)^{d-k}\int_{\de_2 M}\left(\bB\tbA + \bB\Hat\sfA\right).
\end{equation}

\subsection{Residual fields}\label{sec: abBF backgrounds}
We now focus on the last bulk term $\Hat\calS_M:=\int_M \Hat\sfB\,\dd \Hat\sfA$.
Because of the boundary conditions on $\Hat\sfA$ and $\Hat\sfB$, its variations have no boundary terms. Its critical
points are given by $\dd\Hat\sfA=\dd\Hat\sfB=0$. As $\calZ_M$ we now choose an embedding of the appropriate cohomologies. Namely, for $i=1,2$, let us define the subcomplexes
\[
\Omega^\bullet_{\DD i}(M):=\{\gamma\in\Omega^\bullet(M) : \iota_i^*\gamma=0\}
\]
of $\Omega^\bullet(M)$, where $\iota_i$ is the inclusion map of $\de_i M$ into $M$. (Here $\DD$ stands for Dirichlet.)
Observe that the corresponding cohomologies $H^\bullet_{\DD1}(M)$ and $H^\bullet_{\DD2}(M)$
are canonically paired by integration over $M$.\footnote{We have canonical identification with cohomology of pairs $H^\bullet_{\DD1}(M)=H^\bullet(M,\de_1 M)$, $H^\bullet_{\DD2}(M)=H^\bullet(M,\de_2 M)$.}
Hence
$$\calV_M:=H^\bullet_{\DD1}(M)[k]\oplus H^\bullet_{\DD2}(M)[d-k-1]$$
is a finite-dimensional BV manifold.
Using Poincar\'e duality, we may also write $\calV_M=T^*[-1](H^\bullet_{\DD1}(M)[k])=T^*[-1](H^\bullet_{\DD2}(M)[d-k-1])$.
This is the
space of {\backgrounds}.
In the notations of Section~\ref{ss:pertq}, we have
\be \calZ_M=\calV_M\times\calB^\calP_{\de M}\label{e: Z=V x B}\ee
as a trivial bundle. According to our construction (cf.\ Remark \ref{rem: hat H}), the space $\Hat\calH^\calP_M$ is $\HDens(\calZ_M)$.

To define the 
BV Laplacian on $\calV_M$
pick a basis $\{[\chi_i]\}$ of $H^\bullet_{\DD1}(M)$ and its
dual basis $\{[\chi^i]\}$ of $H^\bullet_{\DD2}(M)$ with chosen representatives $\chi_i$ and $\chi^i$
in $\Omega^\bullet_{\DD1}(M)$ and $\Omega^\bullet_{\DD2}(M)$. In particular, we have
$\int_M\chi^i\chi_j=\delta^i_j$.
We write
\begin{align*}
\sfa &= \sum_i z^i\chi_i,\\
\sfb &= \sum_i z^+_i\chi^i,
\end{align*}
where $\{z^i,z^+_i\}$ are canonical coordinates on $\calV_M$ with BV form
\[
\omega_{\calV_M}=\sum_i (-1)^{k+(d-k)\cdot \deg z^i} \delta z^+_i\delta z^i.
\]
Notice that $\deg z^i = k- \deg\chi_i$ and $\deg z^+_i=-\deg z^i-1$.
The BV operator on $\calV_M$ is
\begin{equation}\label{e:Delta}
\Delta_{\calV_M} = \sum_i  (-1)^{k+(d-k)\cdot \deg z^i} \frac\de{\de z^i}\frac\de{\de z^+_i}.
\end{equation}

\subsubsection{Boundary components and \pseudovacua}
Our choice of {\pseudovacua} depends on which components of the boundary we choose as $\de_1 M$ and  $\de_2 M$.

If $\de M$ is connected, there are only two choices: $(\de_1 M=\de M,\de_2 M=\emptyset)$ and $(\de_1 M=\emptyset,\de_2 M=\de M$). The first yields
$\calV_M=H^\bullet(M,\de M)[k]\oplus H^\bullet(M)[d-k-1]$, the second
$\calV_M=H^\bullet(M)[k]\oplus H^\bullet(M,\de M)[d-k-1]$. The two are not BV symplectomorphic to each other
(unless $2k=d-1$).

If $\de M$ is not connected, there are more choices which yield  other, generally inequivalent, moduli spaces.
For example, take $M=\Sigma\times[0,1]$ where $\Sigma$ is a compact $(d-1)$\ndash manifold. Besides  the choices
$(\de_1 M=\de M,\de_2 M=\emptyset)$ and $(\de_1 M=\emptyset,\de_2 M=\de M$),
which yield $\calV_M=T^*[-1](H^\bullet(\Sigma)[d-k-1])$ and $\calV_M=T^*[-1](H^\bullet(\Sigma)[k])$,
we now also have
$\de_1 M=\Sigma\times\{0\},\de_2 M=\Sigma\times\{1\}$ and
$\de_1 M=\Sigma\times\{1\},\de_2 M=\Sigma\times\{0\}$, both of which yield $\calV_M=\{0\}$.

\subsection{The propagator}\label{ss:prop}
We now write
\begin{equation}\label{e:splitting}
\begin{split}
\Hat\sfA &= \sfa + \alpha,\\
\Hat\sfB &=\sfb + \beta,
\end{split}
\end{equation}
where the fluctuation $\alpha$ is required  to restrict to zero on $\de_1 M$, whereas
the fluctuation $\beta$ is required  to restrict to zero on $\de_2 M$. Notice that we have
$\Hat\calS_M=\int_M\beta\,\dd\alpha$. We regard it as a quadratic function
on $\Omega^\bullet_{\DD1}(M)[k]\oplus\Omega^\bullet_{\DD2}(M)[d-k-1]$. Notice that critical points are closed forms.

We now have to fix a Lagrangian subspace $\calL$
of a symplectic complement of $\calV_M$
on which $\Hat\calS_M$ has an {\sf isolated} critical point at the origin
(i.e. $\dd$
has no kernel).
This
can be done, for example, using
the Hodge theory
on manifolds with boundary \cite{Fr,Mo,CDGM}.
Namely, we pick a metric
on $M$ through which we define the Hodge star operator. We assume that the metric has a product structure near the boundary.\footnote{\label{footnote: product structure} In other words, there is a diffeomorphism $\phi$ between a 
neighborhood $U$ of $\de M$ in $M$ and $\de M\times [0,\epsilon)$ for some $\epsilon>0$, such that $\phi|_{\de M}=\mathrm{id}_{\de M}$ and  the the metric on $M$ restricted to $U$ has the form $\phi^*(g_{\de M}+dt^2)$. 
Here $g_{\de M}$ is some Riemannian metric on the boundary and $t\in [0,\epsilon)$ is the 
vertical coordinate on $\de M\times [0,\epsilon)$. 
}
This yields a scalar product on $\Omega^\bullet(M)$, $(\gamma,\lambda):=\int_M\gamma\,*\!\lambda$, and
the Hodge dual $\dd^*$ of the de~Rham differential. We define
\begin{equation}\label{e:Lgf}
\calL=((\dd^*\Omega^{\bullet+1}_{\text{N}2}(M))
\cap\Omega_{\DD1}^\bullet(M))[k]
\oplus
((\dd^*\Omega^{\bullet+1}_{\text{N}1}(M))
\cap\Omega_{\DD2}^\bullet(M))[d-k-1]
\end{equation}
where
\[
\Omega^\bullet_{\text{N}i}(M):=\{\gamma\in\Omega^\bullet(M)\colon \iota_{\de_i M}^**\!\gamma=0\}
\]
is the space of Neumann forms relative to $\de_iM$. The restriction of $\Hat\calS_M$ to $\calL$ is nondegenerate.
In Appendix~\ref{a:Hodge}, see Lemma~\ref{l:ortho},
we show that
$\calL$ is Lagrangian in the complement
of $H^\bullet_{\DD1}(M)[k]\oplus H^\bullet_{\DD2}(M)[d-k-1]$ which, thanks to \eqref{e:HHHarm} and \eqref{e:HarmhHarm}, is embedded into
$\Omega^\bullet_{\NN2,\DD1}(M)[k]\oplus\Omega^\bullet_{\NN1,\DD2}(M)[d-k-1]$ as the space of $(\dd,\dd^*)$\ndash closed forms.

In the notations of Section \ref{ss:pertq}, the coisotropic subbundle $\Tilde \calZ_M$ of $\calF_M\ra \calB^\calP_{\de M}$, generating $\calZ_M$ as its fiberwise reduction, is
$$\Tilde \calZ_M=\calZ_M\times\calL$$
with $\calZ_M$ as in (\ref{e: Z=V x B}).

The propagator can then be explicitly constructed generalizing the construction by Axelrod and Singer \cite{AS} for the boundaryless case. As a different option, 
one can use a topologically constructed propagator following the
philosophy of
\cite{K,BC,SchwTQFT}.

More concretely, we are interested in the integral kernel $\eta$ (a.k.a.~parametrix) of
the chain contraction $K$ of the space of forms $\Omega_{\DD1}^\bullet(M)$ onto the cohomology $H^\bullet_{\DD1}(M)$, which is related to the gauge-fixing Lagrangian by
\begin{equation}\label{e:L via K}
\calL=\mathrm{im}(K)[k]\oplus \mathrm{im}(K^*)[d-k-1].
\end{equation}
One possible strategy is to
choose the Hodge-theoretic chain contraction $K\colon \Omega^\bullet_{\DD1}(M)\to\Omega^{\bullet-1}_{\DD1}(M)$ given by $K=\dd^*/(\Delta_\mathrm{Hodge}+P_\mathrm{Harm})$ where $P_\mathrm{Harm}$ is the projection to (ultra-)harmonic forms (we refer the reader to Appendix \ref{a: Hodge propagator} for details). This choice corresponds, via (\ref{e:L via K}), to the gauge-fixing subspace (\ref{e:Lgf}). 

Being the integral kernel of the inverse of an elliptic operator (composed with $\dd^*$), the propagator $\eta$ restricts to a smooth form away from the diagonal of $M\times M$. 
If we
define $$C_2^0(M)=\{(x_1,x_2)\in M : x_1\not=x_2\}$$
and denote by $\iota_\DDD$ the inclusion of
$$\DDD:=\{x_1\times x_2\in (\de_1 M\times M)\cup (M\times \de_2 M) : x_1\not=x_2\}$$
into $C_2^0(M)$,
we then have $\eta\in\Omega^{d-1}(C^0_2(M),\DDD)$,\footnote{
In fact, the Hodge-theoretic
propagator outlined above satisfies stronger boundary conditions: ultra-Dirichlet (see Appendix \ref{a:Hodge} for the definition) on $\de_1 M$ in the first argument and ultra-Dirichlet on $\de_2 M$ in the second argument, and also ultra-Neumann on $\de_2 M$ in the first argument and ultra-Neumann on $\de_1 M$ in the second argument, see Section \ref{a: Hodge propagator subsubsec}. The same is true for the propagator constructed in Appendix \ref{a:prop}.
} with
\begin{equation}
\label{e: Omega D}
\Omega^\bullet(C_2^0(M),\DDD)=\{\gamma\in\Omega^\bullet(C_2^0):\iota_\DDD^*\gamma=0\}.
\end{equation}
Its properties 
are defined by the formula
\begin{equation}\label{e:eta}
\eta =\frac1{T_M}\frac{(-1)^{kd}}{\ii\hbar}\int_\calL \EE^{\frac\ii\hbar \Hat\calS_M} \pi_1^*\alpha\,\pi_2^*\beta,
\end{equation}
with
\begin{equation}\label{e:T}
T_M = \int_\calL \EE^{\frac\ii\hbar \Hat\calS_M}.
\end{equation}
In (\ref{e:eta}), we denote by $\pi_1,\pi_2$ the projections from $M\times M$ to its first and second factor, and, by abuse of notations, also the corresponding restricted maps $C_2^0(M)\to M$.

\subsubsection{On $T_M$ and torsions}\label{sec: T and torsions}
First we comment on the Gaussian functional integral (\ref{e:T}) which has to be prescribed a mathematical meaning using an appropriate regularization procedure.

In the case $\de M=\emptyset$ and with forms on $M$ taken with coefficients in an acyclic $O(m)$-local system $E$, Schwarz showed in \cite{SchwBF} that $T_M$, understood via zeta-function regularization, is the Ray--Singer torsion (or its inverse, depending on $k$) of the complex $\Omega^\bullet(M,E)$: $T_M=\tau_{RS}(M,E)^{(-1)^{k-1}}$. In the present case, we should think of it as a generalization to the relative complexes (one relevant model being the complex $\Omega^\bullet_{\Hat\DD1\Hat\NN2}(M)$, cf. Appendix \ref{a:Hodge}).

Since we consider forms on $M$ with trivial coefficients, and the trivial local system is not acyclic, $T_M$ is not a number, but a constant (i.e. not depending on a point in $\calV_M$) complex-valued half-density on $\calV_M$, defined up to a sign:\footnote{For the purposes of this paper we are working with partition functions as defined up to a sign. The problem of fixing this sign is akin to fixing the sign of Reidemeister torsion, which requires the introduction of additional orientation data, cf. \cite{TuraevTor}.}
$$T_M\in \mathbb{C}\otimes \mathrm{Dens}^{\frac12}_\mathrm{const}(\calV_M)/\{\pm 1\}\cong \mathbb{C}\otimes\left(\mathrm{Det}\,H^\bullet_{\DD1}(M)\right)^{(-1)^{k-1}}/\{\pm 1\}$$
where $\mathrm{Det}\, H^\bullet_{\DD1}(M)$ is the determinant line of de Rham cohomology of $M$ relative to $\de_1 M$ and, by convention, for $l$ a line, $l^{-1}=l^*$ is the dual line.
A choice of basis  $\{[\chi_i]\}$ in $H^\bullet_{\DD1}(M)$ induces a trivialization of the determinant line $\phi\colon\mathrm{Det}\,H^\bullet_{\DD1}(M)\xrightarrow{\simeq}\mathbb{R}$, which makes $\phi_*T_M\in \mathbb{C}/\{\pm 1\}$ a number (defined up to sign). Choosing a different basis $\{[\Tilde\chi_i]\}$ in $H^\bullet_{\DD1}(M)$ induces a different trivialization $\Tilde\phi$ of the determinant, and one has the transformation property
$$\Tilde\phi_* T_M = (\mathrm{Ber}\,\theta)^{(-1)^k} \phi_* T_M$$
where $\theta$ is the transformation matrix between the two bases, $[\Tilde\chi_i]=\sum_j\theta^i_j [\chi_j]$
and $\mathrm{Ber}\,\theta$ is it's Berezinian (superdeterminant).

The BV integral (\ref{e:T}) does not depend on the choice of $\calL$ (cf. independence of Ray-Singer torsion on the choice of Riemannian metric).

By comparison with the result of \cite{cell_ab_BF} in the combinatorial setting, $T_M$ is expressed in terms of the Reidemeister torsion $\tau(M,\de_1 M)\in \mathrm{Det}\, H^\bullet_{\DD1}(M)/\{\pm 1\}$ as
\begin{equation}\label{e:T via tau}
T_M=\xi\cdot \tau(M,\de_1 M)^{(-1)^{k-1}}
\end{equation}
where the factor $\xi$, originating in the normalization of the integration measure, compatible with gluing, is
\begin{equation}\label{e:xi}
\xi=(2\pi\hbar)^{\sum_{j=0}^d \left(\frac{(-1)^k}{4}+\frac{1}{2} j (-1)^{j-1}\right)\dim H^j_{\DD1}(M)}\cdot \left(e^{-\frac{\pi i}{2}}\hbar\right)^{\sum_{j=0}^d \left(\frac{-(-1)^k}{4}+\frac{1}{2} j (-1)^{j-1}\right)\dim H^j_{\DD1}(M)}\in \mathbb{C}
\end{equation}
Note that, by Milnor's duality theorem for torsions, (\ref{e:T via tau}) can also be written as $T_M=\xi\cdot \tau(M,\de_2 M)^{(-1)^{d-k}}$.

\begin{Rem} In (\ref{e:T via tau}) we use the Reidemeister torsion. On the other hand, the analytic (Ray-Singer) torsion, as defined via zeta-function regularized determinants of Hodge-de Rham Laplacians, is known to differ from the Reidemeister torsion by the factor $2^{\frac{1}{4}\chi(\de M)}$ with $\chi(\de M)$ the Euler characteristic of the boundary (in the case of a product metric near the boundary), see \cite{Luck,Vishik}. This means that the normalization of the functional integral measure in (\ref{e:T}) corresponding to the zeta-function regularization procedure is not the one compatible with discretization and gluing as in \cite{cell_ab_BF}.
\end{Rem}

\begin{Rem}
To be completely
pedantic, we should also include in $T_M$ the factors $\tau(\de_1 M)^{\frac{(-1)^{k-1}}{2}}$ and $\tau(\de_2 M)^{\frac{(-1)^{d-k}}{2}}$, coming from the fact that $T_M$ is also a constant half-density on boundary fields and identification between half-densities and functions is via multiplication by an appropriate power of torsion, cf. Section \ref{rem: half-densities on an elliptic complex}. Note that, for gluing, these boundary torsion factors coming from the two sides of the gluing interface cancel each other due to the relation $\tau(\Sigma)^{(-1)^{k-1}}\cdot \tau(\Sigma)^{(-1)^{d-k}}=1$ for $\Sigma$ a closed $(d-1)$-manifold, arising from Milnor's duality theorem.
\end{Rem}

\subsubsection{Properties of propagators}
\label{sec: properties of propagators}
For the computations, it is also useful to define
\begin{equation}\label{e:etahat}
\Hat\eta :=\frac1{T_M}\frac{(-1)^{kd}}{\ii\hbar}\int_\calL \EE^{\frac\ii\hbar \Hat\calS_M} \pi_1^*\Hat\sfA\,\pi_2^*\Hat\sfB=
\eta + \frac{(-1)^{kd}}{\ii\hbar}\sum_{ij} z^i\pi_1^*\chi_i\,z^+_j\pi_2^*\chi^j.
\end{equation}
By calculating
$\int_\calL \Delta\left(\EE^{\frac\ii\hbar \Hat\calS_M} \pi_1^*\Hat\sfA\,\pi_2^*\Hat\sfB\right)$ in two different ways (taking $\Delta$ out by the chain map property of BV pushforwards -- Theorem \ref{thm: BV Stokes for pushforward}, or by computing the integrand directly), we get the relation
  $(-1)^d\dd\Hat\eta=\frac\hbar\ii\Delta_{\calV_M}\Hat\eta$, which implies
\begin{equation}\label{e:deta}
\dd\eta = (-1)^{d-1}\sum_i (-1)^{d\cdot\deg\chi_i}\pi_1^*\chi_i\,\pi_2^*\chi^i.
\end{equation}
Notice that in the case $\de M=\emptyset$ the sum defines a representative of the Euler class of $M$.

The other characteristic property of $\eta$ is that its integral on the $(d-1)$\ndash cycle
given by  fixing one of the two arguments in $C_2^0(M)$ and letting the other vary on a small $(d-1)$\ndash sphere centered on the first one is normalized to $\pm1$.\footnote{
More precisely, the integral is $+1$, if we fix the second argument and vary the first. In the opposite case, the integral is $(-1)^d$.
} (As a consequence, if the first point is fixed on the boundary, then either the propagator is identically zero due to boundary conditions (\ref{e: Omega D}), or otherwise the integral over the relative cycle given by second point varying on a small {\it half-sphere} is $\pm 1$.)

Instead of using the Hodge-theoretic propagator of Appendix \ref{a: Hodge propagator}, one can construct a ``soft'' propagator along the lines
of \cite{BC,C,CR}. More precisely, one may use the construction for boundaryless manifolds to produce the propagator for manifolds with boundary by a version of the method of image charges, see Appendix~\ref{a:prop}. The soft propagator does not correspond to the gauge-fixing Lagrangian (\ref{e:Lgf}), but to another one, constructed via (\ref{e:L via K}) for the chain contraction
\begin{equation}\label{e:K via eta}
\begin{array}{cccc}
K_\mathrm{soft}\colon & \Omega^\bullet_{\DD1}(M) &\rightarrow & \Omega^{\bullet-1}_{\DD1}(M) \\
& \alpha &\mapsto & (\pi_1)_*(\eta\wedge\pi_2^*(\alpha))
\end{array}
\end{equation}


\begin{Rem}[Change of data]\label{r:cod}
Notice that, once we have fixed representatives $\chi_i$'s and $\chi^i$'s, still $\eta$ is only defined up to the differential
of a form $\lambda\in\Omega^{(d-2)}(C_2^0(M))$. We may also change the representatives $\chi_i$'s and $\chi^i$'s
by exact forms and also perform a change of basis. The latter corresponds to a linear BV transformation of $\calV_M$.
If we denote the former change by
\begin{subequations}\label{e:chidot}
\begin{align}
\dot\chi_i&=\dd\sigma_i, &\sigma_i&\in\Omega^{\deg\chi_i-1}_{\DD1}(M),\\
\dot\chi^i&=\dd\sigma^i, & \sigma^i&\in\Omega^{\deg\chi^i-1}_{\DD2}(M),
\end{align}
\end{subequations}
then we get
\begin{equation}\label{e:etadot}
\dot\eta=\dd\lambda + (-1)^{d-1}\sum_i (-1)^{d\cdot\deg\chi_i}\pi_1^*\sigma_i\,\pi_2^*\chi^i
+(-1)^{d-1}\sum_i (-1)^{(d-1)\cdot\deg\chi_i}\pi_1^*\chi_i\,\pi_2^*\sigma^i.
\end{equation}
Cf.~the classification of infinitesimal deformations of gauge-fixing data for BV pushforwards into types I, II, III in \cite{CMcs}.
\end{Rem}

\begin{Rem}
To study the properties of Feynman diagrams in theories that are perturbations of abelian $BF$ theories,
it is useful to consider the ASFM compactifications of configuration spaces \cite{AS,FM}. The propagator, see Appendix~\ref{a:prop}, extends to the compactification $C_2(M)$, which is a smooth manifold with corners, as a smooth form.
\end{Rem}

\begin{Rem} For $M$ closed, the Hodge propagator of Appendix \ref{a: Hodge propagator} has the property
\begin{equation}\label{e:T-sym}
T^*\eta=(-1)^d\eta
\end{equation}
where the map $T:C_2(M)\to C_2(M)$ sends $(x_1,x_2)$ to $(x_2,x_1)$, which corresponds to the chain contraction $K$ being skew self-adjoint. If $M$ has boundary, one has instead
\begin{equation}\label{e:T-sym with bdry}
T^*\eta=(-1)^d\eta^\mr{op}
\end{equation}
where $\eta^\mr{op}$ stands for the propagator (corresponding to the same metric on $M$) with opposite boundary conditions. For soft propagators, see Appendix \ref{a:prop}, this $T$-symmetry property is not automatic but can always be achieved. In Section \ref{ss:dbling} 
we explain how to recover this property that might have been spoiled by the gluing procedure. Another property that is automatic for the Hodge propagator is
\begin{equation}\label{e:eta eta =0}
(\pi_2)_*(\pi_{12}^*\eta\wedge \pi_{23}^*\eta)=0
\end{equation}
where $\pi_{12},\pi_{23}\colon C_3(M)\to C_2(M)$ are the projections induced from taking the first or the last pair of points in a triple $(x_1,x_2,x_3)$ and $\pi_2: C_3(M)\to M$ takes the middle point in a triple. Property (\ref{e:eta eta =0}) corresponds the property $K^2=0$ of the Hodge chain contraction. Properties (\ref{e:T-sym},\ref{e:T-sym with bdry}) and (\ref{e:eta eta =0}) are useful for simplifications in perturbation theory, but our treatment does not rely on having them.
\end{Rem}

\subsection{Choosing the extensions}\label{ss:ext}
Let us choose the extensions $\Tilde\bA$ and $\Tilde\bB$ of the boundary values $\bA$ and $\bB$ in such a way that the extension $\Tilde\bA$
has support in a neighborhood $\calN_1$ of $\de_1M$ and the extension $\Tilde\bB$ has support in a neighborhood $\calN_2$ of $\de_2M$ with $\calN_1\cap\calN_2=\emptyset$. Then (\ref{e:SPAB}) becomes
\begin{equation}\label{e:S in splitting}
\calS^\calP_M=\int_M \left( \tbB\,\dd\Hat\sfA + \Hat\sfB\,\dd\tbA + \Hat\sfB\,\dd \Hat\sfA \right)
+(-1)^{d-k}\int_{\de_2 M}\bB\Hat\sfA
\end{equation}
and the BV odd-symplectic form (\ref{e:omega BF}) becomes
$$\omega_M=\int_M \left(\delta\Hat\sfB\; \delta\Hat\sfA+\delta\tbB\;\delta\Hat\sfA+\delta\Hat\sfA\;\delta\tbA\right).$$
From the latter equation, we see that, in order to comply with Assumption \ref{assump: omega compatible with splitting},
we are forced to choose the
discontinuous extension in which $\tbA$ and $\tbB$ drop to to zero immediately outside the boundary (cf. Remark \ref{rem: extension}) -- only then does $\omega_M$ become independent of the boundary fields $\bA,\bB$ and attain the form $\omega_M=\int_M \delta\Hat\sfB\;\delta\Hat\sfA$.
The de Rham differential of $\tbA$ in (\ref{e:S in splitting}) is not defined, but this
problem is easily remedied if we integrate by parts
\[
\calS^\calP_M=\int_M \left( \tbB\,\dd\Hat\sfA + (-1)^{d-k}\dd\Hat\sfB\,\tbA + \Hat\sfB\,\dd \Hat\sfA \right)
+(-1)^{d-k}\left(\int_{\de_2 M}\bB\Hat\sfA -\int_{\de_1 M}\Hat\sfB\bA\right).
\]
The action for the discontinuous extension is then simply
\begin{equation}\label{e:action}
\calS^\calP_M=\int_M  \Hat\sfB\,\dd \Hat\sfA
+(-1)^{d-k}\left(\int_{\de_2 M}\bB\Hat\sfA -\int_{\de_1 M}\Hat\sfB\bA\right).
\end{equation}

Thus, with discontinuous extension of boundary fields, Assumption \ref{assump: omega compatible with splitting} and equation (\ref{e:bp}) are satisfied. On the other hand, if we would have chosen a generic extension, the formalism of Section \ref{sec: BV pushforward in families} would not apply, and we would produce partition functions that are not guaranteed to satisfy mQME and may change uncontrollably under a change of gauge-fixing.

\subsection{The state}\label{ss:state}
Using the splitting \eqref{e:splitting}, we may rewrite \eqref{e:action} as
the sum of the quadratic part in fluctuations, the residual part and the source term:
\[
\calS^\calP_M =\Hat S_M + \calS_M^\vac + \calS_M^\text{source},
\]
with
\begin{align*}
\Hat\calS_M&=\int_M\beta\,\dd\alpha,\\
\calS_M^\vac&=(-1)^{d-k}\left(\int_{\de_2 M}\bB\sfa -\int_{\de_1 M}\sfb\bA\right),\\
\calS_M^\text{source}&=
(-1)^{d-k}\left(\int_{\de_2 M}\bB\alpha -\int_{\de_1 M}\beta\bA\right).
\end{align*}
To compute the state we just have to perform the Gaussian integral over the fluctuations $\alpha$ and $\beta$.
Using the notations of Section~\ref{ss:prop}, we get
\begin{equation}\label{e:state}
\Hat\psi_M= T_M\, \EE^{\frac\ii\hbar\calS^\text{eff}_M},
\end{equation}
with the \textsf{effective action}
\begin{equation}\label{e:effS}
\calS^\text{eff}_M=
(-1)^{d-k}\left(\int_{\de_2 M}\bB\sfa -\int_{\de_1 M}\sfb\bA\right)
-(-1)^{d+kd}\int_{\de_2M\times\de_1M}\pi_1^*\bB\,\eta\,\pi_2^*\bA.
\end{equation}
By \eqref{e:Delta} and \eqref{e:deta}, we immediately see that $\Hat\psi_M$ satisfies the mQME \eqref{e:mQME}
with $\Hat\Delta^\calP_M$  given by
$\Delta_{\calV_M}$ acting on the fibers of $\calZ_M=\calV_M\times\calB_M^\calP$
and
with $\Hat\Omega^\calP_M$ the standard quantization of $\calS^\de_{\de M}$ relative to the chosen polarization, acting
on the base of $\calZ_M$:
\[
\Hat\Omega^\calP_M= \ii\hbar (-1)^d \left(
\int_{\de_2M}
\dd\bB\,\frac{\delta}{\delta\bB}  + \int_{\de_1M}\dd\bA\,\frac{\delta}{\delta\bA}\right).
\]

\begin{Rem}[Change of data]
Under the change of data \eqref{e:chidot} and \eqref{e:etadot}, the operator $\Hat\Omega^\calP_M$ does not change,
whereas the state $\Hat\psi_M$ changes as in Remark~\ref{r:psidot} with $\tau=0$ and
$\Hat\chi=
\Hat\psi_M\cdot
\Hat\zeta$ with
\begin{multline*}
\Hat\zeta= \left(\frac\ii\hbar\right)^2 \Big(
\sum_i(-1)^{\deg z^i}\int_{\de_2M}\bB z^i\sigma_i-
\sum_i(-1)^{d-k-\deg z^i}\int_{\de_1M}z^+_i\sigma^i\bA+\\
+(-1)^{d-k+kd}\int_{\de_2M\times\de_1M}\pi_1^*\bB\,\lambda\,\pi_2^*\bA\Big).
\end{multline*}
\end{Rem}

\subsubsection{The space of states}\label{sss:ss}
What is left to describe is the space of states $\Hat \calH^\calP_M$. To do this we first introduce the following vector spaces associated to a $(d-1)$\ndash manifold. For an integer $l$ and a nonnegative integer $n$,
we define $\calH_{\Sigma,l}^n$
as the vector space of $n$\ndash linear functionals on $\Omega^\bullet(\Sigma)[l]$ of the form
\[
\Omega^\bullet(\Sigma)[l]\ni\bD\mapsto
\int_{
\left.\Sigma\right.^n
}\gamma\,\pi_1^*\bD\dots\pi_n^*\bD,
\]
multiplied by ${\tau(\Sigma)}^{\frac{(-1)^{l-1}}{2}}$ (cf. Section \ref{rem: half-densities on an elliptic complex}).
Here $\gamma$
is a distributional form on
$\Sigma^n$;
$\tau(\Sigma)$ is the Reidemeister torsion of $\Sigma$. We then define
\[
\calH^\calP_{\de M}=\prod_{n_1,n_2=0}^\infty
\calH_{\de_2M,d-k-1}^{n_2}
\,\Hat\otimes\,
\calH_{\de_1M,k}^{n_1}
\]
and
\[
\Hat\calH^\calP_M=\calH^\calP_{\de M}\,\Hat\otimes\, \HDens(\calV_M).
\]
This is our model for the space of half-densities on $\calZ_M$.
In this description states are regarded as families in the parameter $\hbar$. Perturbative calculations of partition functions and expectation values of observables for (possibly perturbed) $BF$ theory yield asymptotic states
of the form $$T_M\cdot \EE^{\frac\ii\hbar \calS_M^\mathrm{eff}}\cdot\sum_{j\geq 0}\hbar^j\sum_{n_1,n_2\geq 0}\int_{(\de_1 M)^{n_1}\times (\de_2 M)^{n_2}}R^j_{n_1 n_2}(\sfa,\sfb)\; \pi_{1,1}^*\bA\cdots\pi_{1,n_1}^*\bA\; \pi_{2,1}^*\bB\cdots \pi_{2,n_2}^*\bB$$
where the coefficients $R^j_{n_1n_2}(\sfa,\sfb)$ are distributional forms on $(\de_1 M)^{n_1}\times (\de_2 M)^{n_2}$ with values in half-densities on $\calV_M$. Here $T_M$ is as in (\ref{e:T via tau}), whereas $\calS^\mathrm{eff}_M$ should, in the case of a perturbed $BF$ theory, be replaced by the corresponding zero-loop effective action.

We will compute some examples of states arising as expectation values of observables
in Section~\ref{ss:expv}.


\subsection{Gluing}\label{ss:gluing}
Suppose two manifolds with boundary $M_1$ and $M_2$ have a common boundary component $\Sigma$
($\Sigma\subset\de M_1$ and $\Sigma^\text{opp}\subset\de M_2$, where $\Sigma^\text{opp}$ denotes
$\Sigma$ with the opposite orientation). We want to get the state $\Tilde\psi_M$ for the glued manifold
$M=M_1\cup_\Sigma M_2$ by pairing the states $\Hat\psi_{M_1}$ and $\Hat\psi_{M_2}$.
(More precisely, we start from a manifold with boundary $M$ and cut it along a codimension-one submanifold $\Sigma$ into two manifolds with boundary $M_1$ and $M_2$.)

This pairing is better suited to functional integral computations if we choose transverse polarizations on
$\calF^\de_\Sigma$ viewed as a space of boundary fields coming from $M_1$ or $M_2$. More precisely,
we fix the boundary decompositions
$\de M_1=\de_1 M_1\sqcup \de_2 M_1$ and $\de M_2=\de_1 M_2\sqcup \de_2 M_2$
in such a way that $\Sigma\subset\de_1 M_1$ and $\Sigma^\text{opp}\subset\de_2M_2$. Denoting by $\bA_1^\Sigma$ and $\bB_2^\Sigma$ the coordinates on $\Omega^\bullet(\Sigma)[k]$ and $\Omega^\bullet(\Sigma)[d-k-1]$, respectively,
we get
\begin{equation}\label{e:gluedpsi}
\Tilde\psi_M=\int_{\bA_1^\Sigma,\bB_2^\Sigma}
\EE^{\frac\ii\hbar(-1)^{d-k}\int_\Sigma\bB_2^\Sigma\bA_1^\Sigma}\,\Hat\psi_{M_1}\,\Hat\psi_{M_2}
\end{equation}
as a half-density 
on $\Tilde\calZ_M=\Tilde\calV_M\times\calB^\calP_{\de M}$,
with $\Tilde\calV_M=\calV_{M_1}\times\calV_{M_2}$. Notice that we have
$\de_1M=(\de_1M_1\setminus\Sigma)\cup\de_1M_2$, $\de_2M=\de_2M_1\cup(\de_2M_2\setminus\Sigma)$ and
\begin{multline*}
\calB^\calP_{\de M}=\\
\Omega^\bullet(\de_1M_1\setminus\Sigma)[k]\oplus\Omega^\bullet(\de_2M_1)[d-k-1]
\oplus\Omega^\bullet(\de_1M_2)[k]\oplus\Omega^\bullet(\de_2M_2\setminus\Sigma)[d-k-1]\\
\ni \bA_1'\oplus\bB_1\oplus\bA_2\oplus\bB_2'.
\end{multline*}
The integral may be explicitly computed and yields
\[
\Tilde\psi_M = T_{M_1}\,T_{M_2}\,\EE^{\frac\ii\hbar \Tilde\calS^\text{eff}_{M}}
\]
with
\begin{multline*}
\Tilde\calS^\text{eff}_{M}=
-(-1)^{d-k}\int_\Sigma\sfb_1\sfa_2
+(-1)^{d+kd}\int_{\Sigma\times\de_1M_2}\pi_1^*\sfb_1\,\eta_2\,\pi_2^*\bA_2
-(-1)^{d+kd}\int_{\de_2M_1\times\Sigma}\pi_1^*\bB_1\,\eta_1\pi_2^*\sfa_2-\\
-(-1)^{kd}\int_{\de_2M_1\times\Sigma\times\de_1M_2}
\varpi_1^*\bB_1\,p_1^*\eta_1\,p_2^*\eta_2\,\varpi_3^*\bA_2+\\
+(-1)^{d-k}\left(\int_{\de_2 M_2\setminus\Sigma}\bB_2'\sfa_2 +\int_{\de_2 M_1}\bB_1\sfa_1-
\int_{\de_1 M_2}\sfb_2\bA_2 - \int_{\de_1 M_1\setminus\Sigma}\sfb_1\bA_1'
\right)-\\
-(-1)^{d+kd}\left(\int_{\de_2M_1\times(\de_1M_1\setminus\Sigma)}\pi_1^*\bB_1\,\eta_1\,\pi_2^*\bA_1'
+\int_{(\de_2M_2\setminus\Sigma)\times\de_1M_2}\pi_1^*\bB_2'\,\eta_2\,\pi_2^*\bA_2\right),
\end{multline*}
where $\sfa_i$ and $\sfb_i$, $i=1,2$, are the corresponding $\sfa$ and $\sfb$ variables on $M_i$, and
$\eta_i$ denotes the propagator for $M_i$. In the fourth contribution we also used pullbacks by the following
projections:
\begin{alignat*}{2}
\varpi_1&\colon& \de_2M_1\times\Sigma\times\de_1M_2 &\mapsto \de_2M_1\\
\varpi_3&\colon& \de_2M_1\times\Sigma\times\de_1M_2 &\mapsto \de_1M_2\\
p_1&\colon& \de_2M_1\times\Sigma\times\de_1M_2 &\mapsto \de_2M_1\times\Sigma\\
p_2&\colon& \de_2M_1\times\Sigma\times\de_1M_2 &\mapsto \Sigma\times\de_1M_2.
\end{alignat*}

The propagator $\Tilde\eta$ on $M$ can also be obtained by pairing the states on $M_1$ and $M_2$, see Section~\ref{ss:tildeprop}.

\subsubsection{Reducing the \pseudovacua}\label{ss:reducing backgrounds}
We now wish to reduce the space of {\pseudovacua} by integrating out those appearing in the term
$ \int_\Sigma\sfb_1\sfa_2$. We will refer to them as redshirt \pseudovacua.
More precisely, let
\begin{alignat*}{2}
\tau_1&\colon H^\bullet_{\DD2}(M_1)&\to H^\bullet(\Sigma)\\
\tau_2&\colon H^\bullet_{\DD1}(M_2)&\to H^\bullet(\Sigma)
\end{alignat*}
be the restriction maps induced by the inclusion of $\Sigma$ into $M_1$ and $M_2$.
We denote by $L_1$ ($L_2$) the image of $\tau_1$ ($\tau_2$). We now choose sections
\begin{alignat*}{2}
\sigma_1&\colon L_1 &\to H^\bullet_{\DD2}(M_1)\\
\sigma_2&\colon L_2 &\to H^\bullet_{\DD1}(M_2)
\end{alignat*}
of $\tau_1$ and $\tau_2$.
We will also need the orthogonal complements $L_1^\perp,L_2^\perp \subset H^\bullet(\Sigma)$ with respect to the Poincar\'e pairing on $H^\bullet(\Sigma)$. By Lefschetz duality, $L_i^\perp$ is the image of  $H^\bt(M_i,\de_i M_i\backslash\Sigma)$ in $H^\bullet(\Sigma)$ for $i=1,2$.\footnote{\label{footnote: L^perp}
Indeed, for $[\gamma]\in H^j(\Sigma)$ and $[\alpha]\in H^{d-1-j}_{\DD2}(M_1)$, we have $\langle [\gamma],\tau_1[\alpha] \rangle_\Sigma=\langle B_1[\gamma], [\alpha] \rangle$ where  $\langle,\rangle_\Sigma$ is the Poincar\'e pairing on $H^\bullet(\Sigma)$ and  $\langle,\rangle$ is the Lefschetz pairing between $H^{j+1}_{\DD1}(M_1)$ and $H^{d-1-j}_{\DD2}(M_1)$; $B_1$ is a map in the long exact sequence
$\cdots\rightarrow H^\bt(M_1,\de_1 M_1\backslash\Sigma)\xrightarrow{r_1}H^\bullet(\Sigma)\xrightarrow{B_1} H^{\bullet+1}_{\DD1}(M_1)\rightarrow \cdots$. Therefore, due to nondegeneracy of Lefschetz pairing, $L_1^\perp=\ker B_1=\mathrm{im}(r_1)$. Case of $L_2^\perp$ is treated similarly.
}

Next,
we choose a complement $L_1^\times$ of $L_1\cap L_2^\perp$ in $L_1$ and a complement $L_2^\times$ of $L_1^\perp\cap L_2$ in $L_2$.
Finally, 
denoting
$H^{\bullet}_{\DD2}(M_1)^\#=\ker \tau_1$ and $H^\bullet_{\DD1}(M_2)^\#=\ker\tau_2$,
we end up with the decompositions
\begin{align*}
H^\bullet_{\DD2}(M_1) &= \sigma_1(L_1\cap L_2^\perp)\oplus \sigma_1(L_1^\times) \oplus H^{\bullet}_{\DD2}(M_1)^\#\\
H^\bullet_{\DD1}(M_2) &= \sigma_2(L_1^\perp\cap L_2)\oplus \sigma_2(L_2^\times) \oplus H^{\bullet}_{\DD1}(M_2)^\#
\end{align*}
We use the notations $\sfb_1=\sfb_1^\cap+\sfb_1^\times+\sfb_1^\#$ and
$\sfa_2=\sfa_2^\cap+\sfa_2^\times+\sfa_2^\#$ for the corresponding decompositions of the
{\backgrounds}.
To fix notations for the following, we set
\begin{align*}
H^\bullet_{\DD2}(M_1)' &= \sigma_1(L_1\cap L_2^\perp) \oplus H^{\bullet}_{\DD2}(M_1)^\#  = \tau_1^{-1}(L_1\cap L_2^\perp)\\
H^\bullet_{\DD1}(M_2)' &= \sigma_2(L_1^\perp\cap L_2) \oplus H^{\bullet}_{\DD1}(M_2)^\#  = \tau_2^{-1}(L_1^\perp\cap L_2)
\end{align*}
\begin{align*}
H^\bullet_{\DD1}(M_1)^\circ &= (\sigma_1(L_1\cap L_2^\perp))^*\oplus (H^{\bullet}_{\DD2}(M_1)^\#)^* \subset
H^\bullet_{\DD1}(M_1)=(H^\bullet_{\DD2}(M_1) )^*\\
H^\bullet_{\DD2}(M_2)^\circ &= (\sigma_2(L_1^\perp\cap L_2))^*\oplus (H^{\bullet}_{\DD1}(M_2)^\#)^* \subset
H^\bullet_{\DD2}(M_2)=(H^\bullet_{\DD1}(M_2) )^*
\end{align*}
and
\begin{align*}
\Tilde H^\bullet_{\DD1}(M_1,M_2) &:= H^\bullet_{\DD1}(M_1)^\circ \oplus H^\bullet_{\DD1}(M_2)' \\
\Tilde H^\bullet_{\DD2}(M_1,M_2) &:= H^\bullet_{\DD2}(M_1)' \oplus H^\bullet_{\DD2}(M_2)^\circ
\end{align*}
Notice that classes in $\sigma_1(L_1\cap L_2^\perp)$ and $\sigma_2(L_1^\perp\cap L_2)$ can be extended to the other manifold. The other summands in
the $\Tilde H$'s contain classes that restrict to zero on $\Sigma$ and which can then also be extended. Thus, we get
maps
\begin{subequations}\label{e:h}
\begin{alignat}{2}
h_1 & \colon \Tilde H^\bullet_{\DD1}(M_1,M_2) &\to H^\bullet_{\DD1}(M)\\
h_2 & \colon \Tilde H^\bullet_{\DD2}(M_1,M_2) &\to H^\bullet_{\DD2}(M)
\end{alignat}
\end{subequations}
We will return to this in Section~\ref{ss:redpv}, where we will prove that $h_1$ and $h_2$ are isomorphisms.

Notice that 
we have
$\int_\Sigma\sfb_1\sfa_2=\int_\Sigma\sfb_1^\times\sfa_2^\times$. By writing
\begin{align*}
\sfb_1^\times &= z^{+\times}_{1i}\,\chi^i_{1\times},\\
\sfa_2^\times &= z^i_{2\times}\,\chi^\times_{2i},
\end{align*}
with $\{[\chi^i_{1\times}]\}$ a basis of $\sigma_1(L_1^\times)$ and
$\{[\chi^\times_{2i}]\}$ a basis of $\sigma_2(L_2^\times)$, we also get
\[
\int_\Sigma\sfb_1\sfa_2= (-1)^{k\cdot\deg\chi_{2\times}^i}\;z^{+\times}_{1i}\, z^j_{2\times}\,\Lambda^i_j
\]
with
\begin{equation}\label{e:Lambda}
\Lambda^i_j =
\int_\Sigma \chi^\times_{2j}\,\chi^i_{1\times}.
\end{equation}
Note that the matrix $\Lambda$ is invertible.\footnote{
This is equivalent to the nondegeneracy of the restriction of Poincar\'e pairing  on $H^\bullet(\Sigma)$ to $L_1^\times \otimes L_2^\times\rightarrow \mathbb{R}$. To prove the latter, assume the opposite, i.e. that there is a nonzero $[\alpha]\in L_1^\times$ such that for any $[\beta]\in L_2^\times$, one has $\langle [\alpha],[\beta] \rangle_\Sigma= 0$. Then $[\alpha]$ is orthogonal to the whole $L_2$, since $[\alpha]$ being in $L_1$ is certainly orthogonal to $L_1^\perp\cap L_2$. Hence $[\alpha]\in L_1\cap L_2^\perp$, which is a contradiction to $[\alpha]\in L_1^\times$. Thus we have shown that the left kernel of the pairing $L_1^\times \otimes L_2^\times\rightarrow \mathbb{R}$ vanishes. Vanishing of the right kernel is shown similarly, which finishes the proof of nondegeneracy.
}

We now reduce the space of {\pseudovacua} by integrating over the zero section $\calL^\times$ of
$
T^*[-1](\sigma_1(L_1^\times)[d-k-1]\oplus\sigma_2(L_2^\times)[k])$.
Namely, we integrate out all the $z^{+\times}_{1i}$ and $z^i_{2\times}$ coordinates, the redshirt \pseudovacua, and set
their canonically conjugate variables to zero.
This way we obtain the state
\[
\Check\psi_M = \int_{\calL^\times} \Tilde\psi_M
\]
as a function on $\Check\calZ_M=\Check\calV_M\times\calB^\calP_{\de M}$ with
\begin{equation}\label{e:checkV}
\Check\calV_M = \Tilde H^\bullet_{\DD1}(M_1,M_2)[k]\oplus \Tilde H^\bullet_{\DD2}(M_1,M_2)[d-k-1].
\end{equation}
We denote by $\Check\sfa_1$, $\Check\sfa_2=\sfa_2^\cap+\sfa_2^\#$,
$\Check\sfb_1= \sfb_1^\cap+\sfb_1^\#$ and $\Check\sfb_2$ the corresponding variables. We represent them as $(\dd,\dd^*)$-closed differential forms on $M_1$, $M_2$ 
with appropriate Dirichlet/Neumann boundary conditions, as in (\ref{e:HHHarm},\ref{e:HarmhHarm}).

The integral over $\calL^\times$ can be easily computed and yields
\[
\Check\psi_M = \Check T_M\,\EE^{\frac\ii\hbar\Check\calS^\text{eff}_{M}}
\]
with
\begin{equation}\label{e:T glued}
\Check T_M =\Xi\cdot\frac{T_{M_1}T_{M_2}}{\Ber\Lambda},
\end{equation}
where $\Ber\Lambda$ denotes the Berezinian of $\Lambda$,
and
\begin{multline*}
\Check\calS^\text{eff}_{M}=
(-1)^{d+kd}\left(\int_{\Sigma\times\de_1M_2}\pi_1^*\Check\sfb_1\,\eta_2\,\pi_2^*\bA_2
-\int_{\de_2M_1\times\Sigma}\pi_1^*\bB_1\,\eta_1\pi_2^*\Check\sfa_2\right)-\\
-(-1)^{kd}\int_{\de_2M_1\times\Sigma\times\de_1M_2}
\varpi_1^*\bB_1\,p_1^*\eta_1\,p_2^*\eta_2\,\varpi_3^*\bA_2+\\
+(-1)^{d-k}\left(\int_{\de_2 M_2\setminus\Sigma}\bB_2'\Check\sfa_2 +\int_{\de_2 M_1}\bB_1\Check\sfa_1
-
\int_{\de_1 M_2}\Check\sfb_2
\bA_2 - \int_{\de_1 M_1\setminus\Sigma}\Check\sfb_1\bA_1'
\right)-\\
-(-1)^{d+kd}\left(\int_{\de_2M_1\times(\de_1M_1\setminus\Sigma)}\pi_1^*\bB_1\,\eta_1\,\pi_2^*\bA_1'
+\int_{(\de_2M_2\setminus\Sigma)\times\de_1M_2}\pi_1^*\bB_2'\,\eta_2\,\pi_2^*\bA_2\right)-\\
-\sum_{ij}
(-1)^{d+kd+\deg\chi^\times_{2i}}
V^i_j\,\,
\left(
\int_{\de_2M_1\times\Sigma}\pi_1^*\bB_1\,\eta_1\,\pi_2^*\chi^\times_{2i}-
(-1)^{k+kd}\int_{\de_2 M_2\setminus\Sigma}\bB_2'\chi^\times_{2i}\right)\cdot\\
\cdot\left(
\int_{\Sigma\times\de_1M_2}\pi_1^*\chi^j_{1\times}\,\eta_2\,\pi_2^*\bA_2+
(-1)^{d+k+kd+(d+1)\cdot\deg\chi^\times_{2i}}
\int_{\de_1 M_1\setminus\Sigma}\chi^j_{1\times}\bA_1'
\right).
\end{multline*}
Here we denoted by $V$ the inverse of the matrix $\Lambda$ defined in \eqref{e:Lambda}.

The factor
\begin{equation}\label{e:Xi}
\Xi=(2\pi i)^{\frac12 \dim (\calL^\times)^\mathrm{even}}\cdot \left(\frac{i}{\hbar}\right)^{\frac12 \dim (\calL^\times)^\mathrm{odd}}
=\frac{\xi_M}{\xi_{M_1}\xi_{M_2}} \quad \in \mathbb{C}
\end{equation}
with $\xi$ as in (\ref{e:xi}) appears in (\ref{e:T glued}) because of the $2\pi$, $i$ and $\hbar$ factors coming from the Gaussian integral over a superspace. (The last equality in (\ref{e:Xi}) is non-obvious; we refer the reader to \cite{cell_ab_BF} for details).

From now on we will denote the boundary fields on $M$ by $\Check\bA$ and $\Check\bB$.
The restriction of $\Check\bA$ to $\de_1M_1\setminus\Sigma$ is what we denoted so far by $\bA_1'$, whereas
the restriction of $\Check\bA$ to $\de_1 M_2$ is what we denoted so far by $\bA_2$.
Similarly, restriction of $\Check\bB$ to $\de_2M_1$ is what we denoted so far by $\bB_1$, whereas
the restriction of $\Check\bB$ to $\de_2 M_2\setminus\Sigma$ is what we denoted so far by $\bB_2'$.

For the {\pseudovacua} we will adopt the collective notation $\Check\sfa$ and $\Check\sfb$.
The restriction of $\Check\sfa$ to $M_2$ is what we denoted so far by $\Check\sfa_2$. On the other hand,
the restriction of $\Check\sfa$ to $M_1$ is the sum $\Check\sfa_1 + \sfa_2^\text{ext}$. The extension
$\sfa_2^\text{ext}$ of $\Check\sfa_2$ to $M_1$ is defined by
\[
\int_{M_1}\gamma\sfa_2^\text{ext}=(-1)^{d+(d-1)\cdot\deg\gamma}\int_{M_1\times\Sigma}
\pi_1^*\gamma\,\eta_1\,\pi_2^*\Check\sfa_2=(-1)^{d+(d-1)\cdot\deg\gamma}\int_{M_1\times\Sigma}
\pi_1^*\gamma\,\eta_1\,\pi_2^*\sfa_2^\cap,
\]
where $\gamma$ is a form on $M_1$.
Similarly,
the restriction of $\Check\sfb$ to $M_1$ is what we denoted so far by $\Check\sfb_1$. On the other hand,
the restriction of $\Check\sfb$ to $M_2$ is the sum $\Check\sfb_2 + \sfb_1^\text{ext}$. The extension
$\sfb_1^\text{ext}$ of $\Check\sfb_1$ to $M_2$ is defined by
\[
\int_{M_2}\sfb_1^\text{ext}\mu=(-1)^{d+k+kd}\int_{\Sigma\times M_2}
\pi_1^*\Check\sfb_1\,\eta_2\,\pi_2^*\mu=(-1)^{d+k+kd}\int_{\Sigma\times M_2}
\pi_1^*\sfb_1^\cap\,\eta_2\,\pi_2^*\mu,
\]
where $\mu$ is a form on $M_2$.

With these notations and with the explicit form for the glued propagator $\Check\eta$ of
Appendix~\ref{We now do the final step in computing the propagator on $M$ for the reduced  space of pseudovacua},
we finally
get
\[
\Check\calS^\text{eff}_{M}=
(-1)^{d-k}\left(\int_{\de_2 M}\Check\bB\Check\sfa -\int_{\de_1 M}\Check\sfb\Check\bA\right)
-(-1)^{d+kd}\int_{\de_2M\times\de_1M}\pi_1^*\Check\bB\,\Check\eta\,\pi_2^*\Check\bA,
\]
which, upon the change of notations, coincides with the one in \eqref{e:effS}.

Observe that $\Check T_M$ is equal to $T_M$, by the gluing properties of Reidemeister torsions (cf. e.g. \cite{Milnor}).
This implies $\Check\psi_M=\Hat\psi_M$.

\begin{Rem}\label{rem:a ext, b ext}
Residual fields $\Check{a}$, $\Check{b}$, as constructed above, are represented by closed forms on $M$ which are smooth away from $\Sigma$ but generally discontinuous through $\Sigma\subset M$; however they have a well-defined smooth pullback to $\Sigma$.
\end{Rem}

\begin{Rem}
Representatives of the cohomology $H^\bullet_{\DD1}(M)$, $H^\bullet_{\DD2}(M)$ constructed via the extension defined above are exactly the ones appearing in the differential of the glued propagator of Appendix~\ref{We now do the final step in computing the propagator on $M$ for the reduced  space of pseudovacua}, as in (\ref{e:deta}). This can be checked either by a brute force calculation, or, more concisely, via homological perturbation theory (see \cite{cell_ab_BF}).
\end{Rem}



\subsubsection{The reduced space of \pseudovacua}\label{ss:redpv}
\begin{Lem}
Maps $h_1,h_2$ defined in (\ref{e:h}) are isomorphisms.
\end{Lem}
\begin{proof}
We will consider $h_1$; the case of $h_2$ is treated similarly.

Recall that for a triple of topological spaces $X\supset Y\supset Z$ one has the long exact sequence of cohomology of the triple
\begin{equation}\label{e:LES of a triple}
\cdots\rightarrow H^\bullet(X,Y)\rightarrow H^\bullet(X,Z)\rightarrow H^\bullet(Y,Z)\rightarrow H^{\bullet+1}(X,Y)\rightarrow\cdots
\end{equation}

Consider the triple $X=M$, $Y=M_2\cup \de_1 M_1$, $Z=\de_1 M$. Then the sequence (\ref{e:LES of a triple}) becomes
\begin{multline}\label{e:LES}
\cdots\rightarrow H^\bullet(M,M_2\cup \de_1 M_1)\xrightarrow{\kappa} H^\bullet(M,\de_1 M) \xrightarrow{\lambda}\\
\xrightarrow{\lambda} H^\bullet(M_2\cup \de_1 M_1,\de_1 M)
\xrightarrow{\rho} H^{\bullet+1}(M,M_2\cup \de_1 M_1)\rightarrow\cdots
\end{multline}
Note that, by excision property of cohomology, we have $H^\bullet(M,M_2\cup \de_1 M_1)=H^\bullet(M_1,\de_1 M_1)$ and $H^\bullet(M_2\cup \de_1 M_1,\de_1 M)=H^\bullet(M_2,\de_1 M_2)$. Thus (\ref{e:LES}) becomes
\begin{equation}\label{e:LES2}
\cdots\rightarrow H^\bullet_{\DD1}(M_1)\xrightarrow{\kappa} H^\bullet_{\DD1}(M) \xrightarrow{\lambda} H^\bullet_{\DD1}(M_2)\xrightarrow{\rho}H^{\bullet+1}_{\DD1}(M_1)\rightarrow\cdots
\end{equation}
Therefore for the cohomology of $M$ we have
\begin{equation} \label{e:H(M)}
H^\bullet_{\DD1}(M)\simeq\mathrm{im}(\lambda)\oplus \mathrm{im}(\kappa)= \ker \rho\oplus \frac{H^\bullet_{\DD1}(M_1)}{\mathrm{im}(\rho)}.
\end{equation}
Note that the connecting homomorphism $\rho$ in (\ref{e:LES2}) factorizes as $H^\bullet_{\DD1}(M_2)\xrightarrow{\tau_2}H^\bullet(\Sigma)\xrightarrow{B_1}H^{\bullet+1}_{\DD1}(M_1)$  (with $B_1$ as in Footnote \ref{footnote: L^perp}). This implies
$$\ker \rho=\tau_2^{-1}(\ker B_1)=\tau_2^{-1}(L_1^\perp\cap L_2)=H^\bullet_{\DD1}(M_2)'.$$
For the image of $\rho$ we have $\mathrm{im}(\rho)=B_1(L_2)=B_1(L_2^\times)\subset H_{\DD1}(M_1)$. Its annihilator in $H_{\DD2}^\bullet(M_1)$ is
$$
\mathrm{Ann}(\mathrm{im}\,\rho)=\{[\alpha]\in H_{\DD2}^\bullet(M_1)\,:\, \underbrace{\langle [\alpha], B_1[\gamma] \rangle}_{=\langle \tau_1[\alpha],[\gamma]\rangle_\Sigma}=0 \;\; \forall [\gamma]\in L_2\}
=\tau_1^{-1}(L_2^\perp)=\tau_1^{-1}(L_1\cap L_2^\perp).
$$
Therefore, for the second term in (\ref{e:H(M)}) we have
$$\frac{H^\bullet_{\DD1}(M_1)}{\mathrm{im}(\rho)}=\left( \mathrm{Ann}(\mathrm{im}\rho) \right)^*=\left(\tau_1^{-1}(L_1\cap L_2^\perp)\right)^*= H^\bullet_{\DD1}(M_1)^\circ. $$
Thus we have constructed the isomorphism
$$H^\bullet_{\DD1}(M)\simeq H^\bullet_{\DD1}(M_1)^\circ\oplus H^\bullet_{\DD1}(M_2)'.$$
By inspection of the construction, it is precisely the inverse of $h_1$ of (\ref{e:h}).
\end{proof}



\newcommand{\Szero}{\mathcal{S}_{M,0}}
\newcommand{\Spert}{\mathcal{S}_{M,\text{pert}}}

\section{$BF$-like theories}\label{sec: BF-like}
In this Section we consider interacting theories that deform abelian $BF$ theories.
This means first that as unperturbed theory we consider $n$ copies
of an abelian $BF$ theory,
\[
\Szero = \sum_{i=1}^n \int_M \sfB_i\,\dd\sfA^i,
\]
with $\sfA^i\oplus \sfB_i\in \Omega^\bullet(M)[k_i]\oplus\Omega^\bullet(M)[d-k_i-1]$ for some choice of $k_i$.
Equivalently, we may define $\calF_M=(\Omega^\bullet(M)\otimes V[1])\oplus(\Omega^\bullet(M)\otimes V^*[d-2])$
where $V$ is a graded vector space and write\footnote{We recover the previous notation if
we pick a graded basis $\sfe^i$ of $V$ and its dual basis $\sfe_i$, set $k_i=1-|\sfe_i|$
and write
$\sfA = \sum_{i=1}^n \sfe^i\sfA_i$, $\sfB=\sum_{i=1}^n(-1)^{1-k_i}\sfB^i\sfe_i$.}
\[
\Szero =  \int_M \braket{\sfB}{\dd\sfA},
\]
where $\braket{\ }{\ }$ denotes the canonical pairing between $V^*$ and $V$. The whole Section~\ref{s:abeBF}
can now be extended with obvious modifications.

Next we consider an interacting term
that is the integral of a density\ndash valued function $\calV$ of the fields $\sfA$ and $\sfB$,
\[
\Spert = \int_M \calV(\sfA,\sfB),
\]
such that $\calS_M:=\Szero+\Spert$ solves the classical master equation for $M$ without boundary. We view $\Spert$ as a ``small'' perturbation (cf. Remark \ref{r:S_0+S_pert}). We further require that $\calV$ depends only on the fields, but not on their derivatives.
We consider three examples:
\begin{Exa}[Quantum mechanics]\label{e:QM}
This is the case when $d=1$ and $V=W[-1]$, with $W$ concentrated in degree zero. We denote by $P$ and $Q$
the degree-zero zero forms components of $\sfB$ and $\sfA$, respectively. We choose a volume form
$\dd t$ on $M$ and a function $H$ on $T^*W$. We then set
$\calV(\sfA,\sfB):=H(\sfA,\sfB)\,\dd t = H(Q,P)\,\dd t$. We then have
\[
\calS_M=\int_M\left(
\sum_i P_i\dot Q^i + H(Q,P)
\right)\dd t,
\]
the classical action of mechanics in Hamilton's formalism.
\end{Exa}
\begin{Exa}[AKSZ theories \cite{AKSZ}]\label{exa-AKSZ}
In this case we assume that we are given a function $\Theta$ on $T^*[d-1](V[1])=V[1]\oplus V^*[d-2]$ that has
degree $d$ and Poisson commutes with itself with respect to the canonical graded Poisson structure on
the shifted cotangent bundle. We then set $\calV(\sfA,\sfB)$ to be the top degree part
of $\Theta(\sfA,\sfB)$.
Notice that this is a special case of the construction in \cite{AKSZ}, where the target is not assumed to be a shifted cotangent bundle but just a general graded symplectic manifold with symplectic form of
degree $d-1$. We have three particular cases of interest:
\begin{description}
\item[$BF$ theories] Here we assume $V=\frg$ to be a Lie algebra and set $\Theta=\frac12\braket{b}{[a,a]}$ with
$a\in V[1]$ and $b\in V^*[d-2]$.
\item[Split Chern--Simons theory] If we are given a Lie algebra $\frg$ with an invariant pairing, we can define a
function $\Theta$ of degree $3$ on $\frg[1]$ by $\Theta=\frac16({a},{[a,a]})$.
This fits with our setting if $d=3$
and we have a decomposition of $\frg$, as a vector space, $\frg=V\oplus W$ where $V$ and $W$ are maximally isotropic subspaces. The pairing allows identifying $W$ with $V^*$.
\item[The Poisson sigma model] If $(P,\pi)$ is a Poisson manifold, the Poisson sigma model on $M$
has as its space of fields
\[
\calF_M=\Map(T[1]M,T^*[1]P)
\]
and $\Theta$
is the Poisson bivector field $\pi$ regarded
as a function of degree $d=2$ on $T^*[1]P$. This fits with our setting if $P$ is a vector space $W$ and we set
$V=W[-1]$. More generally, we may perturb the general Poisson sigma model around a constant map
$x\colon M\to P$ and we fit again in our setting with $V=T_xP[-1]$.
\end{description}
\end{Exa}
\begin{Exa}[2D Yang--Mills theory]\label{r:twoDYM}
The classical action of Yang--Mills (YM) theory can be written in the first order formalism as
$\int_M \left(\braket B{F_A} +\frac12g^2(B,{*B})\right)$ where $A$ is a connection on a principal
$G$\ndash bundle over $M$, $F_A$ its curvature,
$B$ a $(d-2)$\ndash form of the coadjoint type, $(\ , \ )$ a nondegenerate, invariant pairing
on the dual $\frg^*$ of the
Lie algebra $\frg$ of $G$,
$*$ the Hodge star for some reference metric, and
$g$ a coupling constant. This action looks like a perturbation of $BF$ theory, with $V=\frg$,
but for $d>2$ the perturbation
$\int_M  (B,{*B})$ breaks the symmetry; 
hence the corresponding BV theory is not
a perturbation of the BV version of $BF$ theory. This is due to the fact that
one of the symmetries of $BF$ theory consists in adding the covariant derivative of a $(d-3)$\ndash form to $B$.
However, for $d=2$ this symmetry is absent, so indeed in two dimensions  YM theory is a perturbation of $BF$ theory. We can write
the corresponding BV action as
\[
\calS_M=\int_M\left(
\braket{\sfB}{\dd\sfA}+
\frac12\braket{\sfB}{[\sfA,\sfA]}+
\frac12g^2 v (B,B)
\right)
\]
where $v$ is the volume form associated to the fixed metric on $M$  
and $B$ denotes the degree zero zero\ndash form in $\sfB$.
More generally, for any coad-invariant function $f$ on $\frg^*$, the BV action
\[
\calS_M=\int_M\left(
\braket{\sfB}{\dd\sfA}+
\frac12\braket{\sfB}{[\sfA,\sfA]}+
v f(B)
\right)
\]
solves the classical master equation on a two\ndash manifold $M$ without boundary and perturbs $BF$ theory.
Notice
that, by degree reasons, we have
\[
\calV(\sfA,\sfB) = \frac12\braket{\sfB}{[\sfA,\sfA]} + v f(B) = \frac12\braket{\sfB}{[\sfA,\sfA]} + v f(\sfB).
\]
We call this theory the \textsf{generalized two\ndash dimensional YM theory}.
\end{Exa}

Notice that, whereas the AKSZ theories of Example~\ref{exa-AKSZ} are topological, quantum mechanics and
YM theory are not.

\begin{Rem}
YM theory in 4 dimensions can also be regarded as a perturbation of a $BF$\ndash like theory \cite{Costello}.
The main difference is that the $\dd$ operator appearing in the unperturbed term is not the de~Rham differential.
This changes the propagator, but the algebraic structure is the same as the one considered in this paper.
\end{Rem}

\subsection{Perturbative expansion}\label{s:BF-like pert}
The assumption that $\calV(\sfA,\sfB)$ does not depend on derivatives of the field implies that
the space of boundary fields on a $(d-1)$\ndash manifold $\Sigma$
is exactly the same as for the unperturbed theory,
$\calF^\de_\Sigma=(\Omega^\bullet(\Sigma)\otimes V[1])\oplus(\Omega^\bullet(\Sigma)\otimes V^*[d-2])$, with the
same symplectic structure $\omega^\de_\Sigma= \delta\alpha^\de_\Sigma$ and
\[
\alpha^\de_\Sigma = (-1)^{d}\int_\Sigma \braket{\sfB}{\delta\sfA}.
\]
On the other hand the perturbation may affect the boundary cohomological vector field $Q^\de_\Sigma$
and the boundary action $\calS^\de_\Sigma$.
\begin{Rem}\label{r:deAKSZ}
In the case of an AKSZ theory, one has \cite{CMR}
\[
\calS^\de_\Sigma = \int_\Sigma \left(
\braket{\sfB}{\dd\sfA} + \Theta(\sfA,\sfB)
\right).
\]
\end{Rem}
\begin{Rem}\label{r:de2YM}
In the case of the generalized two\ndash dimensional YM theory, the non-AKSZ term $vf(B)$ produces a vertical term in $Q_M$. Hence,
$Q^\de_{\de M}$ is the same as for $BF$ theory. As a consequence,
\[
\calS^\de_\Sigma = \int_\Sigma \left(
\braket{\sfB}{\dd\sfA} + \Theta(\sfA,\sfB)
\right)=
\int_\Sigma \left(
\braket{\sfB}{\dd\sfA} + \frac12\braket{\sfB}{[\sfA,\sfA]}
\right).
\]
\end{Rem}
We then proceed as in Section~\ref{s:abeBF} and choose polarizations as in subsection~\ref{s:pol}. Notice that
the term to be added to the action to make it compatible with the polarization now reads
\[
f^\calP_{\de M}=(-1)^{d-1}\int_{\de_2 M}\braket\sfB\sfA.
\]
We denote again by $\bA$ the coordinate on $\calB_1$ and by $\bB$ the coordinate on $\calB_2$, which we have to extend by zero in the bulk.
We have
\begin{align*}
\sfA&= \sfa + \alpha,\\
\sfB&=\sfb + \beta,
\end{align*}
where $\sfa$ and $\sfb$ denote the residual fields, and $\alpha$ and $\beta$ denote the fluctuations.
For the unperturbed part we proceed exactly as in Section~\ref{s:abeBF}, getting
\[
\calS^\calP_M =\Hat\calS_{M,0} + \Hat\calS_{M,\text{pert}}  + \calS_M^\vac + \calS_M^\text{source}
\]
with
\begin{align*}
\Hat\calS_{M,0}&=\int_M\braket\beta{\dd\alpha},\\
\Hat\calS_{M,\text{pert}} &= \calV(\sfa+\alpha,\sfb+\beta),\\
\calS_M^\vac&=(-1)^{d-1}\left(\int_{\de_2 M}\braket\bB\sfa -\int_{\de_1 M}\braket\sfb\bA\right),\\
\calS_M^\text{source}&=
(-1)^{d-1}\left(\int_{\de_2 M}\braket\bB\alpha -\int_{\de_1 M}\braket\beta\bA\right).
\end{align*}
The propagator is determined, exactly like in the abelian case, by $\Hat\calS_{M,0}$. The perturbation term
$\Hat\calS_{M,\text{pert}}$ has to be Taylor expanded around zero and produces the interaction vertices.
In addition we have univalent vertices on the boundary. The Feynman diagrams of the theory with boundary then also contain edges connecting to the boundary.

Ultimately, the perturbative expansion for the state  takes the form
\begin{equation}\label{e: psi pert}
\hat\psi_M=\prod_{i=1}^nT_M^{(k_i)}\cdot \exp\left(\frac\ii\hbar\sum_\Gamma \frac{(-\ii\hbar)^{\mathrm{loops}(\Gamma)}}{|\mathrm{Aut}(\Gamma)|}\int_{C_\Gamma}\omega_\Gamma(\bA,\bB;\sfa,\sfb)\right)
\end{equation}
where $T_M^{(k)}$ is as in (\ref{e:T via tau}), for the field grading shift $k$. In the exponential, we sum over connected Feynman diagrams -- connected oriented graphs $\Gamma$ -- with
\begin{itemize}
\item $n\geq 0$ bulk vertices in $M$ decorated by ``vertex tensors''
$\left.\frac{\de^{s+t}}{\de \sfA_{i_1}\cdots\de \sfA_{i_s}\de\sfB^{j_1}\cdots\de\sfB^{j_t}}\right|_{\sfA=\sfB=0}\calV(\sfA,\sfB)$
where $s,t$ are the out- and in-valencies of the vertex,
\item $n_1\geq 0$ boundary vertices on $\de_1 M$ with single incoming half-edge and no outgoing half-edges decorated by $\bA_i$ evaluated at the point (vertex location) on $\de_1 M$,
\item $n_2\geq 0$ boundary vertices on $\de_2 M$ with single outgoing half-edge and no incoming half-edges decorated by $\bB^i$ evaluated at the point on $\de_2 M$,
\item edges are decorated with the propagator $\eta\cdot\delta^i_j$, with  $\eta$ same as in Section \ref{ss:prop},\footnote{More generally, if the shifts $k_i$ are different for different field components, we put $\eta^{(k_i)}\cdot \delta^i_j$ on the edge, where $\eta^{(k)}$ is the propagator for abelian $BF$ theory with field grading shift $k$.}
\item loose half edges (leaves) are allowed and are decorated with the residual fields $\sfa_i$ (for out-orientation), $\sfb^i$ (for in-orientation).
\end{itemize}
The differential form $\omega_\Gamma(\bA,\bB;\sfa,\sfb)$ on the compactified configuration space $C_\Gamma$ of points on $M$ (with $n$ bulk points, $n_1$ points on $\de_1$ and $n_2$ points on $\de_2$) is the wedge product of the decorations above, with field component indices $i$ contracted according to the combinatorics of $\Gamma$. Note that $\omega_\Gamma$ is a polynomial in boundary and residual fields of order determined by the numbers of boundary vertices and leaves in $\Gamma$.

\begin{Rem}[Short loops]
The perturbative expansion has potential singularities when we contract a fluctuation $\alpha$ with a fluctuation
$\beta$ in the same interaction vertex (short loops).
In AKSZ theories, short loops are absent if
 a unimodularity condition of the target structure is satisfied.\footnote{If the Euler characteristic of $M$ vanishes,
 one does not even have
to impose the unimodularity condition and one can simply disregard short loops.
 This is why, e.g., the Poisson sigma model is well defined on the upper half plane and on the torus for every Poisson structure.
Notice that short loops contributions are needed for the (modified) quantum master equation to hold. To match
\eqref{e:deta}, one has to assign a $(d-1)$\ndash form $\eta_\text{sl}$ to a short loop on $M$ such that
\[
\dd\eta_\text{sl} = (-1)^{d-1}\sum_i (-1)^{d\cdot\deg\chi_i}\chi_i\,\chi^i.
\]
Notice that the right hand side is precisely exact when the Euler characteristic of $M$ vanishes.}
\end{Rem}

Formally the gluing procedure is exactly as in subsection~\ref{ss:gluing}. The integral over the boundary fields forces the matching of the boundary vertices. Next one  has to integrate over the redshirt residual fields.
\begin{Prop}[Gluing]\label{p:gluing4}
Let $M$ be cut along a codimension-one submanifold $\Sigma$ into $M_1$ and $M_2$. Let $\psi_{M_1}$ and $\psi_{M_2}$ be the states
for $M_1$ and $M_2$ with a choice of residual fields and propagators and transverse ($\bA$ vs.\ $\bB$) polarizations on $\Sigma$. Then the gluing of $\psi_{M_1}$ and $\psi_{M_2}$ is
the state $\psi_M$ for $M$ with the consequent choice of residual fields and propagators.
\end{Prop}
\begin{proof}[Sketch of the proof]
The gluing of the prefactors (the torsions) and the BV pushforward on the redshirt residual fields (Mayer--Vietoris)
are as in the abelian theory. The explicit integration over the boundary fields and the redshirt residual fields has the effect to produce the $M$\ndash propagators out of the $M_1$\ndash\ and $M_2$\ndash propagators
(see Appendix~\ref{a:comp}).
\end{proof}

\subsubsection{The full state}\label{sss:fullstate}
The state as described above---to which we will refer as the \textbf{principal part} of the state---is all what we need for gluing purposes. However, it may be incorrect as for the modified quantum master equation.
The problem lies in the fact that in general $\Omega$ will contain higher functional derivatives and one has to be careful in defining them appropriately.

\newcommand{\vev}[1]{\left\langle#1\right\rangle}

Let us start the discussion with the present field theory version of \eqref{e:bp}. We focus on the $\de_1M$ boundary where we work in the $\bA$\ndash representation
(the $\de_2 M$ boundary is treated analogously). There the base coordinate $b$ is $\bA$, whereas the fiber coordinate $p$ is
$\iota^*_{\de_1 M}\sfB=\iota^*_{\de_1 M}(\beta+\sfb)$. In the following we will refer to $\bA$ (and similarly
to $\bB$) as to a \textbf{base boundary field}.
Equation \eqref{e:bp} works indeed. To make this more precise,
we average the functional derivatives at a point by a test form $F$ (a smooth differential form
possibly depending on residual and on base boundary fields). We have
\[
\int_{\de_1 M} F^i\,\frac\delta{\delta\bA^i}\, \calS^\calP_M =(-1)^d\int_{\de_1 M} (\beta_i+\sfb_i)\,F^i.
\]
To move to \eqref{e:mQMEP} we have to assume that a higher functional derivative with respect to $\bA$ applied to $\EE^{\frac\ii\hbar\calS_M^\calP}$ will produce
multiplication by the corresponding power of $\iota^*_{\de_1 M}(\beta+\sfb)$. This also works with the naive definition of a higher functional derivative. For example,
\[
\int_{\de_1 M} F^{ij}\,\frac{\delta^2}{\delta\bA^i\delta\bA^j}\,\EE^{\frac\ii\hbar\calS_M^\calP}=\left(\frac\ii\hbar\right)^2
\EE^{\frac\ii\hbar\calS_M^\calP}\int_{\de_1 M} (\beta_i+\sfb_i)(\beta_j+\sfb_j)\,F^{ij}.
\]
Problems arise when we move to the functional integration. The point is that the right hand side of the above equation now involves a quadratic vertex at the boundary.
To be more precise, the principal part of the state 
can be written as $Z\vev{\EE^{\frac\ii\hbar(\calS_M^\text{res}+\calS_M^\text{source})}}$, where $Z$ is the product of torsions and $\vev{\ }$ denotes the expectation value for
the bulk theory. The problem is that a higher functional derivative of this expectation value may differ from the expectation value of the higher functional derivative, 
for the latter also includes Feynman diagrams that remain connected after removing the boundary vertex corresponding to the insertion of the higher
power of $\beta$.

The way out is to define the higher functional derivative in a way that agrees with the naive expectation above but does not have this problem; at the same time one has to define the product of functionals (as in the exponential) appropriately. This is easily achieved by introducing \textbf{composite fields}    (as e.g. in \cite{Ans,Collins}) as higher powers of $\bA$ at a point and regarding
higher functional derivatives as first-order functional derivatives with respect to the corresponding composite field. To make this fit with the naive expectation where a higher functional derivative concentrates the fields on some diagonal, we should also understand the product of integrals as containing the diagonal contributions for the corresponding composite field. Namely, we set
\begin{multline*}
\int_{\de_1 M} u_i\,\bA^i\bullet\int_{\de_1 M} v_j\,\bA^j:=\\
(-1)^{|\bA^i|(d-1+|v_j|)+|u_i|(d-1)}\left(
\int_{C_2(\de_1 M)} \pi_1^*u_i\pi_2^*v_j\,\pi_1^*\bA^i\pi_2^*\bA^j+
\int_{\de_1 M} u_iv_j [\bA^i\bA^j]
\right),
\end{multline*}
where $u$ and $v$ are smooth differential forms depending on bulk and residual fields and $[\bA^i\bA^j]$ is our notation for the composite field.
Now the operator $\int_{\de_1 M} F^{ij}\,\frac{\delta^2}{\delta\bA^i\delta\bA^j}$ has to be interpreted as
$\int_{\de_1 M} F^{ij}\,\frac{\delta}{\delta[\bA^i\bA^j]}$, so we get
\[
\int_{\de_1 M} F^{ij}\,\frac{\delta^2}{\delta\bA^i\delta\bA^j}\,
\left(\int_{\de_1 M} u_i\,\bA^i\bullet\int_{\de_1 M} v_j\,\bA^j\right)
=
\int_{\de_1 M} u_iv_j F^{ij}
\]
in accordance with our naive expectation.

We now formalize the above construction.
For a multi-index $I=(i_1,\dots,i_p)$,
the symbol $\left[\bA^{I}\right]$, or equivalently $\left[\bA^{i_1}\cdots\bA^{i_p}\right]$, denotes a new
composite field of degree
$k-(p-1)(d-1)$ where $k$ is the sum of the degrees of $\bA^{i_1},\dots,\bA^{i_p}$ (one way to remember this is to think of the composite field as being obtained
by integrating the $\bA$ fields around the point where we evaluate the composite field). 
The functional derivative $\frac{\delta^p}{\delta\bA^{i_1}\cdots\delta\bA^{i_p}}$ is interpreted as
$\frac\delta{\delta\left[\bA^{i_1}\cdots\bA^{i_p}\right]}$. Analogously we consider composite $\bB$\ndash fields and their corresponding functional derivatives.
These operators act on the algebra generated by linear combinations of expressions of the form
\[
\int_{C_{m_1}(\de_1 M)\times C_{m_2}(\de_2 M)} 
L_{I_1\cdots I_2\cdots }^{J_1\cdots J_2\cdots}
\, \pi_{1}^*\left[\bA^{I_1}\right]\cdots
\pi_{m_1}^*\left[\bA^{I_{m_1}}\right]
\, \pi_{1}^*\left[\bB_{J_1}\right]\cdots
\pi_{m_2}^*\left[\bB_{J_{m_2}}\right],
\]
where the $L$s are smooth differential forms depending on the fluctuations and on the residual fields.
The product (denoted by $\bullet$) of two expressions as above is obtained by adding all the possible ways
of restricting to a diagonal in the product of the spaces; whenever we do that, the $\bA$s or the $\bB$s from the different brackets are put together. We give
one more example: 
\begin{multline*}
\int_{C_2(\de_1M)} L_{ij}\,\pi_1^*[\bA^i]\pi_2^*[\bA^j] \bullet \int_{\de_1M} \sfb_k\,[\bA^k]=
\int_{C_3(\de_1M)} \pi_{12}^*L_{ij}\pi_3^*\sfb_k\,\pi_1^*[\bA^i]\pi_2^*[\bA^j] \pi_3^*[\bA^k] +\\+
(-1)^{|\bA^j||\bA^k|}\int_{C_2(\de_1M)} L_{ij}\pi_1^*\sfb_k\,\pi_1^*[\bA^i\bA^k]\pi_2^*[\bA^j] +
\int_{C_2(\de_1M)} L_{ij}\pi_2^*\sfb_k\,\pi_1^*[\bA^i]\pi_2^*[\bA^j\bA^k].
\end{multline*}
With this piece of notation, where $\EE_\bullet$ is the exponential defined by the $\bullet$\ndash product,
we now have that
\[
\int_{\de_1 M} F^I\frac{\delta^{I}}{\delta A^I}\,\vev{\EE_\bullet^{\frac\ii\hbar(\calS_M^\text{res}+\calS_M^\text{source})}}
=\vev{\int_{\de_1 M} F^I\frac{\delta^{I}}{\delta A^I}\,
\EE_\bullet^{\frac\ii\hbar(\calS_M^\text{res}+\calS_M^\text{source})}},
\]
where $F$ is a smooth differential form depending on residual, base boundary and (possibly) composite fields.

\newcommand{\bpsi}{\boldsymbol{\psi}}

Finally, we come to the correct definition of the state, which we call the full state
and write in boldface: 
\[
\hat\bpsi_M = Z \vev{\EE_\bullet^{\frac\ii\hbar(\calS_M^\text{res}+\calS_M^\text{source})}},
\]
where we just have replaced the exponential with the $\bullet$\ndash exponential. In terms of Feynman diagrams we now have additional boundary vertices of higher valency. The combinatorics may be simplified by observing that,
for any form $\gamma$,
\[
\EE_\bullet^{\sum_i\int_{\de_1M}\bA^i\gamma_i}=
\EE^{\sum_{I:|I|>0}\frac{\epsilon_I}{|I|!}\int_{\de_1M}[\bA^I]\gamma_I},
\]
where on the right hand side we have the usual exponential and $\epsilon_I$ is a sign, implicitely determined by
\[
A^{i_1}\gamma_{i_1}\cdots A^{i_p}\gamma_{i_p}=
\epsilon_{i_1\cdots i_p}A^{i_1}\cdots A^{i_p}\gamma_{i_1}\cdots\gamma_{i_p}.
\]
We have an analoguous expression for $\bB$.

Note that, when gluing states we do not see the composite fields (the proof of this statement relies on the explicit formula for the glued propagators). For this purposes it is enough to consider the {principal part}
$\hat\psi_M$ of the state.

In abelian $BF$ theory, $\Omega$ contains functional derivatives up to the first order. For this reason
we did not bother
introducing
the $\bullet$\ndash exponential. Note that the
full state is just
$\Hat\bpsi_M= T_M\, \EE_\bullet^{\frac\ii\hbar\calS^\text{eff}_M}$,
whereas its principal part was $\Hat\psi_M= T_M\, \EE^{\frac\ii\hbar\calS^\text{eff}_M}$.
For perturbed $BF$ theories, the full state however is in general not just the bullet exponential of the effective
action appearing in the principal part.


The strategy for checking the modified quantum master equation as well as the fact that $\Omega$ squares to zero simply relies on computing boundary contributions in the compactified configuration spaces appearing in the Feynman diagram expansion for the state. Before doing this, we make the definition of the space of states and
its algebra of differential operators more precise (essentially, the only addition to the above, is the possibility of products of composite fields, for these contributions are generated by the application of differential operators).

\subsubsection{The space of states}\label{s:4.1.1 states}
In Section~\ref{sss:ss} we gave a description of the space of states for (possibly perturbed) $BF$ theories.
Now we have to refine the structure of the distributional
forms $R^j_{n_1n_2}(\sfa,\sfb)$ to allow for a proper definition of the higher functional derivatives
with respect to $\bA$ and $\bB$ that may appear in $\Omega$ following the discussion of
Section~\ref{sss:fullstate}. We then come to the following definition:
%
%
%
A \textbf{regular functional} on the space of base boundary fields 
is a linear combination of expressions of the form
\[
\int_{C_{m_1}(\de_1 M)\times C_{m_2}(\de_2 M)} \!\!\!\!\!
L_{I_1^1\cdots I_1^{r_1}I_2^1\cdots I_2^{r_2}\cdots}^{J_1^1\cdots J_1^{l_1}J_2^1\cdots J_2^{l_2}\cdots}
\, \pi_{1}^*\prod_{j=1}^{r_1}\left[\bA^{I_1^j}\right]\cdots
\pi_{m_1}^*\prod_{j=1}^{r_{m_1}}\left[\bA^{I_{m_1}^j}\right]
\, \pi_{1}^*\prod_{j=1}^{l_1}\left[\bB_{J_1^j}\right]\cdots
\pi_{m_2}^*\prod_{j=1}^{l_{m_2}}\left[\bB_{J_{m_2}^j}\right]
,
\]
where the $I_i^j$ and $J_i^j$ are (target) multi-indices and
$L_{I_1^1\cdots I_1^{r_1}I_2^1\cdots I_2^{r_2}\cdots}^{J_1^1\cdots J_1^{l_1}J_2^1\cdots J_2^{l_2}\cdots}$
is a smooth differential form on the product of compactified configuration spaces $C_{m_1}(\de_1 M)$ and $C_{m_2}(\de_2 M)$
depending on the residual fields.

We assume that at each point in the configuration space there is a field insertion (otherwise we may integrate that point out and get a new $L$); i.e., we
have the conditions
\begin{gather*}
|I^1_s| + |I^2_s| +\dots+|I^{r_s}_s|>0 \text{ for all }s=1,\dots,m_1,\\
|J^1_s| + |I^2_s| +\dots+|J^{l_s}_s|>0 \text{ for all }s=1,\dots,m_2.
\end{gather*}
The space of the states is the span of the regular functionals (multiplied by $T_M$).

We may extend the bullet product to the regular functionals. Notice that the derivative with respect to
a residual field satisfies the Leibniz rule also with respect to the bullet product.

\begin{Rem}
Note that we have only allowed insertions of $\bA$ and $\bB$ in the states but not of their derivatives.
If we only consider states that may appear from the bulk and from the application of $\Omega$ to them,
it is enough to work with this restricted definition: applying $\Omega$ will produce terms containing
$\dd\bA$ and $\dd\bB$, but it is always possible to integrate by parts and move all the derivatives on the coefficients
(see below).
\end{Rem}

\subsubsection{Operators}\label{s:4.1.2 operators}
We now come to the class of operators we consider acting on the space of states defined above.

One term of $\Omega$ that is always present, as we work in perturbation theory, is the one corresponding to
 abelian $BF$ theory, i.e., the one that
acts by the de~Rham differential (times $\ii\hbar\,(-1)^d$) on $\bA$ and $\bB$  as well as on all
composite fields. We will denote it by $\Omega_0$.
Integrating by parts, we may rewrite the result as an allowed state. Namely,
on a regular functional as above
we get a term wih $L$ replaced by $\dd L$ plus all the terms corresponding to the boundary of the configuration space. As $L$ is smooth, its restriction to the boundary is also smooth and can be integrated on the fibers yielding a smooth form on the base configuration space; the bracketings at the related points are instead put together at the collapsing vertex. For example:
\[
\Omega_0\int_{\de_1 M} L_{IJ}\,[\bA^I][\bA^J]=\pm \ii\hbar\int_{\de_1 M} \dd L_{IJ}\,[\bA^I][\bA^J],
\]
\begin{multline*}
\Omega_0\int_{C_2(\de_1M)} L_{IJK}\, \pi_1^*([\bA^I][\bA^J])\pi_2^*[\bA^K]=\\=
\pm\ii\hbar \int_{C_2(\de_1M)} \dd L_{IJK}\, \pi_1^*([\bA^I][\bA^J])\pi_2^*[\bA^K] \pm\ii\hbar
\int_{\de_1 M} \underline{L_{IJK}}\, [\bA^I][\bA^J][\bA^K] ,
\end{multline*}
with $\underline{L_{IJK}}=\pi_*^\de  L_{IJK}$, where $\pi^\de\colon\de C_2(\de_1M)\to \de_1M$ is the canonical projection. Notice that for any two
regular functionals $S_1$ and $S_2$ we have $\Omega_0(S_1\bullet S_2)=\Omega_0(S_1)\bullet S_2\pm S_1\bullet\Omega_0(S_2)$.

The other generators that we allow are 
products of expressions of of the form
\[
\int_{\de_1 M} L^J_{I^1\cdots I^r}\,\left[\bA^{I^1}\right]\cdots\left[\bA^{I^r}\right]\,\frac{\delta^{|J|}}{\delta\bA^J}
\]
or
\[
\int_{\de_2 M} L_I^{J^1\cdots J^l}\,\left[\bB_{J^1}\right]\cdots\left[\bB_{J^l}\right]\,\frac{\delta^{|I|}}{\delta\bB_I},
\]
where the $L$'s are smooth differential forms on the boundary. We call these  expressions \textbf{simple} operators.
Each of the factors in a product of operators
acts independently on a state. The action of a simple operator on a regular functional is
defined by pairing a derivative with the corresponding composite field in all possible ways. If there are no derivatives (i.e., if $|I|=0$ or $|J|=0$), then the factor
is just $\bullet$\ndash multiplied with the rest. Example:
\[
\left(\sum_{ijk}\int_{\de_1M}L^{ij}_k\,[\bA^k]\frac{\delta^2}{\delta\bA^i\delta\bA^j}\right)
\left(
\sum_r\int_{\de_1M} \sfb_r[\bA^r]\bullet
\sum_s\int_{\de_1M} \sfb_s[\bA^s]
\right)=\pm
\int_{\de_1M}L^{ij}_k \sfb_i\sfb_j[\bA^k].
\]

The algebra of differential operators that we consider is generated by
products of $\Omega_0$ and simple operators. Note that the composition
of two products of simple operators is again a sum of
products of simple operators. This composition 
is easy to describe: each factor acts on a product either by multiplication (in the graded symmetric algebra) or by pairing the multiple derivative with
a corresponding composite field. Restricted to the $\bA$\ndash representation, this algebra is
the space of  $(S_m\times S_n)$\ndash invariants and coinvariants
$\oplus_{m,n} P(m,n)_{S_m \times S_n}$, where
in our case $P$ is the prop envelope of the endomorphism operad of $SV$ tensored with $\Omega^\bullet(\de_1M)$
(with the condition that the arguments must be in $S^+V$).
It was shown in \cite{MeVa} that for a general (dg) operad this construction yields a (dg) associative algebra.
In the $\bB$\ndash representation we get the same description with $V$ replaced by $V^*$
and $\de_1M$ replaced by $\de_2M$.

Example of a composition of products of simple operators (here $\odot$ is the graded commutative product of simple operators):
{\small
\begin{multline*}
\Big(
\underbrace{\int_{\de_1 M}(L_1)^i \frac{\delta}{\delta\bA^i}}_{\lambda_1}
\odot
\underbrace{\int_{\de_1 M}(L_2)^{jk} \frac{\delta^2}{\delta\bA^j\delta\bA^k}}_{\lambda_2}
\Big)\circ
\Big(
\underbrace{\int_{\de_1 M}(N_1)_{pqr}^{s}[\bA^p][\bA^q\bA^r]\frac{\delta}{\delta\bA^s}}_{\nu_1}
\odot \underbrace{\int_{\de_1 M}(N_2)_t[\bA^t]}_{\nu_2}
\Big)\\
= \pm\int_{\de_1 M}(L_1)^i (L_2)^{jk} (N_1)_{ijk}^s\frac{\delta}{\delta \bA^s} \odot
\int_{\de_1 M}(N_2)_t[\bA^t] \pm
\int_{\de_1 M}(L_2)^{jk}(N_1)_{pjk}^s[\bA^p]\frac{\delta}{\delta\bA^s} \odot
\int_{\de_1 M}(L_1)^i (N_2)_i \\
\pm  
\lambda_1
\odot \int_{\de_1 M} (L_2)^{jk}(N_1)_{pjk}^s[\bA^p]\frac{\delta}{\delta\bA^s} \odot \nu_2 
\pm
\lambda_2 
\odot
\int_{\de_1 M} (L_1)^i(N_1)^s_{iqr}[\bA^q\bA^r]\frac{\delta}{\delta \bA^s} \odot \nu_2 
\\
\pm  \lambda_2 
\odot
\nu_1 
\odot
\int_{\de_1 M} (L_1)^i(N_2)_i
+
\lambda_1\odot\lambda_2\odot\nu_1\odot\nu_2.
\end{multline*}
}

We call an operator  \textbf{principal} if it is simple
and
each field insertion is linear (i.e.,
$|I^1|=\dots=|I^r|=1$ or $|J^1|=\dots=|J^l|=1$) or if it is a multiple of $\Omega_0$.
Notice that, on a boundary $\Sigma$, $\Omega_0$ can be viewed as the standard quantization of
\[
\calS^\de_{\Sigma,0}=\int_\Sigma
\braket{\sfB}{\dd\sfA}.
\]
By analogy, we will say that the principal operator
\[
((-1)^d\ii\hbar)^{|J|}\int_{\Sigma} L^J_{i^1\cdots i^r}\,\left[\bA^{i^1}\right]\cdots\left[\bA^{i^r}\right]\,\frac{\delta^{|J|}}{\delta\bA^J}
\]
is the standard quantization, in the $\bA$\ndash representation, of
\[
\int_{\Sigma} L^J_{I}\,\sfA^{I}\,\sfB_J,
\]
where we grouped the indices $i^1,\dots,i^r$ into the multi-index $I$. Similarly, we will say that
\[
((-1)^d\ii\hbar)^{|I|}\int_{\Sigma} L_I^{j^1\cdots j^l}\,\left[\bB_{j^1}\right]\cdots\left[\bB_{j^l}\right]\,\frac{\delta^{|I|}}{\delta\bB^I},
\]
is the standard quantization, in the $\bB$\ndash representation, of
\[
\int_{\Sigma} L^J_{I}\,\sfB_{J}\,\sfA^I,
\]
where we grouped the indices $j^1,\dots,j^r$ into the multi-index $J$.

The $\Omega$ we will get in the modified QME is a linear combination of simple operators. We will call the underlying linear combination of principal terms
its \textbf{principal part}. In most examples we will focus on the principal part only. By the above notation it can be written as the standard quantization of some boundary functional.

\subsection{The modified QME}\label{s:mQME}
In all these theories $\Omega$ may be explicitly obtained by the usual techniques about integrals on
compactified configurations spaces (see, e.g., \cite{CMcs}).
Under the assumption of ``unimodular'' perturbations and ``tractable'' contributions from hidden faces in the bulk we have the following
\begin{Thm}[The mQME]
There is a quantization $\Omega$ of $S^\de$ that squares to zero and such that the modified quantum master equation (mQME) is satisfied. This $\Omega$ is completely determined by graph contributions at the boundary of compactified configuration spaces.
\end{Thm}
We split the proof this result into three Lemmata.

%
\begin{Lem}\label{lemma: mQME}
The modified QME is satisfied with $\Omega=\Omega_0+\Omega_\text{pert}$,
where $\Omega_0$ is the standard
quantization of the unperturbed boundary action and $\Omega_\text{pert}$ is
determined by the boundary configuration space integrals.
\end{Lem}
\begin{proof}[Sketch of the proof and construction of $\Omega_\text{pert}$]
Let $\Gamma$ be a Feynman graph (a disjoint union of $\geq 1$ graphs of the type appearing in the exponential in (\ref{e: psi pert} )) and $\omega_\Gamma$ the corresponding differential form
over the compactified configuration space $C_\Gamma$. Consider Stokes theorem
$\int_{C_\Gamma} \dd\omega_\Gamma = \int_{\de C_\Gamma} \omega_\Gamma$. The left hand side contains terms where $\dd$ acts on an $\bA$ or a $\bB$ and terms where $\dd$ acts on the 
propagator. The former correspond to the action of $\frac1{\ii\hbar}\Omega_0$, the latter when summed over graphs $\Gamma$ assemble, due to (\ref{e:deta}), to the action of $-\ii\hbar\Delta_{\calV_M}$ on the state.
The right hand side contains three classes of terms:
\begin{enumerate}
\item Integrals over boundary components where two vertices collapse in the bulk. The combinatorics of the Feynman diagrams in the expansion ensures that these terms cancel out when we sum over all the diagrams.\footnote{This cancellation relies on the assumption that the perturbed action satisfies the classical master equation, which is equivalent to $\sum_{i=1}^n \pm \frac{\de }{\de \sfA_i}\calV(\sfA,\sfB)\cdot \frac{\de }{\de \sfB^i}\calV(\sfA,\sfB)=0$, which in turn implies a relation on contractions of pairs of vertex tensors.}
\item Integrals over boundary components where more than two vertices collapse in the bulk (``hidden faces''). The usual arguments---vanishing theorems---ensure the vanishing of all these terms apart, possibly, for faces where all the vertices of a connected component of a graph collapse. In all the above mentioned theories, with the exception of Chern--Simons theory, also these terms vanish. In Chern--Simons theory, they may possibly survive, but can be compensated by a framing dependent term (see \cite{AS} and \cite{BC}).
\item \label{lm mQME (3)} Terms where two or more (bulk and/or boundary) vertices collapse together at the boundary 
or a single bulk vertex hits the boundary. The integral on
such a boundary face splits into an integral over a subgraph $\Gamma'$ of $\Gamma$ corresponding
to the collapsed vertices and an integral over $\Gamma/\Gamma'$, the graph obtained by identifying all the vertices
in $\Gamma'$ and forgetting the edges inside $\Gamma'$.
We define the action of $\frac\ii\hbar\Omega_\text{pert}$ by the sum of the boundary contributions of the  $\Gamma'$'s.
If we now sum over all graphs $\Gamma$, all these terms will give $\frac\ii\hbar\Omega_\text{pert}$ applied to the state.
\end{enumerate}
As a result we get the mQME.
\end{proof}

\begin{Rem}
In QM we clearly have $\Omega=0$, by degree reasons.
\end{Rem}
\begin{Rem}
In the (generalized) two\ndash dimensional YM theory, the term $vf(\sfB)$ does not contribute to
$\Omega_\text{pert}$, for the restriction of $v$ to the boundary is zero. As a consequence,
$\Omega$ for the (generalized) two\ndash dimensional YM theory is the same as for $BF$ theory.
\end{Rem}

\begin{Lem}
$\Omega$ squares to zero.
\end{Lem}

\begin{proof}[Sketch of the proof]
This can be done again by the same techniques as in the previous Lemma.
Namely,
let $\Gamma'$ be a graph appearing in the definition of $\Omega_\text{pert}$ and $\sigma_{\Gamma'}$
the corresponding differential form
---a product of
the propagators  $\eta$ and the boundary fields
$\bA$ or $\bB$ ---
over the compactified configuration space $\underline{C}_{\Gamma'}$, 
obtained
by modding out translations along the boundary and scalings.
Consider again Stokes theorem
$\int_{\underline C_{\Gamma'}} \dd\sigma_{\Gamma'} = \int_{\de \underline C_{\Gamma'}} \sigma_{\Gamma'}$.
The left hand side contains only terms where $\dd$ acts on an $\bA$ or a $\bB$, which correspond
to the action of $\frac1{\ii\hbar}\Omega_0$. The right hand side contains again three classes of terms.
The first class contains the terms where two vertices collapse in the bulk (the bulk is now a neighborhood of a point in the boundary); these terms cancel out when we sum over all graphs. The second class contains the terms
where more than two vertices collapse in the bulk; these terms do not contribute by the usual vanishing
theorems. Finally, the third class contains terms when two or more (bulk and/or boundary) vertices collapse together at the boundary or a single bulk vertex hits the boundary.
When we sum over all graphs, 
these terms yield the action of
$\frac\ii\hbar\Omega_\text{pert}$. 
This shows that $\Omega_0\Omega_\text{pert}+\Omega_\text{pert}\Omega_0+\Omega_\text{pert}^2$ vanishes.
Since we know that $\Omega_0^2=0$, we conclude that $\Omega^2=0$.
\end{proof}

\begin{Lem}\label{l:Omegahigher}
$\Omega$ is given by the canonical quantization of $S^\de$ plus (possibly) higher
order corrections. More precisely, the canonical quantization of $S^\de$ corresponds to $\Omega_0$ plus the contributions of\/ $\Omega_\text{pert}$ corresponding to exactly one bulk point approaching the boundary.
\end{Lem}

\begin{proof}[Sketch of the proof] 
Consider, e.g., the $\de_1$ boundary (the $\de_2$ case is treated similarly). Here we are in the $\bA$ representation. In a boundary term
of the type stated in the Lemma, there will be one bulk vertex coming from 
$\calV(\sfA,\sfB)$
and boundary vertices $\braket\bA\beta$ (in $d>1$ there are no contributions
from composite fields as in this particular case they would correspond to a multiple edge which vanishes
by dimensional reasons, for $d>2$, or by parity reasons, for $d=2$).
The bulk $\sfA$s actually only contribute with $\alpha$ as $\sfa$ vanishes on the boundary. So a monomial term
of degree $k$ in $\sfA$ in 
$\calV(\sfA,\sfB)$
will actually yield a boundary graph with $k$ boundary vertices and
with $k$ propagators joining
the bulk vertex to each boundary vertex. The integration is over the configuration space of these $k+1$ vertices
modulo 
horizontal translations (i.e. translations tangent to the boundary) and scalings. All the $\bA$ fields are grouped in the
integration along the boundary. The $\sfB$s in 
$\calV$
correspond to applying
$(-1)^d\ii\hbar\frac\delta{\delta \bA}$ to the rest of the state (as we explained above, if this results in
 a higher functional derivative, it has to be interpreted as the
first order functional derivative with respect to the corresponding composite field).
What remains to be shown is that the coefficients are equal to $1$. This is obvious for $k=0$.
For $k\ge1$,
denote by $g_k$  the result of the integration of the graph with $k$ boundary vertices
(notice that each $g_k$ is a number as we are integrating a $k\times(d-1)$\ndash form
on a $k\times(d-1)$\ndash dimensional space).
The simplest one, $g_1$, corresponds to one bulk vertex
and one boundary vertex joined by a an edge. We fix the horizontal translations by fixing the boundary point and we fix the scalings on the bulk point. The integral yields $1$ precisely because the propagator is normalized. Next, one shows that all other graphs yield the same contribution. This is an application of Stokes' theorem again. Consider a graph with $2$ bulk and $k$ boundary vertices, $k\ge1$,
and exactly one edge joining the bulk vertex 1 to each boundary vertex and to the bulk vertex 2. We take the differential
of the corresponding form and integrate over the corresponding boundary configuration space. Notice that all propagators are closed as we are near the boundary, so we just get an equality between the boundary contributions. There are actually two of them: the first is when the two bulk points collapse together, and this yields $g_k$; the second is
when the bulk point 2 goes to the boundary, and this yields $g_{k+1}$. So we have $g_{k+1}=g_k$ for all $k$,
which, together with $g_1=1$, yields $g_k=1$ for all $k$.

\end{proof}

\begin{Rem}\label{r:changeprop}
If we choose a different propagator, the higher order corrections might change leading to a different, but equivalent,
$\Omega$.
\end{Rem}

\begin{Rem}\label{rem: univ coeff}
Using results from \cite{BC, CR} one sees that the possible higher order corrections depend on global forms, possibly appearing in the action, and on universal coefficients that are 
invariant polynomials of the curvature of the connection used in the construction of the propagator.
The universal coefficients are Chern--Weil representatives of certain universal polynomials, with real coefficients, in the Pontryagin classes of the pull-back of the tangent bundle of $M$ to $\de M$. Note that, by the stability property, these Pontryagin classes coincide in cohomology, $H^{4j}(\de M)$, with classes of the tangent bundle of $\de M$, since $TM|_{\de M}=T\de M\oplus N\de M$ and the last term (the normal bundle to the boundary) is a trivial rank 1 bundle. This implies that, up to equivalence as in Remark~\ref{r:changeprop},
the boundary operator $\Omega$ does not depend on the bulk.
\end{Rem}

The principal part of the operator $\Omega_\mr{pert}$ (see the end of Section~\ref{s:4.1.2 operators})
constructed in the proof of Lemma \ref{lemma: mQME} has the following general structure:
\begin{multline*}
\Omega_\mr{pert}=\sum_{n,k\geq 0}\sum_{\Gamma'_1} \frac{(-\ii\hbar)^{\mr{loops}(\Gamma'_1)}}{|\mr{Aut}(\Gamma'_1)|} \int_{\de_1 M} \left(\sigma_{\Gamma'_1}\right)_{i_1\cdots i_{n}}^{j_1\cdots j_{k}}\bA^{i_1}\cdots \bA^{i_{n}}\left((-1)^d\ii\hbar \frac{\delta}{\delta \bA^{j_1}}\right)\cdots \left((-1)^d\ii\hbar \frac{\delta}{\delta \bA^{j_{k}}}\right)\\ +
\sum_{n,k\geq 0}\sum_{\Gamma'_2} \frac{(-\ii\hbar)^{\mr{loops}(\Gamma'_2)}}{|\mr{Aut}(\Gamma'_2)|} \int_{\de_2 M} \left(\sigma_{\Gamma'_2}\right)_{j_1\cdots j_{k}}^{i_1\cdots i_{n}}\bB_{i_1}\cdots \bB_{i_{n}}\left((-1)^d\ii\hbar \frac{\delta}{\delta \bB_{j_1}}\right)\cdots \left((-1)^d\ii\hbar \frac{\delta}{\delta \bB_{j_{k}}}\right)
\end{multline*}
where $\Gamma'_1$ runs over graphs with
\begin{itemize}
\item $n$ vertices on $\de_1 M$ of valence $1$ with adjacent half-edges oriented inwards and decorated with boundary fields $\bA_{i_1},\ldots,\bA_{j_{n}}$, all evaluated at the point of collapse $x\in \de_1 M$,
\item $k$ inward leaves decorated with variational derivatives in boundary fields $$(-1)^d\ii\hbar\frac{\delta}{\delta \bA_{j_1}},\ldots, (-1)^d\ii\hbar\frac{\delta}{\delta \bA_{j_k}}$$
at the point of collapse,
\item no outward leaves (graphs with them do not contribute).
\end{itemize}
The form $\sigma_{\Gamma'_1}$ on $\de_1 M$ is the universal coefficient of
Remark \ref{rem: univ coeff} and is obtained as the integral, over the compactified configuration space $\underline{C}_{\Gamma'_1}$ (with translations along boundary and scalings modded out), of the product of limiting propagators at the point of collapse and vertex tensors.
The graphs $\Gamma'_2$ correspond to a collapse at a point $y\in \de_2 M$; the Feynman rules for them are similar, but with opposite orientations for boundary vertices and leaves, and with multiplications and derivations in the field $\bB$ instead of $\bA$.

Note that $\Omega$ does not depend on residual fields.

\begin{Rem}[Change of data]
Using the same techniques \cite{CMcs}, one can show that under a change of data, see Remark~\ref{r:cod}, the state changes
consistently: $\frac\dd{\dd t}\psi = (\hbar^2\Delta+\Omega)(\psi\zeta)$, where $\zeta$ can be computed explicitly in
terms of Feynman diagrams.
\end{Rem}

\begin{Rem}[Open problem]
When we glue two states $\psi_{M_1}$ and $\psi_{M_2}$ as in Proposition~\ref{p:gluing4} we get a new state $\psi_M$.
All three states satisfy the mQME as they are Feynman diagram expansions of the theory. This shows that there is
a relation between the operators $\Omega_1$ and $\Omega_2$ on the gluing submanifold $\Sigma$ regarded as a boundary component
of $M_1$ or of $M_2$. Namely, the pairing of $\psi_{M_2}$ with $(\Hat\Omega_1-\Omega_2)\psi_{M_2}$ vanishes,
where $\Hat\Omega_1$ is the functional Fourier transform of $\Omega_1$. Notice that in the pairing we only see the principal parts. This leads then to the conjecture that
$\Hat\Omega_1^\text{princ} = \Omega_2^\text{princ}$.
\end{Rem}

\subsection{The doubling trick}\label{ss:dbling}
On a manifold without boundary one can choose the propagator to be symmetric, up to a sign, under the exchange of $\alpha$ and $\beta$. The boundary polarizations however break this symmetry. This asymmetry persists after gluing, even if at the end we have a closed manifold. One can obviate this as follows.
First we add an additional abelian $BF$ theory with the same field content:
\[
\calS_{M,\text{double}}(\sfA,\sfB,\Check\sfA,\Check\sfB) := \Szero(\sfA,\sfB)+\Spert(\sfA,\sfB) +
\Szero(\Check\sfA,\Check\sfB).
\]
The states for this theory are tensor products of the states for the $(\sfA,\sfB)$\ndash theory
with the states for the abelian $(\Check\sfA,\Check\sfB)$\ndash theory, and we know the latter explicitly. In particular, on
a closed manifold, the partition function of the doubled theory will differ from the one in the original theory just by a multiple of the torsion of $M$. Moreover, the expectation values of
$(\sfA,\sfB)$\ndash observables will be the same for the two theories.
Next we make the change of variables:\footnote{{}In our setting the space of fields is a vector space. In more general
settings, $\sfA$ and $\Check\sfA$ contain a connection in degree zero, so the space of fields is affine.
In this case, $\sfA_1$ will still belong to the same affine space, whereas $\sfA_2$ will belong to its tangent space.}
\begin{align*}
\sfA &= \sfA_1 + \sfA_2, &  \Check\sfA &= \sfA_1 - \sfA_2,\\
\sfB &= \sfB_1 + \sfB_2, & \Check\sfB &= \sfB_1 - \sfB_2,
\end{align*}
We now have
\begin{multline*}
\calS_{M,\text{double}}(\sfA_1,\sfB_1,\sfA_2,\sfA_2)=
2\Szero(\sfA_1,\sfB_1)+
2\Szero(\sfA_2,\sfB_2)+
\Spert(\sfA_1 + \sfA_2,\sfB_1 + \sfB_2).
\end{multline*}
The final step in this construction is the choice a polarization. Our choice will be to choose opposite
polarizations for the fields of type $1$ and those of type $2$. To stick to the notations of
subsection~\ref{s:pol}, on $\de_1 M$ we choose the $\frac{\delta}{\delta\sfB_1}\times\frac{\delta}{\delta\sfA_2}$\ndash polarization and on $\de_2 M$ we choose the $\frac{\delta}{\delta\sfA_1}\times\frac{\delta}{\delta\sfB_2}$\ndash polarization.
We then proceed with the splittings of the fields into boundary, residual and fluctuation fields.
Notice that the propagators for the theories $1$ and $2$ will have opposite boundary conditions and will be
$\frac12$ of the propagators considered before (because of the factor $2$ in front of the $\Szero$'s).
On the other hand, to construct the Feynman diagrams we will always have to contract a factor
$\alpha_1+\alpha_2$ from one vertex with a factor $\beta_1+\beta_2$ from another vertex. This will produce
the average of the two propagators computed in Section~\ref{s:abeBF}
with the two opposite boundary conditions.


\subsection{Quantum mechanics}
We start with the simple case of quantum mechanics, see Example~\ref{e:QM}.
In this case, $\calF^\de=T^*W$ and, by degree reasons, we have $S^\de=0$ and
$\Omega=0$ (as the only connected zero dimensional manifold is a point, we do not write it explicitly as an
index).\footnote{More generally, we could take as target a superspace endowed with BFV data in addition
to a Hamiltonian function. In this case,
$S^\de$ and $\Omega$ may not be trivial.} Also we take $M$ to be the interval $[t_1,t_2]$.

The simplest way to compute QM is with the mixed polarization: namely, we take $\de_1M=\{t_1\}$ and
$\de_2M=\{t_2\}$ (or vice versa). In this case there are no residual fields and we have
$\eta(s,t)=\Theta(s-t)$, with $\Theta$ the Heaviside function. We also have $T_M=1$ (with $T_M$ as in Section \ref{sec: T and torsions}).
If $H=0$, we then simply have
\[
\Psi_{[t_1,t_2],0}=\EE^{-\frac i\hbar\sum_i p_iq^i},
\]
where we use the notation $q=\bA$ and $p=\bB$. Notice that this state is the representaion of the identity operator.
One can easily compute $\langle P_r(\tau)\rangle_0=\EE^{-\frac i\hbar\sum_i p_iq^i}p_r$ and $\langle Q^s(\tau)\rangle_0=\EE^{-\frac i\hbar\sum_i p_iq^i}q^s$ for all $\tau\in(t_1,t_2)$.
Let $\tau_1,\tau_2$ be such that $t_1<\tau_1<\tau_2<t_2$. We then have
\begin{subequations}\label{e:starord}
\begin{align}
\langle Q^s(\tau_2)P_r(\tau_1)\rangle_0 &= \EE^{-\frac i\hbar\sum_i p_iq^i}(q^s p_r+\ii\hbar\,\delta_r^s),\\
\langle P_s(\tau_2)Q^r(\tau_1)\rangle_0 &= \EE^{-\frac i\hbar\sum_i p_iq^i} p_s q^r,\\
\langle P_s(\tau_2)P_r(\tau_1)\rangle_0 &= \EE^{-\frac i\hbar\sum_i p_iq^i}p_sp_r,\\
\langle Q^s(\tau_2)Q^r(\tau_1)\rangle_0 &= \EE^{-\frac i\hbar\sum_i p_iq^i} q^sq^r.
\end{align}
\end{subequations}
Hence, if $f$ and $g$ are functions on $T^*W$, we have
\[
\langle f(Q(\tau_2),P(\tau_2))\;g(Q(\tau_1),P(\tau_1))\rangle_0=\EE^{-\frac i\hbar\sum_i p_iq^i} f\star g(q,p)
\]
where $\star$ is the star product defined by the 
ordering \eqref{e:starord}, i.e. $f\star g=f\,\EE^{\sum_i\ii\hbar \frac{\overleftarrow{\de}}{\de q^i}\frac{\overrightarrow{\de}}{\de p_i}}g$. Finally, if we have a Hamiltonian
function $H$, we may write $\int_M H(Q,P)\dd t$ as a limit of Riemann sums. Taking the expectation value and computing the limit finally yields
\[
\Psi_{[t_1,t_2]}=\EE^{-\frac i\hbar\sum_i p_iq^i}\EE_\star^{\frac\ii\hbar (t_2-t_1)H}(q,p).
 \]

We may also work in the $\bA$\ndash representation on both sides. In this case, we have residual fields
\[
\sfa = zv,\quad \sfb=z^+,
\]
with $v\in\Omega^1([t_1,t_2])$ satisfying $\int_{t_1}^{t_2}v=1$. Notice that $\deg z=-1$ and $\deg z^+=0$ and that $\Delta=-\sum_r\frac{\de^2}{\de z^r\de z^+_r}$.
The corresponding propagator is then
$\eta(s,t)=\Theta(s-t)+\psi(s)$ with $\psi(s)=-\int_{t_1}^{s}v$. It follows that
\[
\Psi_{[t_1,t_2],0}=\EE^{\frac i\hbar\sum_i z^+_i(q_2^i-q_1^i)},
\]
where $q_1$ and $q_2$ denote $\bA$ at $\{t_1\}$ and at $\{t_2\}$.
Notice that we can make a BV integration on residual fields  by choosing the Lagrangian subspace $\{z=0\}$.
The integration over $z^+$ yields, up to a normalization constant, $\delta(q_2-q_1)$, which is
the $q$\ndash representation of the identity operator.
We can now compute $\langle P_r(\tau)\rangle_0=\EE^{\frac i\hbar\sum_i z^+_i(q_2^i-q_1^i)}z^+_r$ and
\begin{multline*}
\langle Q^s(\tau)\rangle_0=\EE^{\frac i\hbar\sum_i z^+_i(q_2^i-q_1^i)}(q_1^s + (q_1-q_2)^s\psi(\tau))=\\=
\EE^{\frac i\hbar\sum_i z^+_i(q_2^i-q_1^i)}q_1^s-
\ii\hbar\Delta(\EE^{\frac i\hbar\sum_i z^+_i(q_2^i-q_1^i)}z^s\psi(\tau)),
\end{multline*}
for all $\tau\in(t_1,t_2)$. Similarly, for $t_1<\tau_1<\tau_2<t_2$, we get
\begin{align*}
\langle Q^s(\tau_2)P_r(\tau_1)\rangle_0 &= \EE^{\frac i\hbar\sum_i z^+_i(q_2^i-q_1^i)}(q_1^s z^+_r+\ii\hbar\delta_r^s)
-\ii\hbar\Delta(\EE^{\frac i\hbar\sum_i z^+_i(q_2^i-q_1^i)} z^s z^+_r\psi(\tau_2))
,\\
\langle P_s(\tau_2)Q^r(\tau_1)\rangle_0 &= \EE^{\frac i\hbar\sum_i z^+_i(q_2^i-q_1^i)}z^+_s q_1^r
-\ii\hbar\Delta(\EE^{\frac i\hbar\sum_i z^+_i(q_2^i-q_1^i)}z^+_s z^r\psi(\tau_1))
,\\
\langle P_s(\tau_2)P_r(\tau_1)\rangle_0 &= \EE^{\frac i\hbar\sum_i z^+_i(q_2^i-q_1^i)} z^+_s z^+_r,\\
\langle Q^s(\tau_2) Q^r(\tau_1)\rangle_0 &=\EE^{\frac i\hbar\sum_i z^+_i(q_2^i-q_1^i)} q_1^s q_1^r
-\ii\hbar\Delta(\EE^{\frac i\hbar\sum_i z^+_i(q_2^i-q_1^i)}(q_1-q_2)^sz^r\psi(\tau_1)\psi(\tau_2)).
\end{align*}
More generally, we have
\[
\langle f(Q(\tau_2),P(\tau_2))\; g(Q(\tau_1),P(\tau_1))\rangle_0= \EE^{\frac i\hbar\sum_i z^+_i(q_2^i-q_1^i)}
f\star g(q_1,z^+)-\ii\hbar\Delta(\cdots).
\]
If we integrate over $z^+$, with $z=0$, we finally get
\[
\int_{z=0} \dd z^+\,\langle f(Q(\tau_2),P(\tau_2))\; g(Q(\tau_1),P(\tau_1))\rangle_0=
f\star g\left(q_2,-\ii\hbar\frac\de{\de q_2}\right)\delta(q_2-q_1).
\]
Finally,
\[
\int_{z=0} \dd z^+\,\Psi_{[t_1,t_2]}=\EE_\star^{\frac\ii\hbar (t_2-t_1)H}\left(q_2,-\ii\hbar\frac\de{\de q_2}\right)\delta(q_2-q_1).
\]

\subsection{Nonabelian $BF$ theories}\label{sec:non-ab BF}
We continue with the case of nonabelian $BF$ theories for a Lie algebra $\frg$, see Example~\ref{exa-AKSZ}.
The bulk BV action is
\[
\calS_M=\int_M\left(
\braket{\sfB}{\dd\sfA}+
\frac12\braket{\sfB}{[\sfA,\sfA]}\right)
\]
and, since this is an AKSZ theory, the boundary BFV action has the same form:
\[
\calS^\de_\Sigma=\int_\Sigma\left(
\braket{\sfB}{\dd\sfA}+
\frac12\braket{\sfB}{[\sfA,\sfA]}\right).
\]
The standard quantization is then
\begin{multline}\label{e:BFOmegastand}
\Omega_\text{stand} = \int_{\de_2M}\left(
(-1)^d\ii\hbar\sum_a\dd\bB_a\,\frac{\delta}{\delta\bB_a}-
\sum_{a,b,c} f^a_{bc} \frac{\hbar^2}2
\bB_a\frac\delta{\delta\bB_b}\frac\delta{\delta\bB_c}
\right)+\\+
 \int_{\de_1M}\left(
(-1)^d\ii\hbar\sum_a \dd\bA^a\,\frac{\delta}{\delta\bA^a}+
\frac12\sum_{a,b,c} f^a_{bc}\,(-1)^d\ii\hbar\,\bA^b\bA^c\frac\delta{\delta\bA^a}
\right),
\end{multline}
where we have introduced a basis for the Lie algebra and denoted the corresponding structure constants by
$f^a_{bc}$. One can easily check that $\Omega_\text{stand}^2=0$.

\begin{Lem}
If $d$ is even, then the principal part of $\Omega$ is $\Omega_\text{stand}$. If $d$ is odd,
then the principal part of $\Omega$ is the standard quantization of
\[
\Tilde \calS^{\de M} = \calS^{\de M} -\ii\hbar \sum_{j=0}^{\left[\frac{d-3}4\right]}
\int_{\de M} \gamma_{j} \tr\ad_{\sfA}^{d-4j},
\]
where $\gamma_{j}$ is a closed $4j$\ndash
form on $\de M$ which is
an invariant polynomial, with universal coefficients, of the curvature of the connection used in the construction of the propagator.
\end{Lem}
\begin{proof}[Sketch of the proof]
As the interaction is cubic, the vertices are at most trivalent. Notice that if the boundary diagram contains a univalent bulk vertex, then the integral is zero by dimensional reasons unless this is the only vertex, in which case we get a contribution to $\Omega_\text{stand}$. This means that in the boundary graph we only have bivalent and trivalent bulk vertices.
We now use the following convention: edges in the graph are oriented pointing from the $A$ vertex to the $B$ vertex.
Notice that the trivalent vertex has one incoming and two outgoing arrows, so it increases the number of outgoing arrows.

On $\de_2M$ we then have outgoing arrows from the boundary and the bulk vertices are either bivalent, with one incoming and one outgoing arrow, or trivalent. Thus, the only possibility is to have only the bivalent vertices and they have to be arranged in a loop.

On $\de_1M$ we have instead arrows pointing to the boundary and the bulk vertices are either bivalent, with two outgoing arrows, or trivalent. Suppose that the graph has $b$ bivalent bulk vertices, $t$ trivalent bulk vertices and
$m$ boundary vertices. By arrow conservation we have $2b+t=m$. Moreover, the total number of arrows is
$(3t+2b+m)/2=2t+2b$. This implies that the form degree is $(2t+2b)(d-1)$. The dimension of the boundary space
is $d(b+t)+(d-1)m-d=(3d-2)b+(2d-1)t-d$. If the dimension is larger than the form degree, then the integral vanishes.
Since the difference between form degree and dimension is $d(1-b)-t$, we get
$d(b-1)+t\le0$. This cannot hold if $b>1$. For $b=1$ we get $t=0$, which is a contribution to $\Omega_\text{stand}$.
Hence we are left with $b=0$ --- i.e., no bivalent vertices --- and $t\le d$. This means that the graph is a wheel from which trees depart to hit the boundary. We claim that this graph vanishes unless each vertex in the wheel is directly connected to a boundary vertex. In fact, if this is not the case, there will be a bulk vertex not in the wheel with two emanating edges that hit the boundary. Integrating a boundary vertex removes the corresponding edge, by normalization of the propagator. Hence, integrating these two boundary vertices leaves a univalent vertex, so the integral vanishes.
Finally consider the wheels with each vertex directly attached to a boundary vertex. Again, integrating the boundary vertices removes the corresponding edges. Hence, the contribution of each such wheel is the same as the contribution of the corresponding loop, as on $\de_2M$.

In figure~\ref{fig:wheelloop} we give an example of a loop and the corresponding wheel that might give a nontrivial contribution.
\begin{figure}[htbp] 
   \centering
\includegraphics[width=3in]{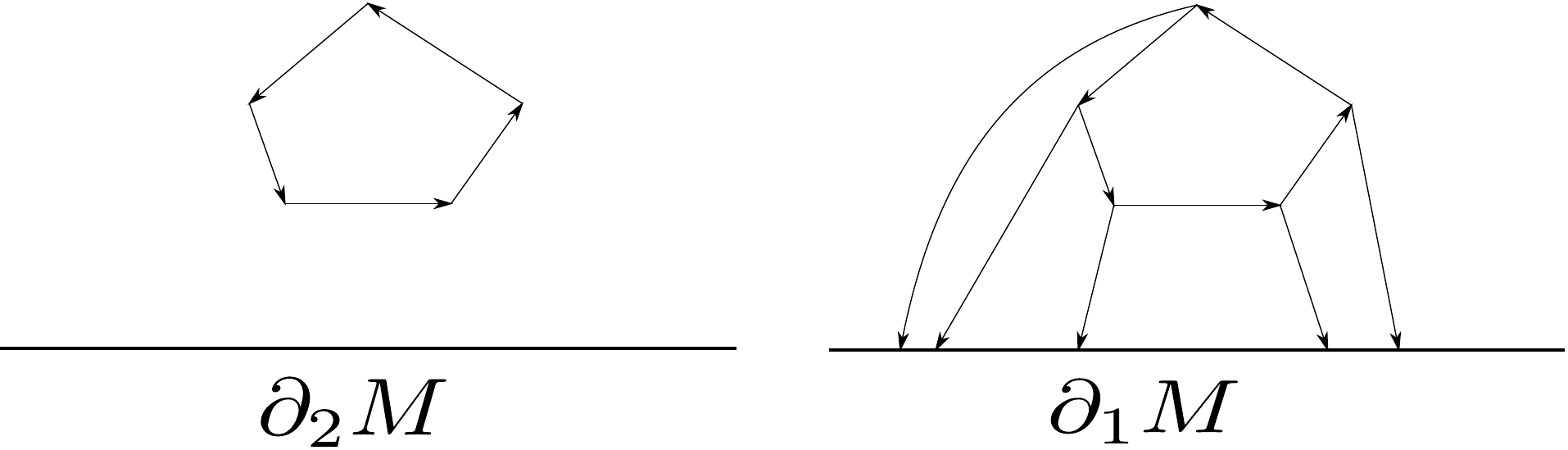}
   \caption{An example of a loop (only the internal edges, not leaves, of the collapsed graph are shown) and the corresponding wheel}
   \label{fig:wheelloop}
\end{figure}

Let us denote by $\beta_k$ the $(d-k)$\ndash form on $\de M$ obtained by integrating the loop with
$k$ vertices. Since the restriction of the propagators to these boundary faces is closed, by Stokes theorem we
get $\dd\beta_{2s+1}=\pm\beta_{2s}$ $\forall s$.

As in 
Remark~\ref{rem: univ coeff},
we now have to recall, see \cite{BC, CR},  that $\beta_k$ is an invariant polynomial in
the curvature of the connection used to define the propagator (if we define the propagator by Hodge decomposition, the connection is the Levi-Civita connection for the chosen metric). In particular, $\beta_k=0$ if its degree is
odd. Moreover, using compatibility of the connection with reduction of the structure group of the tangent bundle to $SO(d)$, we have that $\beta_k$ can be nonzero only if $d-k=0\bmod 4$. The coefficients in the polynomial are universal.

If $d$ is even, then $\deg\beta_{2s+1}$ is odd. This then implies that $\beta_{2s+1}=0$ and
hence also $\beta_{2s}=0$ for all $s$.

If $d$ is odd, then $\beta_{2s}=0$ for all $s$ by the same reason. Moreover, $\beta_{2s+1}$ is zero unless $2s+1=d\bmod 4$. We then denote $\gamma_j:=\beta_{d-4j}$ the potentially non-vanishing polynomials.
%
%
%
\end{proof}

\begin{Rem}\label{rem: gammas change with connection}
If we change the connection in the construction, each polynomial  $\gamma_{j}$ changes
by an exact form $\dd\sigma_j$. Hence, $\Tilde \calS^{\de M}$ changes by
\[
\left\{
\Tilde \calS^{\de M},
\ii\hbar\sum_{j=0}^{\left[\frac{d-3}{4}\right]}
\int_{\de M} \sigma_j \tr\ad_{\sfA}^{d-4j}\right\},
\]
where $\{\ ,\ \}$ is the Poisson bracket associated to $\omega^\de_{\de M}$,
so we see explicitly that we get an equivalent $\Omega$.
\end{Rem}
\begin{Rem}
We do not know if the characteristic classes $\gamma_{j}$ in odd dimension are non zero.
They might vanish if, e.g., we had a vanishing Lemma that ensures that bivalent vertices with consecutive arrows yield
zero. (This is easily shown to be true in two dimensions.)
\end{Rem}
\begin{Rem}
Notice that $\gamma_{0}$ is a closed zero-form. Moreover, this constant is universal (possibly zero).
Denoting it by $c_d$, we get a  contribution
$c_d\int_{\de M}  \tr\ad_{\sfA}^{d}$. Notice that this is the only contribution for $d=3$ and for $d=5$.
In higher odd dimensions there may be other contributions as well.
\end{Rem}
\begin{Exa}
We first consider the example when $M$ is a ball and we work in the $\bB$\ndash representation.
If we denote the propagators as arrows joining $\alpha$ to $\beta$, then we have arrows issuing from the boundary.
The only vertex that reduces the number of arrows corresponds to a term $\beta\sfa\sfa$. Because of the boundary conditions the residual fields $\sfa$ are concentrated in cohomology degree $0$. Hence we get univalent
vertices which vanish upon integration. There are two vertices that preserve the number of arrows: $\sfb\sfa\sfa$ and
$\beta\alpha\sfa$. The first just gives an insertion of residual fields. The second produces loops. However, since $\sfa$ is in degree zero, the form degree of a loop with $n$ vertices is
$n(d-1)$; the dimension of the configuration space is however $nd$, so the integral vanishes. In conclusion,
the state for a ball in nonabelian $BF$ theory is the same as for $\dim\frg$ copies of abelian $BF$ theory plus
the insertion  $\sfb\sfa\sfa$.
In particular,
the effective action in $d$ dimensions reads
\begin{multline*}
\calS^\text{eff}_M(\bB,z,z^+) = (-1)^{d-1} \int_{S^{d-1}} \braket\bB \sfa +
\int_M \frac12\braket{\sfb}{[\sfa,\sfa]}=\\
=(-1)^d\sum_a z^a\int_{S^{d-1}} B_a^{d-1}
+ \frac12\sum_{a,b,c}\,f^a_{bc} z^+_az^bz^c,
\end{multline*}
where we have written $\sfa=z1$ and $\sfb=z^+v$ for a normalized volume form on $M$
with $z\in\frg$ and $z^+$ valued in $\frg^*$. By $B^{d-1}$ we denote the $(d-1)$\ndash form component
of $\bB$ (which has ghost number $-1$).
\end{Exa}

The same computation in the $\bA$\ndash polarization is much more involved as in this case nontrivial graphs appear.
This case may be obtained from the previous one using the generalized Segal--Bargmann transform, which is nontrivial as requires
considering the cylinder $S^{d-1}\times I$ with $\bA$\ndash polarization on both boundary components.

\subsection{2D Yang--Mills theory}
As explained in Example~\ref{r:twoDYM}, the (generalized) two\ndash dimensional YM theory may be treated
as a perturbation of $BF$ theory with the same Lie algebra $\frg$. As the perturbation does not affect the boundary,
we get that $\Omega=\Omega_\text{stand}$ as in \eqref{e:BFOmegastand}.

\subsubsection{Examples}
For simplicity we focus on the abelian case $\frg=\bbR$. The vertices are given by the Taylor expansion
of $f$, $f(x)=\sum_{k=0}^\infty \frac1{k!} f^{(k)}x^k$.

We first consider the example when $M$ is a disk and we work in the $\bB$\ndash representation.
In the bulk we expand $\sfB=\sfb+\beta$. As $\sfb$ is concentrated in form degree $2$, we get
$\int_M vf(\sfB)=\int_M vf(\beta)=\sum_{k=0}^\infty \frac1{k!} f^{(k)}\int_M v\beta^k$.
Each $\alpha$ on the boundary can be paired to a $\beta$ in the interaction. The graphs contributing to the state
are stars with one bulk vertex, with coefficient $vf^{(k)}$, and $k$ boundary vertices. If we denote
by $\alpha_k$ the $k$\ndash form on $(S^1)^k$ obtained by integrating the bulk vertex of such a graph,
we get
\begin{multline}\label{e:YMB}
\calS^\text{eff}_M(\bB,z,z^+) = -\int_{S^1} \bB z+
\sum_{k=0}^\infty \frac1{k!}f^{(k)}\int_{(S^1)^k}\alpha_k\,\pi_1^*\bB\cdots\pi_k^*\bB=\\=
z\int_{S^1} B^1 +
\sum_{k=0}^\infty \frac1{k!}f^{(k)}\int_{(S^1)^k}\alpha_k\,\pi_1^*B^0\cdots\pi_k^*B^0,
\end{multline}
where $B^i$ denotes the $i$\ndash form component of $\bB$.

Next we consider the same example but in the $\bA$\ndash representation. In this case
$\sfb$ is concentrated in form degree $0$. On the other hand, there are no $\alpha$s to pair the $\beta$s.
If we write $\sfb=z^+1$, with $\deg z^+=0$, we get the effective action
\begin{equation}\label{e:YMA}
\calS^\text{eff}_M(\bA,z,z^+)=\int_{S^1}z^+\bA + Vf(z^+)=
Vf(z^+)+ z^+\int_{S^1}A_1
\end{equation}
with $V:=\int_{M}v$ the area of the disk and $A_1$ the $1$\ndash form component of $\bA$ (i.e., the classical field).

One can pass from one polarization to the other by the generalized Segal--Bargmann transform, see Remark~\ref{rem: changing P by attaching a cylinder}.  
To do this we have to consider the cylinder $S^1\times I$ with the topological theory corresponding to the 2D generalized YM theory. This is just $BF$ theory.

Suppose we start from the $\bB$\ndash representation. Then we should consider the cylinder with $\bA$\ndash representation on both end sides. We denote the boundary fields by $\Tilde\bA$ and $\bA$ to distinguish the two boundary components. We write the residual fields as
\[
\sfa = w^1\chi_1 + w u,\qquad \sfb = w^+1+ w^+_1\chi^1,
\]
with $u$ a two-form 
and $\chi_1$, $\chi^1$ one\ndash forms forming a basis in the cohomologies together with $1$.
The effective action reads
\[
\calS^\text{eff}_{S^1\times I}(\Tilde\bA,\bA,w,w^1,w^+,w^+_1)=
w^+\int_{S^1}\Tilde A_1 - w^+_1\int_{S^1}\chi^1\Tilde A_0
-w^+\int_{S^1} A_1 + w^+_1\int_{S^1}\chi^1 A_0.
\]
We now pair the $\bA$ variables with the $\bB$ variables in \eqref{e:YMB}. This yields, after integration,
the exponent
\[
w^+\int_{S^1}\Tilde A_1 - w^+_1\int_{S^1}\chi^1\Tilde A_0-
zw^+_1+Vf(w^+).
\]
We now take the Lagrangian subspace $\{w_1=0,\ z^+=0\}$; integrating out $z$ and $w^+_1$,
and using $\int_{(S^1)^k}\alpha_k=V$,\footnote{In fact, integrating the boundary vertices just removes the edges form the graph;
at the end we are left with $\int_Mv=V$.}
yields the exponent
\[
Vf(w^+)+w^+\int_{S^1}\Tilde A_1
\]
which is \eqref{e:YMA} with a relabeling of the variables.

Next we start from the $\bA$\ndash representation. We then consider the cylinder with $\bB$\ndash representation on both end sides. We now denote the boundary fields by $\Tilde\bB$ and $\bB$ to distinguish the two boundary components.
We write the residual fields as
\[
\sfa=w1+w^1\chi_1,\qquad\sfb=w^+_1\chi^1+w^+u.
\]
We have the effective action
\[
\calS^\text{eff}_{S^1\times I}(\Tilde\bB,\bB,w,w^1,w^+,w^+_1)=
-\int_{S^1}\Tilde B^1w-\int_{S^1}\Tilde B^0w^1\chi_1
+\int_{S^1}B^1w+\int_{S^1} B^0w^1\chi_1.
\]
We pair the $\bB$ variables with the $\bA$ variables in \eqref{e:YMA}. This yields, after integration,
the exponent
\[
-\int_{S^1}\Tilde B^1w-\int_{S^1}\Tilde B^0w^1\chi_1+
Vf(z^+)-z^+w^1,
\]
where we have used $\int_{S^1}\chi_1=1$. We now choose the Lagrangian subspace
$\{z=0,\ w^+_1=0\}$ and integrate out $z^+$ and $w^1$. This yields the exponent
\begin{equation}\label{e:YMBm}
-\int_{S^1}\Tilde B^1w+Vf\left(-\int_{S^1}\Tilde B^0\chi_1\right),
\end{equation}
which differs from \eqref{e:YMB} but actually just by a BV transformation. Recall that
$\int_{(S^1)^k}\alpha_k=(-1)^k V$.  This shows that $\alpha_k$ and
$\beta_k:=(-1)^k V\pi_1^*\chi_1\cdots\pi_k^*\chi_1$ are in the same cohomology class.
Let $\tau_k$ be a path of $k$\ndash forms interpolating between $\alpha_k$ and $\beta_k$; e.g.,
$\tau_k(t)=(1-t)\alpha_k+t\beta_k$, $t\in[0,1]$. We have that $\dot\tau_k=\dd\gamma_k$, for some
$(k-1)$\ndash form $\gamma_k$. We define
\[
\calS^\text{eff}_M(\bB,z,z^+;t) = z\int_{S^1} B^1 +
\sum_{k=0}^\infty \frac1{k!}f^{(k)}\int_{(S^1)^k}\tau_k(t)\,\pi_1^*B^0\cdots\pi_k^*B^0,
\]
Notice that the exponent computed above, equation \eqref{e:YMBm}, with a relabeling of the variables
is $\calS^\text{eff}_M(\bB,z,z^+;1)$,
whereas $\calS^\text{eff}_M(\bB,z,z^+)$ is $\calS^\text{eff}_M(\bB,z,z^+;0)$. We now have
\[
\frac\dd{\dd t}\EE^{\frac\ii\hbar\calS^\text{eff}_M(\bB,z,z^+;t)}=
\EE^{\frac\ii\hbar\calS^\text{eff}_M(\bB,z,z^+;t)}
\frac\ii\hbar
\sum_{k=0}^\infty \frac1{k!}f^{(k)}\int_{(S^1)^k}\dd\gamma_k\,\pi_1^*B^0\cdots\pi_k^*B^0.
\]
Observe that
\[
\int_{(S^1)^k}\dd\gamma_k\pi_1^*B^0\cdots\pi_k^*B^0=
(-1)^k\int_{(S^1)^k}\gamma_k\dd(\pi_1^*B^0\cdots\pi_k^*B^0)=
\frac{(-1)^{k+1}}{\ii\hbar}\Omega\int_{(S^1)^k}\gamma_k\pi_1^*B^0\cdots\pi_k^*B^0.
\]
Since $\Omega\calS^\text{eff}_M(\bB,z,z^+;t)=0$ for all $t$ and all the terms involved are
$\Delta$\ndash closed, we have
\[
\frac\dd{\dd t}\EE^{\frac\ii\hbar\calS^\text{eff}_M(\bB,z,z^+;t)}=
(\hbar^2\Delta+\Omega)\left(
\frac{(-1)^{k+1}}{\hbar^2}\EE^{\frac\ii\hbar\calS^\text{eff}_M(\bB,z,z^+;t)}
\sum_{k=0}^\infty \frac1{k!}f^{(k)}
\int_{(S^1)^k}\gamma_k\pi_1^*B^0\cdots\pi_k^*B^0
\right),
\]
which shows that \eqref{e:YMB} and \eqref{e:YMBm} are equivalent.

\subsection{Split Chern--Simons theory}
The split Chern--Simons theory, see Example~\ref{exa-AKSZ}, can be treated like the nonabelian $BF$ theory; cf. \cite{CMW} for an example of a perturbative calculation.
There are more vertices and what causes more problem is the presence of possibly nonvanishig hidden face contributions, which however can be dealt with using framing (see \cite{BC,CMcs}).

The principal part of the boundary operator $\Omega$ might now have additional contributions to the canonical quantization of
$S^\de$. By dimensional reasons and by the same argument as in Section \ref{sec:non-ab BF}, the corrections are given by
cubic terms with universal numerical coefficients. Hence, the principal part of $\Omega$ will be the canonical quantization of the boundary Chern--Simons action for a
possibly deformed Lie algebra.
We will return to this example in a future paper (for low order results see \cite{CMW}).

\newcommand{\bX}{{\mathbb{X}}}

\subsection{The Poisson sigma model}
The Poisson sigma model, see Example~\ref{exa-AKSZ}, is important in connection to deformation quantization
\cite{Kdq,CFdq}. 
It is also a deformation of abelian $BF$ theory.
Its fields are usually denoted by $\sfX$ and $\sfeta$ instead of $\sfA$ and $\sfB$.
For a source two\ndash manifold $M$ and target $\bbR^n$, we have $\sfX\in\Omega^\bullet(M)\otimes\bbR^n$
and $\sfeta\in\Omega^\bullet(M)\otimes(\bbR^n)^*[1]$.
Given a Poisson bivector
field $\pi$ on $\bbR^n$, the BV action reads
\[
\calS_M=\int_M\left(
\sum_{i=1}^n\sfeta_i\dd\sfX^i+
\frac12 \sum_{i,j=1}^n \pi^{ij}(\sfX)\sfeta_i\sfeta_j
\right).
\]
As an AKSZ theory its boundary BFV action has the same form:
\[
\calS^\de_\Sigma=\int_\Sigma\left(
\sum_{i=1}^n\sfeta_i\dd\sfX^i+
\frac12 \sum_{i,j=1}^n \pi^{ij}(\sfX)\sfeta_i\sfeta_j
\right).
\]
We will denote by $\bX$ and $\bE$ the boundary fields corresponding to $\sfX$ and $\sfeta$, respectively.
The standard quantization of $\calS^\de_\Sigma$ in the $\bX$\ndash  representation is a second-order differential operator, whereas in the $\bE$\ndash representation it is in general of unbounded order (unless
$\pi$ is polynomial).

For the quantization of the PSM one has to pick a background, i.e., a constant map
$x\colon M\to\bbR^n$, and expand around it (by abuse of notation we will write $x$ also for the image of this map).
In the standard quantization of $\calS^\de_\Sigma$ in the $\bE$\ndash representation we Taylor-expand
$\pi$ around $x$, thus getting in general a formal power series in $\sfX$.

Recall that the quantization of the PSM on the upper half plane \cite{CFdq} yields Kontsevich's star product \cite{Kdq}.
This is an associative product on $C^\infty(\bbR^n)[[\ii\hbar]]$. We write
\[
f\star g = fg+\sum_{I,J} 
B^{IJ}\,\frac{\de^{|I|}}{\de x^I}f\,\frac{\de^{|J|}}{\de x^J}g=
fg-\frac{\ii\hbar}2\sum_{ij}\pi^{ij}\,\frac{\de f}{\de x^i}\frac{\de g}{\de x^j} + O(\hbar^2),
\]
where $I$ and $J$ are multi-indices (and $i$ and $j$ are indices) and $B^{IJ}=0$
if $|I|=0$ or $|J|=0$.
\begin{Lem}\label{l:PSMErep}
In the $\bE$\ndash representation, we have
\[
\Omega = \Omega_0 + \int_{\Sigma}
\sum_{IJKRS}
\frac{(-\ii\hbar)^{|K|-|I|-|J|+1}}
{(|K|+|R|+|S|)!}
\,
\de_KB^{IJ}(x)\,[\bE_I\bE_R][\bE_J\bE_S]\,
\frac{\delta^{|K|+|R|+|S|}}{\delta\bE_K\delta\bE_R\delta\bE_S}.
\]
\end{Lem}
Note that $\Omega^2=0$ follows 
 from the associativity of the star product. Also notice that
the principal part of $\Omega$ is
the standard quantization of
\[
\Tilde\calS^\de_\Sigma=\int_\Sigma\left(
\sum_{i=1}^n\sfeta_i\dd\sfX^i+
\frac12 \sum_{i,j=1}^n \Pi^{ij}(\sfX)\sfeta_i\sfeta_j
\right),
\]
where
\[
\Pi^{ij} = \frac{B^{ij}-B^{ji}}{-\ii\hbar}=\frac{x^i\star x^j-x^j\star x^i}{-\ii\hbar} =\pi^{ij}+O(\hbar).
\]

\begin{proof}[Sketch of the proof]
The main remark is that the propagator in a boundary face near the boundary is Kontsevich's propagator.
To see this recall that the propagator on a closed two\ndash manifold $M$
restricts
to the boundary $\de C_2(M)=STM$, with 
$ST$ denoting the sphere bundle of the tangent bundle, to a global angular form $\gamma$.
By choosing a Riemannian metric, we may view
$STM$ as $O(M)\times_{SO(2)}S^1$, where $O$ denotes the orthogonal frame bundle.
The pullback of $\gamma$ to $O(M)\times S^1$ is $\omega-\theta$, where $\omega$ is the normalized
invariant volume form on $S^1$ and $\theta$ some metric connection
(regarded as an $\mathfrak{so}(2)$\ndash valued $1$\ndash form on $O(M)$). The propagator for a manifold
with boundary is constructed by the method of image charges, see Appendix~\ref{a:prop}. Hence, $\theta$ drops
out and $\omega$ gets replaced by Kontsevich's propagator (notice that in higher dimension connection
dependent terms in the propagator survive).

We use the following convention: edges in the graph are oriented pointing from the
$\sfeta$\ndash vertex to the $\sfX$\ndash vertex.

In the $\bE$-representation we have arrows pointing to the boundary and the bulk vertices have two outgoing arrows. If we have $n$ bulk vertices and $m$ boundary vertices,
then the form degree is $2n$, whereas the dimension is $2n+m-2$. Since the propagators do not depend on boundary variables,
we must have equality between dimension and degree for the integral not to vanish: hence, $m=2$.
The resulting graphs are the same as in Kontsevich's star product.
The edges that leave the graph do either correspond to derivatives of the coefficients or get directly attached to a boundary vertex.
\end{proof}

To deal with the $\bX$\ndash representation, we have have to consider graphs on the upper half plane
with opposite boundary conditions as in \cite{CFdq}. These boundary conditions have been considered in
\cite{CFbranes}. In the present setting, we define
\[
\Hat\pi = \sum_{Kij}\frac1{|K|!}\,\theta_i\theta_j\,\de_K\pi^{ij}(x)\,\frac{\de^{|K|}}{\de\theta_K},
\]
where $K$ is a multi-index and the $\theta$s are the coordinates on $\bbR^n[1]$. Since $\pi$ is Poisson,
$\Hat\pi$ is a MC element in the graded Lie algebra of multivector fields on $\bbR^n[1]$. The Poisson sigma model on the upper half plane with the boundary conditions as in \cite{CFbranes} produces a (curved) $A_\infty$\ndash structure
on $C^\infty(\bbR^n[1])[[\ii\hbar]][(\ii\hbar)^{-1}]$ that quantizes $\Hat\pi$.
We write
\begin{multline*}
\mu_k(\phi_1,\dots,\phi_k)=\phi_1\phi_2\,\delta_{k2}+
\sum_{I_1\dots I_k}
A_{I_1\dots I_k}
\,\de^{I_1}\phi_1\cdots\de^{I_k}\phi_k=\\=
\phi_1\phi_2\,\delta_{k2}+
-\frac12\frac{(\ii\hbar)^{k-1}}{k!}\sum_{iji_1\dots i_k}
\theta^i\theta^j\,
\de_{i_1}\cdots\de_{i_k}\pi^{ij}(x)\,
\de^{i_1}\phi_1\cdots\de^{i_k}\phi_k
+ O(\hbar^k)
,
\end{multline*}
where $I_1,\dots, I_k$ are multi-indices and $i,j,i_1,\dots, i_k$ are indices,
and $A_{I_1\dots I_k}=0$ if $|I_r|=0$ for some $r$.
Derivatives
with an upper (multi)index refer to the $\theta$\ndash coordinates:
$\de^i:=\frac\de{\de\theta_i}$. Note that $A_{I_1\dots I_k}$ is a function of $\theta$ (and of
the background $x$).
\begin{Lem}
In the $\bX$\ndash representation, we have
\begin{multline*}
\Omega = \Omega_0 -
\sum_{k=0}^\infty\frac1{k!}
\int_{\Sigma}
\sum_{LI_1\dots I_kR_1\dots R_k}
\frac{(\ii\hbar)^{|L|-(|I_1|+\dots +|I_k|)+1}}{
(|L|+|R_1|+\dots+|R_k|)!
}\,\cdot\\
\cdot\de^L {A_{I_1\dots I_k}}_{|_{\theta=0}}\,
[X^{I_1}X^{R_1}]
\cdots
[X^{I_k}X^{R_k}]\,
\frac{\delta^{|L|+|R_1|+\dots+|R_k|}}
{\delta\bX^L\delta\bX^{R_1}\cdots\delta\bX^{R_k}}.
\end{multline*}
\end{Lem}
Note that $\Omega^2=0$ follows 
 from the $A_\infty$\ndash relations.
\begin{proof}[Sketch of the proof]
The first part of the proof of Lemma~\ref{l:PSMErep} carries over. For the second part,
specific for the chosen representation, we just have to observe that the graphs we obtain
are those appearing in \cite{CFbranes} to define the (curved)  $A_\infty$\ndash structure.
\end{proof}

\subsubsection{Example}
Consider $M$ the disk, $\pi$ a constant Poisson structure structure and
$\de_1 M = \de M=S^1$; i.e., we work in $\bE$\ndash representation. We denote
by $z$ and $z^+$ the coefficients, in $\bbR^n$, for the residual fields.
The effective action is easily computed as
\[
S^\text{eff}_{S^1}(\bE,z,z^+) = \sum_{i=1}^n\int_{S^1} \bE_i z^i +
\frac12 \sum_{i,j=1}^n \pi^{ij}\int_{C_2(S^1)} \pi_1^*\bE_i\,\zeta\, \pi_2^*\bE_j+
\sum_{i,j=1}^n \pi^{ij} z^+_i \int_{S^1} \bE_i \tau,
\]
where $\tau\in\Omega^1(S^1)$ is the result of the integral over the bulk vertex of the graph
with one bulk vertex connected to one boundary vertex and
$\zeta\in\Omega^0(C_2(S^1))$ is the result of the integral over the bulk vertex of the graph
with one bulk vertex connected to two boundary vertices. Notice that $\int_{S^1}\tau=1$ and that
$\zeta$ is a propagator for $S^1$ satisfying $\dd\zeta=\pi_1^*\tau-\pi_2^*\tau$.
It is not difficult to check that $\EE^{\frac\ii\hbar S^\text{eff}_{S^1}(\bE,z,z^+)}$ is
$(\hbar^2\Delta+\Omega)$\ndash closed with
\[
\Omega = \int_{S^1}\left(
\ii\hbar\sum_{i=1}^n \dd \bE_i\frac\delta{\delta\bE_i}+
\frac12 \sum_{i,j=1}^n \pi^{ij} \bE_i\bE_j
\right).
\]

\subsubsection{The deformation quantization of the relational symplectic groupoid}
In the applications to deformation quantization \cite{Kdq,CFdq,CFbranes} one imposes boundary conditions,
for example $\sfeta=0$ if no branes are present.

Let $D_n$ denote the disk with the boundary $S^1$ split into $2n$ intervals $I$ intersecting only at the end points and
with the boundary condition $\sfeta=0$ on alternating intervals. The remaining $n$ intervals are free, so
the space of boundary fields is $\calF^\de_{D_n}=(\calF^\de_I)^n$ with
\[
\calF^\de_I =\Omega^\bullet(I)\otimes\bbR^n\oplus\Omega_0^\bullet(I)\otimes(\bbR^n)^*[1],
\]
with $\Omega_0^\bullet(I)$ denoting the subcomplex of forms whose restriction to the end points is zero.
We will denote by $\calH$ the vector space that quantizes $\calF^\de_I$ in one of the two usual polarizations.

We may then view the state $m_x$ associated to $D_3$ perturbing around a constant solution $X=x$
as a linear map $\calH\otimes\calH\to\calH$.
There are two inequivalent ways to cut $D_4$ into gluings of two $D_3$s. {}From this we see that
$m_x$ defines an associative structure in the $(\hbar^2\Delta+\Omega)$\ndash cohomology for $D_4$.
This provides a way of defining the deformation quantization of the relational symplectic
groupoid of \cite{CCrsg}. 

\newcommand{\bbC}{{\mathbb{C}}}
To compare this result with the deformation quantization of the Poisson manifold $W$, we have to consider also $D_1$. We view
the state $\sigma_x$ associated to it as a linear map $\calH\to\bbC[[\epsilon]]$, with $\epsilon={\ii\hbar}/2$.
If $f$ is a function on $W$, we may also take the expectation value of $f(X(u_0))$ where $u_0$ is a
point in the interior of the interval with the boundary condition. We denote the result by
$\tau_xf$. We may view $\tau_x$ as a linear map $\calC^\infty(W)\otimes\bbC[[\epsilon]]\to\calH$.
Kontsevich's star product is then obtained by composition:
\[
f\star g(x) = \sigma_x(m_x(\tau_xf\otimes\tau_xg)).
\]
\begin{Rem}
Notice that the classical field $X$ on the boundary defines a path in the target $W$. Thus,
if we work in the $\bX$\ndash representation, the degree zero part of $\calH$ is
$\Fun(PW)\otimes\bbC[[\epsilon]]$, where $\Fun(PW)$ denotes a convenient space of functions on the path space
$PW$ of $W$. There is a canonical inclusion $\iota\colon W\to PW$ that maps a point to a constant map with that value. We may regard $\sigma$ as a deformation of $\iota^*\colon\Fun(PW)\to\calC^\infty(W)$.
Given $\nu\in\Omega^1(I)$ with $\int_I\nu=1$, we also have a map
$p\colon PW\to W$, $X\mapsto\int_IX\nu$. We may then regard $\tau$ as a deformation of
$p^*\colon\calC^\infty(W)\to \Fun(PW)$
with $\nu$ the result of integrating the free boundary vertex of the graph with one edge joining
the free boundary vertex to $u_0$.
\end{Rem}



\appendix

\section{The Hodge decomposition for manifolds with boundary}\label{a:Hodge}
In this Appendix we describe a form of Hodge decomposition for manifolds with boundary that in particular shows that \eqref{e:Lgf}
is a gauge fixing. In this Section $M$ is a smooth compact Riemannian manifold with boundary, with the metric having product structure near the boundary (cf. Footnote \ref{footnote: product structure}). We denote by $*$ the Hodge operator and by by $\dd^*$ the corresponding adjoint of the de~Rham differential. We call a form ultra-harmonic if it closed with respect to both $\dd$ and $\dd^*$.\footnote{Notice that  this implies that the form is harmonic,
but, in the presence of a boundary, this is a stronger condition.} We denote by $\hHarm^\bullet(M)$ the space of ultra-harmonic forms on $M$.

\subsection{Ultra-Dirichlet and Ultra-Neumann forms}
For the following construction we need a refinement of the notion of Dirichlet and Neumann forms. Let $M$ be a compact manifold with boundary $\de M$. We fix a given boundary component $\de_i M$.
\begin{Def}
We say that a differential form $\mu$ on $M$ is ultra-Dirichlet relative to $\de_iM$ if the pullbacks to $\de_iM$ of all the even normal derivatives of $\mu$ and the pullbacks of all the odd normal derivatives of $*\mu$ vanish. Similarly, we say that $\mu$ is ultra-Neumann
relative to $\de_iM$ if the pullbacks to $\de_iM$ of all the even normal derivatives of $*\mu$ and the pullbacks of all the odd normal derivatives of $\mu$ vanish. We denote by $\Omega^\bullet_{\hDD i}(M)$ and by $\Omega^\bullet_{\hNN i} (M)$ the spaces of ultra-Dirichlet and
ultra-Neumann forms, respectively. Notice that they are subcomplexes both for $\dd$ and for $\dd^*$.\footnote{This property relies on having a product metric near the boundary.}
\end{Def}
Near the boundary component $\de_iM$, we can write a form $\mu$ as
\[
\mu=\alpha+\lambda\dd t,
\]
where $t$ is the normal coordinate, and $\alpha$ and $\lambda$ are $t$\ndash dependent forms on $\de_i M$.
With this notation, $\mu$ is ultra-Dirichlet if and only if $\left(\frac\dd{\dd t}\right)^n_{|_{t=0}}\alpha=0$ for $n=0,2,4,\dots$ and
$\left(\frac\dd{\dd t}\right)^n_{|_{t=0}}\lambda=0$ for $n=1,3,5,\dots$. It is ultra-Neumann if and only if $\left(\frac\dd{\dd t}\right)^n_{|_{t=0}}\lambda=0$ for $n=0,2,4,\dots$ and
$\left(\frac\dd{\dd t}\right)^n_{|_{t=0}}\alpha=0$ for $n=1,3,5,\dots$.
In the following we are going to need the following formulae:
\begin{align}
\dd\mu&= \dd'\alpha + (\dot\alpha+\dd'\lambda)\dd t,\label{se:dmu}\\
*\mu &= *'\lambda + (*'\alpha)\dd t,\label{se:starmu}\\
\dd^*\mu &= (\dd^{*\prime}\alpha+\dot\lambda) + (\dd^{*\prime}\lambda)\dd t,\label{se:dstarmu}
\end{align}
where $\dd'$ is the de~Rham differential on $\de_i M$, $*'$ is the Hodge operator for the induced metric, $\dd^{*\prime}$ is the formal adjoint of $\dd'$, and the dot denotes  the derivative with respect to $t$. These formulae immediately imply the following
\begin{Lem}
An ultra-harmonic Dirichlet form is ultra-Dirichlet and an ultra-harmonic Neumann form is ultra-Neumann.
\end{Lem}
With a bit more work, we also have the following
\begin{Lem}\label{l:fromND to ultra}
Fix a  a neighborhood $U_i$  of a boundary component $\de_iM$. Let $\mu\in\Omega^k(M)$ for some $0\leq k\leq d$. The following statements hold:
\begin{enumerate}
\item \label{l:fromND to ultra (1)} If $\dd\mu=0$, then there is a $\nu\in\Omega^{k-1}_{\DD i}$ with support in $U_i$ such that $\mu-\dd\nu\in\Omega^k_{\hNN i}$.
Moreover,
\begin{enumerate}
\item \label{l:fromND to ultra 1a} if $\mu\in\Omega^k_{\NN i}$, then $\dd\nu\in\Omega^k_{\NN i}$;
\item if $\mu\in\Omega^k_{\DD i}$, then one can choose $\nu$ as above such that in addition $\mu-\dd\nu\in\Omega^k_{\hDD i}$
\end{enumerate}
\item \label{l:fromND to ultra (2)} If $\dd^*\mu=0$, then there is a $\nu\in\Omega^{k+1}_{\NN i}$ with support in $U_i$ such that $\mu-\dd^*\nu\in\Omega^k_{\hDD i}$.
Moreover,
\begin{enumerate}
\item if $\mu\in\Omega^k_{\DD i}$, then $\dd^*\nu\in\Omega^k_{\DD i}$;
\item if $\mu\in\Omega^k_{\NN i}$, then one can choose $\nu$ as above such that in addition $\mu-\dd^*\nu\in\Omega^k_{\hNN i}$
\end{enumerate}
\end{enumerate}
\end{Lem}
\begin{proof}
For (\ref{l:fromND to ultra (1)}), we pick a
$t$\ndash dependent form $\gamma$ on $\de_i M$ to be determined below. We pull it back to a neighborhood of $\de_iM$ and multiply it by a bump function supported in $U_i$ and equal to $1$ in a neighborhood of $\de_iM$. This will define $\nu$.
In the latter neighborhood we have
$\dd\nu=\dd'\gamma+\dot\gamma\dd t$, so $\mu-\dd\nu=(\alpha-\dd'\gamma)+(\lambda-\dot\gamma)\dd t=:\alpha'+\lambda'\dd t$.
This shows that we can choose $\gamma$ so that $\lambda'=0$. Since $\mu$ is closed, this automatically implies that $\dot\alpha'=0$.
In particular, this shows that $\mu-\dd\nu\in\Omega_{\hNN i}$.
This immediately implies (\ref{l:fromND to ultra 1a}).
If $\mu$ is Dirichlet, then $\alpha|_{t=0}=0$. By choosing $\gamma$ with
$\gamma|_{t=0}=0$, we get $\alpha'|_{t=0}=0$ which, together with $\dot\alpha'=0$, implies that $\alpha'=0$. In conclusion,
$\mu-\dd\nu$ vanishes in a whole neighborhood of $\de_iM$ and in particular is ultra-Dirichlet.

Statement (\ref{l:fromND to ultra (2)}) follows from (\ref{l:fromND to ultra (1)}) by applying Hodge star $\ast: \Omega^\bt \ra \Omega^{d-\bt}$ to all objects and renaming $\ast \mu \mapsto \mu$, $\ast \nu \mapsto\nu$, $k\mapsto d-k$.
\end{proof}

Now, as in Section~\ref{s:abeBF}, we split the boundary of $M$ into two disjoint components $\de_1M$ and $\de_2M$.
The above Lemma can be used in a neighborhood of each boundary component. In particular, we may choose the neighborhoods $U_1$ and $U_2$ to be disjoint. We thus get isomorphisms
\begin{subequations}\label{e:HHrelative}
\begin{align}
H^\bullet_{\hNN1,\hDD2}(M)&\simeq H^\bullet(M,\de_2M)=H^\bullet_{\DD 2}(M),\label{e:HHrelative 1}\\
H^\bullet_{\hNN2,\hDD1}(M)&\simeq H^\bullet(M,\de_1M)=H^\bullet_{\DD 1}(M),\label{e:HHrelative 2}
\end{align}
\end{subequations}
with $H^\bullet_{\hNN1,\hDD2}(M)$ the de~Rham cohomology of $\Omega^\bullet_{\hNN 1}(M)\cap\Omega^\bullet_{\hDD 2}(M)$,
and $H^\bullet_{\hNN2,\hDD1}(M)$ the de~Rham cohomology of $\Omega^\bullet_{\hNN 2}(M)\cap\Omega^\bullet_{\hDD 1}(M)$.\footnote{
In the case of (\ref{e:HHrelative 1}), the map $i_*\colon H^\bullet_{\hNN1,\hDD2}(M)\to H^\bullet_{\DD 2}(M)$ is induced by the inclusion $i\colon \Omega^{\bullet,\mathrm{closed}}_{\hNN1,\hDD2}\to \Omega^{\bullet,\mathrm{closed}}_{\DD 2}$ while the map in the opposite direction $j\colon H^\bullet_{\hNN1,\hDD2}(M)\leftarrow H^\bullet_{\DD 2}(M)$ sends a cohomology class $[\mu]$ of $\mu\in\Omega^{\bullet,\mathrm{closed}}_{\DD 2}$ to the class of the form $\mu-d\nu\in \Omega^{\bullet,\mathrm{closed}}_{\hNN1,\hDD2}$, constructed using (\ref{l:fromND to ultra (1)}) of Lemma \ref{l:fromND to ultra}, in $H^\bullet_{\hNN1,\hDD2}(M)$. These two maps are obviously mutually inverse. One point that requires a comment is that $j$ is well-defined (or, equivalently, that $i_*$ is injective): if $\alpha\in \Omega^n_{\hNN1,\hDD2}(M)$ is exact, i.e. $\alpha=d\beta$ with $\beta\in \Omega^{n-1}_{\DD2}(M)$, then one can find another primitive $\gamma\in \Omega^{n-1}_{\hNN1,\hDD2}(M)$ such that $\alpha=d\gamma$. To construct such $\gamma$, choose a smooth map $\Phi\colon [0,1]\times M\rightarrow M$ such that $\Phi_0=\mathrm{id}_M$, $\Phi_\tau$ the identity on $\de M$ for any $\tau\in [0,1]$, and such that normal derivatives of $\Phi_1$ of all orders vanish on the boundary. Then we construct the primitive as $\gamma=\int_0^1 \Phi^*\alpha + \Phi_1^*\beta$; it satisfies the required boundary conditions. The second isomorphism (\ref{e:HHrelative 2}) is constructed similarly.
}

\subsection{Doubling the manifold (twice)}\label{ss:doubling}
\label{a:Hodge doubling twice}
Pick a second copy of $M$ with opposite orientation and glue it to $M$ along $\de_1M$. This defines a new compact Riemannian manifold with boundary, which we denote by $M'$. On this manifold we can define an (orientation reversing) involution $S_1$ that maps a point in one copy of $M$ to the same point in the other copy.

We now repeat the operation with $M'$ by gluing it to a second copy of itself with opposite orientation along the whole boundary. We now get a compact closed Riemannian manifold $M''$. We can extend the involution $S_1$ to it, but we can also define a new (orientation reversing) involution $S_2$ that maps a point in one copy of $M'$ to the same point in the other copy. Notice that, by construction, the metric on $M''$ is invariant under $S_1$ and $S_2$. As a consequence, pullbacks on differential forms,
$S_1^*$ and $S_2^*$, anticommute with $*$ and commute with $\dd$, and hence also commute with $\dd^*$.

We denote by
$\Omega^\bullet_{S_1^e,S_2^o}(M'')$ the $(\dd,\dd^*)$\ndash subcomplex\footnote{By a $(\dd,\dd^*)$\ndash complex we simply mean a $\mathbb{Z}$\ndash graded vector space which is simultaneously a cochain complex with respect to $\dd$ and a chain complex with respect to $\dd^*$. Since $\dd$ and $\dd^*$ do not commute, this is obviously not a bi-complex.} of forms that are even with respect to $S_1^*$ and odd with respect to $S_2^*$.
Similarly, we denote by $\Omega^\bullet_{S_2^e,S_1^o}(M'')$ the  $(\dd,\dd^*)$\ndash subcomplex of forms that are even with respect to $S_2^*$ and odd with respect to $S_1^*$. Setting $\Omega^\bullet_{\hNN i,\hDD j}(M):=\Omega^\bullet_{\hNN i}(M)\cap \Omega^\bullet_{\hDD j}(M)$,
$i\not= j$ in $\{1,2\}$,
we have the following isomorphisms of  $(\dd,\dd^*)$\ndash complexes:
\begin{alignat*}{2}
q_{12}&\colon\Omega^\bullet_{\hNN1,\hDD2}(M)&\to\Omega^\bullet_{S_1^e,S_2^o}(M''),\\
q_{21}&\colon\Omega^\bullet_{\hNN2,\hDD1}(M)&\to\Omega^\bullet_{S_2^e,S_1^o}(M''),
\end{alignat*}
which are obtained by extending the differential forms from $M$ to $M''$.
Thanks to \eqref{e:HHrelative}, we then get the isomorphisms
\begin{subequations}\label{e:HHHarm}
\begin{align}
H^\bullet_{\DD2}(M) &\simeq H^\bullet _{S_1^e,S_2^o}(M'')=\Harm^\bullet _{S_1^e,S_2^o}(M''),\\
H^\bullet_{\DD 1}(M) &\simeq H^\bullet _{S_2^e,S_1^o}(M'')=\Harm^\bullet _{S_2^e,S_1^o}(M''),
\end{align}
\end{subequations}
where $\Harm^\bullet$ denotes the space of harmonic forms and we have used Hodge's theorem on $M''$. Notice that, by the $q_{ij}$'s,
$\Harm^\bullet _{S_1^e,S_2^o}(M'')$ and $\Harm^\bullet _{S_2^e,S_1^o}(M'')$ are the subspaces of ultra-Harmonic forms in $\Omega^\bullet_{\DD2}(M)$
and $\Omega^\bullet_{\DD1}(M)$, respectively. More precisely,
\begin{subequations}\label{e:HarmhHarm}
\begin{align}
q_{12}^{-1}(\Harm^\bullet _{S_1^e,S_2^o}(M'')) &= \hHarm^\bullet_{\NN1,\DD2}(M),\\
q_{21}^{-1}(\Harm^\bullet _{S_2^e,S_1^o}(M'')) &= \hHarm^\bullet_{\NN2,\DD1}(M).
\end{align}
\end{subequations}

\begin{Lem}\label{l:ortho}
Fix two integers $0\leq k,l\leq d$ satisfying $k+l=d$. Then the symplectic orthogonal of
\[
\calL=
\begin{array}{c}
(\dd^*\Omega^{k+1}_{\NN 2}(M))
\cap\Omega^k_{\DD1}(M)\\
\oplus\\
(\dd^*\Omega^{l+1}_{{\NN}1}(M))
\cap\Omega^l_{\DD2}(M)
\end{array}
\]
in $\Omega^k_{\DD1}(M)\oplus\Omega^l_{\DD2}(M)$ is
\[
\begin{array}{c}
\hHarm^k_{\NN2,\DD1}(M)
\oplus (\dd^*\Omega^{k+1}_{{\NN}2}(M))
\cap\Omega^k_{\DD1}(M)\\
\oplus\\
\hHarm^l_{\NN1,\DD2}(M)
\oplus
(\dd^*\Omega^{l+1}_{{\NN}1}(M))
\cap\Omega^l_{\DD2}(M).
\end{array}
\]
\end{Lem}
\begin{proof}
We have to prove that $\beta\in\Omega^l_{\DD2}(M)$ satisfies $\int_M\beta\,\alpha=0$ for every
$\alpha\in(\dd^*\Omega^{k+1}_{{\NN}2}(M))
\cap\Omega^k_{\DD1}(M)$ if and only if $\beta\in q_{12}^{-1}(\Harm^l_{S_1^e,S_2^o}(M''))\oplus
(\dd^*\Omega^{l+1}_{{\NN}1}(M))
\cap\Omega^l_{\DD2}(M)$.
Similarly, we have to prove that $\alpha\in\Omega^k_{\DD1}(M)$ satisfies $\int_M\beta\,\alpha=0$ for every
$\beta\in(\dd^*\Omega^{l+1}_{\text{N}1}(M))
\cap\Omega^l_{\DD2}(M)$ if and only if $\alpha\in q_{21}^{-1}(\Harm^k_{S_2^e,S_1^o}(M''))\oplus (\dd^*\Omega^{k+1}_{\text{N}2}(M))
\cap\Omega^k_{\DD1}(M)$.

We prove the first statement only, as the proof of the second is identical (by exchanging the role of the boundary indices $1$ and $2$, and interchanging $k$ and $l$).
We start with the (easier) ``if'' part. We write $\alpha=\dd^*\gamma$ with $\gamma\in\Omega^{k+1}_{{\NN}2}(M)$ and $\dd^*\gamma\in\Omega^k_{\DD1}(M)$.
Up to sign, we have that $\int_M\beta\,\alpha$ is
equal to $\int_M (*\beta)\,\dd*\gamma$. Since $\dd^*\beta=0$, this is equal to the boundary term which, up to a sign, is
$\int_{\de M} (*\beta)\,(*\gamma)$. This boundary term vanishes
since $\gamma\in\Omega^{k+1}_{{\NN}2}(M)$ and $\beta\in\Omega^l_{{\NN}1}(M)$.

We now have to prove the ``only if'' part.
Writing $\alpha=\dd^*\gamma$, we have that $\int_M (*\beta)\,\dd*\gamma=0$ for every $\gamma\in\Omega^{k+1}_{\text{N}2}(M)$ with
$\dd^*\gamma\in\Omega^k_{\DD1}(M)$. In particular, we may take $\gamma$ to be a bump form near any point in the bulk and vanishing on the boundary (so that we can integrate by parts). This implies
\[
\dd^*\beta=0.
\]
This in turns implies $\int_{\de M} (*\beta)\,(*\gamma)=0$ for every $\gamma$ as above. Since $\gamma\in\Omega^{k+1}_{\text{N}2}(M)$,
we actually have $\int_{\de_1 M} (*\beta)\,(*\gamma)=0$ for every $\gamma$ as above. If, in a neighborhood of $\de_1 M$, we write $\gamma$
as $\sigma+\lambda\dd t$, we get, as in \eqref{se:dstarmu}, $d^*\gamma=(\dd^{*\prime}\sigma+\dot\lambda) + (\dd^{*\prime}\lambda)\dd t$.
The condition $\dd^*\gamma\in\Omega^k_{\DD1}(M)$ implies that $\dd^{*\prime}\sigma+\dot\lambda$ vanishes on $\de_1M$, but this puts no condition
on the restriction $*'\lambda$
of $*\gamma$ to $\de_1M$. As a consequence, we get that $*\beta$ must vanish on $\de_1M$, i.e.,
\[
\beta\in\Omega^l_{\NN 1}(M)
\]

To summarize, we now know that $\dd^*\beta=0$ and $\beta\in\Omega^l_{\NN1,\DD2}$. Thanks to Lemma~\ref{l:fromND to ultra}, part $(2)$,
picking $\nu$ appropriately near each boundary component, we conclude that there is a $\nu\in\Omega^{l+1}_{\NN1,\NN2}$ with $\dd^*\nu\in\Omega^l_{\DD2}(M)$
such that
$\beta':=\beta-\dd^*\nu$ belongs to $\Omega^l_{\hNN1,\hDD2}(M)$. So $q_{12}(\beta')\in\Omega^l_{S_1^e,S_2^o}(M'')$ and
$\dd^*q_{12}(\beta')=0$. By the Hodge decomposition theorem on $M''$ (which has no boundary), we get
$q_{12}(\beta')\in\Harm^l_{S_1^e,S_2^o}(M'')\oplus \dd^*\Omega^{l+1}_{S_1^e,S_2^o}(M'')$ and hence
$\beta'\in q_{12}^{-1}(\Harm^l_{S_1^e,S_2^o}(M''))\oplus \dd^*\Omega^{l+1}_{\hNN1,\hDD2}(M)$ and, in turn,
\begin{multline*}
\beta\in q_{12}^{-1}(\Harm^l _{S_1^e,S_2^o}(M''))\oplus(\dd^*(\Omega^{l+1}_{\NN1,\NN2}(M)+\Omega^{l+1}_{\hNN1,\hDD2}(M)))\cap\Omega^l_{\DD2}(M)\subset \\
\subset q_{12}^{-1}(\Harm^l_{S_1^e,S_2^o}(M''))\oplus
(\dd^*\Omega^{l+1}_{\text{N}1}(M))
\cap\Omega^l_{\DD2}(M).
\end{multline*}

\end{proof}

\subsection{The Hodge propagator}
\label{a: Hodge propagator}

\subsubsection{Strong and weak Hodge decompositions}
\begin{Def} We say that a cochain complex of real (possibly, infinite-dimensional) vector spaces $(V^\bullet,\dd)$ admits a {\sf strong} Hodge decomposition if it is equipped with a positive inner product $(,)\colon V^j\otimes V^j\rightarrow \mathbb{R}$, $\dd$ has an adjoint $\dd^*\colon V^\bullet\rightarrow V^{\bullet-1}$ with respect to $(,)$ and $V^\bullet$ splits as a direct sum of eigenspaces of the Laplacian $\Delta_\mathrm{Hodge}=\dd\dd^*+\dd^*\dd\colon V^\bullet \rightarrow V^\bullet$. As a consequence, $V^\bullet$ splits as
$$V^\bullet=V^\bullet_\mathrm{Harm}\oplus \dd(V^{\bullet-1})\oplus \dd^*(V^{\bullet+1})$$
with $V^\bullet_\mathrm{Harm}=\ker\Delta_\mathrm{Hodge}\simeq H^\bullet(V)$ the harmonic representatives of cohomology.
\end{Def}

\begin{Def} For a cochain complex $(W^\bullet,\dd)$, we call a {\sf weak} Hodge decomposition a decomposition of the form
\begin{equation}\label{e:weak Hodge}
W^\bullet=\iota(H^\bullet(W))\oplus \dd(W^{\bullet-1})\oplus K(W^{\bullet+1})
\end{equation}
where $\iota\colon H^\bullet(W)\rightarrow W^\bullet$ is a choice of representatives of cohomology, $K\colon W^\bullet\rightarrow W^{\bullet-1}$ is a linear map (the chain contraction) satisfying
$$\dd K+ K \dd= \mathrm{id}-\iota\circ p,\quad K\circ\iota = p\circ K=0$$
with $p\colon W^\bullet\rightarrow H^\bullet(W)$ a choice of projection onto cohomology.
\end{Def}

To a strong Hodge decomposition of $V^\bullet$, one can canonically associate the data of the weak Hodge decomposition of $V^\bullet$, where $\iota$ represents the cohomology class by a harmonic cochain, $p$ takes the cohomology class of the orthogonal projection of the input cochain onto harmonic cochains, and the chain contraction is given by
\begin{equation}\label{e:K Hodge abstract}
K_\mathrm{Hodge}= \dd^*/(\Delta_\mathrm{Hodge}+P_\mathrm{Harm})
\end{equation}
where $P_\mathrm{Harm}=\iota\circ p$ is the orthogonal projection onto harmonic cochains.

\subsubsection{The hierarchy of boundary conditions on differential forms}
\label{a: hierarchy of bc}
Returning to the setting of a Riemannian manifold $M$ with boundary $\de M=\de_1 M\sqcup \de_2 M$, consider the tower of inclusions
\begin{equation}\label{e:tower of bc}
\Omega^\bullet_{\DD1}(M)\supset \Omega^\bullet_{\DD1,\NN2}(M)\supset \Omega^\bullet_{\mathrm{rel}1,\mathrm{abs}2}(M)\supset \Omega^\bullet_{\hDD1,\hNN2}(M).
\end{equation}
Here, following \cite{Ray-Singer,Cheeger}, we say that a form $\alpha$ satisfies {\sf relative} boundary condition on $\de_1 M$ if $\alpha|_{\de_1 M}=\dd^*\alpha|_{\de_1 M}=0$ 
and satisfies {\sf absolute} boundary condition on $\de_2 M$ if $*\alpha|_{\de_2 M}=*\dd\alpha|_{\de_2 M}=0$.
Similarly, we have a tower related to (\ref{e:tower of bc}) by applying the Hodge star to all terms:
\begin{equation}\label{e:tower of bc 2}
\Omega^\bullet_{\DD2}(M)\supset \Omega^\bullet_{\NN1,\DD2}(M)\supset \Omega^\bullet_{\mathrm{abs}1,\mathrm{rel}2}(M)\supset \Omega^\bullet_{\hNN1,\hDD2}(M).
\end{equation}
Note that only the rightmost terms in (\ref{e:tower of bc},\ref{e:tower of bc 2}) are closed with respect to $\dd$ and $\dd^*$. Leftmost terms are closed with respect to $\dd$ but not $\dd^*$, and middle terms are closed with respect to neither (in particular, they are not cochain complexes).

All the graded vector spaces in (\ref{e:tower of bc},\ref{e:tower of bc 2}) 
are equipped with the Hodge inner product $(\alpha,\beta)=\int_M \alpha\wedge *\beta$. On $\Omega^\bullet_{\mathrm{rel}1,\mathrm{abs}2}(M)$
the operators $\dd$ and $\dd^*$ are mutually adjoint, i.e. $(d\alpha,\beta)=(\alpha,\dd^*\beta)$, and the spectral problem for the Laplacian is well-posed, however, as pointed out above, these operators spoil the relative/absolute boundary conditions, i.e. are not endomorphisms of $\Omega^\bullet_{\mathrm{rel}1,\mathrm{abs}2}(M)$.
Moreover, if $\alpha\in \Omega^\bullet_{\mathrm{rel}1,\mathrm{abs}2}(M)$ (or even in $\Omega^\bullet_{\DD1,\NN2}(M)$) is an eigenform of the Laplacian $\Delta_\mathrm{Hodge}$, then it is automatically in $\Omega^\bullet_{\hDD1,\hNN2}(M)$.\footnote{
Indeed, assume that $\alpha\in \Omega^p_{\DD1,\NN2}(M)$ is an eigenform of $\Delta_\mathrm{Hodge}$ with eigenvalue $\lambda$. For $i=1,2$, near $\de_i M$ the Laplacian decomposes as $\Delta_\mathrm{Hodge}=\Delta_\mathrm{Hodge,\de_i}-\frac{\dd^2}{\dd t^2}$ with $t$ the normal coordinate near boundary. Thus,  near $\de_i M$  we have $\alpha=\sum_r \theta_{\de_i,r}^{(p)}\left( a_r \cos(\omega_r^{(p)} t)+b_r \sin(\omega_r^{(p)} t)\right)+ dt\cdot\sum_s \theta_{\de_i,s}^{(p-1)}\left( c_s \cos(\omega_s^{(p-1)} t)+d_s \sin(\omega_s^{(p-1)} t)\right) $ where sums are over the eigenforms $\theta_{\de_i}$ of the boundary Laplacian on $\de_i M$ of degrees $p$ and $p-1$, respectively, with $r,s$ the indices enumerating the boundary spectrum in these degrees. Denoting eigenvalues of the latter by $\mu_{\de_i}$, for the (possibly, imaginary)
frequencies $\omega$ we have $\lambda=(\omega_r^{(p)})^2+\mu_{\de_i,r}^{(p)}= (\omega_s^{(p-1)})^2+\mu_{\de_i,s}^{(p-1)}$. Relative boundary condition on $\de_1 M$ enforces $a_r=d_s=0$, which implies the ultra-Dirichlet condition; similarly, the absolute  boundary condition on $\de_2 M$ enforces $b_r=c_s=0$, which implies the ultra-Neumann condition.
}
The case of $\Omega^\bullet_{\mathrm{abs}1,\mathrm{rel}2}(M)$ vs. $\Omega^\bullet_{\hNN1,\hDD2}(M)$ works analogously.

\subsubsection{The Hodge propagator}
\label{a: Hodge propagator subsubsec}
As follows from the discussion of Sections \ref{a:Hodge doubling twice}, \ref{a: hierarchy of bc}, the complexes $\Omega^\bullet_{\hDD1,\hNN2}(M)$ and $\Omega^\bullet_{\hNN1,\hDD2}(M)$ possess a strong Hodge decomposition, whereas all other terms of (\ref{e:tower of bc},\ref{e:tower of bc 2}) do not. On $\Omega^\bullet_{\hDD1,\hNN2}(M)$ we construct the chain contraction as in (\ref{e:K Hodge abstract}):
\begin{equation}\label{e:K Hodge}
K_\mathrm{Hodge}^{\hDD1,\hNN2}
=\dd^*/(\Delta_\mathrm{Hodge}+P_\mathrm{Harm})\colon\quad
\Omega^\bullet_{\hDD1,\hNN2}(M)\rightarrow \Omega^{\bullet-1}_{\hDD1,\hNN2}(M).\end{equation}
Similarly, on $\Omega^\bullet_{\hNN1,\hDD2}$ we have the chain contraction
$$K_\mathrm{Hodge}^{\hNN1,\hDD2}
=\dd^*/(\Delta_\mathrm{Hodge}+P_\mathrm{Harm})\colon\quad
\Omega^\bullet_{\hNN1,\hDD2}(M)\rightarrow \Omega^{\bullet-1}_{\hNN1,\hDD2}(M).$$

Being the inverse of an elliptic operator (composed with $\dd^*$), the chain contractions above are integral operators
$$K_\mathrm{Hodge}^{\hDD1,\hNN2}= (\pi_1)_*(\eta_\mathrm{Hodge}\wedge \pi_2^*(-)),\qquad
K_\mathrm{Hodge}^{\hNN1,\hDD2}= (\pi_1)_*(\eta'_\mathrm{Hodge}\wedge \pi_2^*(-))$$
with integral kernels $\eta_\mathrm{Hodge}$, $\eta'_\mathrm{Hodge}$ given by smooth $(d-1)$\ndash forms on the configuration space of two points $C_2^0(M)$.
Since the complexes $\Omega^\bullet_{\hDD1,\hNN2}(M)$ and $\Omega^\bullet_{\hNN1,\hDD2}(M)$ are dual to each other by Poincar\'e pairing $\int_M\alpha\wedge\beta$, we have
\begin{equation}\label{e:T eta = eta'}
T^*\eta_\mathrm{Hodge}=(-1)^d\;
\eta'_\mathrm{Hodge}
\end{equation}
where $T: C_2^0(M)\rightarrow C_2^0(M)$ maps $(x_1,x_2)\mapsto (x_2,x_1)$. 
Equation (\ref{e:T eta = eta'}) implies that $\eta_\mathrm{Hodge}$ satisfies boundary conditions $\hDD1,\hNN2$ in the first argument and $\hNN1,\hDD2$ in the second argument; $\eta'_\mathrm{Hodge}$ satisfies the opposite boundary conditions: $\hNN1,\hDD2$ in the first argument and $\hDD1,\hNN2$ in the second argument.

\begin{Def}
We call the form  $\eta_\mathrm{Hodge}\in \Omega^{d-1}(C_2^0(M))$ defined as above, i.e. as the integral kernel of the chain contraction (\ref{e:K Hodge}), the {\sf Hodge propagator} on $M$.
\end{Def}

This is the adaptation of the propagator of Axelrod-Singer \cite{AS} to manifolds with boundary.

Finally, notice that one can use $\eta_\mathrm{Hodge}$ to define the chain contraction of the whole complex $\Omega_{\DD1}^\bullet(M)$, given by the same formula $K_\mathrm{Hodge}=(\pi_1)_*(\eta_\mathrm{Hodge}\wedge\pi_2^*(-))$ (i.e. we extend the domain of $K_\mathrm{Hodge}^{\hDD1,\hNN2}$ by relaxing the boundary conditions from $\hDD1,\hNN2$ to $\DD1$). This defines a weak Hodge decomposition (\ref{e:weak Hodge}) of $\Omega_{\DD1}^\bullet(M)$:
$$ \Omega_{\DD1}^\bullet(M) = \widehat{\mathrm{Harm}}^\bullet_{\DD1,\NN2}(M)\;\oplus \;  \underbrace{\dd\,\Omega^{\bullet-1}_{\DD1}(M)}_{\mathrm{im}(\dd)}  \;\oplus\;
\underbrace{\dd^*\,\Omega^{\bullet+1}_{\NN2}\;\cap \Omega^\bullet_{\DD1}}_{\mathrm{im}(K_\mathrm{Hodge})} $$

\section{Constructing the propagator: ``soft'' method and the method of image charges}\label{a:prop}
Recall that, if $N$ is a closed, compact $d$\ndash manifold, then it is possible to construct a propagator $\eta_N$ on $N$ as in \cite{BC,C,CR}.

Namely, one has first to choose an inclusion $\iota$ of $H^\bullet(N)$ into $\Omega^\bullet(N)$. This determines a representative of the Poincar\'e dual $\chi_\Delta$ of the diagonal $\Delta$ in $N\times N$ and, by restriction, a representative $e_N$ of the Euler class of $N$:
\begin{align*}
\chi_\Delta &= \sum_i (-1)^{d\cdot\deg\chi_i}\pi_1^*\chi^N_i\,\pi_2^*\chi_N^i,\\
e_N &= \sum_i (-1)^{\deg\chi_i}\chi_N^i \chi^N_i,
\end{align*}
where $\pi_1$ and $\pi_2$ are the projections from $N\times N$ to $N$, $\{\chi^N_i\}$ is the image under $\iota$ of a basis of $H^\bullet(N)$ and
$\{\chi^N_i\}$ is the image of the dual basis.

Next one picks a global angular form $\vartheta$ on the sphere bundle $STN$ such
that $\dd\vartheta$ is the pullback of the representative of $e_N$. By explicit construction, one can obtain a $(d-1)$\ndash form $\sigma_N$ with the following properties:
\begin{align*}
\dd\sigma_N &= \pi^*\chi_\Delta,\\
\iota_{\de}^*\sigma_N &= \vartheta,\\
T^*\sigma_N &= (-1)^d\sigma_N,
\end{align*}
where $\pi$ is the projection $C_2(N)\to N\times N$, $\iota_\de$ is the inclusion map $STN=\de C_2(N)\hookrightarrow C_2(N)$, and $T$ is the involution
of $C_2(N)$ the sends $(x,y)$ to $(y,x)$.

It follows that $\eta_N:=(-1)^{d-1}\sigma_N$ is a propagator for the abelian $BF$ theory 
on $N$.

We now want to use the above construction to get a propator for the manifold with boundary $M$ by using a variant of the method of image charges.
First we double it twice to $M''$ as in Appendix~\ref{ss:doubling}. By using the involutions $S_1$ and $S_2$ defined there, we may write
\[
\Omega^\bullet(M'')=\Omega^\bullet_{S_1^e,S_2^e}(M'')\oplus \Omega^\bullet_{S_1^e,S_2^o}(M'')\oplus \Omega^\bullet_{S_1^o,S_2^e}(M'')\oplus \Omega^\bullet_{S_1^o,S_2^o}(M),
\]
and similarly in cohomology.
Notice that, since $S_1$ and $S_2$ are orientation reversing, an $S^e_i$ component is paired to an $S^o_i$ component.
We choose the embedding $\iota\colon H^\bullet(M'')\hookrightarrow\Omega^\bullet(M'')$ to respect this decomposition and construct a propagator $\eta_{M''}$ accordingly.
Next we define
\[
\Breve C_2^0(M''):=\{(x,y)\in M''\times M'' : x\not=y,\ S_1(x)\not=y,\ x\not=S_2(y),\ S_1(x)\not=S_2(y)\}
\]
as a subspace of $C_2^0(M'')$. We extend to $\Breve C_2^0(M'')$ the involutions $S_1$ and $S_2$ as
\begin{align*}
\Breve S_1(x,y)&:=(S_1(x),y),\\
\Breve S_2(x,y)&:=(x,S_2(y)).
\end{align*}
Finally, we denote by $\Breve\eta$ the restriction of the propagator $\eta_{M''}$ to $\Breve C_2^0(M'')$ and define $\eta$ as the extension to the compactification $C_2(M)$
of the restriction to $C_2^0(M)\subset \Breve C_2^0(M'')$ of
\[
\Breve\eta':=
\Breve\eta - \Breve S_1^*\Breve\eta -
\Breve S_2^*\Breve\eta +
\Breve S_1^* \Breve S_2^*\Breve\eta.
\]
It is readily verified that $\eta$ is a propagator on $M$ with respect to the embeddings of $H^\bullet(M,\de_1M)$ and $H^\bullet(M,\de_2 M)$ into
$\Omega^\bullet_{\DD1}(M)$ and $\Omega^\bullet_{\DD2}(M)$ (actually, $\Omega^\bullet_{\hDD1,\hNN2}(M)$ and $\Omega^\bullet_{\hNN1,\hDD2}(M)$)
given by the following forms
\begin{align*}
\chi_i &= 2\iota_M^*\chi_i^{M'',S_1^o,S_2^e},\\
\chi^i &= 2\iota_M^*\chi^i_{M'',S_1^e,S_2^o},
\end{align*}
with $\iota_M$ the inclusion $M\hookrightarrow M''$.

\begin{Rem} The Hodge propagator of Appendix \ref{a: Hodge propagator} is a special case of this construction, corresponding to $\eta_{M''}$ being the Hodge propagator on $M''$.
\end{Rem}

\begin{Rem}[One boundary component]
If we group all the boundary components of $M$ into $\de_1 M$, so $\de_2 M=\emptyset$, the formulae get simplified as follows.
First, we have the decomposition $\Omega^\bullet(M')=\Omega^\bullet_{S_1^e}(M')\oplus \Omega^\bullet_{S_1^o}(M')$, and similarly in cohomology,
where $M'$ is the doubling of $M$ defined in Appendix~\ref{ss:doubling}. We choose the embedding $\iota\colon H^\bullet(M')\hookrightarrow\Omega^\bullet(M')$ to respect this decomposition and construct a propagator $\eta_{M'}$ accordingly. Next we define
\[
\Breve C_2^0(M'):=\{(x,y)\in M'\times M' : x\not=y,\ S_1(x)\not=y\}
\]
as a subspace of $C_2^0(M')$. We extend to $\Breve C_2^0(M')$ the involution $S_1$ as
\[
\Breve S_1(x,y):=(S_1(x),y).
\]
Finally, we denote by $\Breve\eta$ the restriction of the propagator $\eta_{M'}$ to $\Breve C_2^0(M')$ and define $\eta$ as the extension to the compactification $C_2(M)$
of the restriction to $C_2^0(M)\subset \Breve C_2^0(M')$ of
\[
\Breve\eta':=
\Breve\eta - \Breve S_1^*\Breve\eta.
\]
Again, it is readily verified that $\eta$ is a propagator on $M$ with respect to the embedding of $H^\bullet(M,\de_1M)$ into
$\Omega^\bullet_{\DD1}(M)$  (actually, $\Omega^\bullet_{\hDD1}(M)$)
given by the following forms
\begin{align*}
\chi_i &= 2\iota_M^*\chi_i^{M',S_1^o},\\
\chi^i &= \iota_M^*\chi^i_{M',S_1^e},
\end{align*}
with $\iota_M$ the inclusion $M\hookrightarrow M'$.

\end{Rem}

\section{Examples of propagators}
\label{a: exa of prop}
\begin{Exa}[Interval with opposite polarization on the endpoints]
Let $M=[0,1]$ be an interval with coordinate $t$. We set $\de_1 M=\{1\}$, $\de_2 M=\{0\}$. Then the space of {\backgrounds} is empty $\calV_M=0$ and the propagator is
\begin{equation}\label{e:eta interval 12}
\eta(t_1,t_2)=-\Theta(t_2-t_1)\quad \in \Omega^0(C^0_2(M),\DDD)
\end{equation}
with $\Theta(x)=\left\{\begin{array}{ll} 1 & x> 0 \\ 0 & x<0\end{array}\right.$ the step function (which we never have to evaluate at zero, since the diagonal $t_1=t_2$ is removed from the configuration space where $\eta$ is defined). The $\DDD$-boundary condition (\ref{e: Omega D}) simply means that $\eta(t_1,t_2)$ vanishes if either $t_1=1$ or $t_2=0$.
The associated chain contraction of $\Omega_{\DD1}^\bt(M)$ (which is an acyclic complex) is
\begin{equation}\label{e:K [0,1]}
K\colon f+g\,dt \mapsto \int_{0}^1 \eta(t,t_2) g(t_2) dt_2=-\int_t^1 g(t_2)dt_2.
\end{equation}
It satisfies $\dd K+K\dd=\mathrm{id}_{\Omega^\bullet_{\DD1}(M)}$, which is equivalent to $\dd\eta=0$ accompanied by the discontinuity condition \begin{equation}\label{e:eta jump}
\eta(t+0,t)-\eta(t-0,t)=1.
\end{equation}
The propagator (\ref{e:eta interval 12}) is in fact unique 
and does indeed extend to the ASMF compactification, which simply amounts to attaching boundary strata $\{(t+0,t)\,|\, t\in[0,1]\}$ and $\{(t-0,t)\,|\, t\in[0,1]\}$ to $C_2^0$.
\end{Exa}

\begin{Exa}[Interval with same polarization on the endpoints]
\label{exa interval 11}
Consider again the unit interval, but now set $\de_1 M=\{0\}\sqcup\{1\}$ and $\de_2 M=\varnothing$. Then the space of {\backgrounds} is non-empty, since $H^1_{\DD1}(M)=\mathbb{R}=H^0_{\DD2}(M)$ (the other cohomology spaces vanish), and we choose the basis $[\chi_1]=[\dd t]\in H^1_{\DD1}(M)$ and $[\chi^1]=[1]\in H^0_{\DD2}(M)$.  Thus $\calV_M=\mathbb{R}[k-1]\oplus \mathbb{R}[-k]$ and we write the {\backgrounds} as
$\sfa=z^1\cdot \dd t$, $\sfb=z_1^+\cdot 1$ with coordinates $z^1, z_1$ of degrees $k-1$ and $-k$, respectively. We have
\begin{equation}\label{e:eta interval 11}
\eta(t_1,t_2)=\Theta(t_1-t_2)-t_1
\end{equation}
which satisfies the equation
$d\eta=-\dd t_1\wedge 1_{t_2}$ (cf. (\ref{e:deta})), the discontinuity condition (\ref{e:eta jump}) and $\DDD$-boundary condition
$\eta(0,t_2)=\eta(1,t_2)=0$.
\end{Exa}

The case of the interval with both boundary points marked as $\de_2$ works similarly. The propagator in this case is:
\begin{equation}\label{e:eta interval 22}
\eta(t_1,t_2)=-\Theta(t_2-t_1)+t_2.
\end{equation}

\begin{Exa}[Circle]
\label{exa: eta circle}
Let $M=S^1$ be a circle with coordinate $t\in [0,1]$ with points $t=0$ and $t=1$ identified. The basis in cohomology is $[\chi_0]=[1]\in H^0(M)$, $[\chi_1]=[\dd t]\in H^1(M)$; the Poincar\'e-dual basis is $[\chi^0]=[\dd t]\in H^1(M)$, $[\chi^1]=[1]\in H^0(M)$.  Hence $\calV_M=\mathbb{R}[k]\oplus\mathbb{R}[k-1]\oplus \mathbb{R}[-1-k]\oplus\mathbb{R}[-k]$
and the {\backgrounds} are $\sfa=z^0\cdot 1+z^1\cdot dt$, $\sfb=z_0^+\cdot \dd t+z_1^+\cdot 1$ where the coordinates $z^0,z^1,z^+_0,z^+_1$ have degrees $k$, $k-1$, $-1-k$, $-k$ respectively.
The propagator is:
$$\eta(t_1,t_2)=\Theta(t_1-t_2)-t_1+t_2-\frac12$$
it is periodic in $t_1,t_2$ and moreover is a smooth function on the configuration space $C_2^0(S^1)$. It also clearly satisfies the discontinuity condition (\ref{e:eta jump}) and the equation (\ref{e:deta}):
$\dd\eta=-\dd t_1\wedge 1_{t_2}+1_{t_1}\wedge \dd t_2$.
The propagator also satisfies the anti-symmetry property
\begin{equation}\label{e:eta circle antisym}
\eta(t_2,t_1)=-\eta(t_1,t_2).
\end{equation}
\end{Exa}

\begin{Exa}[The 2-sphere]\label{exa: eta sphere}
Let $M=S^2$ be the 2-sphere which we endow with a complex coordinate $z\in\mathbb{C}\cup\{\infty\}$ via stereographic projection. The cohomology is $H^0(M)=H^2(M)=\mathbb{R}$, $H^1(M)=0$ and we choose the basis $[\chi_0]=[1]\in H^0(M)$, $[\chi_1]=[\mu]\in H^2(M)$ with
\begin{equation}\label{e: mu sphere}
\mu=\frac{1}{2\pi}\frac{i\, \dd z\wedge \dd\bar z}{(1+|z|^2)^2}
\end{equation}
the $SO(3)$-invariant volume form on the sphere of total volume $1$. The dual basis is $[\chi^0]=[\mu]\in H^2(M)$, $[\chi^1]=[1]\in H^0(M)$. We have $\calV_M=\mathbb{R}[k]\oplus \mathbb{R}[k-2]\oplus \mathbb{R}[-1-k]\oplus \mathbb{R}[1-k]$, with {\backgrounds} $\sfa=z^0\cdot 1+z^1\cdot \mu$, $\sfb=z_0^+\cdot\mu+z_1^+\cdot 1$; coordinates $z^0,z^1,z^+_0,z^+_1$ have degrees $k$, $k-2$, $-1-k$, $1-k$, respectively. The $SO(3)$-invariant propagator is
\begin{equation}\label{e:eta sphere}
\eta=\frac{1}{2\pi}\; \frac{|1+z_1\bar z_2|^2}{(1+|z_1|^2)\;(1+|z_2|^2)}\;\left(\dd_1\;\mathrm{arg}\left( \frac{z_1-z_2}{1+z_1 \bar z_2}\right)+
\dd_2\;\mathrm{arg}\left( \frac{z_2-z_1}{1+z_2 \bar z_1}\right)\right)
\end{equation}
where $\dd_1=\dd z_1\frac{\de}{\de z_1}+\dd\bar z_1\frac{\de}{\de \bar z_1}$, $\dd_2=\dd z_2\frac{\de}{\de z_2}+\dd\bar z_2\frac{\de}{\de \bar z_2}$ are the de Rham differentials in $z_1$ and $z_2$, respectively.
It is smooth on the configuration space $C_2^0(S^2)$ and extends smoothly to the compactification by the tangent circle bundle of $S^2_\mathrm{diag}$, it satisfies (\ref{e:deta}): $\dd\eta=-\mu_{z_1}\wedge 1_{z_2}-1_{z_2}\wedge \mu_{z_2}$. Instead of the discontinuity property (\ref{e:eta jump}), we have the property
$$\lim_{\epsilon\rightarrow 0} \;\oint_{\phi=0}^{2\pi} \eta(z_1=z_2+\epsilon\cdot e^{i\phi},\; z_2)=1.$$
Moreover, the propagator (\ref{e:eta sphere})
 is symmetric with respect to interchanging $z_1$ and $z_2$:
\begin{equation}\label{e:T eta}
T^*\eta=\eta
\end{equation}
where $T\colon C_2^0(S^2)\rightarrow C_2^0(S^2)$ sends $(z_1,z_2)\mapsto (z_2,z_1)$, cf. (\ref{e:eta circle antisym}).\footnote{Note that one cannot expect such a property for a propagator on a manifold with boundary, as there are different boundary conditions on the two arguments.}
Note also that (\ref{e:eta sphere}) can be obtained as the $SO(3)$-invariant extension of the propagator with $z_2$ fixed to $0$:\footnote{
I.e. we recover the first term of (\ref{e:T eta}) (with $\dd_1$) by pulling back (\ref{e:eta sphere 1-point}) by a $z_2$-dependent M\"obius transformation $F_{z_2}\colon z\mapsto \frac{z-z_2}{\bar z_2\cdot z+1}$ (which is in the image of $SO(3)$ in $PSL(2,\mathbb{C})$). The second term of (\ref{e:eta sphere}) is recovered by enforcing the symmetry (\ref{e:T eta}).
}
\begin{equation}\label{e:eta sphere 1-point}
\eta(z_1,0)=\frac{1}{2\pi}\;\frac{1}{1+|z_1|^2}\;\dd\;\mathrm{arg}(z_1).
\end{equation}
The properties above do not characterize $\eta$ uniquely: one can add to (\ref{e:eta sphere}) a term of the form $\dd\,\Phi(\mathrm{Dist}(z_1,z_2))$ where $\mathrm{Dist}(z_1,z_2)$ is the geodesic distance between the two points with respect to the round metric on $S^2$ and $\Phi$ can be any smooth even function on $\bR/2\pi\mathbb{Z}$.

One can show that (\ref{e:eta sphere}) is in fact the Hodge propagator (cf. Appendix \ref{a: Hodge propagator}) corresponding to the round metric on $S^2$, while shifting $\eta$ by $\dd\,\Phi(\mathrm{Dist}(z_1,z_2))$ destroys this property.
\end{Exa}

\begin{Exa}
Let $M=D$ be a $2$\ndash disk, which we view as the unit disk in the complex plane, or a hemisphere (via stereographic projection) $\{z\in \mathbb{C}\; :\; |z|\leq 1\}$. Set $\de_1 M=\de M$ the boundary circle and $\de_2 M=\varnothing$. We choose the basis vector $[\chi_0]=[2\mu]$ in $H^2_{\DD1}(M)$ and the dual one $[\chi^0]=[1]$ in $H^0_{\DD2}(M)$. Here $\mu$ is given by (\ref{e: mu sphere}); note that $\mu$ has volume $1/2$ on the hemisphere, hence the normalization of the class $[2\mu]$. The space of {\backgrounds} is $\calV_M=\mathbb{R}[k-2]\oplus \mathbb{R}[1-k]$. The propagator can be constructed by the method of Appendix \ref{a:prop} for the propagator (\ref{e:eta sphere}) for the sphere:
$$\eta(z_1,z_2)=\eta_{S^2}(z_1,z_2)-\eta_{S^2}(\bar z_1^{-1},z_2)$$
Here we denoted $\eta_{S^2}$ the propagator (\ref{e:eta sphere}). For the method of image charges, we are using the involution $z\mapsto \bar z^{-1}$ on $S^2$ which has the equator $|z|=1$ as its locus of fixed points.

If instead we assign the boundary circle as $\de_2 M$, the relevant cohomology becomes $H^0_{\DD1}(M)=\mathrm{Span}([1])$, $H^2_{\DD2}(M)=\mathrm{Span}([2\mu])$; the space of {\backgrounds} becomes $\calV_M=\mathbb{R}[k]\oplus \mathbb{R}[-1-k]$. The corresponding propagator is
\begin{equation*}
\eta(z_1,z_2)=\eta_{S^2}(z_1,z_2)-\eta_{S^2}(z_1,\bar z_2^{-1})
\end{equation*}

Another example of a propagator on a disk was considered in \cite{CFvacua}.
\end{Exa}

\subsection{Axial gauge on a cylinder}

The following example comes from the construction of axial gauge-fixing, in the sense of \cite{BCM}, a special case of the construction of tensor product for induction data in homological perturbation theory \cite{discrBF,cell_ab_BF}.

The propagators we construct here are not smooth differential forms on the compactified configuration space, but rather distributional forms on $M\times M$. Properties (\ref{e:deta}) and normalization of the integral over the $(d-1)$-cycle given by one point spanning an infinitesimal sphere around the other point, are replaced by the distributional identity
$d\eta=\delta^{(d)}_{M,\mathrm{diag}}+(-1)^{d-1}\sum_i (-1)^{d\cdot \deg\chi_i}\pi_1^*\chi_i \pi_2^* \chi^i$. Here $\delta^{(d)}_{M,\mathrm{diag}}$ is the distributional $d$-form on $M\times M$ supported on the diagonal, the integral kernel of the identity map $\Omega^\bullet(M)\rightarrow \Omega^\bullet(M)$.

\begin{Exa}[Two distributional propagators on a cylinder]
\label{exa: axial gauge}
Let $\Sigma$ be a closed $(d-1)$- dimensional manifold with $[\chi_{(\Sigma)i}]$ a basis in $H^\bullet(\Sigma)$, $[\chi_{(\Sigma)}^i]$ the dual basis and $\eta_\Sigma\in \Omega^{d-2}(C_2^0(\Sigma))$ a propagator. Let $M=\Sigma\times [0,1]$, with assignments $\de_1 M=\Sigma\times\{1\}$, $\de_2 M=\Sigma\times \{0\}$. Then $H^\bullet_{\DD1}(M)=H^\bullet_{\DD2}(M)=0$ and hence $\calV_M=0$. Then there are the following two distributional propagators on $M$:
\begin{equation}\label{e:eta axial}
\eta^\mathrm{axial}((x_1,t_1),(x_2,t_2)) =  \eta_{[0,1]}(t_1,t_2)\cdot \delta^{(d-1)}(x_1,x_2),
\end{equation}
\begin{multline}\label{e:eta hor}
\eta^\mathrm{hor}((x_1,t_1),(x_2,t_2)) =  -\delta(t_1-t_2)\cdot (\dd t_1-\dd t_2)\cdot\eta_\Sigma(x_1,x_2)+\\
+\sum_i (-1)^{(d-1)\cdot(\deg\chi_{(\Sigma)i}+1)} \eta_{[0,1]}(t_1,t_2)\cdot\chi_{(\Sigma)i}(x_1)\cdot \chi_{(\Sigma)}^i(x_2).
\end{multline}
Here we denote by $t$ the coordinate on $[0,1]$ and $x$ stands for a point of $\Sigma$; $\delta^{(d-1)}(x_1,x_2)$ is the distributional $(d-1)$\ndash form
$\delta^{(d-1)}_{\Sigma,\mathrm{diag}}$;
$\eta_{[0,1]}=-\Theta(t_2-t_1)$ is the propagator (\ref{e:eta interval 12}). The distributional propagators (\ref{e:eta axial},\ref{e:eta hor}) are the integral kernels of the well-defined chain contractions
$$K^\mathrm{axial}= \mathrm{id}_\Sigma \otimes K_{[0,1]},\qquad
K^\mathrm{hor} = K_\Sigma\otimes \mathrm{id}_{[0,1]} + P_{H^\bullet(\Sigma)}\otimes K_{[0,1]}
$$
acting on smooth forms $\Omega^\bullet(M)=\sum_{j=0}^1 \Omega^{\bullet-j}(\Sigma)\Hat\otimes \Omega^{j}([0,1]) $. Here $P_{H^\bullet(\Sigma)}$ is the projection from $\Omega^\bullet(M)$ onto cohomology $H^\bullet(\Sigma)$; $K_\Sigma$ is the chain contraction for $\Sigma$ associated to the propagator $\eta_\Sigma$ via (\ref{e:K via eta}) and $K_{0,1}$ is the chain contraction (\ref{e:K [0,1]}) for the interval. A propagator closely related to (\ref{e:eta hor}), for the case $\Sigma=\mathbb{R}^2$ (which is non-compact and hence outside of the scope of our treatment), was used in \cite{FroehlichKing,Kontsevich_integral,Bar-Natan} for constructing knot invariants. Also note that in the case of $\Sigma$ being a point, both propagators (\ref{e:eta axial},\ref{e:eta hor}) become (\ref{e:eta interval 12}).
\end{Exa}

\begin{Exa}[Cylinder with the same polarization on the top and the base]
As a modification of Example \ref{exa: axial gauge}, we can take $M=\Sigma\times [0,1]$ with $\de_1 M= \de M=\Sigma\times \{0\}\sqcup \Sigma\times \{1\}$ and $\de_2 M=\varnothing$. Then we have $[(-1)^{d-1}\chi_{(\Sigma)i}\cdot \dd t]$ a basis in $H^\bullet_{\DD1}(M)=H^{\bullet-1}(\Sigma)$ and $[\chi_{(\Sigma)}^i]$ the dual basis in $H^{d-\bullet}_{\DD2}(M)=H^{d-\bullet}(\Sigma)$. The space of {\backgrounds} is $\calV_M=H^\bullet(\Sigma)[k-1]\oplus H^\bullet(\Sigma)[d-k-1]$.
The corresponding propagators are:
\begin{equation}\label{e:eta axial 11}
\eta^\mathrm{axial}((x_1,t_1),(x_2,t_2))= \eta^{1-1}_{[0,1]}(t_1,t_2)\cdot\delta^{(d-1)}(x_1,x_2)- \dd t_1\cdot \eta_\Sigma(x_1,x_2),
\end{equation}
\begin{multline}\label{e:eta hor 11}
\eta^\mathrm{hor}((x_1,t_1),(x_2,t_2)) =  -\delta(t_1-t_2)\cdot (\dd t_1-\dd t_2)\cdot \eta_\Sigma(x_1,x_2)+\\
+\sum_i (-1)^{(d-1)\cdot(\deg\chi_{(\Sigma)i}+1)} \eta_{[0,1]}^{1-1}(t_1,t_2)\cdot \chi_{(\Sigma)i}(x_1)\cdot \chi_{(\Sigma)}^i(x_2).
\end{multline}
Here $\eta^{1-1}_{[0,1]}$ is the propagator (\ref{e:eta interval 11}).

The case of the opposite boundary conditions, i.e. $\de_2 M=\Sigma\times \{0\}\sqcup \Sigma\times \{1\}$, $\de_1 M=\varnothing$, works similarly. Now $[\chi_{(\Sigma)}]$ is the basis in $H^\bullet_{\DD1}(M)=H^\bullet(\Sigma)$ and $[\dd t\cdot\chi^i_{(\Sigma)}]$ is the dual basis in $H^{d-\bullet}_{\DD2}(M)= H^{d-1-\bullet}(\Sigma)$. The corresponding space of {\backgrounds} is $\calV_M=H^\bullet(\Sigma)[k]\oplus H^\bullet(\Sigma)[d-k-2]$. Formulae (\ref{e:eta axial 11},\ref{e:eta hor 11}) become
\begin{equation}\label{e:eta axial 22}
\eta^\mathrm{axial}((x_1,t_1),(x_2,t_2))= \eta^{2-2}_{[0,1]}(t_1,t_2)\cdot\delta^{(d-1)}(x_1,x_2)+ \dd t_2\cdot \eta_\Sigma(x_1,x_2),
\end{equation}
\begin{multline}\label{e:eta hor 22}
\eta^\mathrm{hor}((x_1,t_1),(x_2,t_2)) = -\delta(t_1-t_2)\cdot (\dd t_1-\dd t_2)\cdot \eta_\Sigma(x_1,x_2)+\\
+\sum_i (-1)^{(d-1)\cdot(\deg\chi_{(\Sigma)i}+1)} \eta_{[0,1]}^{2-2}(t_1,t_2)\cdot\chi_{(\Sigma)i}(x_1)\cdot \chi_{(\Sigma)}^i(x_2).
\end{multline}
Here $\eta^{2-2}_{[0,1]}$ is the propagator (\ref{e:eta interval 22}).
\end{Exa}

\section{Gluing formula for propagators}
\label{a:comp}

In this Appendix we complement the discussion of gluing of states in abelian $BF$ theory in Section \ref{ss:gluing} by deriving the gluing formula for propagators, first for the convenient non-minimal realization of the space of residual fields (the direct sum of spaces of residual fields for the manifolds being glued), and then for the minimal (reduced) residual fields. In the first case we implicitly use Fubini theorem for the relevant path integrals, representing a path integral for the glued manifold $M=M_1\cup_\Sigma M_2$ as a triple integral: over fields on $M_1$ and $M_2$ with boundary conditions on the interface $\Sigma$ and over the boundary conditions on $\Sigma$.\footnote{This kind of Fubini theorem for path integrals with insertions of local observables,
also for more general field theories, is the grounds for the gluing formula \ref{e: gluing formula}, and is expected to hold. It can be checked directly in the framework of perturbation theory (e.g. for $BF$-like theories of Section \ref{sec: BF-like}) using the calculus of configuration space integrals, cf. Remark \ref{rem: Fubini}.}
We verify by a direct computation that the resulting glued propagator does indeed satisfy the defining properties of a propagator on $M$, as stated in Section \ref{sec: properties of propagators}, -- Theorem \ref{thm: gluing of propagators}.
(Thus, using also the Mayer-Vietoris formula for torsions \cite{Vishik}, one can prove a posteriori the relevant case of Fubini theorem for path integrals.)
Also, in \cite{cell_ab_BF} we give a different derivation of the same gluing formula for propagators in the language of chain contractions, using standard constructions of homological perturbation theory; from the latter point of view, the desired properties of the propagator are satisfied automatically.

\subsection{Expectation values in abelian $BF$ theory}\label{ss:expv}
We expand the discussion in Section~\ref{ss:state}.
We are in particular interested in the expectation values of the
fields $\sfA$ and $\sfB$. In the interior of $M$ they do not differ from $\Hat\sfA$ and $\Hat\sfB$, so
for test forms $\gamma$ and $\mu$ with support away from $\de M$,
we have
\begin{align*}
\langle\int_M\gamma\sfA\rangle &:=\int_\calL \EE^{\frac\ii\hbar\calS^\calP_M} \int_M\gamma\sfA =
\left(
\int_M\gamma\sfa + (-1)^{d+(d-1)\cdot\deg\gamma}\int_{M\times\de_1M}\pi_1^*\gamma\,\eta\,\pi_2^*\bA\right)\Hat\psi_M\\
\langle\int_M\sfB\mu\rangle &:=\int_\calL \EE^{\frac\ii\hbar\calS^\calP_M} \int_M\sfB\mu =
\left(
\int_M\sfb\mu - (-1)^{d-k+kd}\int_{\de_2M\times M}\pi_1^*\bB\,\eta\,\pi_2^*\mu\right)\Hat\psi_M
\end{align*}
Next, we are interested in the expectation value of $\sfA$ and $\sfB$ located at two different points (i.e., we assume the supports of $\gamma$ and $\mu$ to be disjoint),
\begin{multline*}
\langle\int_M\gamma\sfA\;\int_M\sfB\mu\rangle=
\Big[
(-1)^{d\cdot\deg\gamma}\ii\hbar\int_{M\times M} \pi_1^*\gamma\,\eta\,\pi_2^*\mu+\\+
\left(
\int_M\gamma\sfa + (-1)^{d+(d-1)\cdot\deg\gamma}\int_{M\times\de_1M}\pi_1^*\gamma\,\eta\,\pi_2^*\bA\right)\cdot \\
\cdot \left(
\int_M\sfb\mu -(-1)^{d-k+kd}\int_{\de_2M\times M}\pi_1^*\bB\,\eta\,\pi_2^*\mu\right)
\Big]
\,\Hat\psi_M.
\end{multline*}
This is related to the propagator by
\[
\int_{M\times M} \pi_1^*\gamma\,\eta\,\pi_2^*\mu=
\frac1{T_M}\frac{(-1)^{d\cdot\deg\gamma}}{\ii\hbar}
\langle\int_M\gamma\sfA\;\int_M\sfB\mu\rangle_{z=z^+=\bA=\bB=0}.
\]

\subsection{Gluing propagators for nonreduced \pseudovacua}\label{ss:tildeprop}
Using the discussion in Section~\ref{ss:expv}, we can compute the propagator
$\Tilde\eta\in\Omega^{d-1}(C_2(M))$
on $M$ with the choice $\Tilde\calV_M=\calV_{M_1}\times\calV_{M_2}$
of {\pseudovacua} described in Section~\ref{ss:gluing}.\footnote{This propagator may actually be discontinuous through $\Sigma$ (however the pullback to $\Sigma$ is well-defined), but this is not a problem. See also Section~\ref{ss:dbling}}

By $\gamma_i\in\Omega^\bullet(M)[d-k]$ and $\mu_i\in\Omega^\bullet(M)[k+1]$, $i=1,2$, we denote test forms with support in the interior of $M_i$.
We recover the propagator by
computing first, similarly to what we did in \eqref{e:etahat}, a ``state'' $\Check{\Tilde\eta}$ by
\[
\int_{M\times M} \pi_1^*\gamma_i\, \Check{\Tilde\eta}
\,\pi_2^*\mu_j
=
\frac1{T_{M_1}T_{M_2}}\frac{(-1)^{kd}}{\ii\hbar}
\langle \int_{M\times M}\pi_1^*(\gamma_i\sfA)\pi_2^*(\sfB\mu_j) \rangle
\]
and then setting all the boundary and residual fields to zero. This way, we get \[\Tilde\eta=\Check{\Tilde\eta}|_{\bA_1'=\bB_1=\bA_2=\bB_2'=\sfa_1=\sfb_1=\sfa_2=\sfb_2=0}.\]

For $i=j$ (where we assume the supports of $\gamma_i$ and $\mu_i$ to be disjoint), we have
\[
\langle \int_{M\times M}\pi_1^*(\gamma_i\sfA)\pi_2^*(\sfB\mu_i) \rangle=
(-1)^{d\cdot(k+\deg\gamma_i)}\int_{\bA_1^\Sigma,\bB_2^\Sigma}
\EE^{\frac\ii\hbar(-1)^{d-k}\int_\Sigma\bB_2^\Sigma\bA_1^\Sigma}\,
\langle\int_{M_i}\gamma_i\sfA_i\;\int_{M_i}\sfB_i\mu_i\rangle.
\]
This yields for $i=j=1$,
\begin{multline*}
\langle \int_{M\times M}\pi_1^*(\gamma_1\sfA)\pi_2^*(\sfB\mu_1) \rangle=
\Big[
(-1)^{kd}\,\ii\hbar\int_{M_1\times M_1} \pi_1^*\gamma_1\,\eta_1\,\pi_2^*\mu_1+\\+
(-1)^{d\cdot (k+\deg\gamma_1)}\Big(
\int_{M_1}\gamma_1\sfa_1 + (-1)^{d+(d-1)\cdot\deg\gamma_1}\left(\int_{M_1\times(\de_1M_1\setminus\Sigma)}\pi_1^*\gamma_1\,\eta_1\,\pi_2^*\bA_1'+\right.\\
\left.+\int_{M_1\times\Sigma}\pi_1^*\gamma_1\,\eta_1\pi_2^*\sfa_2
-
(-1)^{k+kd}\int_{M_1\times\Sigma\times\de_1M_2} \varpi_1^*\gamma_1\,p_1^*\eta_1\,p_2^*\eta_2\,\varpi_3^*\bA_2\right)
\Big)\cdot\\
\cdot\left(
\int_{M_1}\sfb_1\mu_1 - (-1)^{d-k+kd}\int_{\de_2M_1\times M_1}\pi_1^*\bB_1\,\eta_1\,\pi_2^*\mu_1\right)
\Big]%
\,\Tilde\psi_M.
\end{multline*}
Similarly, for $i=j=2$ we get,
\begin{multline*}
\langle \int_{M\times M}\pi_1^*(\gamma_2\sfA)\pi_2^*(\sfB\mu_2) \rangle=
\Big[
(-1)^{kd}\ii\hbar\int_{M_2\times M_2} \pi_1^*\gamma_2\,\eta_2\,\pi_2^*\mu_2+\\+
(-1)^{d\cdot(k+\deg\gamma_2)}\left(
\int_{M_2}\gamma_2\sfa_2 + (-1)^{d+(d-1)\cdot\deg\gamma_2}\int_{M_2\times\de_1M_2}\pi_1^*\gamma_2\,\eta_2\,\pi_2^*\bA_2\right)\cdot\\
\cdot\Big(
\int_{M_2}\sfb_2\mu_2 + (-1)^{d-k+kd}\Big(-\int_{(\de_2M_2\setminus\Sigma)\times M_2}\pi_1^*\bB_2'\,\eta_2\,\pi_2^*\mu_2+\\+
\int_{\Sigma\times M_2}\pi_1^*\sfb_1\,\eta_2\,\pi_2^*\mu_2+
(-1)^{k+kd}
\int_{\de_2M_1\times\Sigma\times M_2}
\varpi_1^*\bB_1\,p_1^*\eta_1\,p_2^*\eta_2\,\varpi_3^*\mu_2
\Big)\Big)
\Big]%
\,\Tilde\psi_M.
\end{multline*}

As a consequence,
if $i=j$, we simply get $\int_{M\times M} \pi_1^*\gamma_i\, {\Tilde\eta}
\,\pi_2^*\mu_i=\int_{M\times M} \pi_1^*\gamma_i\, {\eta_i}
\,\pi_2^*\mu_i$; viz., the propagator $\Tilde\eta$ on $M$ coincides with the propagator $\eta_i$ on $M_i$ when both arguments are in $M_i$.

For $i\not=j$, we have instead
\[
\langle \int_{M\times M}\pi_1^*(\gamma_i\sfA)\pi_2^*(\sfB\mu_j) \rangle
= (-1)^{d\cdot(k+\deg\gamma_i)}\int_{\bA_1^\Sigma,\bB_2^\Sigma}
\EE^{\frac\ii\hbar(-1)^{d-k}\int_\Sigma\bB_2^\Sigma\bA_1^\Sigma}\,
\langle\int_{M_i}\gamma_i\sfA_i\rangle\,\langle\int_{M_j}\sfB_j\mu_j\rangle.
\]
The simpler case is when $i=2,j=1$, for in this case the observables do not depend on $\bA_1^\Sigma,\bB_2^\Sigma$.
By Section~\ref{ss:expv} and by \eqref{e:gluedpsi}, we simply get
\begin{multline*}
\langle \int_{M\times M}\pi_1^*(\gamma_2\sfA)\pi_2^*(\sfB\mu_1) \rangle=\\
=(-1)^{d\cdot(k+\deg\gamma_2)}\left(
\int_{M_2}\gamma_2\sfa_2 + (-1)^{d+(d-1)\cdot\deg\gamma_2}\int_{M_2\times\de_1M_2}\pi_1^*\gamma_2\,\eta_2\,\pi_2^*\bA_2\right)\cdot\\
\cdot\left(
\int_{M_1}\sfb_1\mu_1 - (-1)^{d-k+kd}\int_{\de_2M_1\times M_1}\pi_1^*\bB_1\,\eta_1\,\pi_2^*\mu_1\right)
\,\Tilde\psi_M,
\end{multline*}
This implies that $\Check{\Tilde\eta}$, and hence $\Tilde\eta$ vanishes, when the first argument is on $M_2$ and the second argument is on $M_1$.

Next, we come to the case $i=1$ and $j=2$. In this case the observables give nontrivial extra contributions to the integration over
$\bA_1^\Sigma,\bB_2^\Sigma$. We get
\[
\langle \int_{M\times M}\pi_1^*(\gamma_1\sfA)\pi_2^*(\sfB\mu_2) \rangle= C_1+C_2+C_4+C_4
\]
with
\begin{multline*}
C_1=(-1)^{d\cdot(k+\deg\gamma_1)}\left(
\int_{M_1}\gamma_1\sfa_1 + (-1)^{d+(d-1)\cdot\deg\gamma_1} \int_{M_1\times(\de_1M_1\setminus\Sigma)}\pi_1^*\gamma_1\,\eta_1\,\pi_2^*\bA_1'\right)\cdot\\
\cdot\left(
\int_{M_2}\sfb_2\mu_2 - (-1)^{d-k+kd}\int_{(\de_2M_2\setminus\Sigma)\times M_2}\pi_1^*\bB_2'\,\eta_2\,\pi_2^*\mu_2\right)
\,\Tilde\psi_M,
\end{multline*}
\begin{multline*}
C_2=
(-1)^{d+dk+\deg\gamma_1}\left(\int_{M_1\times\Sigma}\pi_1^*\gamma_1\,\eta_1\pi_2^*\sfa_2-
(-1)^{k+kd}\int_{M_1\times\Sigma\times\de_1M_2}
\varpi_1^*\gamma_1\,p_1^*\eta_1\,p_2^*\eta_2\,\varpi_3^*\bA_2
\right)\cdot \\
\cdot\left(
\int_{M_2}\sfb_2\mu_2 - (-1)^{d-k+kd}\int_{(\de_2M_2\setminus\Sigma)\times M_2}\pi_1^*\bB_2'\,\eta_2\,\pi_2^*\mu_2\right)
\,\Tilde\psi_M,
\end{multline*}
\begin{multline*}
C_3=(-1)^{d-k+d\cdot\deg\gamma_1}\left(
\int_{M_1}\gamma_1\sfa_1 + (-1)^{d+(d-1)\cdot\deg\gamma_1}\int_{M_1\times(\de_1M_1\setminus\Sigma)}\pi_1^*\gamma_1\,\eta_1\,\pi_2^*\bA_1'\right)\cdot \\
\cdot\left(\int_{\Sigma\times M_2}\pi_1^*\sfb_1\,\eta_2\,\pi_2^*\mu_2+
(-1)^{k+kd}
\int_{\de_2M_1\times\Sigma\times M_2}
\varpi_1^*\bB_1\,p_1^*\eta_1\,p_2^*\eta_2\,\varpi_3^*\mu_2 \right)
\,\Tilde\psi_M,
\end{multline*}
\begin{multline*}
C_4=(-1)^{d+kd+(d-1)\cdot\deg\gamma_1}\,\Big[
 \ii\hbar\int_{M_1\times\Sigma\times M_2}
\varpi_1^*\gamma_1\,p_1^*\eta_1\,p_2^*\eta_2\,\varpi_3^*\mu_2+\\ %
+(-1)^{d-k}\left(\int_{M_1\times\Sigma}\pi_1^*\gamma_1\,\eta_1\pi_2^*\sfa_2-
(-1)^{k+kd}\int_{M_1\times\Sigma\times\de_1M_2}
\varpi_1^*\gamma_1\,p_1^*\eta_1\,p_2^*\eta_2\,\varpi_3^*\bA_2
\right)\cdot \\
\cdot\left(\int_{\Sigma\times M_2}\pi_1^*\sfb_1\,\eta_2\,\pi_2^*\mu_2+
(-1)^{k+kd}
\int_{\de_2M_1\times\Sigma\times M_2}
\varpi_1^*\bB_1\,p_1^*\eta_1\,p_2^*\eta_2\,\varpi_3^*\mu_2 \right)
\Big]
\,\Tilde\psi_M.
\end{multline*}
This finally implies
\[
\int_{M\times M} \pi_1^*\gamma_1\,{\Tilde\eta}
\,\pi_2^*\mu_2= (-1)^{d+(d-1)\cdot\deg\gamma_1}
\int_{M_1\times\Sigma\times M_2}
\varpi_1^*\gamma_1\,p_1^*\eta_1\,p_2^*\eta_2\,\varpi_3^*\mu_2.
\]
In other words, when the first argument is on $M_1$ and the second on $M_2$, the propagator $\Tilde\eta$ is simply obtained by
taking the product of $\eta_1$ and $\eta_2$ and integrating out the middle point over $\Sigma$.

\subsection{The glued propagator for reduced \pseudovacua}\label{We now do the final step in computing the propagator on $M$ for the reduced  space of pseudovacua}
We now do the final step in computing the propagator on $M$ for the reduced  space of {\pseudovacua} $\Check\calV_M$ of \eqref{e:checkV}. Namely, we define $\Check\eta$ as a $(d-1)$\ndash form on $C_2(M)$ by
\[
\int_{M\times M} \pi_1^*\gamma_i\, \Check{\eta}
\,\pi_2^*\mu_j
=
\frac1{\Check T_M}\frac\ii\hbar\left(\int_{\calL^\times}
\langle \int_{M\times M}\pi_1^*(\gamma_i\sfA)\pi_2^*(\sfB\mu_j) \rangle\right)\Big|_{=0}
\]
Notice that this simply amounts to integrating out the redshirt variables $\sfa_2^\times$ and $\sfb_1^\times$. Since we put all remaining {\pseudovacua} and all boundary fields to zero, the only summands which contribute are those which contain no redshirt variables and those that contain exactly one $\sfa_2^\times$ and one $\sfb_1^\times$ variables. By Gaussian integration, the latter terms produce a pairing by the inverse $V$ of the matrix
$\Lambda$ defined in \eqref{e:Lambda}.
We then get the following:\footnote{
We use notation $\eta(x_1,x_2)$ for the value of a propagator at $(x_1,x_2)\in C_2^0(M)$ as an element of the exterior power of the cotangent bundle: $\eta(x_1,x_2)\in \wedge^{d-1}T^*_{(x_1,x_2)}C_2^0(M)=\oplus_{p=0}^{d-1} (\wedge^pT^*_{x_1}M)\otimes (\wedge^{d-1-p} T^*_{x_2}M)$.
}
\begin{equation}\label{e:glued eta 11}
\Check\eta(x_1,x_2)=\eta_1(x_1,x_2)-\sum_{ij}(-1)^{\deg\chi_{2i}^\times}V^i_j\int_{y\in\Sigma}\eta_1(x_1,y) \chi_{2i}^\times(y)\chi_{1\times}^j(x_2)\qquad \mbox{for}\;x_1,x_2\in M_1,
\end{equation}
\begin{equation}\label{e:glued eta 22}
\Check\eta(x_1,x_2)=\eta_2(x_1,x_2)-\sum_{ij}(-1)^{\deg\chi_{2i}^\times}V^i_j\int_{y\in\Sigma} \chi_{2i}^\times(x_1)\chi_{1\times}^j(y)\eta_2(y,x_2) \qquad \mbox{for}\;x_1,x_2\in M_2,
\end{equation}
\begin{equation}\label{e:glued eta 21}
\Check\eta(x_1,x_2)=-\sum_{ij}(-1)^{d+\deg\chi_{2i}^\times}V^i_j\; \chi_{2i}^\times(x_1)\chi_{1\times}^j(x_2) \qquad \mbox{for}\;x_1\in M_2,\; x_2\in M_1,
\end{equation}
\begin{multline}\label{e:glued eta 12}
\Check\eta(x_1,x_2)=(-1)^d\int_{y\in\Sigma}\eta_1(x_1,y)\eta_2(y,x_2)+\\
+\sum_{ij}(-1)^{\deg\chi_{2i}^\times}V^i_j\int_{y\in\Sigma}\int_{z\in\Sigma} \eta_1(x_1,y)\chi_{2i}^\times(y)\chi_{1\times}^j(z)\eta_2(z,x_2) \qquad \mbox{for}\;x_1\in M_1,\;x_2\in M_2.
\end{multline}

Pictorially we can represent the four cases of gluing as in figures~\ref{fig:gluing1}, \ref{fig:gluing2}, \ref{fig:gluing3}, \ref{fig:gluing4}. Our graphic notations are as follows:  propagators in the submanifolds are denoted by arrows, with the convention that on the l.h.s.\ the propagator vanishes when its tail goes to the boundary, whereas on the r.h.s\ it vanishes when its head goes to the boundary; the propagator in the glued manifold is denoted by a point--dash arrow; a dashed line denotes cohomology classes at its endpoints; finally, a bullet denotes a point on which we integrate.

\begin{figure}[htbp] 
   \centering
   \includegraphics[width=2.5in]{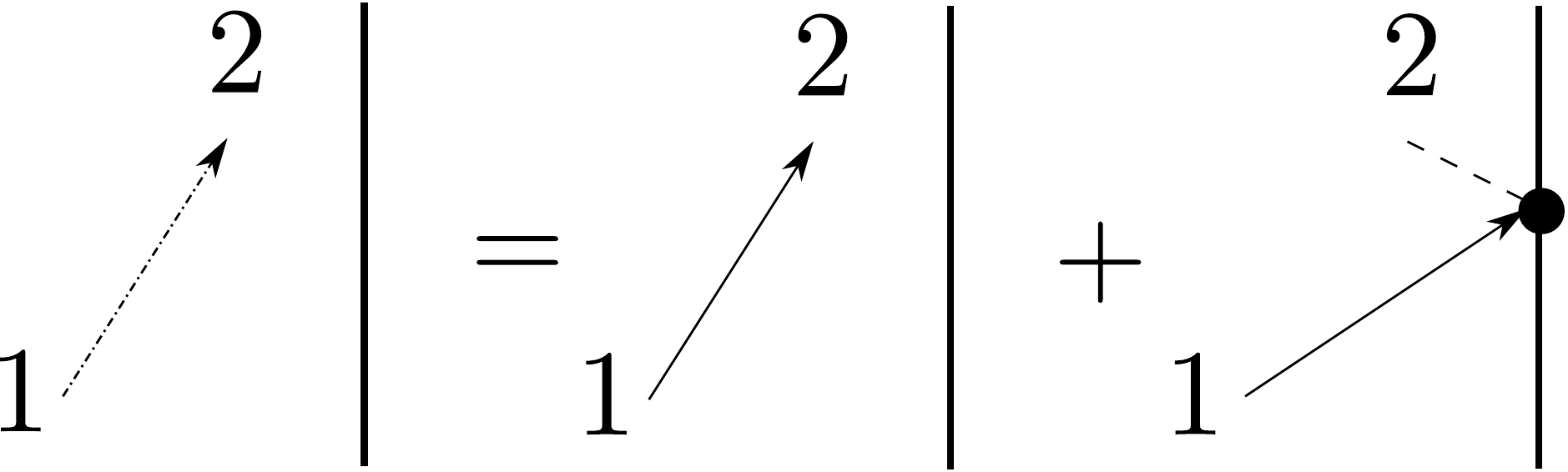}
   \caption{Gluing of propagators: first case}
   \label{fig:gluing1}
\end{figure}

\begin{figure}[htbp] 
   \centering
   \includegraphics[width=2.5in]{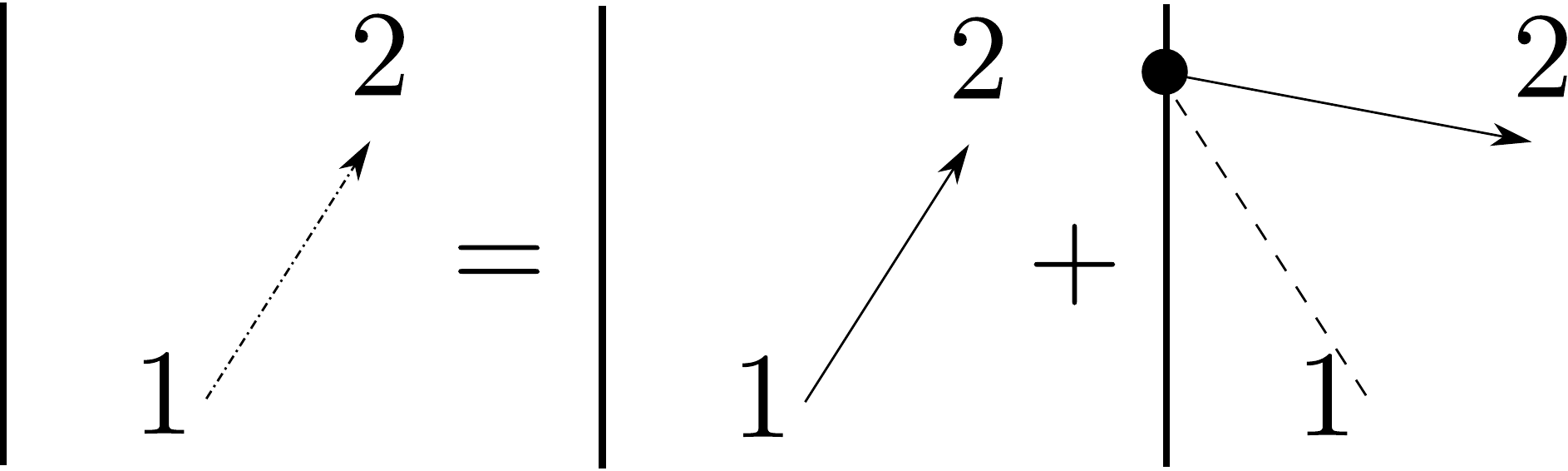}
   \caption{Gluing of propagators: second case}
   \label{fig:gluing2}
\end{figure}

\begin{figure}[htbp] 
   \centering
   \includegraphics[width=2.5in]{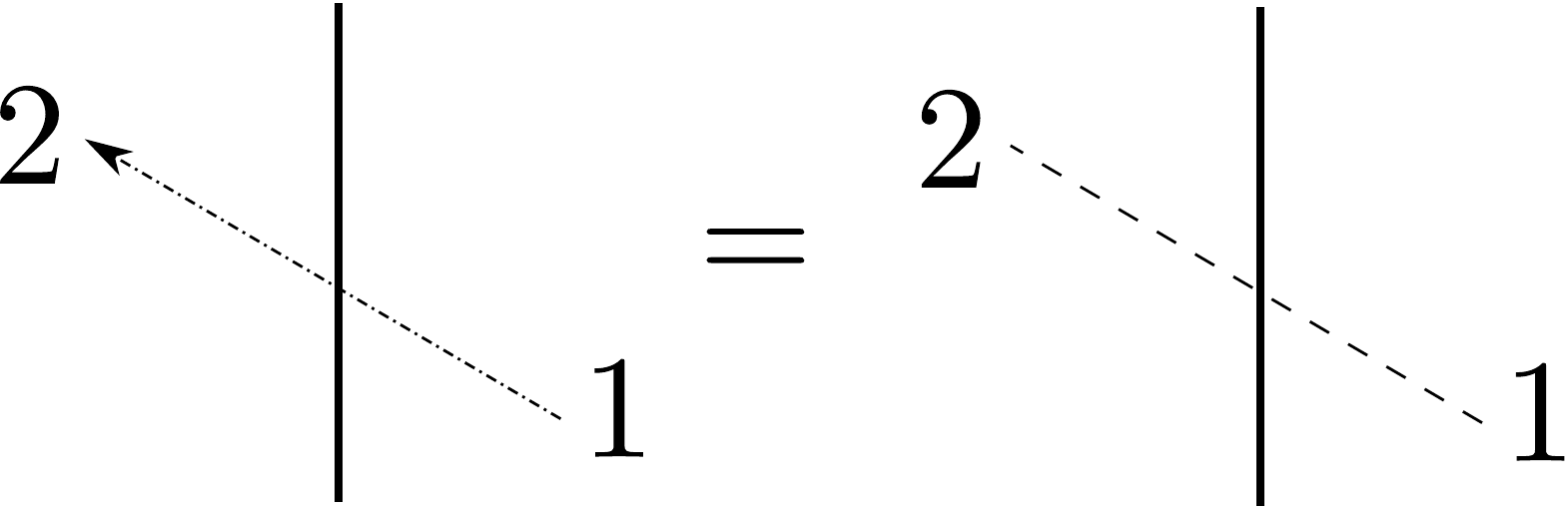}
   \caption{Gluing of propagators: third case}
   \label{fig:gluing3}
\end{figure}

\begin{figure}[htbp] 
   \centering
   \includegraphics[width=2.5in]{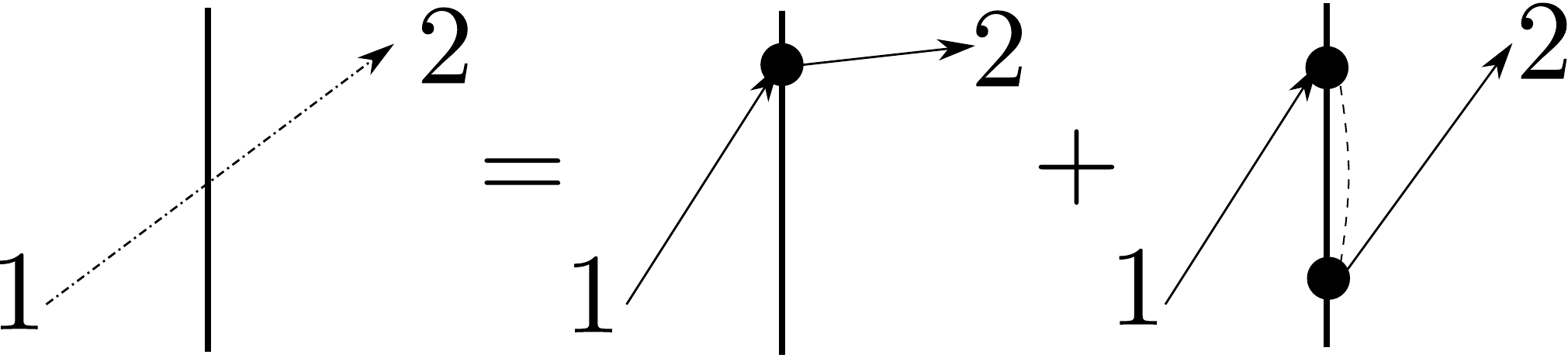}
   \caption{Gluing of propagators: fourth case}
   \label{fig:gluing4}
\end{figure}

The above construction shows heuristically that $\Check\eta$  should be a propagator. This is indeed the case:
\begin{Thm}\label{thm: gluing of propagators}
Form $\Check\eta\in \Omega^{d-1}(C_2(M))$ defined by (\ref{e:glued eta 11}--\ref{e:glued eta 12}) is a propagator on $M$.
\end{Thm}

\begin{proof} The property $\lim_{\epsilon\to 0}\int_{x_1\in S^{d-1}_{x_2,\epsilon}} \Check\eta(x_1,x_2)=1$, where $S^{d-1}_{x_2,\epsilon}$ is the sphere of radius $\epsilon$ (w.r.t.~some fixed metric) centered at $x_2$, follows immediately from the respective property of propagators $\eta_1$ and $\eta_2$. Similarly, one has $\lim_{\epsilon\to 0}\int_{x_2\in S^{d-1}_{x_1,\epsilon}} \Check\eta(x_1,x_2)=(-1)^d$.

Let us check the property (\ref{e:deta}) for $\Check\eta$. For $x_1,x_2\in M_1$, we have from (\ref{e:glued eta 11}) the following:
\begin{multline}\label{e:deta glued check 11}
\dd\Check\eta(x_1,x_2)=\dd\eta_1(x_1,x_2)-\sum_{ij}(-1)^{\deg\chi_{1\times}^j}  V^i_j \int_{y\in\Sigma} \dd\eta_1(x_1,y)\chi_{2i}^\times(y)\chi_{1\times}^j(x_2)\\
=\sum_{I}(-1)^{d-1+d\cdot\deg\chi_{1I}}\chi_{1I}(x_1)\chi_1^{I}(x_2)- \sum_{ij}\sum_{I}(-1)^{d-1+d\cdot\deg\chi_{1I}+\deg\chi_{1\times}^j}  V^i_j \int_\Sigma \chi_{1 I}(x_1)\chi_1^{I}(y)\chi_{2i}^\times(y)\chi_{1\times}^j(x_2)\\
=\sum_{I}(-1)^{d-1+d\cdot\deg\chi_{1I}}\chi_{1I}(x_1)\chi_1^{I}(x_2)-\sum_{ij}\sum_{l} (-1)^{d-1+d\cdot\deg\chi_{1l}^\times}\underbrace{ V^i_j  \Lambda^l_i}_{\delta^l_j} \chi_{1 l}^\times(x_1)\chi_{1\times}^j(x_2)
\\=
\sum_{\alpha}(-1)^{d-1+d\cdot\deg\chi_{1\alpha}^\circ}\chi_{1\alpha}^\circ(x_1)\chi_{1'}^\alpha(x_2).
\end{multline}
Here $I,\alpha$ are the indices for the bases in $H^\bt_{\DD1}(M_1), H^\bt_{\DD1}(M_1)^\circ$ and the dual bases in $H^\bt_{\DD2}(M_1), H^\bt_{\DD2}(M_1)'$ (cf. Section \ref{ss:reducing backgrounds} for notations). Here we used the property (\ref{e:deta}) for $\eta_1$ and the orthogonality of pullbacks to $\Sigma$ of classes from $H^\bt_{\DD2}(M_1)'$ to pullbacks of classes from $H^\bt_{\DD1}(M_2)^\times$.
By a similar computation, for $x_1,x_2\in M_2$ we obtain from (\ref{e:glued eta 22}) that
\begin{equation}\label{e:deta glued check 22}
\dd\Check\eta(x_1,x_2)=
\sum_\beta (-1)^{d-1+d\cdot\deg\chi_{2\beta}'}\chi_{2\beta}'(x_1)\chi_{2\circ}^\beta(x_2)
\end{equation}
where $\beta$ is an index for the basis in $H^\bt_{\DD1}(M_2)'$ and the dual one in $H^\bt_{\DD2}(M_2)^\circ$. For the case $x_1\in M_2$, $x_2\in M_1$, (\ref{e:glued eta 21}) implies immediately that
\begin{equation}\label{e:deta glued check 21}
\dd\Check\eta(x_1,x_2)=0.
\end{equation}
Lastly, for $x_1\in M_1$, $x_2\in M_2$, we have from (\ref{e:glued eta 12}) the following:
\begin{multline}\label{e:deta glued check 12}
\dd\Check\eta(x_1,x_2)=\int_{y\in\Sigma}-\dd\eta_1(x_1,y)\, \eta_2(y,x_2)+ (-1)^{d}\eta_1(x_1,y)\,\dd\eta_2(y,x_2)+\\ +\sum_{ij}(-1)^{\deg\chi_{2i}^\times}\Tilde V^i_j \int_{y\in\Sigma}\int_{z\in\Sigma}\dd\eta_1(x_1,y)\chi_{2i}^\times(y) \chi_{1\times}^j(z)\eta_2(z,x_2)+ \eta_1(x_1,y)\chi_{2i}^\times(y) \chi_{1\times}^j(z)\dd\eta_2(z,x_2)\\
=\sum_\alpha (-1)^{d+\deg\chi_{1\alpha}^\circ}\chi_{1\alpha}^\circ(x_1) \left(\int_{y\in\Sigma}\chi_{1'}^\alpha(y)\eta_2(y,x_2)\right)-\sum_\beta (-1)^{d\cdot\deg\chi_{2\beta}'} \left(\int_{y\in\Sigma}\eta_1(x_1,y)\chi_{2\beta}'(y)\right) \chi_{2\circ}^\beta(x_2).
\end{multline}
Here we are replacing $\dd\eta_1$, $\dd\eta_2$ everywhere with the respective r.h.s.~of (\ref{e:deta}); cancellation of redshirt cohomology classes works similarly to (\ref{e:deta glued check 11}).

Finally, notice that (\ref{e:deta glued check 11}--\ref{e:deta glued check 12}) assembles into property (\ref{e:deta}) for $\Check\eta$ on the glued manifold $M$, with a particular choice of representatives of cohomology of $M$. Namely, for $H^\bt_{\DD1}(M)\simeq H^\bt_{\DD1}(M_1)^\circ\oplus H^\bt_{\DD1}(M_2)'$, we extend representatives $\chi_{1\alpha}^\circ(x)$ by zero into $M_2$ and we extend representatives $\chi_{2\beta}'(x)$ as $(-1)^d\int_{y\in\Sigma}\eta_1(x,y)\chi_{2\beta}'(y)$ into $M_1$ (note that this extension, though being generally non-smooth, has the property of having well-defined pullback to $\Sigma$). Similarly, for $H^\bt_{\DD2}(M)\simeq H^\bt_{\DD2}(M_1)'\oplus H^\bt_{\DD2}(M_2)^\circ$, we extend representatives $\chi_{2\circ}^\beta(x)$ by zero into $M_1$, while representatives $\chi_{1'}^\alpha(x)$ are extended into $M_2$ as $-(-1)^{(d-1)\cdot\deg\chi_{1'}^\alpha}\int_{y\in\Sigma} \chi_{1'}^\alpha \eta_2(y,x)$. (Cf. the construction of residual fields $\Check\sfa$, $\Check\sfb$ on $M$ in Section\ref{ss:reducing backgrounds} and Remark \ref{rem:a ext, b ext}).

The fact that $\Check\eta$ has well-defined pull-back as one of the points restricts to $\Sigma$ (and thus that $\dd\Check\eta$ does not contain a delta-function on $\Sigma$) follows from computing respective limits of (\ref{e:glued eta 11}--\ref{e:glued eta 12}) as one of the points approaches a point on $\Sigma$. For this one uses that, for $\alpha\in \Omega^\bt(M_1)$, one has $\lim_{x\to x_0}-(-1)^{(d-1)\cdot\deg\alpha}\int_{y\in\Sigma}\alpha(y)\eta_2(y,x)= \alpha(x_0)$ where $x_0\in\Sigma$ and likewise for $\beta\in\Omega^\bt(M_2)$ one has $\lim_{x\to x_0}(-1)^d\int_{y\in\Sigma}\eta_1(x,y)\beta(y)=\beta(x_0)$. (These properties follow from the normalization of the integral over a small sphere for $\eta_1$, $\eta_2$, cf. Section \ref{sec: properties of propagators}).

This finishes the proof.

\end{proof}

\section{Examples of gluing of propagators}
\label{a: exa of gluing of prop}

\begin{Exa}\label{exa: gluing intervals 21 cup 21} Let $M_1=[0,1]$, $M_2=[1,2]$ be two intervals. We glue the right endpoint of $M_1$ to the left endpoint of $M_2$ to form $M=M_1\cup_{\{1\}} M_2=[0,2]$. We denote the coordinate on $M$ by $t\in [0,2]$. We set $\de_1 M_1=\de_2 M_2=\{1\}=\Sigma$, $\de_2 M_1=\{0\}$, $\de_1 M_2=\{2\}$. All the relevant cohomology (and hence spaces of {\backgrounds}) vanish for $M_1,M_2,M$. Using the propagator (\ref{e:eta interval 12}) for $M_1$, $M_2$, we obtain by the gluing construction of Appendix \ref{a:comp} the following propagator on $M$:
$$ \Check\eta(t_1,t_2)=\left\{\begin{array}{lll}
\eta_1(t_1,t_2) & \mbox{if}& t_1,t_2\in [0,1], \\
\eta_2(t_1,t_2) & \mbox{if}& t_1,t_2\in [1,2], \\
0 & \mbox{if}& t_1\in [1,2],\; t_2\in [0,1] \\
-\eta_1(t_1,1)\cdot\eta_2(1,t_2) & \mbox{if}& t_1\in[0,1],\;t_2\in [1,2],
\end{array}\right.=
\left\{\begin{array}{lll}
-1 & \mbox{if} & t_1<t_2,\\
0 & \mbox{if} & t_1>t_2.
\end{array}\right. $$
This is precisely the propagator (\ref{e:eta interval 12}) for the glued interval.
\end{Exa}

\begin{Exa}\label{exa: gluing intervals 11 cup 22}
In the setting of Example \ref{exa: gluing intervals 21 cup 21}, let us change the labelling of boundary to $\de_1 M_1=\{0\}\sqcup\{1\}$, $\de_2 M_1=\varnothing=\de_1 M_2$,
$\de_2 M_2=\{1\}\sqcup \{2\}$. The glued interval $M=[0,2]$ has $\de_1 M=\{0\}$, $\de_2 M=\{2\}$. Here one has {\backgrounds} both on $M_1$ and $M_2$ (cf. Example \ref{exa interval 11}), but no {\backgrounds} on $M$. Thus the whole space $\calV_{M_1}\oplus \calV_{M_2}$ consists of redshirt {\backgrounds}.
The relevant cohomology is:
\begin{multline*}
H^1_{\DD1}(M_1)=\mathrm{Span}(\underbrace{[\dd t]}_{[\chi_{10}]}),\quad
H^0_{\DD1}(M_2)=\mathrm{Span}(\underbrace{[1]}_{[\chi_{20}]}),\\
H^0_{\DD2}(M_1)=\mathrm{Span}(\underbrace{[1]}_{[\chi_{1}^0]}),\quad
H^1_{\DD2}(M_2)=\mathrm{Span}(\underbrace{[\dd t]}_{[\chi_{2}^0]}).
\end{multline*}
We also have $L_1=L_1^\times=L_2=L_2^\times=H^0(\{1\})=\mathbb{R}\cdot 1$.
For the propagator on $M_1$ we take (\ref{e:eta interval 11}) and on $M_2$ we take (\ref{e:eta interval 22}) where we make the shift $t_{1,2}\mapsto t_{1,2}-1$ (since now we parametrize $M_2$ by the coordinate $t\in[1,2]$), i.e. $\eta_1(t_1,t_2)=\Theta(t_1-t_2)-t_1$ for $t_1,t_2\in [0,1]$, $\eta_2(t_1,t_2)=-\Theta(t_2-t_1)+t_2-1$ for $t_{1},t_2\in [1,2]$. Formulae of Appendix \ref{We now do the final step in computing the propagator on $M$ for the reduced  space of pseudovacua} yield:
\begin{multline*}
\Check\eta(t_1,t_2)=\left\{\begin{array}{lll}
\eta_1(t_1,t_2)-\eta_1(t_1,1)\chi_{20}(1)\chi_1^0(t_2) & \mbox{if}& t_1,t_2\in [0,1], \\
\eta_2(t_1,t_2)-\chi_{20}(t_1)\chi_1^0(1)\eta_2(1,t_2) & \mbox{if}& t_1,t_2\in [1,2], \\
\chi_{20}(t_1)\chi_1^0(t_2) & \mbox{if}& t_1\in [1,2],\; t_2\in [0,1], \\
-\eta_1(t_1,1)\cdot\eta_2(1,t_2)+\eta_1(t_1,1)\chi_{20}(1)\chi_1^0(1)\eta_2(1,t_2) & \mbox{if}& t_1\in[0,1],\;t_2\in [1,2]
\end{array}\right.\\
=
\left\{\begin{array}{lll}
\Theta(t_1-t_2)& \mbox{if}& t_1,t_2\in [0,1], \\
\Theta(t_1-t_2)& \mbox{if}& t_1,t_2\in [1,2], \\
1 & \mbox{if}& t_1\in [1,2],\; t_2\in [0,1], \\
0 & \mbox{if}& t_1\in[0,1],\;t_2\in [1,2]
\end{array}\right.\quad  =\qquad \Theta(t_1-t_2)\;\; \mbox{for}\;\; t_{1,2}\in [0,2].
\end{multline*}
Thus we obtain exactly the propagator (\ref{e:eta interval 12}), where we have to make a change of coordinates $t\mapsto 2-2t$ to switch from $[0,1]$ with $2-1$ boundary condition to $[0,2]$ with $1-2$ boundary condition.
\end{Exa}

\begin{Exa} Consider gluing two intervals as in Example \ref{exa: gluing intervals 11 cup 22} but in addition let us identify the points $t=0$ and $t=2$. Thus we are gluing a circle $M=S^1$ out of two intervals $M_1=[0,1]$, $M_2=[1,2]$ along two points $\Sigma=\{0\}\sqcup\{1\}$. Then we have no redshirt {\backgrounds}, $\calV_M=\calV_{M_1}\oplus \calV_{M_2}$, with $\calV_{M_1}$, $\calV_{M_2}$ as in Example \ref{exa: gluing intervals 11 cup 22}. For the glued propagator, we obtain:
\begin{multline*}
\Check\eta(t_1,t_2)=\left\{\begin{array}{lll}
\eta_1(t_1,t_2) & \mbox{if}& t_1,t_2\in [0,1], \\
\eta_2(t_1,t_2) & \mbox{if}& t_1,t_2\in [1,2], \\
0 & \mbox{if}& t_1\in [1,2],\; t_2\in [0,1], \\
-\eta_1(t_1,1)\eta_2(1,t_2)+\eta_1(t_1,0)\eta_2(0,t_2) & \mbox{if}& t_1\in[0,1],\;t_2\in [1,2]
\end{array}\right.\\
=\left\{\begin{array}{lll}
\Theta(t_1-t_2)-t_1 & \mbox{if}& t_1,t_2\in [0,1], \\
\Theta(t_1-t_2)+t_2-2 & \mbox{if}& t_1,t_2\in [1,2], \\
0 & \mbox{if}& t_1\in [1,2],\; t_2\in [0,1], \\
-t_1+t_2-1 & \mbox{if}& t_1\in[0,1],\;t_2\in [1,2].
\end{array}\right.
\end{multline*}
This does not coincide with the propagator of Example \ref{exa: eta circle} but is also a valid propagator for the circle, corresponding to different representatives of cohomology of
$M=S^1$ -- the representatives obtained from the gluing procedure for {\backgrounds} of Section \ref{ss:reducing backgrounds}:
\begin{multline*}\Check\chi_0=\left\{\begin{array}{lll} -\eta_1(t,1)\chi_{20}(1)+\eta_1(t,0)\chi_{20}(2)&\mbox{on}& M_1 \\ \chi_{20} &\mbox{on}& M_2 \end{array}\right. \quad=1,\\ \Check\chi_1= \left\{\begin{array}{lll} \chi_{10}&\mbox{on}& M_1 \\ 0 &\mbox{on}& M_2 \end{array}\right. \quad
=\Theta(1-t)\cdot dt,\quad
\Check\chi^0=
\left\{\begin{array}{lll} 0 &\mbox{on}& M_1 \\ \chi_2^0 &\mbox{on}& M_2 \end{array}\right.  \quad=
\Theta(t-1)\cdot dt,\\
\Check\chi^1=\left\{\begin{array}{lll} \chi_{1}^0 &\mbox{on}& M_1 \\
\chi_{1}^0(0)\eta_2(2,t)-\chi_{1}^0(1)\eta_2(1,t)&\mbox{on}& M_2 \end{array}\right. \quad=1.
\end{multline*}
With these representatives, we have equation (\ref{e:deta}) for $\Check\eta$.
Note that these representatives are not continuous (but still closed). Also, the propagator $\Check\eta$ is continuous (for $t_1\neq t_2$) but not differentiable when one of the points hit $\Sigma=\{0\}\sqcup\{1\}$.
\end{Exa}

\subsection{Attaching a cylinder with axial gauge-fixing}

\begin{Exa}[Attaching a cylinder with opposite polarizations on top and bottom]
\label{exa: attaching 21 cylinder}
Let $M_2$ be some $d$-manifold and $\Sigma\subset \de_2 M_2$ a boundary component (or a union of several boundary components).
Set $M_1=\Sigma\times [0,1]$ with $\de_1 M_1=\Sigma\times \{1\}$ (the gluing interface) and $\de_2 M_1=\Sigma\times \{0\}$. Assume that on $M_2$ we have fixed a basis in cohomology $[\chi_{2i}]\in H^\bullet_{\DD1}(M_2)$ together with its dual $[\chi_2^i]\in H^{d-\bullet}_{\DD2}(M_2)$ and fixed a propagator $\eta_2$. Attaching the cylinder (which has $\calV_{M_1}=0$) does not change cohomology, so $\calV_M=\calV_{M_2}$; there are no redshirt {\backgrounds}. Denote by $\phi: M\rightarrow M_2$ the deformation retraction of $M$ onto $M_2$ which is constant on $M_2$ and collapses the cylinder $M_1=\Sigma\times [0,1]$ onto the top $\Sigma\times \{1\}$.  Choosing the gauge-fixing of Example \ref{exa: axial gauge}, we have the glued representatives of cohomology $\Check\chi_i=\phi^* \chi_{2i}$, $\Check\chi^i=\phi^*\chi_2^i$ (the latter are identically zero on $M_1$), for both choices of the propagator on $M_1$.
If we take $\eta_1=\eta^\mathrm{axial}$ (\ref{e:eta axial}), for the glued propagator we obtain
\begin{multline*}
\Check\eta(x_1,x_2) = \eta_2(x_1,x_2),  \quad\mbox{for}\;\; x_1,x_2\in M_2,\\
\Check\eta((y_1,t_1),(y_2,t_2)) = -\Theta(t_2-t_1)\cdot\delta^{(d-1)}(y_1,y_2),  \quad\mbox{for}\;\; (y_i,t_i) \in \Sigma\times [0,1], \\
\Check\eta(x_1,(y_2,t_2)) = 0 ,\quad\mbox{for} \;\; (y_2,t_2) \in \Sigma\times [0,1],\;\; x_1\in M_2,\\
\Check\eta((y_1,t_1),x_2)= \underbrace{\eta_2(y_1,x_2)}_{(\phi\times\mathrm{id})^*\eta_2},\quad\mbox{for} \;\; (y_1,t_1) \in \Sigma\times [0,1],\;\; x_2\in M_2.
\end{multline*}
Taking instead $\eta_1=\eta^\mathrm{hor}$, we obtain the glued propagator
\begin{multline*}
\Check\eta(x_1,x_2) = \eta_2(x_1,x_2,),  \quad\mbox{for}\;\; x_1,x_2\in M_2,\\
\Check\eta((y_1,t_1),(y_2,t_2)) =
-\delta(t_1-t_2)\cdot (\dd t_1-\dd t_2)\cdot\eta_\Sigma(y_1,y_2)-\\
-\sum_i (-1)^{(d-1)\cdot(\deg\chi_{(\Sigma)i}+1)}  \Theta(t_2-t_1)\cdot \chi_{(\Sigma)i}(y_1)\cdot \chi_{(\Sigma)}^i(y_2)
,\quad\mbox{for}\;\; (y_i,t_i) \in \Sigma\times [0,1], \\
\Check\eta(x_1,(y_2,t_2)) = 0, \quad\mbox{for} \;\; (y_2,t_2) \in \Sigma\times [0,1],\;\; x_1\in M_2,\\
\Check\eta((y_1,t_1),x_2)= dt_1\,\delta(1-t_1)\int_\Sigma \eta_\Sigma(y_1,y')\eta_2(y',x_2)
+\underbrace{\sum_i \chi_{(\Sigma)i}(y_1)\cdot \int_\Sigma\chi_{(\Sigma)}^i(y')\;\eta_2(y',x_2)}_{ ((\phi^*P_{H^\bullet(\Sigma)}\iota_\Sigma^*)\otimes\mathrm{id})\;\eta_2}  ,\\ \quad\mbox{for} \;\; (y_1,t_1) \in \Sigma\times [0,1],\;\; x_2\in M_2.
\end{multline*}
Here $[\chi_{(\Sigma)i}]$ is some basis in $H^\bullet(\Sigma)$ and $[\chi_{(\Sigma)}^i]$ the dual one;  $\iota_\Sigma$ is the embedding of $\Sigma$ into $M_2$, $P_{H^\bullet(\Sigma)}$ is the projection to the representatives of cohomology $H^\bullet(\Sigma)$.
\end{Exa}

\begin{Exa}[Changing the polarization by attaching a cylinder]
\label{exa: attaching 11 cylinder}
Let us change the setup of Example \ref{exa: attaching 21 cylinder} by setting
$\de_1 M_1=\Sigma\times\{0\}\sqcup \Sigma\times \{1\}$. I.e. we attach a cylinder with 1-1 boundary condition, which can be viewed as a way to change the boundary condition on $M_2$, since $\Sigma\subset \de_2 M_2$ but $\Sigma\times\{0\}\subset \de_1 M$. In the notations of Section \ref{ss:reducing backgrounds}, we have $L_1=H^\bullet(\Sigma)$; $L_2=L_2^\times\subset H^\bullet(\Sigma)$ is generally nontrivial. The glued cohomology is:
\begin{multline*}
H^\bullet_{\DD1}(M)=\Tilde H^\bullet_{\DD1}(M_1,M_2)=\dd t\cdot\frac{H^{\bullet-1}(\Sigma)}{L_2}\;\oplus \ker \tau_2,\\
H^\bullet_{\DD2}(M)=\Tilde H^\bullet_{\DD2}(M_1,M_2)= 1_t\cdot L_2^\perp\oplus \underbrace{\mathrm{Ann}\; \sigma_2(L_2)}_{\subset H^\bullet_{\DD2}(M_2)}
\end{multline*}
We have a generally non-empty space of redshirt {\backgrounds}:
$$\calV^\times_{M_1M_2}= \left(\dd t\cdot L_2\oplus \sigma_2(L_2)\right)[k]\oplus
\left( 1_t\cdot \frac{H^\bullet(\Sigma)}{L_2^\perp}\;\oplus \frac{H^\bullet_{\DD2}(M_2)}{\mathrm{Ann}\;\sigma_2(L_2)} \right)[d-k-1]
$$
Choosing $\eta_1=\eta^\mathrm{axial}$ (\ref{e:eta axial 11}) as the propagator on $M_1$, we obtain the following glued propagator:
\begin{multline*}
\Check\eta(x_1,x_2) = \eta_2(x_1,x_2)-\sum_{ij}(-1)^{\deg\chi_{2i}^\times}V^i_j\int_\Sigma \chi_{2i}^\times(x_1) \chi_{1\times}^j(y') \eta_2(y',x_2),   \quad\mbox{for}\;\; x_1,x_2\in M_2,\\
\Check\eta((y_1,t_1),(y_2,t_2)) = (\Theta(t_1-t_2)-t_1)\cdot\delta^{(d-1)}(y_1,y_2)- \dd t_1\cdot \eta_\Sigma(y_1,y_2)-\\ -
\sum_{i,j} (-1)^{d+\deg\chi_{2i}^\times}V^i_j \left(t_1\chi_{2i}^\times(y_1)\chi_{1\times}^j(y_2)+ \dd t_1\int_\Sigma\eta_\Sigma(y_1,y') \chi_{2i}^\times(y')\chi_{1\times}^j(y_2)\right)
\quad\mbox{for}\;\; (y_i,t_i) \in \Sigma\times [0,1], \\
\Check\eta(x_1,(y_2,t_2)) = -\sum_{i,j} (-1)^{d+\deg\chi_{2i}^\times}V^i_j\; \chi_{2i}^\times(x_1)\chi_{1\times}^j(y_2), \quad\mbox{for} \;\; (y_2,t_2) \in \Sigma\times [0,1],\;\; x_1\in M_2,\\
\Check\eta((y_1,t_1),x_2)= t_1\eta_2(y_1,x_2)+ \dd t_1 \int_\Sigma \eta_\Sigma(y_1,y') \eta_2(y',x_2)-\\  -
\sum_{i,j} (-1)^{\deg\chi_{2i}^\times}V^i_j \left( t_1\int_\Sigma\chi_{2i}^\times(y_1)\chi_{1\times}^j(y')\eta_2(y',x_2)+ \dd t_1\int_{\Sigma\times \Sigma}\eta_\Sigma(y_1,y') \chi_{2i}^\times(y')\chi_{1\times}^j(y'')\eta_2(y'',x_2)\right),\\
\quad\mbox{for} \;\; (y_1,t_1) \in \Sigma\times [0,1],\;\; x_2\in M_2.
\end{multline*}
Here we have chosen some basis $[\chi_{1\times}^i]$ in $H^\bullet(\Sigma)/L_2^\perp$ and a basis $[\chi_{2i}^\times]$ in $\sigma_2(L_2)\subset H^\bullet_{\DD1}(M_2)$; $V^i_j$ are the matrix elements of the inverse matrix of $\Lambda$ defined by (\ref{e:Lambda}).
The corresponding representatives of $H^\bullet_{\DD1}(M)$ are extensions of $\frac{H^{\bullet-1}(\Sigma)}{L_2}\cdot dt$ by zero to $M_2$ and extensions of $\ker\tau_2$ by zero to $M_1$. For $H^\bullet_{\DD2}(M)$, we extend $L_2^\perp\cdot 1_t$ by the corresponding representatives of $H(M_2,\de_2 M_2\backslash\Sigma)$ given by $\chi^\mathrm{ext}(x)=-(-1)^{(d-1)\cdot\deg\chi}\int_\Sigma \chi(y')\eta_2(y',x)$. Elements of $\mathrm{Ann}\; \sigma_2(L_2)$ are extended by zero on $M_1$.

Next, if instead we choose $\eta_1$ as $\eta^\mathrm{hor}$ (\ref{e:eta hor 11}), we obtain the following:
\begin{multline*}
\Check\eta(x_1,x_2) = \eta_2(x_1,x_2)-\sum_{ij}(-1)^{\deg\chi_{2i}^\times}V^i_j\int_\Sigma \chi_{2i}^\times(x_1) \chi_{1\times}^j(y') \eta_2(y',x_2),   \quad\mbox{for}\;\; x_1,x_2\in M_2,\\
\Check\eta((y_1,t_1),(y_2,t_2)) =
-\delta(t_1-t_2)\cdot (\dd t_1-\dd t_2)\cdot\eta_\Sigma(y_1,y_2)+\\
+\sum_i (-1)^{(d-1)\cdot(\deg\chi_{(\Sigma)i}+1)}\; (\Theta(t_1-t_2)-t_1)\cdot \chi_{(\Sigma)i}(y_1)\cdot \chi_{(\Sigma)}^i(y_2)-
\\
- t_1\sum_{i,j}(-1)^{d+\deg\chi_{2i}^\times} V^i_j
\chi_{2i}^\times(y_1)\chi_{1\times}^j(y_2)
,
\quad\mbox{for}\;\; (y_i,t_i) \in \Sigma\times [0,1], \\
\Check\eta(x_1,(y_2,t_2)) = -\sum_{i,j}(-1)^{d+\deg\chi_{2i}^\times} V^i_j\; \chi_{2i}^\times(x_1)\chi_{1\times}^j(y_2), \quad\mbox{for} \;\; (y_2,t_2) \in \Sigma\times [0,1],\;\; x_1\in M_2,\\
\Check\eta((y_1,t_1),x_2) =\dd t_1\,\delta(1-t_1)\int_\Sigma \eta_\Sigma(y_1,y')\eta_2(y',x_2)
+ t_1 \sum_l \chi_{(\Sigma)l}(y_1)\cdot \int_\Sigma\chi_{(\Sigma)}^l(y') \eta_2(y',x_2)- \\ -
t_1 \sum_{i,j} (-1)^{d\cdot\deg\chi_{2i}^\times} V^i_j \chi_{2i}^\times(y_1)\int_\Sigma\chi_{1\times}^j(y')\eta_2(y',x_2) ,
\quad\mbox{for} \;\; (y_1,t_1) \in \Sigma\times [0,1],\;\; x_2\in M_2.
\end{multline*}

\end{Exa}

\subsection{Gluing Kontsevich's propagators on two half-planes}
This example falls slightly outside of the scope of our construction as the manifolds in question are non-compact, but we find it otherwise instructive.

Let $\Pi_+=\{z\in\mathbb{C}\; |\; \mr{Im}(z)\geq 0\}$ and $\Pi_-=\{z\in\mathbb{C}\; |\; \mr{Im}(z)\leq 0\}$ be the upper and lower halves of the complex plane. On $\Pi_+$ one has the Kontsevich's propagator
\begin{equation}\label{e:eta Kontsevich}
\eta_{\Pi_+}(z,w)=\frac{1}{2\pi}\dd\arg\frac{z-w}{\bar z-w}
\end{equation}
and on $\Pi_-$ one has
\begin{equation}\label{e:eta Kontsevich lower}
\eta_{\Pi_-}(z,w)=-\eta_{\Pi_+}(\bar w,\bar z)=\frac{1}{2\pi}\dd\arg\frac{z-w}{z-\bar w}.
\end{equation}
Here we regard the real line $\bR\subset \bC$ as the $\de_1$-boundary of $\Pi_+$ and $\de_2$-boundary of $\Pi_-$, which corresponds to boundary conditions $\eta_{\Pi_+}(z,w)|_{z=0}=\eta_{\Pi_-}(z,w)|_{w=0}=0$.

\begin{Rem} One can recover the propagators (\ref{e:eta Kontsevich},\ref{e:eta Kontsevich lower}) from the Euclidean ($SO(2)\ltimes\bR^2$-invariant) propagator on the plane, $\eta_{\bR^2}(z,w)=\frac{1}{2\pi}\dd\arg(z-w)$, via the method of image charges of Appendix \ref{a:prop}. Indeed, we have
$$\eta_{\Pi_+}(z,w)=\eta_{\bR^2}(z,w)-\eta_{\bR^2}(\bar z,w)\qquad\mbox{for}\quad \mr{Im}(z)>0,\;\mr{Im}(w)>0$$
and
$$\eta_{\Pi_-}(z,w)=\eta_{\bR^2}(z,w)-\eta_{\bR^2}(z,\bar w)\qquad\mbox{for}\quad \mr{Im}(z)<0,\;\mr{Im}(w)<0.$$
\end{Rem}

Let us calculate the glued propagator $\Check\eta$ on the plane $\bR^2\simeq \bC$. In this example we may regard $\Pi_\pm$ as disks relative to a point on the boundary ($\{\infty\}\in \Bar\Pi_\pm$) and $\bC$ as $\bC P^1$ relative to a point; the corresponding relative cohomology vanishes, so there are no residual fields (neither before nor after gluing).

The non-trivial case is $\mr{Im}(z)>0$ and $\mr{Im}(w)<0$, then we calculate
\begin{multline}\label{e:gluing Kontsevich eta}
\Check\eta(z,w)=\int_{x\in\bR}\eta_{\Pi_+}(z,x)\wedge \eta_{\Pi_-}(x,w)\\
=-\frac{1}{(2\pi)^2}\int_{x\in \bR} \left(\dd_z \log\frac{z-x}{\bar z-x}\wedge \dd_x \log\frac{x-w}{x-\bar w}+
\dd_x \log\frac{z-x}{\bar z-x}\wedge \dd_w \log\frac{x-w}{x-\bar w}\right)\\
=-\frac{1}{(2\pi)^2}\int_{x\in\bR}\left(
\left(\frac{\dd z}{z-x}-\frac{\dd\bar z}{\bar z-x}\right)\wedge
\frac{(w-\bar w)\cdot \dd x}{(x-w)(x-\bar w)}+
\frac{(z-\bar z)\cdot \dd x}{(x-z)(x-\bar z)}\wedge
\left(\frac{\dd w}{w-x}-\frac{\dd\bar w}{\bar w-x}\right)
\right)\\
=-\frac{2\pi\ii}{(2\pi)^2}\left( \frac{\dd z}{z-w}-\frac{\dd\bar z}{\bar z-\bar w}+
\frac{\dd w}{w-z}-\frac{\dd\bar w}{\bar w-\bar z} \right)=
-\frac{\ii}{2\pi}\dd\log\frac{z-w}{\bar z-\bar w}=\frac1\pi \dd\arg (z-w).
\end{multline}
Here the integral over $x$ is computed straightforwardly by residues. Note that on the r.h.s. of (\ref{e:gluing Kontsevich eta}) we obtained {\it twice} the Euclidean 
propagator on the plane $\eta_{\bR^2}$.

Thus, the full result for the glued propagator on the plane is:
$$\Check\eta(z,w)=\left\{\begin{array}{lll}
\frac{1}{2\pi}\dd\arg\frac{z-w}{\bar z-w} & \mbox{if} & \mr{Im}(z)>0,\; \mr{Im}(w)>0, \\
\frac{1}{2\pi}\dd\arg\frac{z-w}{z-\bar w} & \mbox{if} & \mr{Im}(z)<0,\; \mr{Im}(w)<0, \\
0 & \mbox{if} & \mr{Im}(z)<0,\; \mr{Im}(w)>0, \\
\frac1\pi \dd\arg (z-w)  & \mbox{if} & \mr{Im}(z)>0,\; \mr{Im}(w)<0
\end{array}\right.$$

\begin{Rem}
Note that reversing the assignment for boundary conditions on $\Pi_\pm$ (i.e. regarding $\bR$ as $\de_2\Pi_+$ and as $\de_1\Pi_-$) yields a new glued propagator on the plane, $\Check\eta_\mr{reversed}(z,w)=\Check\eta(w,z)$. Therefore, applying to $\Check \eta$ the doubling trick of Section \ref{ss:dbling}, we obtain the symmetrized propagator on the plane
$$\eta_\mr{sym}(z,w)=\frac12 (\Check \eta(z,w)+\Check \eta(w,z))=\frac{1}{2\pi}\dd\arg(z-w),$$
which is again the Euclidean propagator $\eta_{\bR^2}$.
\end{Rem}

\section{On semi-classical BV theories via effective actions}
\label{appendix: semi-classics and eff actions}

Here we outline the setup for perturbative quantization in formal neighborhoods of solutions of equations of motion, done in a family over the body of the Euler-Lagrange moduli space. Over every point of the moduli space we allow a hierarchy (a poset) of ``realizations'', and one can pass from ``larger'' to ``smaller'' realizations via BV pushforwards. Thus, this setup has a version of Wilson's renormalization flow (in a family over the Euler-Lagrange moduli space) built into it. We also consider in detail a 1-dimensional example with realizations associated to triangulations of a circle.

\subsection{General setup}
\label{sec: appF general setup}

We assume that a classical BV theory $M\mapsto (\calF,Q,\omega,\calS)$ is fixed. Let $\mathcal{M}_M=\mathcal{EL}_M/Q$ be the graded odd-symplectic Euler-Lagrange moduli space (see \cite{CMR} for details) and $\mathcal{M}_M^{\mathrm{gh}=0}=EL_M/Q$ its body, i.e. the set of gauge-equivalence classes of (degree zero) solutions of Euler-Lagrange equations (in our notation, $\mathcal{EL}_M$ is the graded zero locus of $Q$ and $EL_M$ is its body; $Q$ in the denominator stands for passing to the quotient over the distribution induced by $Q$ on the zero locus).

For $M$ a space-time manifold, fix 
$x_0\in EL_M$
a solution of Euler-Lagrange equations.
Also, fix a ``formal exponential map'' $\phi(x_0,\bullet)$ from an open subset $U\subset T_{x_0}\calF$ containing the origin to $\calF$, satisfying $\phi(x_0,0)=x_0$ and $\dd\phi(x_0,\bullet)|_{(x_0,0)}=\mathrm{id}\colon T_{x_0}\calF\rightarrow T_{x_0}\calF$.\footnote{In the case when $\calF$ has linear structure, one natural choice is to set $\phi(x_0,\theta)=x_0+\theta$.} For simplicity, we assume that $\phi$ has the ``Darboux property'', i.e. that the 2-form $\phi(x_0,\bullet)^*\omega\in \Omega^2(U)_{-1}$ is constant on $U$.

The $\infty$-jet of $Q$ at $x_0$ defines, via the map $\phi$, an $L_\infty$ algebra $(T_{x_0}[-1]\mathcal{F},\{l_n\}_{n\geq 1})$ where $l_n$ are the $n$-linear operations on $T_{x_0}[-1]\mathcal{F}$. Moreover, this algebra is cyclic, with invariant (i.e. cyclic) inner product of degree $-3$ given by $\omega_{x_0}$.\footnote{Degree $-3$ comes about for the following reason. For $V$ a $\mathbb{Z}$-graded vector space, a degree $-1$ symplectic form on $V[1]$ corresponds to a degree $-3=-1+2(-1)$ inner product on $V$. Factor $2$ appears because the inner product is a binary operation; first $-1$ is the degree of the symplectic form and second $-1$ comes from the shift from $V[1]$ to $V$.} The data of this algebra are related to the ``linearization'' of the action $S$ at $x_0$ by
$$\calS(\phi(x_0,\theta))=\calS(x_0)+\sum_{n\geq 1}\frac{1}{(n+1)!} \omega_{x_0}(\theta,l_n(\underbrace{\theta,\cdots,\theta}_n))$$
where $\theta\in U\subset T_{x_0}\mathcal{F}$ is a tangent vector.

We have a poset (more precisely, a downward directed category) $R$ of deformation retracts of the complex $(T_{x_0}[-1]\mathcal{F}, l_1)$ compatible with the inner product\footnote{\label{footnote: retracts}
For a cochain complex $(V^\bullet,\dd)$ with inner product $\langle,\rangle: V^{j}\otimes V^{k-j}\rightarrow \mathbb{R}$ (for the case in hand, $k=3$) with cyclic property $\langle \dd a,b\rangle=-(-1)^{|a|}\langle a, \dd b \rangle$, we say that $(V'^\bullet,\dd',\langle,\rangle')$ is a {\sf deformation retract compatible with the inner product}, if a chain inclusion $i: V'^\bullet\hookrightarrow V^\bullet$ and a chain projection $p: V^\bullet\twoheadrightarrow V'^\bullet$ are given and have the following properties. Maps $i$ and $p$ should induce identity on cohomology and should satisfy $p\circ i=\mathrm{id}_{V'}$ and $\langle a, i(b')\rangle=\langle p(a),b'\rangle'$. It follows that the splitting $V=i(V')\oplus \ker p$ is orthogonal with respect to $\langle,\rangle$ and induces the pairing $\langle,\rangle'$ on $V'$. If additionally $K:V^\bullet\mapsto V^{\bullet-1}$ is a chain contraction of $V$ onto $V'$ (i.e. $\dd K+K\dd=\mathrm{id}-i\,p$, $K^2=Ki=pK=0$ and $K$ is skew self-adjoint), then we say that the triple $(i,p,K)$ is a {\sf retraction compatible with the inner product} from $V$ onto $V'$ and denote $V\stackrel{(i,p,K)}{\rightsquigarrow} V'$. We view retractions as morphisms in the category of retracts. The composition rule is $(i_1,p_1,K_1)\circ (i_2,p_2,K_2)=(i_2 i_1, p_1 p_2, K_2+i_2 K_1 p_2)$.  The space of retractions between a cochain complex and its fixed retract, inducing a fixed isomorphism on cohomology via $i_*,p_*$, is {\sf contractible}.  We will be omitting ``compatible with inner products'' for retracts and retractions, as it is always assumed throughout this Appendix, unless stated otherwise.}
(we call them ``realizations''), which inherit, via homotopy transfer, an ``induced'' cyclic $L_\infty$ structure
$$(\calV_{x_0,r}[-1],\{l_n^{x_0,r}\}_{n\geq 1},\omega_{x_0,r}).$$
Here $r\in R$ is the label of the particular retract. Note that the operations $l_n^{x_0,r}$ depend on a particular choice of retraction $T_{x_0}[-1]\calF\wavy[g]\calV_{x_0,r}[-1]$; different choices of $g$ induce isomorphic cyclic $L_\infty$ structures on $\calV_{x_0,r}[-1]$.
(We refer the reader to \cite{CMcs} for details on homotopy transfer for cyclic $L_\infty$ algebras.)

\begin{Rem}\label{rem: minimal backgrounds}
Of particular interest is the final (``minimal'') object $r_\mathrm{min}$ of $R$, which corresponds to cohomology of $l_1$ (with the induced cyclic $L_\infty$ structure). The case when all induced operations on $\calV_{x_0,r_\mathrm{min}}[-1]=H^\bullet_{l_1}$ vanish corresponds to the gauge equivalence class $[x_0]$ being a {\sf smooth point} of the EL moduli space $\mathcal{M}_M$, cf. Appendix C of \cite{CMR}. In this case, the tangent space to $\mathcal{M}_M$ at $[x_0]$ is $H^\bullet_{l_1}[1]$; in particular, the tangent space at $[x_0]$ to the body $\mathcal{M}_M^{\mathrm{gh}=0}$ is $H^1_{l_1}$.
We have from homological perturbation theory the $L_\infty$ morphism (extending the chosen embedding $i:H_{l_1}\to T_{x_0}[-1]\calF$ by higher polylinear operations) from $\calV_{x_0,r_\mr{min}}[-1]$ to $T_{x_0}[-1]\calF$; the latter defines a non-linear map of formal pointed dg manifolds $\Tilde i: \calV_{x_0,r_\mr{min}}\to T_{x_0}\calF$. Assuming that $[x_0]$ is a smooth point of $\calM_M$, we have, by reduction by the $Q$-distribution of the map $\calV_{x_0,r_\mr{min}}\xrightarrow{\Tilde i}T_{x_0}\calF\xrightarrow{\phi(x_0,\bt)}\mathcal{EL}_M\subset\calF$, a formal exponential map $\Psi(x_0,\bt)\colon \calV_{x_0,r_\mr{min}}\to \calM_M$.
\end{Rem}

The graded vector space $\calV_{x_0,r}$ is our space of (formal) {\backgrounds}. A perturbative BV theory assigns to the pair $(x_0,r)$ and a retraction $T_{x_0}[-1]\calF\wavy[g] \calV_{x_0,r}[-1] $
(the {\sf gauge-fixing data})  ``the state''
$$\psi_{x_0,r}^g=e^{\frac{i}{\hbar}\left(\calS(x_0)+\sum_{n\geq 1}\frac{1}{(n+1)!} \omega_{x_0,r}(y,l_n^{x_0,r,g}(y,\ldots,y))\right)}\cdot \psi_{x_0,r,g}^{\geq 1\,\mathrm{loops}}$$
where $\psi_{x_0,r,g}^{\geq 1\,\mathrm{loops}}\in \mathrm{Dens}^{\frac12}_\mr{formal}(\calV_{x_0,r})[[\hbar]]=\HDens_\mr{const}(\calV_{x_0,r})\otimes \Hat S^\bt \calV_{x_0,r}^*[[\hbar]]$ is a half-density on $\calV_{x_0,r}$ which is a formal power series in $y$, a coordinate on $\calV_{x_0,r}$, as well as in $\hbar$; we put the index $g$ on operations $l_n$ to emphasize their dependence on gauge-fixing.

\subsubsection{Axioms}
\begin{enumerate}
\item Let $r\wavy[P] r'$ be an ordered pair of realizations with a fixed morphism $P=(i,p,K)$ between them (in the sense of Footnote \ref{footnote: retracts}), i.e. we have
\begin{equation} \label{Y split}
\calV_{x_0,r}=i(\calV_{x_0,r'})\oplus \underbrace{\Tilde{\calV}}_{\ker p}
\end{equation}
-- a splitting into a retract and an acyclic subcomplex w.r.t. $l_1$, which is orthogonal w.r.t. $\omega_{x_0,r}$ and induces $\omega_{x_0,r'}$ on the first term. Then the states for $r$ and $r'$ are related by a BV pushforward:
\begin{equation}\label{e:psi aggregation}
\psi_{x_0,r'}^{P\circ g}=P_* \psi_{x_0,r}^g=\int_{\mathcal{L}\subset \Tilde{\calV}} \psi_{x_0,r}^g
\end{equation}
where $\calL=\mathrm{im}\ K$ -- the gauge-fixing Lagrangian defined by the chain contraction.
\item The state satisfies the quantum master equation
$\Delta \psi_{x_0,r}^g=0$ where $\Delta$ is the canonical BV Laplacian on half-densities on $\calV_{x_0,r}$.
\item 
Changing the gauge-fixing data $g$ changes the state $\psi_{x_0,r}^g$ by a
$\Delta$-exact term (i.e. the corresponding effective action changes by a canonical BV transformation).
\item Allowing $x_0$ to vary, one has a hierarchy (parametrized by $r$) of graded vector bundles $\mathrm{Dens}^{\frac12}_\mr{formal}(\calV_{\bullet,r})$ over the EL moduli space $EL_M/Q=\mathcal{M}_M^{\mathrm{gh}=0}$. Note that one can indeed compare realizations $r$ over open subsets of $\mathcal{M}_M$ via homological perturbation theory. The bundle $\mathrm{Dens}^{\frac12}_\mr{formal}(\calV_{\bullet,r})$ is typically defined over $\mathcal{M}_M$ minus some singular strata (if $r$ is too small, so that the increase of cohomology of $l_1$ over the singular locus obstructs the extension). The bundle corresponding to the minimal realization $r_\mathrm{min}$, defined over the smooth locus of $\mathcal{M}_M^{\mathrm{gh}=0}$, is endowed with flat Grothendieck connection\footnote{
This connection corresponds to the possibility to translate an infinitesimal tangential shift along the base (the moduli space) into a fiber shift in the degree zero part of $\calV_{\bt,r_\mr{min}}$. The terminology is motivated by the terminology of formal geometry \cite{GK}, see also \cite{BCM}.
}
$\nabla_\mr{G}$, and the minimal realization of the state is a horizontal section:
$$\nabla_\mathrm{G}\psi_{\bullet,r_\mathrm{min}}^g=0.$$
We assume here that the gauge-fixing data $T_{x_0}[-1]\calF\wavy[g]\calV_{x_0,r_\mr{min}}[-1]$ is chosen  in a family over $\calM^\mr{gh=0}_M$.
\end{enumerate}

\begin{Rem}\label{rem: nabla_G}
The connection $\nabla_\mr{G}$ is constructed as follows. For $[x_0]$ a smooth point of $\calM_M$, the restriction of the map $\Psi$ of Remark \ref{rem: minimal backgrounds} to degree zero residual fields yields the formal exponential map $\Psi^0(x_0,\bt)\colon \calV_{x_0,r_\mr{min}}^0\to \calM_M^\mr{gh=0}$. We define
\begin{equation}\label{e:nabla_G}
\begin{array}{ccccc}
\nabla_\mr{G} & \colon & T_{x_0}\calM^\mr{gh=0}_M & \to & \mathfrak{X}_\mr{formal}(\calV_{x_0,r_\mr{min}}^0) \\
&& v & \to & \mathbf{T}\big(\underbrace{a}_{\in\calV^0_{x_0,r_\mr{min}}}\mapsto \underbrace{-(d_a\Psi^0(x_0,\bt))^{-1}v}_{\in T_a \calV^0_{x_0,r_\mr{min}}}\big)
\end{array}
\end{equation}
We understand formal vector fields on $\calV^0_{x_0,r_\mr{min}}$ as endomorphisms of $\HDens_\mr{formal}(\calV^0_{x_0,r_\mr{min}})$ and extending trivially to residual fields of nonzero degree, as endomorphisms of $\HDens_\mr{formal}(\calV_{x_0,r_\mr{min}})$; $\mathbf{T}$ stands for converting an actual vector field (defined in a neighborhood of the origin) on $\calV^0_{x_0,r_\mr{min}}$ to a formal vector field, via taking $\infty$-jet in $a$ at the origin. Thus (\ref{e:nabla_G}) does indeed define $\nabla_\mr{G}$ as an Ehresmann connection on the bundle $\HDens_\mr{formal}(\calV_{\bt,r_\mr{min}})$ over $\calM^\mr{gh=0}_M$.
\end{Rem}

\subsubsection{Number-valued partition function}
By Remark \ref{rem: minimal backgrounds}, the minimal realization of the state $\psi_{x_0,r_\mathrm{min}}^g\in\HDens_\mr{formal}(H^\bullet_{l_1}[1])=\HDens_\mr{formal}(T_{[x_0]}\mathcal{M}_M)$ defines
a half-density on the EL moduli space. One can define the {\sf number-valued partition function} of the theory as a BV integral over a Lagrangian submanifold in the EL moduli space (assuming that it converges):
\begin{equation}\label{Z number 1}
Z_M:=\int_{\calL\subset \mathcal{M}_M} \left.\psi^g_{\bullet,r_\mathrm{min}}\right|_0\qquad \in \mathbb{C}.	
\end{equation}
Here $\left.\psi^g_{\bullet,r_\mathrm{min}}\right|_0$ refers to putting degree zero residual fields in $\psi^g_{\bt,r_\mr{min}}$ to zero.

A special case of this construction is as follows. Assume that the body of the moduli space $\mathcal{M}_M^{\mathrm{gh}=0}$ contains an open dense subset $\Tilde{\mathcal{M}}_M$ such that, for any $[x_0]\in \Tilde{\mathcal{M}}_M$, one has $H^i_{l_1}=0$ for $i\neq 1,2$ (note that, by Poincar\'e duality/cyclicity, the vector spaces $H^1_{l_1}$ and $H^2_{l_1}$ are mutually dual). Then $\mathcal{M}_M$ has an open dense subset of the form $T^*[-1]\Tilde{\mathcal{M}}_M$. The minimal realization of the state $\psi^g_{x_0,r_\mathrm{min}}\in \mathrm{Dens}^{\frac12}_\mr{formal}(H^\bullet_{l_1}[1])_0=
\mathrm{Dens}_\mr{formal}(T_{[x_0]}\mathcal{M}_M^{\mathrm{gh}=0})$ defines a (fiberwise, formal in fiber direction) density on the tangent bundle of the moduli space $\mathcal{M}_M^{\mathrm{gh}=0}$ and thus its restriction to the zero-section can be integrated. The integral
\begin{equation}\label{Z number}
Z_M=\int_{\Tilde{\mathcal{M}}_M}\left.\psi^g_{\bullet,r_\mathrm{min}}\right|_0\qquad \in \mathbb{C},
\end{equation}
if it converges, is a special gauge-fixing for the BV integral (\ref{Z number 1}), with some singular strata of the moduli space removed,
corresponding to the Lagrangian submanifold  $\Tilde{\mathcal{M}}_M\subset T^*[-1]\Tilde{\mathcal{M}}_M$. 

\subsection{Example: non-abelian $BF$ theory on polygons twisted by a background connection}\label{appF: sec 1D}
\subsubsection{Model on a circle}\label{appF: circle}
Consider the 1-dimensional non-abelian $BF$ theory on a circle (cf. Example \ref{exa-AKSZ}).
We view the circle as being parametrized either by $t\in \mathbb{R}$ defined modulo $1$, or by $t\in [0,1]$ with the points $t=0$ and $t=1$ identified.
Additionally, we will assume that the Lie algebra of coefficients $\g$ is equipped with an invariant non-degenerate inner product $\langle,\rangle$, so that $\g^*$ can be identified with $\g$.

The space of fields of the model is $\calF=\Omega^\bullet(S^1,\g)[1]\oplus \Omega^\bt(S^1,\g)[-1]$ and the action is
\begin{equation}\label{e:S for BF on a circle}
\calS(\sfA,\sfB)=\oint_{S^1} \langle \sfB, d\sfA+\frac12 [\sfA,\sfA]\rangle
= \oint_{S^1} \langle \sfB^{(0)}, \dd \sfA^{(0)}+[\sfA^{(1)},\sfA^{(0)}] \rangle +
\langle \sfB^{(1)},\frac12 [\sfA^{(0)},\sfA^{(0)}]\rangle.
\end{equation}
Here on the r.h.s. we expressed the action in terms of homogeneous components of fields, $\sfA=\sfA^{(0)}+\sfA^{(1)}$, $\sfB=\sfB^{(0)}+\sfB^{(1)}$ where the upper index is the de Rham degree of the component; the internal degrees (ghost numbers) are
\begin{equation}\label{appF: gh}
|\sfA^{(0)}|=1,\quad |\sfA^{(1)}|=0,\quad |\sfB^{(0)}|=-1,\quad |\sfB^{(1)}|=-2.
\end{equation}

Note that the only classical (i.e. degree zero) field is $\sfA^{(1)}$ -- the connection 1-form on the circle. The classical action (i.e. $\calS$ restricted to degree zero fields) is identically zero. However, there is a gauge symmetry generated by the BV action, $\sfA^{(1)}\mapsto h \sfA^{(1)}h^{-1}+h\dd h^{-1}$ with $h:S^1\to G$. Here $G$ is the simply-connected Lie group integrating $\g$.
Thus the ghost number zero part of the Euler-Lagrange moduli space is
$$ \calM^\mathrm{gh=0}= \frac{\{\sfA^{(1)}\in \Omega^{1}(S^1,\g)\}}{\sfA^{(1)}\sim h\sfA^{(1)}h^{-1}+h\dd h^{-1}\quad\quad \forall h\in C^\infty(S^1,G)} \simeq G/G $$
where $G/G$ stands for the 
stratified manifold of conjugacy classes in $G$, arising as holonomy $U=\calP\exp \int_0^1 A^{(1)}$ of the connection defined by $\sfA^{(1)}$ around the circle modulo conjugation $U\mapsto h(0)\cdot U\cdot h(0)^{-1}$ by a group element (the value $h(0)$ of the generator of the gauge transformation at the base point on the circle).

Fix a background flat connection $A_0\in \Omega^1(S^1,\g)$.
The formal exponential map is 
$$\begin{array}{ccccc}
\phi(A_0,-)\colon &\, & T_{A_0}\calF= \Omega^\bt(S^1,\g)[1]\oplus \Omega^\bt(S^1,\g)[-1] & \to & \calF= \Omega^\bt(S^1,\g)[1]\oplus \Omega^\bt(S^1,\g)[-1] \\
& & (\Hat \sfA,\Hat \sfB) & \mapsto & (A_0+\Hat \sfA,\Hat \sfB)
\end{array}
$$
Here the pair $(\Hat \sfA, \Hat \sfB)=(\Hat \sfA^{(0)}+\Hat \sfA^{(1)},\Hat \sfB^{(0)}+\Hat \sfB^{(1)})$ is the formal variation of the field (which is allowed to have ghost number $\neq 0$ components in addition to a formal variation of the connection $\Hat \sfA^{(1)}$).
In the notations of Section \ref{sec: appF general setup}, $x_0=A_0$ and $\theta= (\Hat \sfA, \Hat \sfB)$.
We have
\begin{equation}\label{e:S for BF on circle twisted}
\calS(\phi(A_0;\Hat \sfA,\Hat \sfB))=\calS(A_0+\Hat \sfA,\Hat \sfB)=
\oint_{S^1} \langle \Hat \sfB, \dd_{A_0} \Hat \sfA+ \frac12 [\Hat \sfA, \Hat \sfA] \rangle
\end{equation}
where $\dd_{A_0}=\dd+[A_0,-]\colon\, \Omega^0(S^1,\g)\to \Omega^1(S^1,\g)$ is the de Rham operator twisted by the background connection. We denote $U=\calP\exp \int_0^1 A_0$ the holonomy around the circle.

It is convenient to introduce the complex of quasi-periodic forms (or, equivalently, forms on the universal covering of the circle equivariant w.r.t. covering transformations):
\begin{equation*}
\Omega^j_U=\{\alpha\in \Omega^j(\mathbb{R},\g)\;|\; \alpha(t+1)=U\alpha(t) U^{-1}\}
,\qquad j=0,1
\end{equation*}
with ordinary de Rham differential $\alpha\mapsto \dd\alpha$. As a complex, $(\Omega^\bt_U,\dd)$ is isomorphic to $(\Omega^\bt(S^1,\g),\dd_{A_0})$ with isomorphism given by
\begin{equation}\label{e:Phi}
\begin{array}{clccc}
\Phi\colon & \,  & (\Omega^\bt_U,\dd) & \stackrel{\sim}{\to} & (\Omega^\bt(S^1,\g),\dd_{A_0}) \\
& & \alpha(t) & \mapsto & U_t^{-1} \alpha(t) U_t
\end{array}
\end{equation}
where $U_t=\calP\exp\int_0^t A_0$ is the holonomy along the interval $[0,t]$. Note that $\Phi$ sends quasi-periodic forms to strictly periodic.

Denoting by $\Hat \sfa=\Phi^{-1}\Hat \sfA$, $\Hat \sfb=\Phi^{-1}\Hat \sfB$ the reparametrized fields, we can write (\ref{e:S for BF on circle twisted}) as
$$\calS(A_0+\Phi\Hat\sfa,\Phi\Hat\sfb)=\oint_{S^1} \langle \Hat\sfb, \dd\Hat\sfa+\frac12 [\Hat\sfa,  \Hat\sfa]\rangle$$
which looks exactly like the original non-twisted action (\ref{e:S for BF on a circle}) but is defined on quasi-periodic forms $(\Hat\sfa,\Hat\sfb)\in \Omega^\bt_U[1]\oplus \Omega^\bt_U[-1]$.

\subsubsection{Polygon realizations}
Now let us introduce the realization of the theory associated to equipping the circle with cell decomposition with $N\geq 1$ $0$-cells (vertices) and $N$ $1$-cells (edges), thus realizing the circle as an $N$-gon. We denote this realization $r_N$; we also denote this cell decomposition of $S^1$ by $T_N$. Next, we introduce the complex of quasi-periodic cell cochains on $T_N$ (or, equivalently, cochains on the covering cell decomposition $\Tilde T_N$ of $\mathbb{R}$, equivariant w.r.t. covering transformations):
$$C^j_U(T_N)=\{\alpha\in C^j(\Tilde T_N,\g)\; |\; \tau^* \alpha = U\alpha U^{-1}\},\qquad j=0,1$$
where $\tau: \Tilde T_N\to \Tilde T_N$ is the covering transformation corresponding to going around $S^1$ once in the direction of orientation. We equip $C^\bt_U(T_N)$ with a coboundary operator $d$ induced from the standard cellular coboundary operator on $C^\bt(\Tilde T_N,\g)$ (acting trivially in $\g$ coefficients).

As a graded vector space (but not as a complex) $C^\bt_U(T_N)$ is isomorphic to $C^\bt(T_N,\g)$. We introduce the cellular bases $\{e_k\}$, $\{e_{k,k+1}\}$ in $C^0(T_N)$ and $C^1(T_N)$, respectively, with $k=0,1,\ldots,N-1$. The coboundary operator of $C^\bt_U(T_N)$ then operates as
$$x_0e_0+\cdots +x_{N-1}e_{N-1}\mapsto (x_1-x_0)e_{01}+\cdots+(x_{N-1}-x_{N-2})e_{N-2,N-1}+(x_N-x_{N-1})e_{N-1,N}$$
where $x_0,\ldots,x_{N-1}\in\g$ are coefficients in the Lie algebra and $x_N\colon = U x_0 U^{-1}$. For $x\in\g$ and $k\in\mathbb{Z}$, we identify the 0-cochain $xe_k$ with the element $\sum_{p=-\infty}^\infty U^p x U^{-p} \Tilde e_{k-pN}\in C^0_U(T_N)$ and likewise the 1-cochain $xe_{k,k+1}$ with the element $\sum_{p=-\infty}^\infty U^p x U^{-p} \Tilde e_{k-pN,k+1-pN}\in C^1_U(T_N)$. Here $\Tilde e_l$, $\Tilde e_{l,l+1}$ with $l\in \mathbb{Z}$ stand for the cellular bases in 0- and 1-cochains of $\Tilde T_N$.

In complete analogy with the discussion above, we introduce the dual cell decomposition of the circle $T_N^\vee$ and the corresponding complex of quasi-periodic cochains $C^\bullet_U(T_N^\vee)$. We denote the cellular bases in  cochains of $T_N^\vee$ by $\{e_k^\vee\}$, $\{e_{k-1,k}^\vee\}$, with $k=0,\ldots,N-1$. We think of cells of $\Tilde T_N$ as being slightly displaced in the direction of orientation w.r.t. the corresponding cells of $T_N$. The intersection pairing is: $\langle e_k,e^\vee_{l-1,l} \rangle=\delta_{k,l}$, $\langle e_{k,k+1},e^\vee_l\rangle=\delta_{k,l}$.

We define the space of residual fields in realization $r_N$ to be
$$\calV_{A_0,r_N}\colon= C^\bt_U(T_N)[1]\oplus C^\bt_U(T_N^\vee)[-1]$$
parametrized by $\sfa=\sum_{k=0}^{N-1} \sfa_k e_k+\sfa_{k,k+1} e_{k,k+1}$ and
 $\sfb=\sum_{k=0}^{N-1} \sfb_k e_k^\vee+\sfb_{k-1,k} e_{k-1,k}^\vee$ where all the coefficients take values in $\g$ and have ghost numbers
$|\sfa_k|=1$, $|\sfa_{k,k+1}|=0$, $|\sfb_k|=-1$, $|\sfb_{k-1,k}|=-2$. The odd-symplectic form on $\calV_{A_0,r_N}$ comes from the intersection pairing: $\omega_{A_0,r_N}=\langle \delta \sfb, \delta \sfa \rangle=\sum_{k=0}^{N-1} \langle \delta\sfb_k ,\delta\sfa_{k,k+1} \rangle+\langle \delta\sfb_{k-1,k},\delta\sfa_k\rangle$.

Assume that $k$-th vertex of the polygon is geometrically realized as the point $t=t_k$ on the circle with $t_0=0<t_1<\cdots<t_{N-1}<t_N=1$ (e.g. one can choose $t_k=k/N$). We construct a retraction\footnote{This is a retraction without any compatibility with inner product.
} $\Omega^\bt_U\wavy[(i,p,K)] C^\bt_U(T_N)$ with
\begin{multline}\label{e:polygon i}
i\colon \sum_{k=0}^{N-1}x_k e_k+x_{k,k+1}e_{k,k+1} \mapsto\\
\mapsto \sum_{k=0}^{N-1} \left(\frac{t_{k+1}-t}{t_{k+1}-t_k}\; x_k+ \frac{t-t_k}{t_{k+1}-t_k}\; x_{k+1}+\frac{\dd t}{t_{k+1}-t_k}x_{k,k+1}\right)\;(\Theta(t-t_k)-\Theta(t_{k+1}-t))
\end{multline}
\begin{equation}\label{e:polygon p}
p\colon f(t)+ g(t)\; \dd t \mapsto \sum_{k=0}^{N-1} f\left(t_k\right) e_k+ \left(\int_{t_k}^{t_{k+1}} g(t')\dd t'\right) e_{k,k+1},
\end{equation}
\begin{equation}\label{e:polygon K}
K\colon g(t)\; \dd t \mapsto \sum_{k=0}^{N-1} \left( \int_{t_k}^t g(t')\dd t'-\frac{t-t_k}{t_{k+1}-t_k}\int_{t_k}^{t_{k+1}} g(t')\dd t' \right) \;(\Theta(t-t_k)-\Theta(t_{k+1}-t)).
\end{equation}
Here $\Theta(t)$ is the Heaviside step function.
This retraction defines in a unique way a retraction compatible with the inner product
\begin{equation} \label{e:doubled retraction}
\Omega^\bt_U\oplus \Omega^\bt_U[-2] \wavy[(i\oplus p^\vee,p\oplus i^\vee, K\oplus K^\vee)] C^\bt_U(T_N)\oplus C^\bt_U(T_N^\vee)[-2]
\end{equation}
where the superscript $\vee$ for maps $i,p,K$ stands for the adjoint map w.r.t.~the Poincar\'e pairing between the two copies of $\Omega^\bt_U$ and intersection pairing between cochains of $T_N$ and $T_N^\vee$. This, upon composition with the isomorphism (\ref{e:Phi}) gives the gauge-fixing data $T_{A_0}[-1]\calF\wavy[g] \calV_{A_0,r_N}[-1]$. The state for the realization $r_N$ is defined as the corresponding BV pushforward
$$\psi_{A_0,r_N}^g=\int_{\mathrm{im}K[1]\oplus \mathrm{im}K^\vee[-1]\subset \Tilde \calV}\EE^{\frac\ii\hbar\calS(A_0+\Phi i\sfa+\alpha,\Phi p^\vee \sfb+\beta)} (d\alpha)^{1/2}(d\beta)^{1/2}(d\sfa)^{1/2}(d\sfb)^{1/2}$$
with $(\sfa,\sfb)\in \calV_{A_0,r_N}$ the residual fields and $(\alpha,\beta)\in \Tilde\calV$ fluctuations. This integral can be computed exactly, following \cite{discrBF}, and yields
\begin{multline}\label{e:psi for BF on polygon}
\psi_{A_0,r_N}^g=\EE^{\frac\ii\hbar\left(\sum_{k=0}^{N-1}\left\langle \sfb_{k-1,k} , \frac12 [\sfa_k,\sfa_k]\right\rangle+\left\langle \sfb_k, F(\ad_{\sfa_{k,k+1}})\circ(\sfa_{k+1}-\sfa_k)+ [\sfa_{k,k+1},\frac{\sfa_k+\sfa_{k+1}}{2}] \right\rangle\right)}\cdot\\
\cdot \prod_{k=0}^{N-1}{\det}_\g G(\ad_{\sfa_{k,k+1}})\cdot
\xi_{r_N}\cdot
(d\sfa)^{1/2}(d\sfb)^{1/2}\qquad\qquad \in \HDens(\calV_{A_0,r_N})
\end{multline}
where we introduced the notation $F,G$ for the two functions
$$ F(x)=\frac{x}{2}\coth \frac{x}{2},\qquad G(x)=\frac{2}{x}\sinh\frac{x}{2}.$$
The factor $\xi_{r_N}=\left(\EE^{-\frac{\pi \ii}{2}}\hbar\right)^{N\dim\g}$ comes from the normalization of the integration measure, cf. (\ref{e:xi=exp}) and \cite{cell_ab_BF}.

The quantum master equation reads
$$\left(\sum_{k=0}^{N-1}\left\langle \frac{\de}{\de \sfa_k}, \frac{\de}{\de \sfb_{k-1,k}} \right\rangle + \left\langle \frac{\de}{\de \sfa_{k,k+1}}, \frac{\de}{\de \sfb_{k}} \right\rangle \right)\psi^g_{A_0,r_N}=0.$$
It can be checked by an explicit computation, cf. Section 5.5.1 of \cite{discrBF}.

For $N=1$, (\ref{e:psi for BF on polygon}) becomes
\begin{multline}\label{appF: psi for T_1}
\psi^g_{A_0,r_1}=\EE^{\frac\ii\hbar \left( \left\langle \sfb_{-1, 0} , \frac12 [\sfa_0,\sfa_0]  \right\rangle +  \left\langle \sfb_0 , F(\ad_{\sfa_{01}})\circ (\Ad_U-\id)\circ \sfa_0 + \frac12 \ad_{\sfa_{01}}\circ (\Ad_U+\id)\circ \sfa_0  \right\rangle  \right)}\cdot \\
\cdot {\det}_\g G(\ad_{\sfa_{01}})\cdot
\left(\EE^{-\frac{\pi \ii}{2}}\hbar\right)^{\dim\g}\cdot
(d\sfa_0)^{1/2}(d\sfa_{01})^{1/2}(d\sfb_0)^{1/2}(d\sfb_{-1,0})^{1/2}\qquad \in
\HDens(\underbrace{\g[1]\oplus \g\oplus\g[-1]\oplus\g[-2]}_{\calV_{A_0,r_1}}).
\end{multline}

One can define cellular aggregation morphisms $\calV_{A_0,r_N}\wavy[\mathrm{agg}_k^\varkappa] \calV_{A_0,r_{N-1}}$ corresponding to merging edges $[k,k+1]$ and $[k+1,k+2]$ together, for $k=0,1\ldots,N-1$. Here $\varkappa\in [0,1]$ is a parameter of the morphism. To define $\mathrm{agg}_k^\varkappa$, we start by introducing a retraction $C^\bt_U(T_N) \wavy[i_k^\varkappa,p_k^\varkappa,K_k^\varkappa] C^\bt_U(T_{N-1})$ (without compatibility with inner products):
\begin{multline*}
i_k^\varkappa\colon \sum_{l=0}^{N-2} x_l e_l+x_{l,l+1} e_{l,l+1} \mapsto
\left(\sum_{l=0}^{k-1} x_l e_l+x_{l,l+1} e_{l,l+1}\right)+ x_k e_k+ ((1-\varkappa)\cdot x_k+\varkappa\cdot x_{k+1}) e_{k+1} + \\
+ x_{k,k+1}(\varkappa\cdot e_{k,k+1} + (1-\varkappa)\cdot e_{k+1,k+2})+
\left(\sum_{l=k+1}^{N-2} x_{l} e_{l+1}+x_{l,l+1} e_{l+1,l+2}\right),
\end{multline*}
\begin{multline*}
p_k^\varkappa\colon \sum_{l=0}^{N-1} x_l e_l+x_{l,l+1} e_{l,l+1} \mapsto \\
\mapsto
\left(\sum_{l=0}^{k-1} x_l e_l+x_{l,l+1} e_{l,l+1}\right) +x_k e_k +(x_{k,k+1}+x_{k+1,k+2}) e_{k,k+1}+ \left(\sum_{l=k+2}^{N-1} (x_l e_{l-1}+x_{l,l+1} e_{l-1,l}\right),
\end{multline*}
$$
K_k^\varkappa\colon \sum_{l=0}^{N-1} x_{l,l+1} e_{l,l+1} \mapsto ((1-\varkappa)\cdot x_{k,k+1}-\varkappa \cdot x_{k+1,k+2})\cdot e_{k+1}.
$$
Next, we define the corresponding aggregation morphism between spaces of residual fields (now, a retraction compatible with the inner product) by the doubling construction as in  (\ref{e:doubled retraction}):
$$ \underbrace{C^\bt_U(T_N)\oplus C^\bt_U(T_N)[-2]}_{\calV_{A_0,r_N}[-1]}\quad \wavy[\mathrm{agg}_k^\varkappa\colon=\;(i_k^\varkappa\oplus p_k^{\varkappa\vee},p_k^\varkappa\oplus i_k^{\varkappa\vee},K_k^\varkappa\oplus K_k^{\varkappa\vee})]\quad \underbrace{C^\bt_U(T_{N-1})\oplus C^\bt_U(T_{N-1})[-2]}_{\calV_{A_0,r_{N-1}}[-1]}  $$

One has the automorphic property of the state (\ref{e:psi for BF on polygon}) with respect to aggregations:
\begin{equation}
(\mathrm{agg}_k^\varkappa)_* \psi_{A_0,r_N}^g=\psi_{A_0,r_{N-1}}^g,
\end{equation}
cf. (\ref{e:psi aggregation}), which can be checked by calculating the BV pushforward explicitly; the computation is analogous to the one in Section 3.2.2 of \cite{AM}. Note that the BV pushforward yields {\sf precisely} the state for the standard gauge-fixing (\ref{e:polygon i},\ref{e:polygon p},\ref{e:polygon K}), not up to a $\Delta$-exact term.\footnote{This corresponds to the observation that gauge-fixings $\mathrm{agg}_k^\varkappa\circ g_{r_{N}}$ and $g_{r_{N-1}}$ for the realization $r_{N-1}$ precisely coincide if we place the $(k+1)$-st vertex in $r_N$ at the point $t_{k+1}^{r_N}=(1-\varkappa)t_k^{r_N}+\varkappa t_{k+2}^{r_N}$ on $S^1$ and assign $t_l^{r_{N-1}}=t_l^{r_N}$ for $l=0,\ldots,k$ and $t_l^{r_{N-1}}=t_{l+1}^{r_N}$ for $l=k+1,\ldots,N-1$. Also note that the state cannot depend on the positions of vertices of the polygon, since the continuum theory is diffeomorphism-invariant. For clarity, here we indicated the realization explicitly.}

\subsubsection{Minimal realization}
The complex $C^\bt_U(T_1)$ is
$$0\to \g\xrightarrow{\Ad_U-\id} \g \to 0.$$
Its cohomology (which is the same as cohomology of $\Omega^\bt_U$ and $C^\bt_U(T_N)$ for any $N$) is
$$H^0_U=\g_U,\qquad H^1_U=\g/\g_U^\perp\simeq \g_U$$
where we denoted $\g_U\subset\g$ the subspace comprised of elements of the Lie algebra commuting with the holonomy $U$; $\g_U^\perp$ is the orthogonal complement of $\g_U$ in $\g$ w.r.t.~the inner product $\langle,\rangle$.  According to Remark \ref{rem: minimal backgrounds}, we set
$$\calV_{A_0,r_\mr{min}}=H^\bt_U[1]\oplus H^\bt_U[-1].$$
We denote the corresponding residual fields $\sfa^{(0)}, \sfa^{(1)}, \sfb^{(0)}, \sfb^{(1)}\in \g_U$
with upper index standing for the form (or cochain) degree, as in Section \ref{appF: circle}; ghost numbers are as in (\ref{appF: gh}).

We have a retraction $C^\bt_U(T_1)\wavy[(i_\mr{min},p_\mr{min},K_\mr{min})] H^\bt_U$ where $i_\mr{min},p_\mr{min}$ correspond to the inclusion of the first summand and the projection onto the first summand in the splitting $\g=\g_U\oplus \g_U^\perp$ in degrees 0 and 1. The chain homotopy $K_\mr{min}$ is $(\Ad_U-\id)^{-1}$ on $\g_U^\perp\subset C^1_U(T_1)$ and vanishes on $\g_U\subset C^1_U(T_1)$. By doubling, as in (\ref{e:doubled retraction}), we produce a gauge-fixing morphism
$\calV_{A_0,r_1}\wavy[P]\calV_{A_0,r_\mr{min}}$. 
The state in the minimal representation $\psi_{A_0,r_\mr{min}}^g$ with gauge-fixing $g=P\circ g^{r_1}$ can be computed from (\ref{appF: psi for T_1}) as a BV pushforward $\psi_{A_0,r_\mr{min}}^g=P_*\psi_{A_0,r_1}^{g_{r_1}}$. The result is:
\begin{multline}\label{appF: psi min}
\psi_{A_0,r_\mr{min}}^g=\EE^{\frac\ii\hbar\left(
\langle\sfb^{(1)},\frac12 [\sfa^{(0)},\sfa^{(0)}] \rangle+
\langle \sfb^{(0)},[\sfa^{(1)},\sfa^{(0)}] \rangle
\right)}
\cdot \\ \cdot {\det}_\g G(\ad_{\sfa^{(1)}})\cdot
{\det}_{\g_U^\perp}\left( F(\ad_{\sfa^{(1)}})\circ (\Ad_U-\id)+\frac12 \ad_{\sfa^{(1)}}\circ (\Ad_U+\id) \right)\cdot \\
\cdot\xi_{r_\mr{min}}\cdot
(d\sfa^{(0)})^{1/2}(d\sfa^{(1)})^{1/2}(d\sfb^{(0)})^{1/2}(d\sfb^{(1)})^{1/2}
\qquad \in \HDens(\underbrace{\g_U[1]\oplus\g_U\oplus\g_U[-1]\oplus \g_U[-2]}_{\calV_{A_0,r_\mr{min}}}).
\end{multline}
Here $\xi_{r_\mr{min}}=\left(\EE^{-\frac{\pi \ii}{2}}\hbar\right)^{\mathrm{rk}(G)}$ with $\mr{rk}(G)=\dim G/G$ the rank of the group $G$.

Note that there is an open dense subset $\Bar G\subset G$ consisting of elements $U\in G$ such that $\g_U$ is a maximal abelian (Cartan) subalgebra of $\g$. These are the group elements with the ``maximal'' conjugacy class; the set of these maximal conjugacy classes $\Bar G/G$ is the smooth locus of the moduli space $\calM^\mr{gh=0}=G/G$.

In the case when the holonomy of the background connection satisfies $U\in\Bar G$, the result (\ref{appF: psi min}) simplifies to
\begin{equation}\label{appF: psi min simplified}
\psi^g_{A_0,r_\mr{min}}={\det}_{\g_U^\perp}\left(\Ad_{U\cdot \exp(\sfa^{(1)})}-\id\right)\cdot \xi_{r_\mr{min}}\cdot
(d\sfa^{(0)})^{1/2}(d\sfa^{(1)})^{1/2}(d\sfb^{(0)})^{1/2}(d\sfb^{(1)})^{1/2}.
\end{equation}

Allowing the background connection $A_0$ to vary as long as the holonomy $U$ is in $\Bar G$, we view $\psi^g_{-,r_\mr{min}}$ as a section of the vector bundle
\begin{equation}\label{appF: bundle of 1/2dens}
\HDens(\calV_{-,r_\mr{min}})\to \Bar G
\end{equation}
(where the dash stands for the background connection and the bundle projection consists in taking the holonomy of the connection). Simultaneous conjugation of $U$ and the residual fields by group elements $h\in G$ induce an action of $G$ on the bundle (\ref{appF: bundle of 1/2dens}) by bundle automorphisms. The section $\psi^g_{-,r_\mr{min}}$ is equivariant w.r.t.~the $G$ action. Thus we can regard $\psi^g_{[-],r_\mr{min}}$ as a section of the bundle over the smooth locus of the moduli space
\begin{equation}\label{appF: bundle over G/G}
\HDens(\calV_{[-],r_\mr{min}})\to \Bar G/G,
\end{equation}
where $[-]$ stands for the gauge equivalence class of the background connection.

We can introduce a partial connection\footnote{``Partial'' means here that we only define covariant derivatives along vector fields tangent to a particular distribution on $\Bar G$.} on the bundle (\ref{appF: bundle of 1/2dens}):
$$ \begin{array}{rcc}
T_U\Bar G\simeq \g\supset \g_U &\to & \mathfrak{X}(\HDens(\calV_{A_0,r_\mr{min}})) \\
v & \mapsto & -\left\langle v,\frac{\de}{\de \sfa^{(1)}} \right\rangle
\end{array} $$
Passing to the quotient by $G$ action in (\ref{appF: bundle of 1/2dens}), we obtain the Grothendieck connection on the bundle over the moduli space (\ref{appF: bundle over G/G}). Explicitly, we can write $$\nabla_\mr{G}=\dd-\left\langle \frac{\de}{\de \sfa^{(1)}} , U^{-1}\dd U \right\rangle.$$
Since the state (\ref{appF: psi min simplified}) manifestly only depends on the combination $U\cdot\exp(\sfa^{(1)})$, it satisfies the horizontality condition
$$\nabla_\mr{G}\psi^g_{[-],r_\mr{min}}=0.$$

The formal exponential map $\Psi^0$ of Remark \ref{rem: nabla_G} sends $\sfa^{(1)}\mapsto U\cdot \exp(\sfa^{(1)})$.

\begin{Rem} Note that the full Euler-Lagrange moduli space of the model $\calM$, as opposed to $\calM^\mr{gh=0}$, contains formal directions spanned by $\sfa^{(0)}, \sfb^{(1)}$ on which the state (\ref{appF: psi min simplified}) does not depend. Therefore there is no choice of gauge-fixing Lagrangian $\calL\subset \calM$ which would produce a convergent nonzero integral (\ref{Z number 1}) for the number-valued partition function.
\end{Rem}

\begin{Rem} The one-dimensional model presented here admits a meaningful generalization to graphs. Under certain assumptions on a graph, one can define the number-valued partition function. We plan to present this generalization in a future paper.
\end{Rem}

\subsection{Example: partition function of 2D non-abelian $BF$ theory on a closed surface}
Consider non-abelian $BF$ theory on a closed surface $\Sigma$ of genus $\gamma\geq 2$. We fix a compact-simply connected Lie group $G$ with Lie algebra $\g$. As in Section \ref{appF: sec 1D}, we identify $\g^*$ with $\g$ using a non-degenerate invariant inner product $\langle,\rangle$ on $\g$.

We have fields
$(\sfA=\sum_{k=0}^2 \sfA^{(k)},\sfB=\sum_{k=0}^2 \sfB^{(k)})\in\calF=\Omega^\bt(\Sigma,\g)[1]\oplus \Omega^\bt(\Sigma,\g)$ where the upper index stands for the form degree. Ghost numbers are $|\sfA^{(k)}|=1-k$, $\sfB^{(k)}=-k$. The moduli space of classical solutions of equations of motion is
$$\calM^\mr{gh=0}=\frac{\{(\sfA^{(1)},\sfB^{(0)})\in \Omega^1(\Sigma,\g)\oplus \Omega^0(\Sigma,\g)\;|\; \dd\sfA^{(1)}+\frac12 [\sfA^{(1)},\sfA^{(1)}]=0,\; \dd\sfB^{(0)}+[\sfA^{(1)},\sfB^{(0)}]=0 \}}{(\sfA^{(1)},\sfB^{(0)})\sim (h\sfA^{(1)}h^{-1}+h\dd h^{-1},h\sfB^{(0)}h^{-1})\quad \forall h:\Sigma\to G }$$
It projects onto the moduli space of flat $G$-connections on $\Sigma$, $M_{\Sigma,G}=\Hom(\pi_1(\Sigma),G)/G$ (by taking holonomy of $\sfA^{(1)}$), with fiber $H^0_{\dd_{\sfA^{(1)}}}$ where $\dd_{\sfA^{(1)}}\colon \Omega^\bt(\Sigma,\g)\to \Omega^{\bt+1}(\Sigma,\g)$ is the de Rham operator twisted by the flat connection $\sfA^{(1)}$. Note also that $H^1_{\dd_{\sfA^{(1)}}}\simeq T_{[\sfA^{(1)}]}M_{\Sigma,G}$ -- the tangent space to moduli space of flat connections, here $[\sfA^{(1)}]$ is the class of the connection modulo gauge transformations.

Denote $M_{\Sigma,G}^\mr{irred}\subset M_{\Sigma,G}$ the moduli space of {\sf irreducible} flat connections (i.e. those with $H^0_{\dd_{\sfA^{(1)}}}=0$). For a surface $\Sigma$ of genus $\geq 2$ (which we requested precisely for this reason), $M_{\Sigma,G}^\mr{irred}$ is open dense in $M_{\Sigma,G}$. Hence, $M_{\Sigma,G}^\mr{irred}\subset \calM^\mr{gh=0}$ is an open dense subset; the inclusion maps $[\sfA^{(1)}]\mapsto [(\sfA^{(1)},0)]$. Moreover, the odd cotangent bundle $T^*[-1]M_{\Sigma,G}^\mr{irred}$ is open dense in the full (i.e. not just the ghost number zero part) Euler-Lagrange moduli space $\calM$.

Fix a classical solution of equations of motion of the form $x_0=(A_0,0)$ with $A_0$ an irreducible flat connection. Then $H^\bt_{\dd_{A_0}}$ is concentrated in degree $1$, thus the minimal realization for the space of residual fields on the background defined by $(A_0,0)$ is $\calV_{(A_0,0),r_\mr{min}}=H^1_{d_{A_0}}\oplus H^1_{\dd_{A_0}}[-1]\simeq T^*[-1] T_{[A_0]}M_{\Sigma,G}$. We denote an element of $\calV_{(A_0,0),r_\mr{min}}$ by $(\sfa^{(1)},\sfb^{(1)})$.
The state in the minimal realization is given by
\begin{equation}\label{appF: 2D psi PI}
\psi_{(A_0,0),r_\mr{min}}(\sfa^{(1)},\sfb^{(1)})=\int_\calL \EE^{\frac\ii\hbar \calS(A_0+\sfa^{(1)}+\alpha,\sfb^{(1)}+\beta)} \left.\left((d\alpha)^{1/2}(d\beta)^{1/2}\right)\right|_\calL (d\sfa^{(1)})^{1/2} (d\sfb^{(1)})^{1/2}
\end{equation}
where $\alpha,\beta$ are fluctuations over which we integrate, restricting to  the gauge-fixing Lagrangian $\calL$. The action in the exponential expands as
$$\calS(A_0+\sfa^{(1)}+\alpha,\sfb^{(1)}+\beta)=\int_\Sigma \langle \beta,\dd_{A_0}\alpha+\frac12 [\alpha,\alpha]+[\sfa^{(1)},\alpha]+\frac12 [\sfa^{(1)},\sfa^{(1)}] \rangle.$$
The path integral on the r.h.s. of (\ref{appF: 2D psi PI}) can be calculated perturbatively and yields
$$\psi_{(A_0,0),r_\mr{min}}(\sfa^{(1)},\sfb^{(1)})= T_\Sigma\; \EE^{W(\sfa^{(1)})}$$
where $T_\Sigma$ is as in (\ref{e:T via tau}) adjusted for the nontrivial local system defined by $[A_0]$ and $W(\sfa^{(1)})$ is the sum of 1-loop graphs (a collection of binary trees with leaves decorated by $\sfa^{(1)}$ with roots attached to the cycle); $W$ is a function on $H^1_{\dd_{A_0}}$ with zero of order at least 2 at the origin. One can calculate $T_\Sigma$ explicitly:
$$T_\Sigma = \underbrace{(2\pi\hbar)^{n/2}(\EE^{-\frac{\pi\ii}{2}}\hbar)^{3n/2}}_{\xi} \underbrace{\frac{\omega^{\wedge n}}{n!}}_{\tau(\Sigma,[A_0])=d\sfa^{(1)}\sim (d\sfa^{(1)})^{1/2}(d\sfb^{(1)})^{1/2}}\qquad \in \mr{Dens} (T_{[A_0]}M_{\Sigma,G})\simeq \HDens (T_{[(A_0,0)]}\calM)$$
where $n=(\gamma-1)\dim G=\frac12 \dim M_{\Sigma,G}$ and $\omega$ is the Atiyah-Bott symplectic structure on $M_{\Sigma,G}$; factor $\xi$ is as in (\ref{e:xi}). The symplectic volume form $\omega^{\wedge n}/n!$ coincides with the Reidemeister torsion of the surface equipped with the non-acyclic local system defined by the flat connection $A_0$ (in adjoint representation), cf. e.g. \cite{Witten}.

Now we can define the number-valued partition function of the theory as in (\ref{Z number}). Since it only depends on the value of $\psi_{(A_0,0),r_\mr{min}}$ at the origin of the tangent space to the moduli space, we do not need to know the function $W$ to define $Z_\Sigma$. Explicitly, we obtain that the partition function is, up to the factor $\xi$, the symplectic volume of the moduli space of flat connections on $\Sigma$  \cite{Witten}:
$$Z_\Sigma=\xi\int_{M^\mr{irred}_{\Sigma,G}}\frac{\omega^{\wedge n}}{n!}=
\xi\cdot\# z(G)\cdot \mr{Vol}(G)^{2\gamma-2}\sum_R \frac{1}{(\dim R)^{2\gamma-2}}.$$
Here $\# z(G)$ is the number of elements in the center of $G$ and the sum in l.h.s. runs over irreducible representations $R$ of $G$.

\end{document}